\documentclass[11pt,a4paper,twoside,openany,final]{book}
\usepackage{titlesec,blindtext}
\usepackage[dvipsnames]{xcolor}
\definecolor{gray75}{gray}{0.4}
\newcommand{\hsp}{\hspace{20pt}}
\titleformat{\chapter}[hang]{\Large\bfseries}{\thechapter\hsp\textcolor{gray75}{|}\hsp}{0pt}{\Large\bfseries}
\titleformat{name=\chapter,numberless}[hang]{\Large\bfseries}{\textcolor{gray75}{|}\hsp}{0pt}{\Large\bfseries}
\titlespacing{\chapter}{0pt}{0pt}{*1}
\renewcommand{\contentsname}{\sffamily\textcolor{red}{Another title
    for the table of contents}}

\usepackage{fullpage}
\usepackage{etex}
\usepackage[utf8]{inputenc}
\usepackage[portuges,english]{babel}
\usepackage[T1]{fontenc}
\usepackage{lmodern}
\usepackage[square,numbers,sort&compress]{natbib}
\usepackage[fleqn]{amsmath}
\usepackage{amssymb,amsthm}
\usepackage[printonlyused,smaller,withpage]{acronym}
\usepackage{booktabs}
\usepackage{hepunits}
\usepackage{units}

\usepackage{multirow}
\usepackage{graphicx}

\usepackage{epigraph}
\epigraphsize{\small\itshape}
\usepackage{csquotes}

\usepackage{slashed}
\usepackage{caption}
\usepackage{subcaption}
\usepackage{setspace}
\usepackage{notoccite}
\usepackage{listings}
\usepackage{enumerate}

\definecolor{webgreen}{rgb}{0,.5,0}
\definecolor{webbrown}{rgb}{.6,0,0}
\usepackage[hyperfootnotes=false,pdfpagelabels]{hyperref}  
\hypersetup{        colorlinks=true, linktocpage=true, pdfstartpage=3, pdfstartview=FitV,            breaklinks=true, pdfpagemode=UseNone, pageanchor=true, pdfpagemode=UseOutlines,    plainpages=false, bookmarksnumbered, bookmarksopen=true, bookmarksopenlevel=1,    hypertexnames=true, pdfhighlight=/O,   urlcolor=gray75, linkcolor=gray75, citecolor=gray75,
            pdftitle={Scalar Fields in Particle Physics},     pdfauthor={Leonardo Antunes Pedro},            }

\usepackage{mciteplus}

\newcommand{\Br}{\text{Br}}
\newcommand{\eq}[1]{(\ref{#1})}
\newcommand{\abs}[1]{\left|#1\right|}
\newcommand{\re}{\text{Re}}
\newcommand{\im}{\text{Im}}

\graphicspath{{./Images/}}

\newcommand\cleartooddpage{\clearpage
  \ifodd\value{page}\else\null\thispagestyle{empty}\clearpage\fi}

\renewcommand{\textflush}{flushepinormal}

\theoremstyle{plain}\newtheorem{thm}{Theorem}[section]
\newtheorem{lem}[thm]{Lemma}
\newtheorem{prop}[thm]{Proposition}
\newtheorem{rmk}[thm]{Note}
\newtheorem*{rmk*}{Note}
\newtheorem*{cor}{Corollary}
\theoremstyle{definition}
\newtheorem{defn}[thm]{Definition}
\newtheorem*{defn*}{Definition}

\theoremstyle{remark}
\newtheorem*{note}{Note}

\makeatletter
\renewcommand{\@epitext}[1]{
\itshape \begin{minipage}{\epigraphwidth}\begin{\textflush} #1
\end{\textflush}\end{minipage}\vspace{1ex}}

\AtBeginDocument{\renewcommand\tableofcontents{    \if@twocolumn
      \@restonecoltrue\onecolumn
    \else
      \@restonecolfalse
    \fi
    \chapter*{\contentsname
      \@mkboth{           \MakeUppercase\contentsname}{\MakeUppercase\contentsname}}
    \@starttoc{toc}    \if@restonecol\twocolumn\fi
    }
\renewcommand\listoffigures{    \if@twocolumn
      \@restonecoltrue\onecolumn
    \else
      \@restonecolfalse
    \fi
    \section*{\listfigurename}      \@mkboth{\MakeUppercase\listfigurename}              {\MakeUppercase\listfigurename}    \@starttoc{lof}    \if@restonecol\twocolumn\fi
    }
\renewcommand\listoftables{    \if@twocolumn
      \@restonecoltrue\onecolumn
    \else
      \@restonecolfalse
    \fi
    \section*{\listtablename}      \@mkboth{\MakeUppercase\listtablename}              {\MakeUppercase\listtablename}    \@starttoc{lot}    \if@restonecol\twocolumn\fi
    }

}
\makeatother

\begin{document}

\pagestyle{plain}

\frontmatter
\clearpage{}
\thispagestyle {empty}

\hspace{-1.6cm}\includegraphics[bb=1.8600000000000001cm 0cm 10cm 3.1600000000000001cm,scale=0.3]{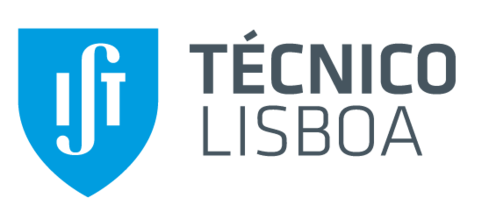}
~

~\vspace{0.8cm}

\begin{center}
\vspace{-2.4cm}\textbf{\Large \uppercase{Universidade de Lisboa}}
\par\end{center}{\Large \par}

\begin{center}
\textbf{\Large \uppercase{Instituto Superior T\'{e}cnico}}
\par\end{center}{\Large \par}

~

~

~

~

~

~

\begin{center}
\textbf{\huge Scalar Fields in Particle Physics}
\par\end{center}{\Large \par}

\textbf{\large ~}{\large \par}

\begin{center}
{\large \vspace{-0.5cm}}\textbf{\large Leonardo Antunes Pedro}
\par\end{center}{\large \par}

~

~

\begin{flushleft}
{\large \vspace{-0.5cm}Supervisor: Doctor Gustavo da Fonseca Castelo Branco}
\par\end{flushleft}{\large \par}

\begin{flushleft}
{\large \vspace{-0.5cm}~}
\par\end{flushleft}{\large \par}

{\large ~}{\large \par}

~

\begin{center}
{\large \vspace{-0.5cm}Thesis approved in public session to obtain
the PhD Degree in}
\par\end{center}{\large \par}

\begin{center}
{\large \vspace{-0.3cm}}\textbf{\large Physics}
\par\end{center}{\large \par}

\begin{center}
{\large Jury final classification: Pass With Merit}{\Large }
\par\end{center}{\Large \par}

\begin{center}
{\large ~}
\par\end{center}{\large \par}

\begin{center}
{\large ~}
\par\end{center}{\large \par}

\begin{center}
{\large \vspace{-0.5cm}}\textbf{\large Jury}
\par\end{center}{\large \par}

\begin{flushleft}
{\large \vspace{-0.2cm}Chairperson: Chairman of the IST Scientific
Board}
\par\end{flushleft}{\large \par}

\begin{flushleft}
{\large \vspace{-0.3cm}Members of the Committee:}
\par\end{flushleft}{\large \par}

\begin{flushleft}
{\large \vspace{-0.3cm}\hspace{0.8cm}Doctor Francisco Jos\'e Botella Olcina}
\par\end{flushleft}{\large \par}

\begin{flushleft}
{\large \vspace{-0.3cm}\hspace{0.8cm}Doctor Gustavo da Fonseca Castelo Branco}
\par\end{flushleft}{\large \par}

\begin{flushleft}
{\large \vspace{-0.3cm}\hspace{0.8cm}Doctor Jorge Manuel Rodrigues
Crispim Rom\~ao}
\par\end{flushleft}{\large \par}

\begin{flushleft}
{\large \vspace{-0.3cm}\hspace{0.8cm}Doctor Maria Margarida Nesbitt Rebelo da Silva}
\par\end{flushleft}{\large \par}

\begin{flushleft}
{\large \vspace{-0.3cm}\hspace{0.8cm}Doctor Paulo Andr\'e de Paiva Parada}
\par\end{flushleft}{\large \par}

\begin{flushleft}
{\large \vspace{-0.3cm}\hspace{0.8cm}Doctor David Emanuel da Costa}
\par\end{flushleft}{\large \par}

~

\begin{center}
\vspace{0.3cm}\textbf{\large 2015}\vspace{-2.0cm}
\par\end{center}

\clearpage{}
\cleartooddpage

\clearpage{}
\thispagestyle {empty}

\hspace{-1.6cm}\includegraphics[bb=1.8600000000000001cm 0cm 10cm 3.1600000000000001cm,scale=0.3]{Images/logo}
~

~\vspace{0.8cm}

\begin{center}
\vspace{-2.4cm}\textbf{\Large \uppercase{Universidade de Lisboa}}
\par\end{center}{\Large \par}

\begin{center}
\textbf{\Large \uppercase{Instituto Superior T\'{e}cnico}}
\par\end{center}{\Large \par}

~

~

~

~

\begin{center}
\textbf{\huge Scalar Fields in Particle Physics}
\par\end{center}{\Large \par}

\textbf{\large ~}{\large \par}

\begin{center}
{\large \vspace{-0.5cm}}\textbf{\large Leonardo Antunes Pedro}
\par\end{center}{\large \par}

~

~

\begin{flushleft}
{\large \vspace{-0.5cm}Supervisor: Doctor Gustavo da Fonseca Castelo Branco}
\par\end{flushleft}{\large \par}

\begin{center}
{\large \vspace{0.0cm}Thesis approved in public session to obtain
the PhD Degree in}
\par\end{center}{\large \par}

\begin{center}
{\large \vspace{-0.3cm}}\textbf{\large Physics}
\par\end{center}{\large \par}

\begin{center}
{\large Jury final classification: Pass With Merit}{\Large }
\par\end{center}{\Large \par}

\begin{center}
{\large \vspace{0.0cm}}\textbf{\large Jury}
\par\end{center}{\large \par}

\begin{flushleft}
{\large \vspace{-0.2cm}Chairperson: Chairman of the IST Scientific
Board}
\par\end{flushleft}{\large \par}

\begin{flushleft}
{\large \vspace{-0.3cm}Members of the Committee:}
\par\end{flushleft}{\large \par}

\begin{flushleft}
{\large \vspace{-0.3cm}\hspace{0.8cm}Doctor Francisco Jos\'e Botella Olcina, Full Professor,}
\par\end{flushleft}{\large \par}

\begin{flushleft}
{\large \vspace{-0.3cm}\hspace{0.8cm}Instituto de F\'isica Corpuscular, Universitat
de Val\`encia, Spain}
\par\end{flushleft}{\large \par}

\begin{flushleft}
{\large \vspace{-0.1cm}\hspace{0.8cm}Doctor Gustavo da Fonseca Castelo Branco,
Full Professor,}
\par\end{flushleft}{\large \par}

\begin{flushleft}
{\large \vspace{-0.3cm}\hspace{0.8cm}Instituto Superior T\'ecnico,
Universidade de Lisboa}
\par\end{flushleft}{\large \par}

\begin{flushleft}
{\large \vspace{-0.1cm}\hspace{0.8cm}Doctor Jorge Manuel Rodrigues
Crispim Rom\~ao, Full Professor,}
\par\end{flushleft}{\large \par}

\begin{flushleft}
{\large \vspace{-0.3cm}\hspace{0.8cm}Instituto Superior T\'ecnico,
Universidade de Lisboa}
\par\end{flushleft}{\large \par}

\begin{flushleft}
{\large \vspace{-0.1cm}\hspace{0.8cm}Doctor Maria Margarida Nesbitt Rebelo da Silva, Principal Researcher,}
\par\end{flushleft}{\large \par}

\begin{flushleft}
{\large \vspace{-0.3cm}\hspace{0.8cm}Instituto Superior T\'ecnico,
Universidade de Lisboa}
\par\end{flushleft}{\large \par}

\begin{flushleft}
{\large \vspace{-0.1cm}\hspace{0.8cm}Doctor Paulo Andr\'e de Paiva Parada, Assistant Professor,}
\par\end{flushleft}{\large \par}

\begin{flushleft}
{\large \vspace{-0.3cm}\hspace{0.8cm}Faculdade de Ci\^encias,
Universidade da Beira Interior}
\par\end{flushleft}{\large \par}

\begin{flushleft}
{\large \vspace{-0.1cm}\hspace{0.8cm}Doctor David Emanuel da Costa, Assistant  Researcher,}
\par\end{flushleft}{\large \par}

\begin{flushleft}
{\large \vspace{-0.3cm}\hspace{0.8cm}Instituto Superior T\'ecnico,
Universidade de Lisboa}
\par\end{flushleft}{\large \par}

\begin{center}
{\large \vspace{0.5cm}}\textbf{\large Funding Institutions}
\par\end{center}{\large \par}

\begin{center}
{\large \vspace{-0.2cm}Funda\c{c}\~ao para a Ci\^encia e a Tecnologia}
\par\end{center}{\large \par}

\begin{center}
\textbf{\large 2015}\vspace{-2.0cm}
\par\end{center}

 \selectlanguage{portuges}
\chapter*{Campos escalares em Física de Partículas}
\paragraph{Leonardo Antunes Pedro}
\paragraph{Doutoramento em Física}
\paragraph{Orientador} Doutor Gustavo da Fonseca Castelo Branco
\section*{Resumo}
Alargar o sector escalar ajuda a estudar o mecanismo de Higgs e alguns problemas do Modelo Padrão.

Implementamos a correspondência entre os estados elementares dependentes de gauge e os estados assimptóticos não-perturbativos invariantes de gauges não-abelianas, necessários para estudar a fenomenologia não-perturbativa de dois-dubletos-Higgs.

A violação de sabor e CP nos dados experimentais obedece a um padrão hierárquico, acomodado pelo Modelo Padrão.
Definimos a condição de Violação Mínima de Sabor com seis espuriões em teorias de campo efectivas, implicando violação de sabor e CP inteiramente dependente das matrizes de mistura dos fermiões mas independente da hierarquia das massas dos fermiões; é invariante sobre o grupo de renormalização.

Estudamos a fenomenologia de modelos de dois-dubletos-Higgs, que verificam a condição definida como consequência de uma simetria;
novas partículas escalares leves, mediando correntes neutras que violam o sabor, são permitidas pelos dados de sabor sem coeficientes de sabor extra; testámos os modelos com bibliotecas de C++ ligadas pela biblioteca simbólica GiNaC e propomos mais bibliotecas para uma procura por correntes neutras que violam o sabor.

Mapeamos as representações do grupo de Poincare complexas para as reais, derivamos a equação de Dirac livre requerendo localizabilidade covariante das representações e estudamos Localização e simetrias de gauge.

\paragraph{Palavras-chave}modelo de dois-dubletos-Higgs; estados assimptóticos; mecanismo de Higgs; Violação de Sabor Mínima; Correntes Violadoras de Sabor; cálculo simbólico; representação real; grupo de Poincare; Localização; spinor de Majorana.

\selectlanguage{english}
\chapter*{Scalar Fields in Particle Physics}
\section*{Abstract}
Extending the scalar sector helps in studying the Higgs mechanism and some Standard Model problems.

We implement the correspondence between the gauge-dependent elementary states
and the non-perturbative non-abelian gauge-invariant asymptotic states, necessary to study the non-perturbative phenomenology of two-Higgs-doublet models.

The Flavour and CP violation in experimental data follows a hierarchical pattern, accounted by the Standard Model.
We define the Minimal Flavour Violation condition with six spurions in effective field theories,
implying Flavour and CP violation entirely dependent on the fermion mixing matrices 
but independent of the fermion masses hierarchy; it is renormalization-group invariant.

We study the phenomenology of renormalizable two-Higgs-doublet models which verify the
defined condition as consequence of a symmetry; new light physical scalars, mediating Flavour Changing Neutral Currents, are allowed by flavour data without flavour coefficients beyond the Standard Model; we tested the models with C++ libraries linked by the symbolic skills of the GiNaC library and we propose more libraries supporting a systematic search for Flavour Changing Neutral Currents.

We also map the complex to the real Poincare group representations, derive the free Dirac equation requiring covariant localizability of the representations and study Localization and gauge symmetries in Quantum Field Theory.

\paragraph{Keywords} two-Higgs-doublet model; Asymptotic states; Higgs mechanism; Minimal Flavour Violation; Flavour Changing Neutral Currents; computer algebra system; real representation; Poincare group; localization; Majorana spinor.

\chapter*{Acknowledgments}
\epigraph{}{Da daaa da. Daaaa da da.--- \emph{Joana Pais Pedro (2014)}}

\small

When a person asks me to solve a problem which is time consuming and it is clear what is 
the purpose, sometimes I study if there is a solution to another problem which would serve the same purpose in a better way. 
My sister says ``tu \'{e}s um chato'' and my friends used to say ``m\'{o}, s\'{o} \'{e}s esperto para a escola!'', 
so the person usually don't like my alternative solution (``m\'{o}'' is an exclamation used in Algarve).

I don't make it easy for them but still, my family and friends, supervisor and collaborators, professors and colleagues
have provided me all the necessary guidance, autonomy and support to do this thesis and related work. 
Without their help, I would have other concerns and I couldn't focus in studying and solving physics problems.
I am thankful to:
\begin{itemize}\itemsep1pt \parskip0pt \parsep0pt
\item My family and friends;
\item My supervisor and collaborator Gustavo Branco; my collaborators Francisco Botella, Adrian Carmona, Miguel Nebot, Margarida Rebelo;
\item Members of the CFTP Lisboa and Physics and Mathematics Departments of T\'{e}cnico-Lisboa, IFIC Val\`{e}ncia, ETH Z\"{u}rich, CERN-TH and Institut f\"{u}r Physik UNI-GRAZ, where the work took place.
Including Renato Fonseca, Jos\'{e} Mour\~{a}o, David Emmanuel-Costa and Jos\'{e} Nat\'{a}rio for their help related to the real representations of the Poincare group;  Nuno Ribeiro, David Forero, Lu\'{i}s Lavoura  and Jorge Rom\~{a}o for their help related to the phenomenology of the BGL models; Axel Maas, Wolfgang Schweiger, Elmar Biernat and Ant\'{o}nio Figueiredo for their kind hospitality and help related to the non-perturbative phenomenology; Axel Maas also for his availability for an ongoing collaboration;
\item Current and former members of the LIP-Lisboa CMS group, including Pedro Martins, Pedro Silva, Jo\~{a}o Pela and Michele Gallinaro for their advices and availability for a collaboration which did not happened due to lack of time; also Joaquim Silva-Marcos from CFTP Lisboa and Palash Pal from Saha Institute of Nuclear Physics for the same reasons.
\end{itemize}

I acknowledge the support---provided through the grant
SFRH/BD/70688/2010 from the \textit{Fundação
para a Ciência e a Tecnologia}---of the Portuguese
State. This thesis is based in part in the following publications and citeable papers (co)authored by the present author:
\begin{itemize}\itemsep1pt \parskip0pt \parsep0pt
\item ``Physical constraints on a class of two-Higgs doublet models with FCNC at tree level'',\\ with F. J. Botella, G. C. Branco, A. Carmona, M. Nebot, and M. N. Rebelo,\\
\href{http://dx.doi.org/10.1007/JHEP07(2014)078}{JHEP 7 (2014) 78};
\item ``On the real representations of the Poincare group'', \href{http://arxiv.org/abs/1309.5280}{arXiv:1309.5280}, 2013;
\item ``The Majorana spinor representation of the Poincare group'', 
\href{http://arxiv.org/abs/1307.1853}{arXiv:1307.1853}, 2013.
\end{itemize}
The last two papers were not published yet because the author needs to discuss them further 
within the physics community (in particular mathematical physics), a process that requires some time and patience.
This is expected, for estabilished authors in HEP the process of publishing a paper in a good journal can easily take six months, so for someone not yet estabilished
it can easily go beyond one year.
\normalsize

\cleartooddpage

\tableofcontents{}

\listoftables
\listoffigures

\section*{Acronyms}\small
\begin{acronym}[TDMA]
\acro{2HDM}{two-Higgs-doublet model}
\acro{ATLAS}{A Toroidal LHC ApparatuS}
\acro{BR}{Branching Ratio}
\acro{BGL}{Branco\textendash{}Grimus\textendash{}Lavoura}
\acro{BSM}{Beyond the Standard Model}
\acro{CL}{Confidence Level}
\acro{cLFV}{charged Lepton Flavor Violation}
\acro{CLIC}{Compact Linear Collider}
\acro{CMS}{Compact Muon Solenoid}
\acro{CP}{Charge-Parity}
\acro{CPT}{Charge-Parity-Time reversal}
\acro{DM}{Darkmatter}
\acro{EDM}{Electric Dipole Moment}
\acro{EFT}{Effective Field Theory}
\acro{EW}{Electroweak}
\acro{EWSB}{Electroweak symmetry breaking}
\acro{FCNC}{Flavour Changing Neutral Current}
\acro{MET}{Missing Transverse Energy}
\acro{MFV2}{Minimal Flavor Violation with two spurions}
\acro{MFV6}{Minimal Flavor Violation with six spurions}
\acro{GIM}{Glashow\textendash{}Iliopoulos\textendash{}Maiani}
\acro{GNS}{Gelfand-Naimark-Segal}
\acro{GUT}{Grand unified theory}
\acro{ILC}{International linear collider}
\acro{LEP}{Large electron\textendash{}positron collider}
\acro{LFC}{Lepton flavor conservation}
\acro{LFV}{Lepton Flavor Violation}
\acro{LHC}{Large Hadron Collider}
\acro{MFV}{Minimal flavour violation}
\acro{MIA}{Mass insertion approximation}
\acro{MSSM}{Minimal Supersymmetry Standard Model}
\acro{nuMSM}[$\nu$MSM]{minimal extension of the Standard Model by three right-handed neutrinos}
\acro{PS}{Pati-Salam}
\acro{PT}[$\mathrm{p_T}$]{transverse momentum}
\acro{QCD}{Quantum chromodynamics}
\acro{RG}{Renormalization group}
\acro{RGE}{Renormalization group equation}
\acro{SM}{Standard Model}
\acro{SUSY}{Supersymmetry, Supersymmetric}
\acro{VEV}{Vacuum expectation value}
\acro{MEG}{Muon to electron and gamma}
\acro{NP}{New Physics}
\acro{NH}{Normal hierarchy}
\acro{IH}{Inverted hierarchy}
\acro{CKM}{Cabibbo\textendash{}Kobayashi\textendash{}Maskawa}
\acro{PMNS}{Pontecorvo-Maki-Nakagawa-Sakata}
\end{acronym}

\normalsize
\clearpage{}
\cleartooddpage
\mainmatter
\clearpage{}\chapter{Introduction}
\label{chap:Introduction}

\section{Particle Physics}
\begin{epigraphs}
\qitem{
To explain all nature is too difficult a task for any one man or even for any one age. 
It is much better to do a little with certainty, and leave the rest for others that come after you,
than to explain all things by conjecture without making sure of any thing.}
{--- \textup{Isaac Newton (1704)\cite{newtonbio}}}

\qitem{
It appears therefore that certain phenomena in electricity and magnetism lead to the same conclusion 
as those of optics, namely, that there is an \ae{thereal} medium pervading all bodies, and modified only
in degree by their presence; that the parts of this medium are capable of being set in motion by electric
currents and magnets; that this motion is communicated from one part of the medium to another by forces 
arising from the connexions of these parts; that under the action of these forces there is a certain yielding 
depending on the elasticity of these connections; and that therefore energy in two different forms may exist 
in the medium, the one form being the actual energy of motion of its parts, and the other
being the potential energy stored up in the connexions, in virtue of their elasticity.

Thus, then, we are led to the conception of a complicated mechanism capable of a vast variety of motion, 
but at the same time so connected that the motion of one part depends, according to definite relations, 
on the motion of the other parts, these motions being communicated by forces arising from the relative
displacement of the connected parts, in virtue of their elasticity. Such a mechanism must be subject to 
the general laws of Dynamics, and we ought to be able to work out all the consequences of its motion, 
provided we know the form of the relation between the motions of the parts.}
{--- \textup{James C. Maxwell (1865)\cite{maxwell}}}

\qitem{Radiation in free space as well as isolated 
material particles are abstractions, their properties in the quantum theory 
being definable and observable only through their interactions with other 
systems. Nevertheless, these abstractions are indispensable for a description
of experience in connection with our ordinary space-time view.}
{--- \textup{Niels Bohr (1928)\cite{bohr}}}
\end{epigraphs}

In Particle Physics, we do not know what the elementary particles are. 
Since any experimental
apparatus is built from elementary particles, we can only measure the 
effects of the particles we are studying on the particles from the apparatus
when they interact. For instance, in Astronomy or in Chemistry we may study 
the structure of the stars or the molecules, using the knowledge about the 
interaction properties of its components.
In Particle Physics we study the interaction properties of the elementary particles instead.

\pagebreak
In a simplified view, we can try to divide Particle
Physics into 3 main areas:
\begin{description}\itemsep1pt \parskip0pt \parsep0pt
\item[Mathematical/Computational Physics] mathematical/computational 
tools allow to derive many logical consequences and construct 
simulation tools from our knowledge and assumptions about particles.
\item[Theoretical Particle Physics] physics models are studied and constructed,
such that they are compatible with what we know, assume and its logical consequences.
The model's predictions to be compared with the
experimental data are calculated, often using computer simulations.
\item[Experimental Particle Physics] the experiments are built and
conducted. Using simulations, the expected data
compatible with the theoretical predictions is calculated and compared with
the experimental data, producing more knowledge about particles.
\end{description}

Today, an excellent particle physicist is likely to excel in one of
different subjects such as: algebra, geometry, computer science,
statistics or electronics. The result is that we
can explain, within the experimental uncertainty, an impressive range
of physical phenomena\cite{particle}. The recent discovery of a Higgs boson,
crucial for the logical consistency of the Standard Model\cite{lifeboson,*makingof,*glorious}, 
is the icing on the cake\cite{higgs1,*higgs2,*higgs3,*higgs4,ATLAS,*CMS}. 
In 2010, my master thesis was about the search for a charged
Higgs boson in the early data of the \acs{LHC} using the \acs{CMS} experiment, and I
could check myself in many high energy phenomena that the simulations
explained the experimental data\cite{CMS-PAS,*CMS-AN10,master}. 
In my doctorate studies I could check myself in many low energy phenomena that the
Standard Model explained the experimental results\cite{bglanalysis}.

The achievements of Particle Physics should not induce in us a blind
confidence in everything we think is true about particles.
The success of this field of science,
where we do not know the internal structure of its objects of study---the
elementary particles---, can only come from  critical thinking and
hard work, as it happens.

\clearpage

\section{Contributions from Social and Computer Sciences, Mathematics and Quantum Foundations}
\begin{epigraphs}
\qitem{Everyone is sure of this 
 \emph{[the hypothesis that errors are normally
   distributed]}, Mr. \emph{[Gabriel]} Lippman told me one day, since the
   experimentalists believe that it is a mathematical theorem, and the
   mathematicians that it is an experimentally
  determined fact.}
 {--- \textup{Henri Poincar\'{e}, \emph{Calcul des probabilit{\'e}s} (1912)}}
\qitem{The 1960s was a golden age for particle physics thanks to remarkable advances in accelerator physics-progress 
matched by the increased power and sophistication of particle detectors.\emph{[...]}
In vibrant fields of observational science, practitioners cannot be too dogmatic or doctrinaire for the simple reason
that their ideas will soon be put to the test.\emph{[...]}
And this difference feeds back into improved sociology throughout the entire scientific community.
On the other hand, when there is no fear factor, there is no penalty for dogmatism.
And so dogmatism often emerges.\emph{[...]} one should exhibit at least as much skepticism and doubt as certainty, 
and as much tolerance for other points of view as is the case in a strongly data-driven environment.
}{--- \textup{James Bjorken, \emph{Data Matters}, News from ICTP 112 (2004)}}
\qitem{I tried once in a talk to describe
the different approaches to progress in
physics like different religions. You have
prophets, you have followers --- each
prophet and his followers think that they
have the sole possession of the truth.\emph{[...]}
The problem with a lot of physicists is that they have a tendency to ``follow the leader'':
as soon as a new idea comes up, ten people write ten or more papers on it, and the effect is that
everything can move very fast in a technical direction.
But big progress may come from a different direction;
you do need people who are exploring different avenues.
}{--- \textup{M. Atiyah, Interview during the 
\emph{Abel Prize} celebrations (2004)}}
\end{epigraphs}

There are many good reasons for the scientific
community to be organized along a finite number of directions of research which reflect the 
progress achieved so far. But those reasons have very little to do with how to progress further.
In the same way that the knowledge about nutrition or hydrodynamics
contributes to the improvement of the swimmers' performances in the Olympics; also the knowledge about
Social and Computer Sciences and Mathematics contributes to the improvement of the scientists' performance.

From Social Sciences we know that
autonomy and demonstrations of respect may increase our creativity and productivity\cite{creativity}, 
but we are not capable of making rational judgments whenever we try. This knowledge is based on 
the people's tendency to use the same types of reasoning for both simple and complex problems, which often succeeds on the simple problems and fails on the complex ones\cite{failure};
the Nobel winning economic Prospect theory stating that people's decision making under risk is not based on the final outcome, but on the potential 
value of losses and gains evaluated using heuristics\cite{thinking};
recently,
a study suggested that social influence substantially biases rating dynamics in systems designed to harness collective intelligence\cite{socialbias}. 
Concerning scientists, it was argued that the research in the biosciences fits a tournament
economic structure, which induces not only high productivity but also to publish quickly with the postdoc and graduate 
students as the primary labor input\cite{biosciences}; 
the provisional results of an ongoing study about the \acs{LHC}
suggest that the traditional philosophical model---where the selection of 
rival theories is based on the merits of each theory---do not fit Particle Physics, 
for theorists the personal skills seem to be a major factor when choosing theories to work on\cite{empiricallhc}.

The modern information and communication technologies brought new tools which are changing how research is done, with increased transparency, collaboration and accessibility\cite{opening}.
Computational Physics is today one important branch of physics\cite{computationalphysics}, with contributions to General Relativity\cite{nr} and Particle Physics---such as the generation of the renormalization group equations for Gauge Theories\cite{rge,*rge2}, reduction of Feynman integrals to master integrals using a computer algebra system\cite{reduze,ginac}, 
or the implementation of on-shell methods for one-loop amplitudes\cite{blackhat}.
The experience accumulated and the innovation over the years on the statistical data analysis in Particle Physics is now crucial to inferring results from the huge amount of information collected by the experiments\cite{stathep}; which in turn can be compared with the predictions of the electroweak sector\cite{gfitter} and flavour structure\cite{ckmfitter2} of the Standard Model, using global fits;
the event generators\cite{nlo} and Lattice simulations\cite{latticeflavour,maas} are crucial for the calculation of many theoretical predictions,
taking advantage of the increasing power of parallel computing. 

From Mathematics the functional renormalization group unified the renormalization methods by expressing the Wilson's idea of effective action which is iteratively calculated by successive elimination of the high-energy degrees of freedom\cite{frg},
 it was discovered the Hopf algebra structure of renormalization in perturbative quantum field theory, which allowed to develop a new approach to Feynman diagrams calculation\cite{kreimer};
the non-commutative geometry generalizes geometry with Hilbert space operators\cite{noncommutativegeo}; the geometry of jet bundles generalizes the notion of tangent vectors\cite{jetbundle}; 
the algebra of generalized functions allows well defined multiplications of Dirac deltas\cite{generalized}.
From Quantum Foundations, individual quantum systems can now be measured and manipulated\cite{nobelcat};
the Consistent Histories approach to Quantum Mechanics\cite{consistent} is an example showing that the orthodox Quantum Mechanics can be improved.

\clearpage
\section{Beyond the Standard Model: a modular approach}
\label{sec:modular}
\begin{epigraphs}
\qitem{If we want things to stay as they are, things will have to change.}
{--- \textup{G. Tomasi di Lampedusa, \emph{Il Gattopardo} (1958)}}
\qitem{Wightman and others have questioned for approximately fifty years whether
mathematically well-defined examples of relativistic, non-linear quantum field theories exist.\emph{[...]}

The answers are partial, for in most of these field theories one replaces the
Minkowski space-time $\mathbb{M}^4$ by a lower-dimensional space-time $\mathbb{M}^2$ or $\mathbb{M}^3$, 
or by a compact approximation such as a torus. (Equivalently in the Euclidean formulation
one replaces Euclidean space-time $\mathbb{R}^4$ by $\mathbb{R}^2$ or $\mathbb{R}^3$.)
Some results are known for Yang-Mills theory on a four-torus $\mathbb{T}^4$ approximating $\mathbb{R}^4$,
and while the construction is not complete, there is ample indication that known methods could be extended
to construct Yang-Mills theory on $\mathbb{T}^4$.

In fact, at present one does not know any non-trivial relativistic field theory
that satisfies the Wightman (or any other reasonable) axioms in four-dimensions.
So even having a detailed mathematical construction of Yang-Mills theory on a
compact space would represent a major breakthrough.\emph{[...]}

One presumably needs to revisit known results at a deep level, simplify the methods,
and extend them. New ideas are needed to prove the existence of a mass gap that is uniform in
the volume of space-time. Such a result presumably would enable the study of the
limit as $\mathbb{T}^4\to \mathbb{R}^4$.\emph{[...]}

It is suspected that four-dimensional quantum gauge theory with gauge group SU(N)
(or SO(N), or Sp(N)) may be equivalent to a string theory with 1/N as the string
coupling constant. Such a description might give a clear-cut explanation of the
mass gap and confinement, and perhaps a good starting point for a rigorous proof
(for sufficiently large N).
}{---\textup{A. Jaffe \& E. Witten (2006)\cite{prize}}}
\qitem{You probably know Figure 2 of the Introduction
to the Review of Particle Physics (PDG)\emph{\cite{pdg}}, which shows the development
of several experimental quantities with time.
Every now and then, all those measured values show significant jumps, pointing either to a common systematic shift or to the effect of biased analyses.\emph{[...]}

If a measurement on a quantity has already been published, every new data analysis may
have two possible outcomes: either it agrees with the previous measurement or it does not.
In the first case, the physicist who performs the new measurement will probably be content
(usually he or she has achieved a smaller error), lean back, and finish the analysis without
thinking more deeply about it. In the case of a not too large disagreement
(about one to three standard deviations), however, the scenario becomes very different: the physicist would be somewhat worried and would have a closer look for potential problems.\emph{[...]}
In this way, the new measurement becomes heavily biased towards yielding a result
close to the original value.\emph{[...]}

What can we learn from these examples? The answer is quite simple: free yourself
from any prejudice in regard to the expected result! Do not care about previous
measurements and theory expectations. At best, you only compare your result to
others once the analysis is completely finished.
}{--- \textup{Rainer Wanke (2013)\cite{systs}}}
\end{epigraphs}

The quantum chromodynamics, electroweak and flavour sectors of the Standard Model have been supported by the experimental results.
The Standard Model (when general relativity is included) cannot account for the experimental results on 
neutrino masses and mixing, baryon asymmetry, dark matter, 
Cosmic Microwave Background fluctuations\cite{lhcsymposium}. 

Then there is a number of so-called ``problems'' of the Standard Model and general relativity,
that do not satisfy our criteria of what a theory should be, among others: quantum gravity; cosmological constant(dark energy)\cite{cosmological};
hierarchy; strong \acs{CP}; arbitrariness of the parameters of the Standard Model; meta-stability of the vaccuum;
accidental suppression of \aclp{FCNC}, \aclp{EDM} and proton decay;
the lack of a nonperturbate definition for a Quantum Field Theory with gauge interactions, such as the Standard Model. The mentioned accidental suppressions are particularly relevant in models trying to explain the remaining problems.

Examples of alternatives/extensions to the Standard Model include Inflaton, Supersymmetry, Seesaw, Grand Unified Theories, Strings, more (discrete) symmetries, Axion, vector-like quarks.
Attempts to define non-perturbatively a Quantum Field Theory with gauge interactions
involve string theory or space-times with dimensions lower than 4, Euclidean metric or toroidal topology\cite{prize}.

It is remarkable that the \ac{nuMSM} and one inflaton field can already account for all the experimental results which do not support the Standard Model and general relativity (with enough statistical significance), and it is admitted that this effective model may be valid up to the Planck mass scale, such that the solution to the hierarchy and the cosmological constant problems lies in quantum gravity\cite{Altarelli,nuMSM,*nuMSMinflaton,*scale}. Moreover, a nonperturbative approach to the Effective Field theory quantization of gravity seems promising\cite{frg_gravity}, despite the fact that either the perturbation theory is not pertubatively renormalizable or it lacks unitarity. 
So, what we are called for today when developing a better theory, is not so much to account for unexpected experimental results,
but mostly to improve our understanding of the Standard Model(and its simple extensions) and general relativity and the experimental results supporting these theories.

Then there are at least 3 strategies:
\begin{enumerate}[1)]\itemsep1pt \parskip0pt \parsep0pt
\item rewrite the (possibly whole) theory based on a partial solution for quantum gravity, hopefully accounting for all the experimental results\cite{strings,qg};
\item explore possible solutions based on simple extensions to the Standard Model which do not change the Standard Model principles (Quantum Field Theory, gauge symmetry, etc) and 
so are easier to support based on the existing experimental results;
\item take advantage of the modular structure of the Standard Model and general relativity to clarify, improve and unite some modules (Poincare representations, Yang-Mills-Higgs theory, etc.).
\end{enumerate}
The strategy 2) is most useful to the understanding of the interplay between experimental results and theory, but the progress is limited by the bounds allowed by the Standard Model principles.

The strategy 3) is too general to be useful by itself because we do not know, in general, what are the optimal boundaries of each module to achieve progress. We  have to evaluate case by case, based on the understanding of the experimental results and theory. Note that the Standard Model was the theory that emerged after the work of many people over the years who certainly had the motivation of developing a better theory, the result was a very modular theory.
With access to the most of the knowledge about a module, after a few years of study 
one might be able to develop that module, which will be then integrated by other people who know 
about other modules. 

So, in the second part of the thesis we will follow strategy 3), supported by the understanding of the 
experimental results and theory acquired in the first part of the thesis which follows the strategy 2).

In the first part of the thesis we will focus on the Higgs bosons.

In Chapter 2 we implement the correspondence between the gauge-dependent elementary states
and the non-perturbative non-abelian gauge-invariant asymptotic states, necessary to study the non-perturbative phenomenology of two-Higgs-doublet models.

In Chapter 3 we define the Minimal Flavour Violation condition with six spurions in effective field theories,
implying Flavour and CP violation entirely dependent on the fermion mixing matrices 
but independent of the fermion masses hierarchy; and show that it is one-loop renormalization-group invariant.
We tested the models of Chapter 4 with C++ libraries linked by the symbolic skills of the GiNaC library and we propose more libraries supporting a systematic search for Flavour Changing Neutral Currents.

In Chapter 4 we study the phenomenology of renormalizable two-Higgs-doublet models which verify the
defined condition as consequence of a symmetry; new light physical scalars, mediating Flavour Changing Neutral Currents, are allowed by flavour data without flavour coefficients beyond the Standard Model; we tested the models with C++ libraries linked by the symbolic skills of the GiNaC library and we propose more libraries supporting a systematic search for Flavour Changing Neutral Currents.

In the second part of the thesis we focus on mathematical scalar fields: real and complex numbers.
In Chapter 5 we map the complex to the real Poincare group representations and derive the free Dirac equation requiring covariant localizability of the representations. In Chapter 6 we study Localization and gauge symmetries in Quantum Field Theory.

Finally, chapter \ref{chap:Conclusion} presents some concluding remarks.

\clearpage{}\cleartooddpage
\cleartooddpage
\clearpage{}\phantomsection
\addcontentsline{toc}{part}{I\ \ Higgs mediated Flavour Violation}

\chapter{Non-perturbative phenomenology of the two-Higgs-doublet model}
\begin{epigraphs}
\qitem{A quasi-particle in a superconductor is a
mixture of bare electrons with opposite electric charges
(a particle and a hole) but with the same spin; correspondingly 
a massive Dirac particle is a mixture of bare
fermions with opposite chiralities, but with the same
charge or fermion number. Without the gap or the
mass, the respective particle would become an eigenstate 
of electric charge or chirality.}{--- \emph{
Y. Nambu \& G. Jona-Lasinio (1960)\cite{nambu}}}

\qitem{
The continuum formulation based on perturbation methods and the lattice
(Wilson) formulation of gauge quantum field theories seemingly lead to contradictory
results, in particular when applied to Higgs models, since in the Wilson
formulation all the gauge-dependent Green functions vanish and there cannot be
spontaneous symmetry breaking.\textup{[...]}

Thus, the role of the local order parameter $\overline{\varphi}$ in the standard picture
appears merely as a way of fixing a system of local coordinates, with the result that the
physical degrees of freedom are described by multiplets of fields which, since they depend
on such a coordinate system in field space, are gauge-dependent. That role of the
parameter $\overline{\varphi}$ (of the standard picture) is also in agreement with the
result that there is no phase transition between the confinement and the Higgs regime.}
{--- \textup{J. Frohlich \& G. Morchio \& F. Strocchi (1981)\cite{higgsphenomenon}}}

\qitem{
The construction of physical charged states is one of the basic problems of gauge
field theories. It is deeply related to the solution of the infrared problem in QED,
since a physical charged particle must be accompanied by its radiation field , i.e., by a
``cloud'' of soft photons. Moreover, the possibility of constructing color-charged states
is at the root of the confinement problem.}
{--- \textup{F. Strocchi (2013)\cite{nonperturbativefoundations}}}
\end{epigraphs}

The lattice simulations of the two-Higgs-doublet model with a $SU(2)_L$ gauge symmetry 
indicate that the non-perturbative effects
may be important in some regions of parameters\cite{lattice2HDM,*maas2HDM}. For instance, the lattice simulations reveal that the non-perturbative effects are important for one Higgs doublet when the mass of the Higgs boson is below the mass of the W boson\cite{maasplot,*maas2,*maas1} (corresponding to a QCD-like domain in the phase diagram); also for a top-bottom-Higgs system the non-perturbative effects may affect the (in)stability of the Higgs potential\cite{massbounds}, so conclusions about the (meta)stability of the vacuum of the Standard Model based on perturbative methods may be premature\cite{lhcsymposium}.

This should not be a surprise, as in \ac{QCD}
it is well known that for some parameter space the perturbative methods work very well, while in others they are simply of no use and people must use non-perturbative methods such as Lattice simulations.
In the case of Electroweak theory, we have been using mostly perturbative methods because they do apply in the parameter space where they have been tested. If the perturbative methods did not work, the experiments would have noticed it and people would be using non-perturbative methods just as it happens in Quantum Chromodynamics. Therefore, there is no paradox in the fact that the perturbative methods have produced good results so far for the Electroweak theory. Of course that the fact that the coupling constant in Electroweak theory is small increases a lot the chances that the perturbative methods will work in an arbitrary region of parameter space---when compared with Quantum Chromodynamics---but this is in no way a guarantee, specially in extensions of these  non-abelian gauge theories---not yet understood non-perturbatively---with more and different types of degrees of freedom.
  
Following the standard perturbative treatment of the two-Higgs-doublet model\cite{accidental,*accidental2,Branco:2011iw,O'Neil:2009nr} our goal in this chapter is to extend the non-perturbative 
formulation of the Electroweak model with one Higgs doublet\cite{higgsphenomenon} to the two-Higgs-doublet model,
allowing for additional studies of these non-perturbative effects (for a general $SU(2)$ two-Higgs-doublet model with or without 
$U(1)_Y$ gauge or fermions).
This chapter also serves the purpose of an introduction to the Electroweak theory and two-Higgs-doublet model,
used in the next chapters to study flavour violation. 
Note however that if the reader simply wants to do perturbation theory for phenomenological studies, we suggest instead the reader to follow the standard reference\cite{Branco:2011iw}.
We follow the convention used in the reference\cite{signs} for the signs and constants. 

\section{Custodial symmetry and the Higgs mechanism}

In this section we follow an argument of L. Susskind from the 1970's\cite{banks,*nobreaking2,*nobreaking}.
What follows is at the classical field theory level. Consider the Lagrangian, 
\begin{align*}
\mathcal{L}&\equiv ((D^\mu\phi)^\dagger(D_\mu\phi)-V(\phi)-\frac{1}{4}W_{\mu\nu}^aW^{a\mu\nu}\\
V(\phi^\dagger\phi)&\equiv-\frac{m_h^2}{2} \phi^\dagger\phi+\frac{g^2m_h^2}{8m_W^2}(\phi^\dagger\phi)^2\\
D_\mu&\equiv\partial_\mu+igW_\mu^a\frac{\sigma^a}{2}\\
W_{\mu\nu}^a&\equiv -\frac{i}{g}tr([D_\mu,D_\nu]\sigma^a)=\partial_\mu W_\nu^a-\partial_\nu W_\mu^a-g\epsilon^{abc}W^b_\mu W^c_\nu
\end{align*}
Where $V(\phi^\dagger\phi)$ is the Higgs Potential,
$\phi$ is the Higgs doublet , 
$D_\mu$ is the covariant derivative dependent on the gauge field $W_\mu^a$,  
$W_{\mu\nu}^a$ is the gauge field strength tensor and finally $g$ is the coupling constant, $m_h$ and $m_W$ are the masses of the 
higgs and W bosons and $v\equiv \frac{2}{g}m_W$ is the \ac{VEV} (at classical field theory level),
i.e. the potential is minimum for $\phi^\dagger\phi=\frac{v^2}{2}$. $\epsilon^{abc}$ is the Levi-Civita symbol and $\sigma^a$ 
are the Pauli matrices.

We define the custodial $SU(2)_R$ transposed doublet $\Phi\equiv [\widetilde{\phi}\ \phi]$,
where $\widetilde{\phi}_a\equiv \epsilon_{ab}\phi_b^*$, $(a,b=1,2)$. Note that
$\widetilde{\phi}$ transforms as $\phi$ under a local $SU(2)_L$ transformation.

We can check that the Lagrangian can be rewritten as:
\begin{align*}
\mathcal{L}=\frac{1}{2}tr((D^\mu\Phi)^\dagger(D_\mu\Phi)-V(\frac{1}{2}tr(\Phi^\dagger\Phi))-\frac{1}{4}W_{\mu\nu}^aW^{a\mu\nu}
\end{align*}

The Lagrangian is invariant under $SU(2)_L \times SU(2)_R$ where $SU(2)_L$ is the local gauge symmetry and 
$SU(2)_R$ is the global custodial symmetry. For $(L,R)\in SU(2)_L \times SU(2)_R$ then $\Phi\to L\Phi R$.

We can redefine $\Phi=\rho\Theta$, with $\rho$ a positive scalar field and $\Theta$ a $SU(2)$ matrix valued field.

Then we go to the unitary gauge, by using as $SU(2)_L$ transformation  $\Theta^\dagger\in SU(2)_L$, implying  $\Phi\to \Theta^\dagger\Phi=\rho$.
We can check that the condition $\Phi=\rho\ 1$ is invariant under the global custodial symmetry $(R^\dagger,R)\in SU(2)_L\times SU(2)_R$  $\Phi\to R^\dagger\Phi R$, 
therefore the unitary gauge only fixes the local gauge transformations,
any global custodial transformation $(R^\dagger,R)\in SU(2)_L\times SU(2)_R$ conserves the unitary gauge condition.

Applying the same transformation $\Theta^\dagger\in SU(2)_L$ to $D_\mu$:
\begin{align*}
\Theta^\dagger D_\mu\Theta&=\partial_\mu+igW_\mu^j\frac{\sigma^j}{2}\\
W_\mu^j&\equiv -\frac{i}{g}tr(\Theta^\dagger(\partial_\mu+igW^a_\mu\frac{\sigma^a}{2})\Theta \sigma^j)
\end{align*}
(note the $j$ index instead of the $a$ index).

Since the change of variables $D_\mu\to \Theta^\dagger D_\mu\Theta$ is a gauge transformation by $\Theta$,
then the Lagrangian in the unitary gauge is:

$\mathcal{L}=\frac{1}{2}tr(((\partial^\mu+igW^{\mu j}
\frac{\sigma^j}{2})\rho)^\dagger(\partial_\mu+igW_\mu^k\frac{\sigma^k}{2})\rho)-V(\rho^2)-\frac{1}{4}W_{\mu\nu}^jW^{j\mu\nu}$

The minimum of the potential is unique in terms of $\rho$, hence it seems that there is no symmetry breaking in the unitary gauge. 
This manipulation is not necessarily useful at the perturbative (quantum) level 
because the unitary gauge often increases the mathematical complexity of the calculations at the loop level due to renormalization related
issues\cite{unitarygauge}. 
Moreover, expanding around the vacuum is only valid for small perturbations, 
as it breaks once $\rho$ is allowed to be close to zero\cite{symmetrybreaking}, so it is not useful at the non-perturbative level.
However, it is useful as it shows us that the symmetry breaking is gauge dependent and hence not necessarily physical.

\section{Asymptotic states}

The Phase diagram of the $SU(2)_L$ Yang-Mills-Higgs lattice theory
is connected, which implies that it may be that there is no qualitative physical difference 
between the confinement mechanism and the (non-abelian) Higgs mechanism\cite{greensite,maasplot,*maas2,*maas1,higgsphenomenon}.
Moreover, after certain incomplete gauge fixings (e.g. Coulomb or Landau, see figure \ref{fig:pd}) some global subgroup of the 
local gauge symmetry does indeed break spontaneously, but the location of the breaking in the phase
diagram depends on the choice of gauge fixing\cite{ambiguity}.
\begin{figure}[htb!]
\centering
      \includegraphics[width=0.5\textwidth]{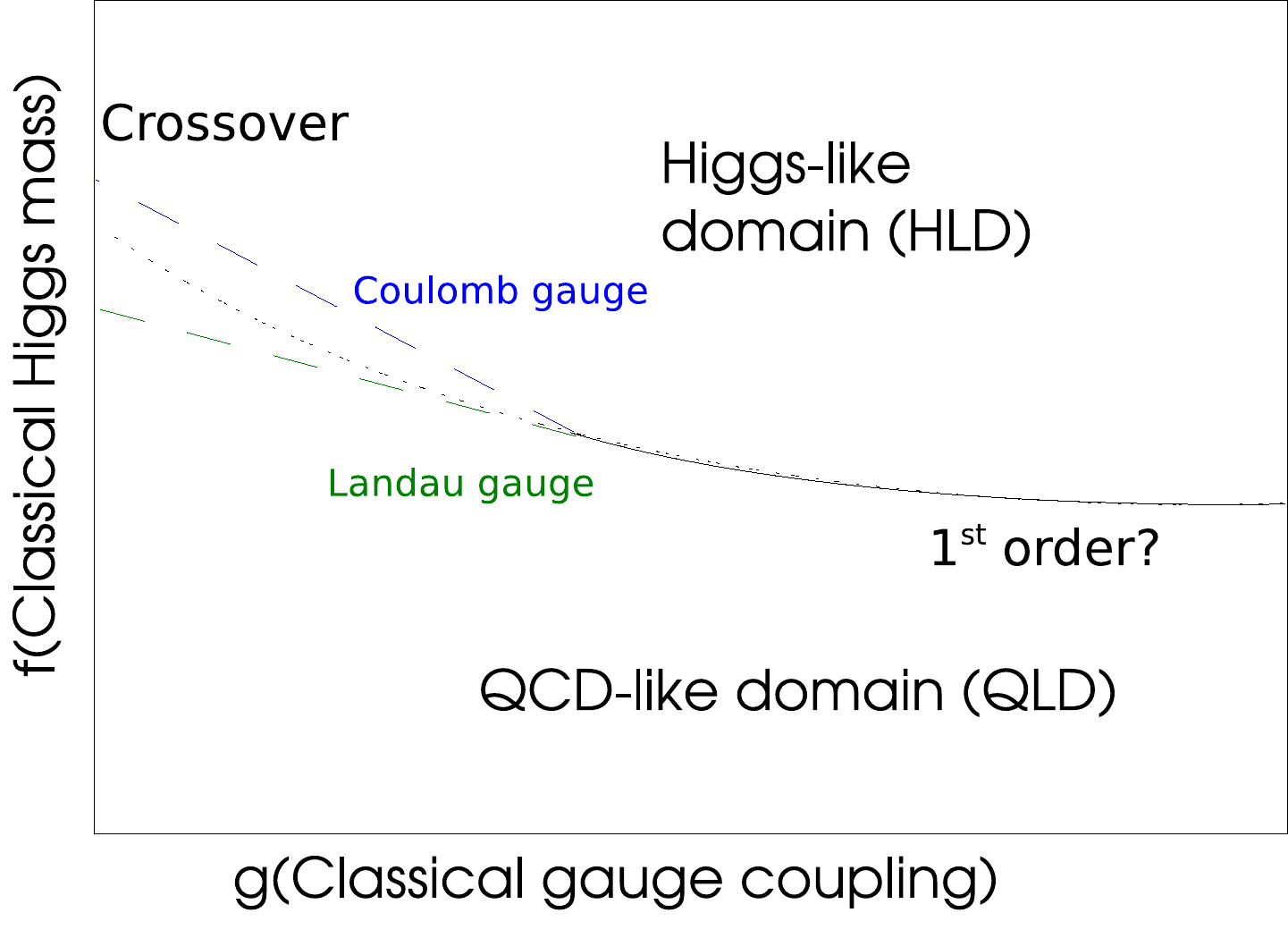}
\caption{A sketch of the phase diagram\cite{maasplot,*maas2,*maas1}, based on a quantitative version\cite{ambiguity}.
Roughly, the axis $g$ is inversely proportional to the gauge coupling, while $f$ is proportional to the vacuum expectation value (after an appropriate gauge fixing).
The solid line locates the phase transition, while the dashed lines locate the breaking of the corresponding global subgroup of the 
local gauge symmetry. The dashed lines do not coincide where there is no phase transition (the solid line is absent).}
\label{fig:pd}
\end{figure}

In a confinement region, only the bound states which are $SU(2)_L$ gauge singlets can be asymptotic states, with the Higgs doublet used to construct such singlets. 
Therefore, it may be that in the Higgs region also only the bound states which are $SU(2)_L$ singlets can be asymptotic states.
We can check that the classical fields in the unitary gauge are related with composite states which are $SU(2)_L$ gauge singlets.
In the Higgs region, we can fix a convenient gauge to do perturbation theory, expanding the Higgs doublet around a point $\Phi= v\Theta_0+\Phi_1$ that minimizes the potential, with $\Theta_0$ a $SU(2)_L$ matrix  
(in a gauge that allows it):
\begin{align*}
\Phi^\dagger\Phi&\approx v(v+\Theta_0^\dagger\Phi_1+\Phi_1^\dagger\Theta_0)+...\\
\Phi^\dagger D_\mu\Phi&\approx v^2\Theta_0^\dagger D_\mu\Theta_0+...
\end{align*}
The leading terms of the gauge singlets match the elementary fields in the unitary gauge.

It remains to be checked the contribution from the next-to leading terms of the singlets, since there are measured precision electroweak observables which must be accounted for.
Assuming that the next-to leading terms produce only scattering states, then the center-of-mass energy of such states starts at the sum of the masses of the elementary fields.
Since the Higgs is among the most heavy gauge-dependent elementary fields, the scattering state's energy spectrum starts far from the mass of the gauge-dependent elementary field.
Therefore, the contribution from the next-to-leading terms when considering center-of-mass energies close to the mass of the gauge-dependent elementary field is expected to be small.
Moreover, there are theoretical arguments indicating that the standard perturbative expansion assuming a gauge-dependent vacuum
expectation value cannot be asymptotic to gauge-dependent correlation functions\cite{higgsphenomenon} and so 
the standard perturbative expansion is not necessarily absent of problems with deviations.

A good analogy is the Kinoshita-Lee-Nauenberg theorem stating that any unitary theory is perturbatively infrared finite,
when all possible initial and final states are summed in a finite energy window then the infrared divergences cancel,
including those with soft photons\cite{kln1,*kln2,*kln}. In this case, as in the Higgs case, the correct procedure is to sum the scattering states, but unlike the Higgs case, the photon is massless and so the scattering states contribute at center-of-mass energies  near the energy of the mass of the elementary state and so it is crucial to take them into account to obtain physically meaningful results avoiding infrared divergences.
The point is that the elementary fields with the vacuum attached have the same quantum numbers as those with the Higgs bosons (i.e. the perturbations from the vacuum) attached, hence they cannot be distinguished (except in an approximate way by the energy spectrum) and the corresponding diagrams should be all summed.

The existence of bound state excitations is in principle possible and could change the predictions of perturbation theory just like in quantum electrodynamics, but there is no evidence so far from lattice simulations that such excitations are expected in the Higgs regime of $SU(2)_L$ Yang-Mills-Higgs theory\cite{maasplot}.

\section{Background symmetries}
\label{sec:bs}
A non-dynamical background field, simply background field or spurion, is a field entering in the definition of the Lagrangian but it is not a variable of the Lagrangian. It may be a non-trivial representation of a group of background symmetries of the Lagrangian, but it cannot be changed when minimizing the action and so there are no Noether's conserved currents associated with such background symmetries. It could be a non-trivial representation of the Poincare group (hence the name field), in such case the Poincare symmetry turns to a background symmetry, we are not interested in such case here.  When calculating the observables, the background fields are replaced by numerical values at each space-time point. The observables are invariant under the action of the group of the background symmetries. In the literature, the action of a group of background symmetries may be called a reparametrization\cite{group2HDM2}, a basis transformation that do not change the functional form of the Lagrangian or a spurion analysis\cite{group2HDM}, weak-basis transformations are also a group of background symmetries\cite{Botella:2012ab}. The background fields may also be considered as source fields\cite{georgi}, because a source field is an example of a background field, moreover if the background field is null, then the background symmetry is a symmetry of the Lagrangian with a conserved Noether's current associated with each continuous symmetry.

\section{two-Higgs-doublet model}

In this section we use the classification of the accidental symmetries of
the two-Higgs-doublet model\cite{accidental,*accidental2} and follow a similar notation.
What follows is at the classical field theory level.
A comment about the notation is in order, we use matrices with well defined commutation relation instead of the Higgs 
doublet indices for the same reason that people at some point started using Dirac 
gamma matrices instead of spinor indices: it may be advantageous; working with a real vector or a complex vector satisfying 
a ``Majorana condition'' is not only isomorphic, but it should lead to a similar notation once we completely avoid indices.
As we will see in the end of the chapter, we do not agree that such notation cannot be applied for studies of the full
theory with both scalars and fermions\cite{Branco:2011iw} and we find it useful for the non-perturbative formulation of the
full theory. One disadvantage is that the literature is mostly written with index notation due to historical reasons.

Let $\phi$ be a $8$ dimensional real vector.
Let $S_{k}$ and $A_k$ ($k=1,2,3$) be respectively symmetric and skew-symmetric $8\times 8$ real 
matrices which all anti-commute and are orthogonal.
The set of all possible products of $S_{k},A_k$ form a basis for the $8\times 8$ real matrices.

Let $\epsilon_{jkl} A_{k}A_{l}$ be the generators of the gauge $SU(2)_L$ transformations, that is, 
$D_\mu\equiv\partial_\mu+gW_\mu^j\epsilon_{jkl} A_{k}A_{l}$.
The matrices $1, \Sigma_a$ $a=1,...,5$ 
where $\Sigma_k\equiv S_k$ ($k=1,2,3$) $\Sigma_4\equiv A_1 A_{2}A_{3}$
and  $\Sigma_5\equiv  \Sigma_1 \Sigma_2\Sigma_3\Sigma_4$ form a basis for the symmetric matrices conserved by the generators of $SU(2)_L$.
Note that $\Sigma_a$ anti-commmute with each other.
The matrices $[\Sigma_a,\Sigma_b]$ form a basis for the skew-symmetric matrices conserved by the generators of $SU(2)_L$ and are the generators of a $Spin(5)$ group
(note that $Spin(5)$ is the double cover of the $SO(5)$ group and that in $U u_a \Sigma_a U^\dagger= v_b\Sigma_b$ with $U\equiv e^{\theta_{ab}[\Sigma_a,\Sigma_b]}$ $u_a$ and $v_b$ are related by a a $SO(5)$ transformation).

If we promote the parameters of the Higgs potential to background fields, 
the Lagrangian is invariant under the gauge group $SU(2)_L$ and the group of background symmetries $Spin(5)$. 
Therefore the physical observables are invariant under the action of the group $SU(2)_L\times Spin(5)$.
We can promote the $8$ dimensional real vector to a $8$ dimensional complex vector verifying a Majorana condition, which is the tensor product of 
a $4$ dimensional complex representation of $Spin(5)$ and a $2$ dimensional complex representation of $SU(2)_L$.
Note that $Spin(5)$ is the double cover of $SO(5)$.

The $SU(2)_L$ invariant operators for Lorentz scalars and vectors include: 
\begin{itemize}
\item $\phi^\dagger\phi$ (singlet under $SO(5)$ and Lorentz scalar);
\item $\phi^\dagger\Sigma_a\phi$ $(a=1,2,3,4,5)$ (5 representation of $SO(5)$ and Lorentz scalar);
\item $Tr([D_\mu,D_\nu][D^\mu,D^\nu])$ (singlet under $SO(5)$ and Lorentz scalar);
\item $\phi^\dagger D_\mu  [\Sigma_a,\Sigma_b]\phi$ $(a,b=1,2,3,4,5)$ (10 representation of $SO(5)$ and Lorentz vector).
\end{itemize}

Other $SU(2)_L$ invariant operators include compositions of the above mentioned operators, such as:
\begin{itemize}
\item $(\phi^\dagger D_\mu  [\Sigma_a,\Sigma_b]\phi)(\phi^\dagger D^\mu  [\Sigma_a,\Sigma_b]\phi)$;
\item $\phi^\dagger\Sigma_a\phi\phi^\dagger\Sigma_a\phi$;
\item $\partial_\mu\phi^\dagger\phi$;
\item $\partial_\mu\phi^\dagger\Sigma_a\phi$.
\end{itemize}

After gauge-fixing for a suitable gauge, 
we can expand $\sqrt{2}\phi=v \phi_0+\varphi$ around a reference point $\frac{v}{\sqrt{2}}\phi_0$ in a gauge orbit minimizing the potential, with $\phi_0^\dagger\phi_0=1$ 
and $v$ is the \ac{VEV} (at classical field theory level).
Without loss of generality we assume that the chosen orbit verifies $u_a\Sigma_a\phi_0=\phi_0$,
with the $SO(5)$ vector $u$ normalized $u_au_a=1$. Note that $1=\frac{1}{5!}\epsilon_{abcdf}\Sigma_a\Sigma_b\Sigma_c\Sigma_d\Sigma_{f}$ and so
 $\Sigma_a=\frac{1}{4!}\epsilon_{abcdf}\Sigma_b\Sigma_c\Sigma_d\Sigma_{f}$. Also $u_a=\phi_0^\dagger\Sigma_a\phi_0$.

We define $\phi_1\equiv\frac{1+u_a\Sigma_a}{2}\phi$ and $\phi_2\equiv\frac{1-u_a\Sigma_a}{2}\phi$ and we call them Higgs doublets---so to speak, as they are not 
really complex doublets but four dimensional real vectors.

There is a correspondence between the standard gauge dependent fields and the
gauge-invariant ones, the $SU(2)_L$ gauge-invariant states which describe the theory are:
\begin{itemize}\itemsep1pt \parskip0pt \parsep0pt
\item $\epsilon_{abcdf}u_f\phi^\dagger_1 D_\mu \Sigma_c\Sigma_d\phi_1$;
\item $\phi^\dagger_1\phi_1$
\item $\frac{1}{3!}\epsilon_{abcdf}u_f\phi^\dagger\Sigma_b\Sigma_c\Sigma_d\phi$ $(a,b,c,d,f=1,...,5)$;
\end{itemize}

Without loss of generality due to the background symmetry, by reparametrization of the Higgs potential
we assume that the chosen orbit verifies $u_a=v\delta_{5a}$ and so $\Sigma_5\phi_0=\phi_0$.
We now choose the reference point $v\phi_0$ to be constant in the fixed gauge and to verify the correspondence
between the custodial and gauge generators:
\begin{align*}
\epsilon_{jkl} A_{k}A_{l}\phi_0=\epsilon_{jkl} \Sigma_{k}\Sigma_{l}\phi_0\ (j,k,l=1,2,3)
\end{align*} 
which is equivalent to $A_3\Sigma_3\phi_0=\phi_0$ and $A_2\Sigma_2\phi_0=\phi_0$.
Then, $\phi^0$ conserves a $SO(3)\times Spin(3)$ background symmetry,
whose generators are $(\Sigma_4 \Sigma_j(1+\Sigma_5)/2-\epsilon_{jkl} A_{k}A_{l})$ and $\Sigma_4 \Sigma_j(1-\Sigma_5)$, respectively.

Keeping only the first non-constant terms in the expansion we get:
\begin{align*}
\phi^\dagger_1 D_\mu \epsilon_{jkl} \Sigma_{k}\Sigma_{l}\phi_1&\approx \frac{v^2}{2} W_\mu^j\\
\phi^\dagger_1\phi_1&\approx \frac{v^2}{2}+v\phi_0^\dagger\varphi\\
\phi^\dagger\Sigma_a\phi&\approx v\phi_0^\dagger\Sigma_a\varphi\ (a=1,...,4)
\end{align*}
Therefore $\phi_0^\dagger\Sigma_a\varphi$ $(a=1,...,4)$ selects the components of $\varphi$ correspondent to the second Higgs doublet;
$\phi_0^\dagger\Sigma_5\varphi=\phi_0^\dagger\varphi$ selects only the component of the first  Higgs doublet aligned with the reference point $v\phi_0$,
the remaining components of $\phi_1$ correspond to the would-be goldstone bosons and constitute the longitudinal degrees of
freedom of $W_\mu^j$.
A $Spin(5)$ transformation with generators $[\Sigma_a,\Sigma_b]$ will induce a $SO(5)$ transformation on
the states $\phi^\dagger_0 \Sigma_a\varphi$ $(a=1,...,5)$ usually identified as the Higgs boson fields---the vacuum vector $u_a$ will change accordingly.

The Higgs potential is:
\begin{align*}
V(\phi)=\mu_a\phi^\dagger\Sigma_a\phi
+\lambda_{ab}(\phi^\dagger\Sigma_a\phi\phi^\dagger\Sigma_b\phi)
\end{align*}
where $a,b=0,1,...5$ and $\Sigma_0\equiv 1$. Hence for $a,b\neq 0$, $\mu_0,\lambda_{00}$ are singlets, $\mu_a,\lambda_{0a}$ are $5$ dimensional representations of $SO(5)$ and $\lambda_{ab}$ is a tensor of $SO(5)$. 

For instance, the most general $Spin(4)$ symmetric potential is:
\begin{align*}
V(\phi)=\mu_0\phi^\dagger \phi+\mu_5\phi^\dagger\Sigma_5\phi+\frac{1}{2}\lambda_{00}(\phi^\dagger\phi)^2+\lambda_{05}(\phi^\dagger\phi)(\phi^\dagger\Sigma_5\phi)+\frac{1}{2}\lambda_{55}(\phi^\dagger\Sigma_5\phi)^2
\end{align*}
The terms in $\Sigma_5$ breaks the symmetry $Spin(5)\to Spin(4)$\cite{accidental}.

The most general minimum verifies $O_{5a}\Sigma_a\phi_0=\phi_0$, where $O\in SO(5)$. The minimum breaks the generators of 
$Spin(4)$ which  do not commute with $O_{5a}\Sigma_a$.
Without lost of generality, we can choose a basis such that $O_{51}=O_{52}=O_{53}=0$. 
Then for $O_{54}\neq 0$ the symmetry conserved by the minimum is $Spin(3)$ with generators $\epsilon_{jkl}\Sigma_k\Sigma_l$ 
and there are three broken generators of $Spin(4)$ namely $\Sigma_j\Sigma_4$, so we expect $3$ massless 
goldstone bosons.

Hence, to avoid goldstone bosons the minimum verifies $\pm \Sigma_5\phi_0=\phi_0$. We can choose a basis such that $\Sigma_5\phi_0=\phi_0$.

We simplify further, considering the Maximally-Symmetric 2HDM\cite{MS2HDM}. 
The potential is:
\begin{align*}
V(\phi)=\mu_0\phi^\dagger \phi+\mu_5\phi^\dagger\Sigma_5\phi+\frac{1}{2}\lambda_{00}(\phi^\dagger\phi)^2
\end{align*}
Then we get the stability condition $\mu_0+\mu_5=-\frac{1}{2}\lambda_{00}v^2$
and the minimum conditions $m_h^2=\lambda_{00}v^2>0$ and $m_H^2=-2\mu_5>0$.
The term in $\mu_5$ breaks softly the symmetry $Spin(5)\to Spin(4)$\cite{accidental}, giving the same mass $m_H$ to the 
Higgs states  $\phi_0^\dagger\Sigma_a\varphi$ ($a=1,2,3,4$) which are now mass eigenstates
\cite{O'Neil:2009nr}---these states are related to the states
$H^\pm$, $R$ and $I$ defined in the next section.

In the Higgs basis the potential is rewritten as:
\begin{align*}
V(H_1,H_2)=\mu_0(H_1^\dagger H_1+H_2^\dagger H_2)+\mu_5(H_1^\dagger H_1-H_2^\dagger H_2)+\frac{1}{2}\lambda_{00}(H_1^\dagger H_1+H_2^\dagger H_2)^2
\end{align*}

This Higgs potential will be used in lattice studies in future work.
\note{comment about why to mention such Higgs potential}

\section{Photons}

We now consider a Lagrangian invariant under the $U(1)_Y$ gauge symmetry with generator $\Sigma_1\Sigma_2$: 
\begin{align*}
\mathcal{L}\equiv ((D^\mu+\Sigma_1\Sigma_2\frac{g'}{2}B^\mu)\phi)^\dagger(D_\mu+\Sigma_1\Sigma_2\frac{g'}{2}B_\mu\phi)-V(\phi)-\frac{1}{4}W_{\mu\nu}^aW^{a\mu\nu}-\frac{1}{4}B_{\mu\nu}B^{\mu\nu}
\end{align*}

Where the $B_\mu$ is the $U(1)_Y$ gauge field,  
$B_{\mu\nu}$ is the gauge field strength tensor and finally $g'$ is the $U(1)_Y$ coupling constant.
All other symbols as in the previous sections, in particular 
$D_\mu\equiv\partial_\mu+gW_\mu^a\epsilon_{abc} A_{b}A_{c}$.
Then we are left with a background symmetry which is the semi-direct product $(U(1)_Y\times Spin(3))\rtimes Z_4$ of the
$Spin(3)$ group whose generators are $\Sigma_3\Sigma_4$, $\Sigma_3\Sigma_5$, $\Sigma_4\Sigma_5$
(the only ones that commute with $\Sigma_1\Sigma_2$) and the $Z_4$ group generated by the \acs{CP} background transformation
$\phi(x^0,\vec{x})\to \Sigma_2\Sigma_3\phi(x^0,-\vec{x})$. Note that $B_\mu$ transforms under \ac{CP} according to $B_\mu(x^0,\vec{x})\to -B^\mu(x^0,-\vec{x})$;
while $U(1)_Y\times Spin(3)$ is a normal subgroup, the \acs{CP} background transformation is not.
Any background transformation may be written as the product of an element of $U(1)_Y\times Spin(3)$ and an element of $Z_4$ 
(either the identity or the above defined \acs{CP} transformation).

The neutral vacuum condition is that the orbit minimizing the potential $v\phi_0$ must be aligned along a linear combination of 
$\Sigma_{3,4,5}$ which all commute with the $U(1)_Y$ generator $\Sigma_1\Sigma_2$. By reparametrization we choose $\Sigma_5\phi_0=\phi_0$.
We define $\phi_1\equiv\frac{1+\Sigma_5}{2}\phi$ and $\phi_2\equiv\frac{1-\Sigma_5}{2}\phi$.
There is a correspondence between the standard gauge dependent fields and the
$SU(2)_L$ gauge-invariant ones, the $SU(2)_L$ gauge-invariant states (but $U(1)_Y$ dependent) which describe the theory are:
\begin{itemize}
\item $\mathcal{W}^+_\mu\equiv \phi^\dagger_1 D_\mu  (\Sigma_2+i\Sigma_1)\Sigma_3\phi_1$;
\item $\mathcal{Z}_\mu\equiv\cos\theta_W\phi^\dagger_1 D_\mu \Sigma_1\Sigma_2\phi_1-\sin\theta_W \frac{v^2}{2}  B_\mu$;
\item $\mathcal{A}_\mu\equiv\sin\theta_W\phi^\dagger_1 D_\mu \Sigma_1 \Sigma_2\phi_1+\cos\theta_W\frac{v^2}{2} B_\mu$;
\item $\phi^\dagger_1\phi_1$;
\item $\phi^\dagger\Sigma_a\phi$ $(a=3,4)$;
\item $\mathcal{H}^+\equiv \phi^\dagger(\Sigma_1-i\Sigma_2)\phi$;
\end{itemize}
Where $\theta_W$ is the weak mixing angle with $\cos\theta_W\equiv \frac{g}{\sqrt{g^2+g^{'2}}}$ and $\sin\theta_W\equiv \frac{g'}{\sqrt{g^2+g^{'2}}}$.
Note that $U(1)_Y$ is an abelian gauge symmetry and so it is not related with the confinement effect,
unlike the non-abelian $SU(2)_L$ gauge symmetry.
We can check that under a gauge transformation  $U(1)_Y$ where $\phi\to e^{\Sigma_1\Sigma_2 \frac{\vartheta}{2}}\phi$, we get:
\begin{align*}
\mathcal{W}^+_\mu&\to e^{i\vartheta}\mathcal{W}^+_\mu\\
\mathcal{A}_\mu&\to \mathcal{A}_\mu-\frac{1}{g\sin\theta_W}\partial_\mu\vartheta\\
\mathcal{H}^+&\to e^{i\vartheta}\mathcal{H}^+
\end{align*}
The remaining states are invariant under $U(1)_Y$.

Under the \acs{CP} transformation we get:
\begin{align*}
\mathcal{W}^+_\mu(x^0,\vec{x})&\to (\mathcal{W}^{+\mu})^*(x^0,-\vec{x})\\
\mathcal{Z}_\mu(x^0,\vec{x})&\to -\mathcal{Z}^{\mu}(x^0,-\vec{x})\\
\mathcal{A}_\mu(x^0,\vec{x})&\to -\mathcal{A}^{\mu}(x^0,-\vec{x})\\
\mathcal{H}^+&\to (\mathcal{H}^+)^*\\
\phi^\dagger\Sigma_3\phi&\to -\phi^\dagger\Sigma_3\phi
\end{align*}

We now choose the reference point minimizing the potential $\frac{v}{\sqrt{2}}\phi_0$ ---used in the expansion $\sqrt{2}\phi=v\phi^0+\varphi$---
to be constant in the fixed gauge and to verify the correspondence between the custodial and gauge generators:
\begin{align*}
\epsilon_{jkl} A_{k}A_{l}\phi_0=\epsilon_{jkl} \Sigma_{k}\Sigma_{l}\phi_0\ (j=1,2,3)
\end{align*} 
Then the reference point conserves the electromagnetic charge with generator $(\Sigma_1\Sigma_2-A_1A_2)$, that is, $(\Sigma_1\Sigma_2-A_1A_2)\phi_0=0$.
Keeping only the first non-constant terms in the expansion we get:
\begin{align*}
\mathcal{W}^+_\mu&\approx \frac{v^2}{2} (W_\mu^1-iW_\mu^2)\\
\mathcal{Z}_\mu&\approx \frac{v^2}{2} (\cos\theta_W W_\mu^3-\sin\theta_W B_\mu)\\
\mathcal{A}_\mu&\approx \frac{v^2}{2} (\sin\theta_W W_\mu^3+\cos\theta_W B_\mu)\\
\phi^\dagger_1\phi_1&\approx \frac{v^2}{2}+v\phi_0^\dagger\varphi\\
\phi^\dagger\Sigma_a\phi&\approx v\phi_0^\dagger\Sigma_a\varphi\ (a=3,4)\\
\mathcal{H}^+&\approx v\phi_0^\dagger(\Sigma_1-i\Sigma_2)\varphi
\end{align*}
Now the standard gauge dependent fields are $W_\mu^+\equiv \frac{1}{\sqrt{2}}(W_\mu^1-iW_\mu^2)$, 
$Z_\mu\equiv (\cos\theta_W W_\mu^3-\sin\theta_W B_\mu)$, the photon field 
$A_\mu\equiv (\sin\theta_W W_\mu^3+\cos\theta_W B_\mu)$, the charged Higgs boson $H^+\equiv \frac{1}{\sqrt{2}}\phi_0^\dagger(\Sigma_1-i\Sigma_2)\varphi$,
the \ac{CP} pseudoscalar $I\equiv \phi^\dagger_0 \Sigma_3\varphi$ and finally the scalars
$R\equiv \phi^\dagger_0 \Sigma_4\varphi$ and the Higgs boson 
$H^0\equiv \phi^\dagger_0\Sigma_5 \varphi=\phi^\dagger_0\varphi$.

We can check that $(H^0,R,I)$ transforms as a $SO(3)$ vector under a background $Spin(3)$ transformation.
Also, the vacuum direction $u\equiv(\phi_0^\dagger\Sigma_5\phi_0,\phi_0^\dagger\Sigma_4\phi_0,\phi_0^\dagger\Sigma_3\phi_0)$
will transform in the same way and defines the Higgs basis.

In general the vector Higgs mass eigenstates $(h_1,h_2,h_3)$ will result from a $SO(3)$ rotation
of the Higgs basis states $(H^0,R,I)$, with angles determined by the Higgs potential.
Writing $h_j=n_{ja}\phi^\dagger_0 \Sigma_a\varphi$, with $n_{ja}n_{ja}=1$, the $SO(3)$ rotation $n$ relates 
the Higgs basis with the basis of mass eigenstates.

\section{Fermions}

In the previous section, by reparametrization we 
could choose a reference point verifying $\Sigma_5\phi_0=\phi_0$. In this section we will start by not doing it
due to the Higgs couplings to the fermions.

Consider a fermionic field $Q_L$ verifying $\Sigma_1\Sigma_2Q_L=iQ_L$ and $\Sigma_5Q_L=Q_L$,
therefore $Q_L$ is isomorphic to a complex doublet of $SU(2)_L$. 
It transforms under the gauge symmetry $SU(2)_L$, in the same way as $\phi$.
The bar $\overline{Q_{L}}\equiv (Q_{L})^\dagger\gamma^0$ stands for the usual Dirac spinor adjoint.
$Q_L$ already fixes $\Sigma_5$ and we do not want this choice to be reparametrized,
as a consequence the most general reference point does not yet verify $\Sigma_5\phi_0=\phi_0$.

Let $d_R$, $u_R$ be fermionic fields, singlets under $SU(2)_L$.
We set the hyper-charges of the gauge symmetry $U(1)_Y$ as $Q_L(1/6_Y)$, $d_R(-1/3_Y)$, $u_R(2/3_Y)$, i.e.
for $\phi\to e^{\Sigma_1\Sigma_2 \frac{\vartheta}{2}}\phi$ then $Q_L\to e^{i \frac{\vartheta}{6}}Q_L$.

The most general $SU(2)_L$ gauge invariant products of $\phi$ and $Q_L$ are complex linear combinations of
$\overline{Q_L}\phi$, $\overline{Q_L}i\Sigma_3\phi$, $\overline{Q_L}i\Sigma_2\phi$, $\overline{Q_L}i\Sigma_1\phi$ and its hermitian conjugates.
Note that as we have seen the basis of symmetric matrices commuting with the generators of $SU(2)_L$
is $\{1,\Sigma_a\}$, of skew-symmetric matrices is $\{[\Sigma_a,\Sigma_b]\}$ with $a,b,=1,...,5$, for a total of $16$ matrices.
Due to the two projectors in $Q_L$, we must divide the total by $4$ which leaves us with 4 
linearly independent products.

From the above discussion, the most general invariant form for the Yukawa couplings with the quarks is:
\begin{align*}
-{\mathcal{L}}_{Y_Q} &=\overline{Q_{L}}\ \Gamma_d \phi\ d_{R}+
\overline{Q_{L}}\ \Sigma_3\Sigma_1\Gamma_u \phi\ u_{R}+ \text{h.c.} \\
\Gamma_{d,u}&\equiv \frac{1+\Sigma_5}{2}(\Gamma_{d,u\,1r}+\Gamma_{d,u\,1i}\Sigma_3\Sigma_4+\Gamma_{d,u\,2r}\Sigma_4\Sigma_5
+\Gamma_{d,u\,2i}\Sigma_5\Sigma_3)
\end{align*}
with $\Gamma_{d,u 1,2 r,i}$ self-conjugate and acting as real scalars on $\phi$.

The background symmetry group $Spin(3)$
acts on $\phi$ and $\Gamma_{u,d}^\dagger$ in the same way with generators $\Sigma_3\Sigma_4$, $\Sigma_4\Sigma_5$ and $\Sigma_3\Sigma_5$.
Note that we could make $\Gamma_{d,u}$ complex by using the property $iQ_L=\Sigma_1\Sigma_2Q_L$ and rewriting $(\Gamma_{d\,1r},\Gamma_{u\,1i})$
as a complex term. By keeping $\Gamma_{d,u}$ self-conjugate, the $Spin(3)$ generators appear explicitly.

We can now assume without lost of generality by reparametrization of $\Gamma_{d,u}$, 
that the reference point minimizing the potential verifies $\Sigma_5\phi_0=\phi_0$.

We then define the two complex doublets in the Higgs basis as:
\begin{align*}
H_1&\equiv\frac{1-i\Sigma_1\Sigma_2}{2}\frac{1+\Sigma_5}{2}\phi\\
H_2&\equiv \Sigma_4\Sigma_5\frac{1-i\Sigma_1\Sigma_2}{2}\frac{1-\Sigma_5}{2}\phi
\end{align*} 
Also, $\widetilde{H}_j\equiv \Sigma_3\Sigma_1 H_j^*$.

The Yukawa couplings for the quarks are then rewritten as:
\begin{eqnarray*}
-\frac{v}{\sqrt 2}{\mathcal{L}}_{Y_Q} &=&\overline{Q_{L}}\ H_1 M_dd_{R}+
\overline{Q_{L}}\ H_2 N_d^0 d_R+\overline{Q_{L}}\
\widetilde{H}_{1}M_u u_{R}+\overline{Q_{L}}\ \widetilde{H}_{2} N_u^0 u_{R} + \text{h.c.} 
\end{eqnarray*}
With $M_{u,d}\equiv \Gamma_{u,d1r}+i\Gamma_{u,d1i}$, $N_{u,d}^0\equiv \Gamma_{u,d2r}+i\Gamma_{u,d2i}$.
The matrices $M_d\equiv U_L diag(m_d,m_s,m_b)U_R^{d\dagger}$, $M_u\equiv U_LV^\dagger diag(m_u,m_c,m_t)U_R^{u\dagger}$  are the quark mass matrices and $N_{d,u}^0$ are matrices not necessarily diagonal in the quark mass eigenstate basis
which may induce Higgs mediated \aclp{FCNC} at tree level. The \ac{CKM} matrix is $V$. 
 The lepton sector with three right handed neutrinos is analogous in the absence of Majorana masses to the quark sector, with the \ac{PMNS} matrix replacing the \ac{CKM}; since the Majorana mass terms in seesaw type I (as in the \ac{nuMSM}) are gauge singlets, the nonperturbative formalism can be extended to seesaw type I, however for the purposes of this chapter we do not need to enter into such detail. 

Promoting the $M_{u,d}$ and $N_{u,d}^0$ matrices to background fields (spurions), there is an additional background flavour symmetry for the quarks $SU(3)_Q\times SU(3)_U\times SU(3)_D$  and for the leptons in the absence of Majorana masses $SU(3)_\ell\times SU(3)_e\times SU(3)_\nu$. In such case, there is also a background CP(charge-parity) symmetry. The gauge group of the full Lagrangian is $SU(3)_C\times SU(2)_L\times U(1)_Y$, with the $SU(3)_C$ corresponding to the chromodinamics of the quarks as in the Standard Model and hence not discussed it here.

The fermion fields are the following representations of the groups 
(the numbers represent the dimension of the complex representation, hyper-charge $Y$ in the end, singlet representations by omission):
$Q_L(3_Q,3_C,2_L,1/6_Y)$, $u_R(3_U,3_C,2/3_Y)$ $d_R(3_D,3_C,-1/3_Y)$, $\ell_L(3_\ell,2_L,-1/2_Y)$,  $e_R(3_e,-1_Y)$, $\nu_R(3_\nu,0_Y)$.
Finally there is an abelian background symmetry $U(1)^3$ in addition to the global symmetry $U(1)_{nb}\times U(1)_{nl}$ related to the baryonic and leptonic (no Majorana masses) numbers---including $U(1)_Y$ that means one $U(1)$ for each of the 
6 above fermion fields.

There is a correspondence between the standard gauge dependent fields and the
$SU(2)_L$ gauge-invariant ones, the $SU(2)_L$ gauge-invariant states which describe the theory
transform under $e^{i \frac{\vartheta}{2}}\in U(1)_Y$ as:
\begin{align*}
H_1^\dagger Q&\to e^{-i\frac{1}{3}\vartheta}H_1^\dagger Q\\
\tilde{H}_1^\dagger Q&\to e^{i\frac{2}{3}\vartheta}\tilde{H}_1^\dagger Q\\
H_1^\dagger \ell&\to e^{-i\vartheta}H_1^\dagger \ell\\
\tilde{H}_1^\dagger \ell&\to \tilde{H}_1^\dagger \ell
\end{align*} 

The corresponding leading terms of the expansion after gauge fixing are proportional to: 
\begin{align*}
d_L&\equiv h_1^\dagger Q\\ 
u_L&\equiv \tilde{h}_1^\dagger Q\\
e_L&\equiv h_1^\dagger L\\
\nu_L&\equiv\tilde{h}_1^\dagger L
\end{align*}
where $h_1\equiv \frac{1-i\Sigma_1\Sigma_2}{2}\frac{1+\Sigma_5}{2}\phi_0$. Therefore, after gauge fixing in a suitable gauge, we can write the Lagrangian
for the Higgs-quark interactions which we will study in the next chapter, in the basis defined as in Ref.\cite{bglanalysis}, as:
\begin{eqnarray}
{\mathcal L}_Y (\mbox{quark, Higgs})& = & - \overline{d_L} \frac{1}{v}\,
[M_d H^0 + N_d^0 R + i N_d^0 I]\, d_R   \nonumber \\
&&- \overline{{u}_{L}} \frac{1}{v}\, [M_u H^0 + N_u^0 R + i N_u^0 I] \,
u_R   \label{rep}\\
& & - \frac{\sqrt{2} H^+}{v} (\overline{{u}_{L}} N_d^0  \,  d_R 
- \overline{{u}_{R}} {N_u^0}^\dagger \,    d_L ) + \text{h.c.} \nonumber 
\end{eqnarray}

\section{Higgs doublets in an arbitrary Higgs basis}

So far, whenever we wanted to define two doublets from $\phi$, we used a projection aligned with the vacuum.
However, we may be interested in imposing a symmetry, say a transformation that only acts on an arbitrary
second doublet. This symmetry is defined for an arbitrary basis, say the reference basis.
There are then two important $Spin(3)$ transformations:
the rotation relating the reference basis with the Higgs basis, where the vacuum is along $\Sigma_5$
and the \acs{CP} transformation is along $\Sigma_3$; 
the rotation relating the reference basis with the higgs mass eigenstates basis.

\cleartooddpage
\chapter{Higgs mediated Flavour Violation}
\begin{epigraphs}
\qitem{
One should always
keep in mind that every selection cut not only reduces the acceptance, but may also
lead to systematic problems if the acceptance as a function of the cut variable is not
well understood.\textup{[...]}

Sometimes a cut on a badly described quantity cannot be avoided. An example
is particle identification, where one normally has to apply hard selection criteria in
order to get the background under control. In these cases, it is necessary to not rely
on the simulation, but to use the data themselves to determine the acceptances, for
example by using similar, but well-known channels.
}{--- \textup{Rainer Wanke (2013)\cite{systs}}}
\qitem{In the bottom-up approach one constructs effective field theories involving only light degrees
of freedom including the top quark and Higgs boson in which the structure of the effective 
Lagrangians is governed by the symmetries of the SM and often other hypothetical symmetries.\emph{[...]}
 
On the other hand in the top-down approach one constructs first a specific model with heavy
degrees of freedom. For high energy processes, where the energy scales are of the order of the
masses of heavy particles one can directly use this ``full theory'' to calculate various processes in
terms of the fundamental parameters of a given theory. For low energy processes one again constructs
the low energy theory by integrating out heavy particles. The advantage over the bottom-up
approach is that now the Wilson coefficients of the resulting local operators are calculable in terms
of the fundamental parameters of this theory. In this manner correlations between various observables
belonging to different mesonic systems and correlations between low energy and high-energy
observables are possible. Such correlations are less sensitive to free parameters than individual
observables and represent patterns of flavour violation characteristic for a given theory. These 
correlations can in some models differ strikingly from the ones of the SM and of the MFV approach.
}{--- \emph{Andrzej J. Buras (2013)\cite{burastalk,*burascorrelations}}}
\qitem{Minimal flavor violation (MFV) is the assumption that there are two, and only two, spurions that break
the global $SU(3)_Q\times SU(3)_U\times SU(3)_D$ flavor symmetry\emph{[...]}
We emphasize that, while this definition of MFV implies that flavor changing couplings in the quark
sector depend on the CKM parameters, the converse is not true: It is not the case that any model where
flavor changing couplings are determined by the CKM parameters is MFV. Thus, the models proposed
in Ref. \cite{Botella:2009pq}\emph{[e.g. BGL models]} are not MFV as defined here.}{--- \emph{
Dery \& Efrati \& Hiller \& Hochberg \& Nir (2013)\cite{Dery:2013aba}}}
\end{epigraphs}

\section{Correlations are important in data analysis}

Systematic uncertainties are not only good guesses.
The experimental results, including calibrations, depend from each other and 
from the theoretical results. The systematic uncertainties quantify the uncertainty
of a given experimental result due to the uncertainty on those external inputs.

A final result presented as (mean $\pm$ errors) is ok to draw conclusions such as:\\ 
\emph{was a Higgs boson detected in the CMS experiment?}

However, to answer to questions such as: \emph{is a pattern of the \emph{(say)} BGL model present in the data
from many experiments?} Then we need to combine different data distributions which are function
of common parameters which do not depend on the experiment.
That is the reason why the combined results of the Higgs searches of ATLAS and CMS is not the
combination of the final results presented individually by ATLAS and CMS\cite{combination,*combination2,*combination3}.
The combination of different analysis within the same experiment corresponding to different decay channels is a related exercise\cite{cmscombination,*measurements}. 

In the remaining of this section we will see the example of charged Higgs searches at the LHC using the CMS experiment.
The message to retain is that it is by the correlation of different channels that we may attribute to new physics and
not to miscalculated systematic uncertainties the cause of a deviation in data with respect to the Standard Model expectation.

\subsection*{Example: charged Higgs searches at the LHC}
This study\cite{CMS-PAS,*CMS-AN10,master} focused on the search for the charged Higgs in the mass range $80 \leq M_{\mathrm{H}^\pm} \leq 160 \mathrm{GeV/c^2}$ assuming that the charged Higgs decays always to a tau lepton and neutrino, i.e. the Branching Ratio $\mathrm{BR(H^{+}->\tau^{+}\nu)=1}$.
Twenty exclusive categories of events were defined where the signal-enriched categories of events include hadronic and leptonically decaying taus~(diagram of fig. \ref{fig:data}), while background-dominated categories are used mainly to constraint the systematic errors. The systematic uncertainties are as usual included in the Likelihood function, but since the categories of events are sensible to most of these systematic errors, the final results are stable against changes (within the same order of magnitude) of most of the systematic uncertainties. The event yields observed in figure \ref{fig:data} are compared to Monte-Carlo simulations, 
the expected events from different processes are presented stacked summing up to the total expectation. In this sense
the line ttbar represents the total expectation from the Standard Model and the red line, Higgs, represents the total expectation from the two-Higgs-doublet model with $\mathrm{BR(H^{+}->\tau^{+}\nu)=1}$ for a charged Higgs mass of $\mathrm{120GeV/c^2}$.
The uncertainties of the simulation are represented by the vertical bars.
	
It is expected that $t\overline{t}$ events involving a charged Higgs to have large \ac{MET} and that one of the b-jets has softer \ac{PT} than in the Standard Model. The categories of the first plot in figure \ref{fig:data} are purposely indented to be enriched in such events and that is the main justification for the choice of a category with \acs{MET}>40 GeV and two jets with different \acs{PT} thresholds. The first required $e/\mu$ comes from the leptonic decay of the W boson (see the diagram of figure \ref{fig:data}).

The considered two-Higgs-doublet model distinguishes from the Standard Model as it is expected to yield a larger number of events
with taus at the cost of a decrease on the number of events with hard $e/\mu$. Note that the $e/\mu$ produced directly in the W decay are $\mathrm{p_T}$ harder than those that are produced in the decay of a tau coming from a charged Higgs or W bosons, due to the emission of two extra neutrinos in the decay of the tau.
	
When comparing the number of $e/\mu$ with the number of taus, the normalization error is restricted. A larger number of taus might be due to an unexpected increase of the tau fake rate or efficiency. The number of fake taus is therefore controlled
by the category $\tau_{fake}$. The distinction between $\tau_{hard}$ and $\tau_{soft}$ ensures further
that an eventual excess of true taus will correspond to taus coming from the charged Higgs. In order to restrict the errors on the tau fake rate and efficiency, specific categories are used and shown in the second and third plots. These are fake-enriched categories from W+jets events and real tau-enriched categories from $Z\to\tau\tau$.
The separation between Z and W events is achieved by the ZLike cut defined using the invariant mass of the leptons and missing transverse energy.

Finally, the categories of the fourth plot intend to select $t\overline{t}$ events where the W boson decays to quarks (see the diagram of figure \ref{fig:data}). Two b-tagged jets are required to reject W and QCD
events. Note that since we are looking to the relation between $e/\mu$ and taus, the exclusion limits
will be approximately independent of the assumed b-tag uncertainty. The
categories of the fifth plot where b-tag is not required are intended to control the mistag
efficiencies which contaminate the fourth category which depends on b tagging.
	
\begin{figure}[h!]
   \includegraphics[width=\textwidth]{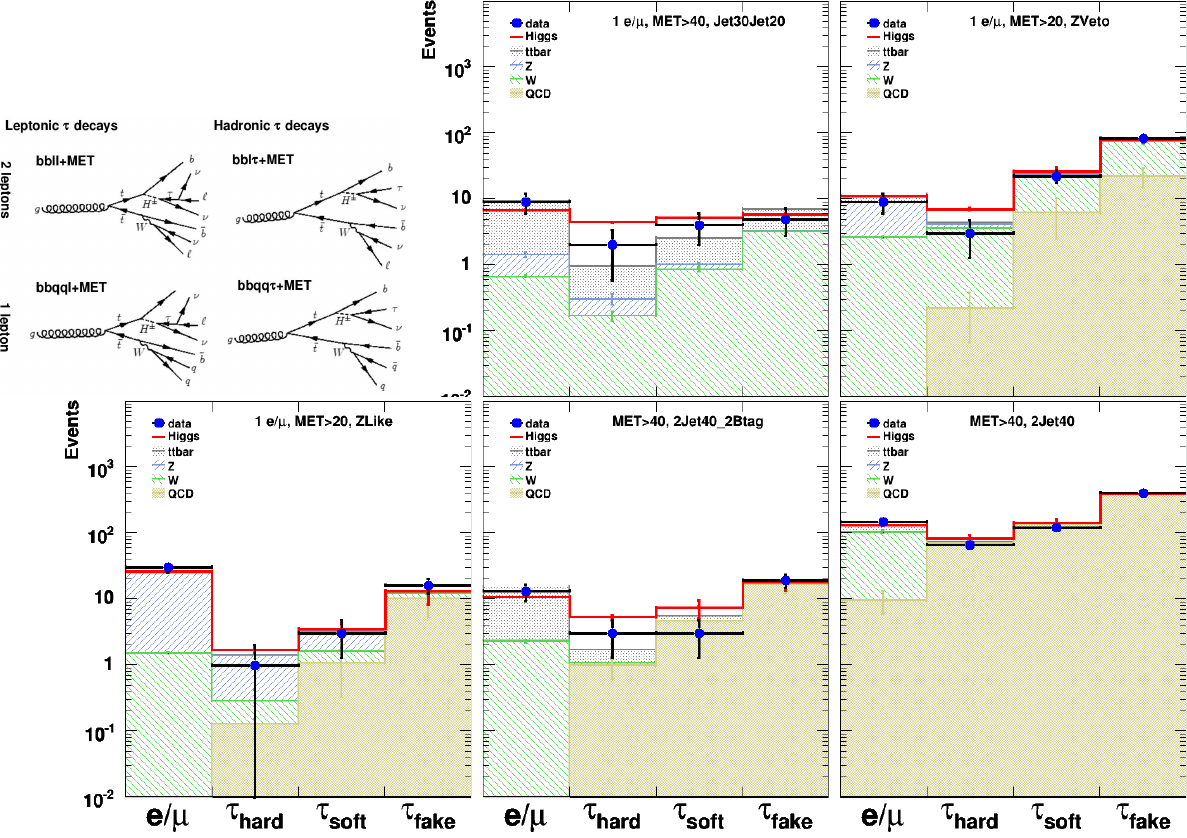}
\caption{
Diagrams of the $t\overline{t}$ decays involving a charged Higgs decaying to a tau lepton and neutrino.
Results for 2.2 $pb^{-1}$ of data at 7~TeV. The event selection used is written in each plot (for instance, the $e/\mu$ category of the first plot has the event selection two $e/\mu$, missing transverse energy MET>40 GeV and two jets, one with transverse momentum $\mathrm{p_T}$>30 GeV/c and another with $\mathrm{p_T}$>20 GeV/c).
Figure from the Master's thesis in reference\cite{master}.
   }
\label{fig:data}
\end{figure}
\clearpage
\section{Top-down and Bottom-up approaches}

We may find solutions to the problems of the Standard Model (see Section~\ref{sec:modular})
by extending it, for instance considering Grand Unified Theories or Supersymmetry.
Then we are many times confronted with the problem of the suppression of the 
Flavour Changing Neutral Currents, which in the Standard Model are accidentally suppressed 
by the \acl{GIM} mechanism\cite{gim}.
This motivates the question, how much does the experimental data constrain the 
Flavour Changing Neutral Currents which would signal New Physics?

In this section we mainly follow Buras' ideas\cite{burastalk,*burascorrelations}.
In the search for Flavor Changing Neutral Currents there are bottom-up and
top-down approaches. As we have seen in the previous section, due to the systematic
uncertainties, correlations between different channels are important to retrieve meaningful
results from experimental data. On the other hand, there is not a new physics model
which can solve all the problems of the Standard Model, specially one predictive enough to
produce useful correlations. Therefore, no useful approach is
completely bottom-up or top-down. 

In this sense, by bottom-up approach we mean 
effective field theory involving degrees of freedom up to the electroweak scale in which 
the effective Lagrangian is defined by the symmetries of the Standard Model and other 
hypotheses (e.g. Minimal Flavour Violation to be defined in the next section, or simplified models\cite{sms,*sms2,*smsatlas,*smscms}).
In short, it is the Standard Model interpreted as an Effective Field Theory extended with a few extra
local operators. With some exceptions such as transitions by two quantum numbers of flavour\cite{DF21,*DF22},
the fewer extra operators the better and so the hypotheses which reduce the number of extra
operators such as Minimal Flavour Violation play an important role.

By top-down approach we mean the study of a renormalizable model defined by hypotheses
(e.g. symmetries, hopefully with physical meaning) which remain valid at both low and high energy scales.
These allow correlations between most observables (low and high energy, involving different flavours, hadronic and leptonic, etc.).
These correlations will depend on the parameters of the model and so, the fewer extra parameters with
respect to the Standard Model the better. Of course, the known models which are predictive enough to produce 
useful correlations can only explain by themselves few, if any, problems of the Standard Model.
Therefore the main goal of these models is to help us in the search for 
Flavour Changing Neutral Currents, by predicting patterns of flavour violation
which may be different from the Standard Model and Minimal Flavour Violation hypothesis (see Figure~\ref{fig:corr}).

\begin{figure}[h!]
   \center
   \includegraphics[width=.32\textwidth]{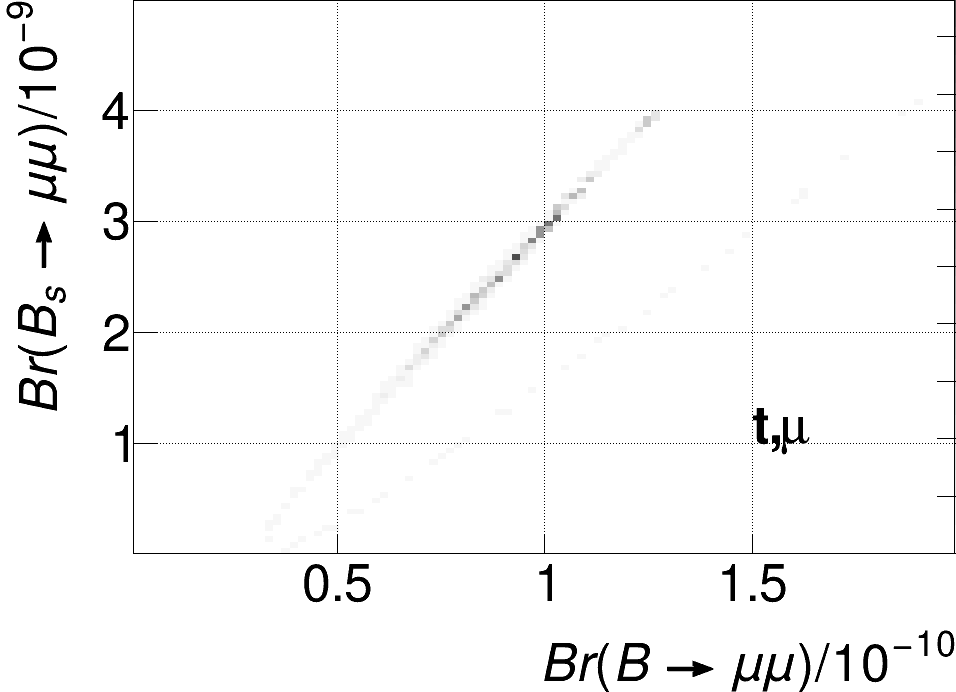} 
   \includegraphics[width=.32\textwidth]{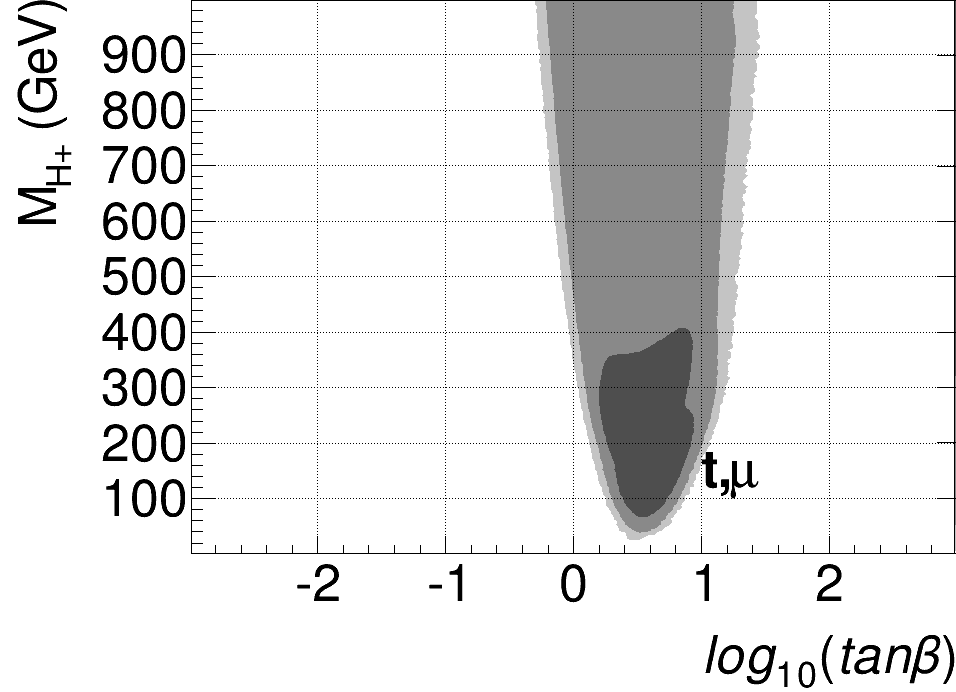} 
   \includegraphics[width=.32\textwidth]{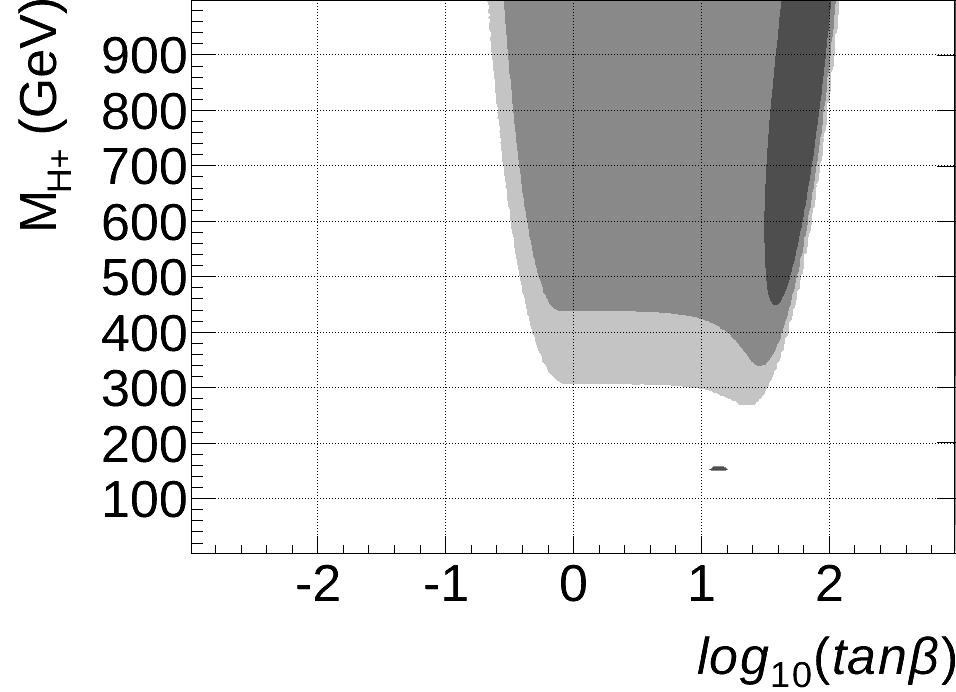}
   \caption{Left: example of a correlation for the $(t,\mu)$ BGL model, defined in Sec.\ref{sec:bgl};\\
Center: Parameter space of the $(t,\mu)$ BGL model compatible with flavour data;\\
Right: Parameter space of the type II two-Higgs-doublet model compatible with the same flavour data.
The dark regions in all figures above represent the regions of the
parameter space non-excluded by the considered data with 68\%, 95\% and 99\% \acs{CL}.
}
\label{fig:corr}
\end{figure}

We should not neglect the fact that the approaches are complementary,
for instance, the Yukawa aligned two-Higgs-doublet models may be used as the structure of 
a ultraviolet completion of a SM-like Higgs sector with free couplings\cite{smfree}.

\section{Minimal Flavour Violation}
Following the recent discovery by \acs{ATLAS} and \acs{CMS}\cite{ATLAS,*CMS} of
a particle which may be consistently interpreted as
a Standard-Model-like Higgs boson, comes the question whether the
scalar sector is larger than in the Standard Model and in particular whether there
are more Higgs doublets. There are at least two Higgs doublets
in many extensions of the Standard Model, in particular in models
with spontaneous CP violation\cite{Lee:1973iz} and in supersymmetric models.
The two-Higgs-doublet models\cite{Branco:2011iw,Djouadi:2005gj,*Gunion:1989we}
without extra symmetries, have in general flavour changing neutral currents 
which unless suppressed are not supported by the experimental results. 
The introduction of a discrete symmetry leading to natural flavour conservation\cite{Glashow:1976nt},
or the hypothesis of aligned Yukawa couplings in flavour space\cite{Pich:2009sp} suppress the Flavour Changing Neutral Currents
by avoiding them at tree-level. There are many phenomenological studies in the literature of Flavour Changing Neutral Currents in the context of the two-Higgs-doublet models\cite{Hadeed:1985xn,*Luke:1993cy,*Cvetic:1998uw,*Mohapatra:2013cia,Crivellin:2013wna}.

An alternative to suppress the Flavour Changing Neutral Currents is the principle
of Minimal Flavour Violation---either with two spurions\cite{Buras:2000dm,D'Ambrosio:2002ex,*Bobeth:2005ck,Dery:2013aba}
or with six spurions\cite{Botella:2009pq}). 
A consequence of the Minimal Flavour Violation principle is that there are 
non-vanishing Flavour Changing Neutral Currents at tree level, 
but they are only dependent on the \ac{CKM} matrix.
This is in contrast with the general two-Higgs-doublet model where there
is a large number of parameters which can be expressed in terms of
various unitary matrices arising from the misalignment in flavour
space between pairs of Hermitian flavour matrices\cite{Botella:2012ab}.
The search for the allowed parameter space in two Higgs doublet models 
for a variety of scenarios can be found in the literature\cite{Basso:2012st,*Cheon:2012rh,*Altmannshofer:2012ar,*Celis:2013rcs,*Barroso:2013zxa,*Grinstein:2013npa,*Eberhardt:2013uba,*Craig:2013hca,*Ferreira:2013qua,*Chang:2013ona,*Celis:2013ixa,*Harlander:2013qxa}. 

The two-Higgs-doublet Lagrangian describing the Yukawa couplings of the quarks was defined in Eq.~\ref{rep}.

In the following we give a definition of Minimal Flavour Violation with six spurions
which was previously defined in a different way\cite{Botella:2009pq}.
This definition generalizes the definition of Minimal Flavour Violation with two spurions,
using the same mathematical formalism\cite{Dery:2013aba}.

\subsection{Minimal Flavour Violation with six spurions}
The necessary and sufficient 
condition for Minimal Flavour Violation with six spurions in an Efective Field Theory
is that there are only six linearly independent background fields(see section~\ref{sec:bs})
breaking the flavour $SU(3)_Q\times SU(3)_U\times SU(3)_D$
and \acs{CP} global symmetries of the effective Lagrangian, 
with three of them transforming as $(3,\overline{3},1)$ under the flavour group
and admitting a simultaneous singular
value decomposition for all the three, while
the other three transform as $(3,1,\overline{3})$ under the flavour symmetry and also 
admitting a simultaneous singular value decomposition for all the three.

Any complex matrix $Y$ admits a singular value decomposition\cite{svd} $Y=U_L D U_R^\dagger$, with $U_{L,R}$ unitary and 
$D$ diagonal with non-negative real entries. The set of entries of $D$ is unique (the position of the numbers in the diagonal is not).
A set of $n$ complex matrices $\{Y_j\}$ with $j=1,...,n$ admits 
a simultaneous singular value decomposition, i.e. there are $U_{L,R}$ unitary matrices and a set of $n$ diagonal matrices 
$\{D_j\}$ with non-negative real entries such that $Y_j=U_L D_j U_R^\dagger$, if and only if the matrices within the sets
$\{Y_jY_k^\dagger\}$ and $\{Y_j^\dagger Y_k\}$ commute with each other\cite{svd}.

If the Minimal Flavour Violation is verified, then we can check that $M_d$ must be a background field transforming as 
$(3,1,\overline{3})$ under the flavour group, its unique singular value decomposition is $M_d=U_L D_d U_R^{d\dagger}$ with the diagonal quark matrix $D_d=diag(m_d,m_sm_b)$. 
There are two more linearly independent background fields with simultaneous singular value decompositions,
say $M_d'$ and $M_d''$. Then the set $\{D_d, U_L^\dagger M_d' U_R, U_L^\dagger M_d'' U_R\}$ is a basis of the
3 dimensional space of diagonal matrices. 

Note that the linear dependence condition implies that the spurions are non-zero and so are normalizable.
After a change of basis, we can consider instead the basis
$\{P_j\}$ with the entries $(P_j)_{kl}\equiv\delta_{kj}\delta_{jl}$ and so the linearly independent background fields can be chosen as 
$Y_{dj}\equiv U_L P_j U_R^{d\dagger}$ and $Y_{uj}\equiv U_L V^\dagger P_j U_R^{u\dagger}$ with $V$ the \ac{CKM} matrix and 
$U_L$, $U_R^u$, $U_R^d$ unitary matrices whose vectors are triplet representations 
of $SU(3)_Q$, $SU(3)_U$ and $SU(3)_D$ respectively.
Then, $M_{u,d}=\sum_j m_{u,d j} Y_{u,d\ j}$.

Also $N_{u,d}^0$ must be a function of the background fields 
and transform under the action of the flavour group 
in the correct way. Therefore $N_{u,d}^0=\sum_i p_i(Y_{u,d j}Y_{u,d j}^\dagger) Y_{u,d\ i}$, where $p_i$ are generic matrix polinomials 
of the terms $Y_{dj}Y_{dj}^\dagger$ and $Y_{uj}Y_{uj}^\dagger$ for $j=1,2,3$.
Up to second order in $V$, we can write them as:
\begin{align}
\label{eq:6expansion}
\begin{split}
N_d/v &= a_{2i} P_i  + a_{3ijk} P_i V^\dagger P_j V P_k+...\\
N_u/v &= b_{2i} P_i  + b_{3ijk} P_i V P_j V^\dagger P_k +...
\end{split}
\end{align}
Where $N_d\equiv U_L^\dagger N_d^0 U_R^{d}$ and $N_u\equiv U_L^\dagger V N_u^0 U_R^{u}$,  $D_d\equiv U_L^\dagger M_d U_R^{d}$ and $D_u\equiv U_L^\dagger V M_u U_R^{u}$ 
are the matrices in the basis of quark mass eigenstate.

We can rearrange the terms to obtain an equivalent expansion\cite{Botella:2009pq} to the one above~\ref{eq:6expansion}:
\begin{align*}
N_d &= \lambda_1 D_d + \lambda_{2i} P_i D_d + \lambda_{3ijk} P_i V^\dagger P_j V P_kD_d+...\\
N_u &= \tau_1 D_u + \tau_{2i} P_i D_u + \tau_{3ijk} P_i V P_j V^\dagger P_k D_u +...
\end{align*}
Note that $P_iD_d=m_{di}P_i$ and the mass can be absorbed by a redefinition of the coefficients of the expansion using the \ac{VEV} $v$ to make them dimensionless.

As a consequence the flavour changing couplings of the quarks are only dependent
on the Cabibbo-Kobayashi-Maskawa matrix $V$ and in the limit $V\to 1$
there are no Higgs mediated Flavour Changing Neutral Currents at tree level\cite{Botella:2009pq}.
The coefficients $a,b,\lambda,\tau$ in the expansion are real so that not only the flavour symmetry but also 
CP symmetry is a background symmetry of the Lagrangian. If the coefficients would have complex phases then 
the \acs{CP} symmetry would no longer be a background symmetry and the \ac{CKM} matrix would 
be the only source of \acs{CP} violation, these flavour blind phases might lead to interesting phenomenology\cite{nfcmfv}.

The necessary and sufficient 
condition for Minimal Flavour Violation with two spurions
is that there are only two linearly independent background fields
breaking the global flavour symmetry $SU(3)_Q\times SU(3)_U\times SU(3)_D$ 
and CP symmetry, with one of them transforming as $(3,\overline{3},1)$ under the flavour symmetry,
while the other transforms as $(3,1,\overline{3})$ under the flavour symmetry\cite{Dery:2013aba}.
From the previous expansion, it is straightforward to check that
up to second order in $\frac{D_{u,d}}{v}$, we can write $N_{d,u}^0$ in the basis of quark mass eigenstate 
as\cite{Botella:2012ab}:
\begin{align*}
N_d &= \lambda_1 D_d + \lambda_{2}\frac{D_d^2}{v^2}  D_d + \lambda_{3} V^\dagger \frac{D_u^2}{v^2} V D_d + ...\\
N_u &= \tau_1 D_u + \tau_{2} \frac{D_u^2}{v^2} D_u + \tau_{3} V\frac{D_d^2}{v^2} V^\dagger D_u + ...
\end{align*}
Therefore, the condition of flavour changing couplings dependent on \ac{CKM} is 
necessary but not sufficient for Minimal Flavour Violation with two spurions.

\subsection{Renormalization group evolution of Minimal Flavour Violation}

The extension of Minimal Flavour Violation to the leptonic sector is essential to study 
the renormalization group evolution of the two-Higgs-doublet model with Minimal Flavour Violation.
We consider the case of Dirac neutrinos (no Majorana masses)
with lepton flavour structure analogous to the quark flavour structure, 
hence the number of spurions is multiplied by two.

The one-loop renormalization group equations for the case of Dirac neutrinos (no Majorana masses) can be found in the literature\cite{Botella:2011ne}. These equations preserve the Minimal Flavour Violation conditions with either six or two spurions
(i.e. twelve or four spurions counting with the lepton sector).

Explicitly for twelve spurions, both $N_{d}$ and $M_d$ are given by $\sum_i p_{d i}'(Y_{d j}Y_{d j}^\dagger,Y_{u j}Y_{u j}^\dagger)Y_{d i}$ 
and the products such as $N_{d}M_d^\dagger$ are given by $q_{d}'(Y_{d j}Y_{d j}^\dagger,Y_{u j}Y_{u j}^\dagger)$ where $p_{di}'$ and $q_d'$ are matrix polynomials in 
$Y_{d j}Y_{d j}^\dagger,Y_{u j}Y_{u j}^\dagger$. The same is valid for  $N_{u}$ and $M_u$  and the corresponding matrices 
of the lepton sector. Then we can write the one-loop renormalization group equations as:
\begin{align*}
\mu\frac{d}{d\mu} M_d=\sum_i p_{d i}(Y_{d j}Y_{d j}^\dagger,Y_{u j}Y_{u j}^\dagger)Y_{d i}\\
\mu\frac{d}{d\mu} N_d=\sum_i q_{d i}(Y_{d j}Y_{d j}^\dagger,Y_{u j}Y_{u j}^\dagger)Y_{d i}
\end{align*}
Where $p_{di}$ and $q_{di}$ are matrix polynomials in $Y_{d j}Y_{d j}^\dagger$ and $Y_{u j}Y_{u j}^\dagger$.
Note that the coefficients of the polynomials are function of $tr(Y_{l j}Y_{l j}^\dagger)$ and 
$tr(Y_{\nu j}Y_{\nu j}^\dagger)$ where $Y_{l,\nu j}$ are the spurions of the lepton sector analogous to $Y_{d,u j}$,
therefore as already mentioned we cannot study 
the renormalization group evolution of the either the quark or lepton sectors alone. 

The renormalization group equations for all the matrices $M_d$, $N_d$, $M_u$, $N_u$ and the corresponding ones from the lepton sector
are analogous to the above equations, for different coefficients of the polynomials.
Therefore, the one-loop renormalization equations preserve the form  
$\sum_i p'_{d i}(Y_{d j}Y_{d j}^\dagger,Y_{u j}Y_{u j}^\dagger)Y_{d i}$ for $N_{d}$ and $M_d$ and the same applies 
for  $N_{u}$ and $M_u$  and the corresponding matrices of the lepton sector.
We conclude that the condition for Minimal Flavour Violation is preserved,
i.e. only twelve spurions break the flavour symmetry and \acs{CP} symmetry(for real expansion coefficients). 

Following the same reasoning we can check that the condition for Minimal Flavour Violation with four spurions is also preserved.
This result is consistent with the claim that Minimal Flavour Violation is renormalization group invariant\cite{nfcmfv}, 
based on one-loop and numerical two-loop studies in the context of Supersymmetry\cite{mfvrge1,*mfvrge2}. 
For different lepton sectors the number of versions of what is Minimal Flavour Violation is multiplied, but we do not expect
significant surprises\cite{mlfvrge,Botella:2011ne}.
The ideal situation would be to derive the form of the equations and consequent 
invariance under the renormalization group from the background symmetry,
hopefully for all orders of perturbation theory and for all possible models of Minimal Flavour Violation. 

Using the analogous of the argument used in the expansion of $N_d$,
the symmetry imposes $\mu\frac{d\mu}{d\mu} M_d$ and $\mu\frac{d\mu}{d\mu} N_d$ to be of the form given in the above equations at all orders of perturbation theory, but this is only valid at the classical level as anomalies may appear at the quantum level.
Also note that we need to assume that the full theory including all the high degrees of freedom respect the Minimal Flavour Violation condition\cite{nfcmfv}.

Concluding, we have shown that the condition of Minimal Flavour Violation 
with Dirac Neutrinos is renormalization group invariant at one-loop in the two-Higgs-doublet model,
a question left open by the previous studies of Minimal Flavour Violation with twelve spurions\cite[end of sec.2]{Botella:2011ne}.
\subsection{Comparison of Minimal Flavour Violation definitions}
Comparing the two definitions, we can say that Minimal Flavour Violation with two spurions
has less degrees of freedom but at the cost of assuming that not only the flavour changing couplings 
are ruled by the \ac{CKM} hierarchy but also the flavour conserving couplings are ruled by the 
quark mass hierarchy. Perhaps the \ac{CKM} and quark masses hierarchies are related\cite{Botella:2015yfa},
but the precise relation is far from clear (just look to the analogous lepton masses and \acs{PMNS} matrix)
and so to assume a particular relation may not be advantageous.

The definition of Minimal Flavour Violation with six spurions 
respects the same mathematical formalism---and as a consequence it is as renormalization group invariant---as 
the one with two spurions, 
fulfilling the goal of flavour changing couplings and \acs{CP} violation
determined by the \acs{CKM} matrix with greater generality.

There are examples of models with interesting phenomenological applications 
which have flavour changing couplings and \acs{CP} violation
determined by the \ac{CKM} matrix which do not verify the condition for 
Minimal Flavour Violation with two spurions\cite{gmfv1,*gmfv2,*gmfv3}.
In the next section we will study one explicit example of a class of models 
which do not verify the condition for Minimal Flavour Violation with two spurions
but verify the condition with six spurions, the \acs{BGL} models.

Note that the \ac{CKM} matrix $V$ is assumed arbitrary (a spurion)
in particular $V$ can be replaced by the identity 1. If $V=V_{CKM}$ is fixed to its experimental values
then the expansion has nine terms and can reproduce any matrix\cite{Ellis:2009}.
An alternative application of Minimal Flavour Violation is to parametrize any model with the fixed $V_{CKM}$.
Therefore, only the first few terms of the expansion in $V$ are physically relevant.

\section{BGL models}
\label{sec:bgl}

An interesting alternative to Natural Flavour Conservation is provided by the \ac{BGL} 
models \cite{Branco:1996bq,Botella:2009pq,Botella:2011ne}, where there are 
non-vanishing Flavour Changing Neutral Currents at tree level, 
but they are naturally suppressed as a result of a continuous global symmetry of all the terms of the Lagrangian 
except the Higgs potential, where the symmetry is softly broken. 

The extension of \ac{BGL} models to the leptonic sector is essential to study 
the renormalization group evolution\cite{Botella:2011ne} and their phenomenology\cite{bglanalysis}.
We consider the extension of the two-Higgs-doublet models with three right-handed neutrinos.
For simplicity, in this section we only consider Dirac type neutrinos,
where no Majorana mass terms are added to the Lagrangian, the general case
can be found in the literature\cite{Botella:2011ne}. 

The \ac{BGL} models with Dirac neutrinos verify the Minimal Flavour Violation condition with twelve
spurions\cite{Botella:2009pq} but not with only four spurions\cite{Dery:2013aba}.
Therefore, due to one-loop renormalization group invariance of Minimal Flavour Violation in the
two-Higgs-doublet model with Dirac neutrinos, in the \ac{BGL} models the fact that the Flavour Changing 
Neutral Currents are only dependent on the \ac{CKM} and \ac{PMNS} matrices is stable under one-loop renormalization group,
a result which was specifically shown before for the \ac{BGL} models\cite{Botella:2011ne}.

The neutral and the charged Higgs interactions obtained from the quark
sector are of the form given by Eq.~\ref{rep}. 
In terms of the quark mass eigenstates $u, d$, the Yukawa couplings
are:
\begin{multline}
{\mathcal L}_Y (\mbox{quark, Higgs} ) =\\  - \frac{\sqrt{2} H^+}{v} \bar{u} \left(
V N_d \gamma_R - N^\dagger_u \ V \gamma_L \right) d +  \mbox{h.c.} - 
\frac{H^0}{v} \left(  \bar{u} D_u u + \bar{d} D_d \ d \right) -  \\
 -  \frac{R}{v} \left[\bar{u}(N_u \gamma_R + N^\dagger_u \gamma_L)u+
\bar{d}(N_d \gamma_R + N^\dagger_d \gamma_L)\ d \right] + \\
 +  i  \frac{I}{v}  \left[\bar{u}(N_u \gamma_R - N^\dagger_u \gamma_L)u-
\bar{d}(N_d \gamma_R - N^\dagger_d \gamma_L)\ d \right]
\end{multline}
where $\gamma_{L}$ and $\gamma_{R}$ are the left-handed and right-handed chirality projectors, respectively.

The flavour structure of the quark sector of two Higgs doublet models
is characterized by the four matrices $M_d$, $M_u$, $N_d^0$,
$N_u^0$. For the leptonic sector we have the corresponding matrices
which we denote by $M_\ell$, $M_\nu$, $N_\ell^0$, $N_\nu^0$.

To obtain a structure for the $M,N$ matrices such
that the Flavour Changing Neutral Currents are completely controlled
by the \ac{CKM} mixing matrix $V$, \acl{BGL} imposed the following
symmetry on the quark and scalar sector of the
Lagrangian\cite{Branco:1996bq}:
\begin{equation}
Q_{Lj}^{0}\rightarrow \exp {(i\tau)}\ Q_{Lj}^{0}\ ,\qquad
u_{Rj}^{0}\rightarrow \exp {(i2\tau)}u_{Rj}^{0}\ ,\qquad \Phi_2\rightarrow \exp {(i\tau )}\Phi_{2}\ ,  \label{S symetry up quarks}
\end{equation}

where $e^{i2\tau}\neq 1$ , with all other quark fields transforming 
trivially under the symmetry. The index $j$ can be fixed as either 1,
2 or 3. The Higgs doublets $\Phi_1$ and $\Phi_2$ are defined in an arbitrary basis,
not necessarily the Higgs basis defined by the vacuum. Alternatively the symmetry may be chosen as:
\begin{equation}
Q_{Lj}^{0}\rightarrow \exp {(i\tau )}\ Q_{Lj}^{0}\ ,\qquad
d_{Rj}^{0}\rightarrow \exp {(i2\tau )}d_{Rj}^{0}\ ,\quad \Phi
_{2}\rightarrow \exp {(- i \tau)}\Phi_{2}\ .  \label{S symetry down quarks}
\end{equation} 
The symmetry given by Eq.~(\ref{S symetry up quarks}) leads to Higgs mediated
Flavour Changing Neutral Currents in the down sector, whereas the symmetry specified by 
Eq.~(\ref{S symetry down  quarks}) leads to Flavour Changing Neutral Currents in the up
sector at tree-level. These two alternative choices of symmetry combined with the
three possible ways of fixing  the index $j$ give rise to six
different realizations of two-Higgs-doublet models with the flavour structure, in the
quark sector, controlled by the \ac{CKM} matrix. 

In the leptonic sector, with Dirac neutrinos, there is  perfect
analogy with the quark sector. The Flavour Changing Neutral Currents 
completely controlled by the \ac{PMNS} matrix $U$ are
enforced by one of the following symmetries. Either 
\begin{equation}
L_{Lk}^{0}\rightarrow \exp {(i\tau )}\ L_{Lk}^{0}\ ,\qquad \nu
_{Rk}^{0}\rightarrow \exp {(i2\tau )}\nu _{Rk}^{0}\ ,\qquad \Phi
_{2}\rightarrow \exp {(i\tau )}\Phi _{2} \ , \label{S symetry neutrinos}
\end{equation}
or
\begin{equation}
L_{Lk}^{0}\rightarrow \exp {(i\tau )}\ L_{Lk}^{0}\ ,\qquad
\ell_{Rk}^{0}\rightarrow \exp {(i2\tau )}\ell_{Rk}^{0}\ ,\qquad \Phi
_{2}\rightarrow \exp {(-i \tau )}\Phi _{2} \ , \label{S symetry charged leptons}
\end{equation}
with all other leptonic fields
transforming trivially under the symmetry. The index $k$ can be fixed
as either 1, 2 or 3. 

This defines the \ac{BGL} models that we analyse in the next chapter.
There are thirty six different  models corresponding to the
combinations of the six possible different implementations in each
sector. To combine  the symmetry given by
Eq.~(\ref{S symetry up quarks}) with the one given by 
Eq.~(\ref{S symetry charged leptons}) an overall change of sign is
required, in one set of transformations. 

The symmetry given by Eq.~(\ref{S symetry up quarks}) for $k=3$ 
imposes that the matrices $N_d$, $N_u$ are of the
form\cite{Branco:1996bq}:
\begin{equation}
(N_d)_{ij} = \tan\beta (D_d)_{ij} - 
\left( \tan\beta +  \cot\beta\right) 
(V^\dagger)_{i3} (V)_{3j} (D_d)_{jj}\,, \label{24}
\end{equation}
whereas
\begin{equation}
N_u = \tan\beta
\mbox{diag} \ (m_u, m_c, 0)-\cot\beta \mbox{diag} \ (0, 0, m_t)\,. \label{25}
\end{equation}
In these equations only one new parameter not present in the Standard Model
appears, $\tan\beta$ defined by the Higgs potential. It is the presence of the above
symmetry, which prevents the appearance of additional free parameters.
As a result, \ac{BGL} models are very constrained but their phenomenology
crucially depends on the variant of the \ac{BGL} model considered. For
example with the choice $j=3$ leading to Eqs.~(\ref{24}), (\ref{25}), Higgs mediated \ac{FCNC}
are controlled by the elements of the third row of $V$. This leads, in a natural way,
to a very strong suppression in the neutral currents entering in the ``dangerous''
$\Delta S = 2$ strangeness violating processes contributing to $K^0 - \bar K^0$
transitions. Indeed, in this variant of \ac{BGL} models, the couplings entering in
the tree level $\Delta S = 2$ transition are proportional to
$|V_{td}V^\ast_{ts}|$ leading to a $\lambda ^{10}$ suppression in the
Higgs mediated $\Delta S = 2$ transition, where $\lambda\approx 0.2$
denotes the Cabibbo parameter. With this strong suppression even 
light neutral Higgs, with masses of the order $10^2$ GeV are
allowed. This strong natural suppression makes this variant of \ac{BGL}
models specially attractive. Figure \ref{fig:brs} shows the profile of the decays of a particular \ac{BGL} model with $j=3$;
in general the $CP$ even neutral mass eigenstates are linear
combinations of the fields $H^0$ and $R$ with the mixing
parameters determined by the Higgs potential, in the figure it was assumed no mixing. We can see that the decays involving muons dominate over the ones involving taus which may be advantageous in searches at the LHC.

\begin{figure}[h!]
   \center
   \includegraphics[width=.42\textwidth]{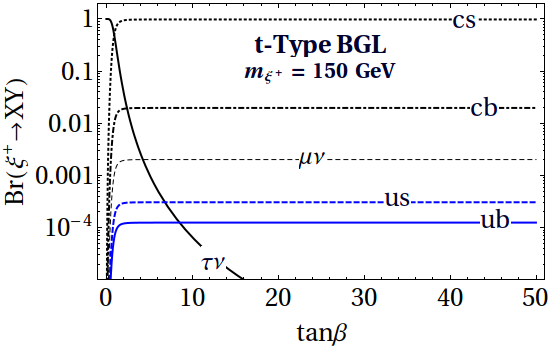} 
   \includegraphics[width=.42\textwidth]{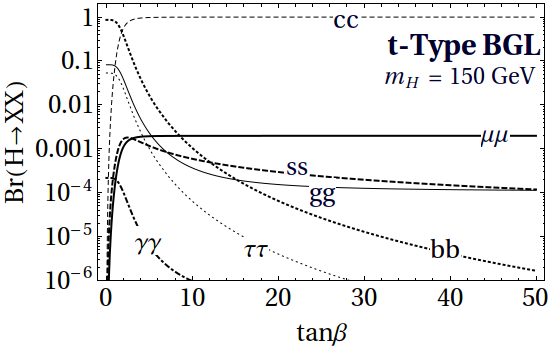}
   \caption{The charged Higgs (left) and neutral Higgs boson R (right) branching ratios to two-body final states;
for the $(t,\tau)$ BGL model (notation defined soon in this section). The plots are from reference\cite{Bhattacharyya:2014nja}.
}
\label{fig:brs}
\end{figure}

The six different \ac{BGL} models can be fully
defined\cite{Botella:2009pq} by:
 \begin{eqnarray}
N_d = \tan{\beta} D_d - \left( \tan\beta +  
\cot\beta\right)   P_{j}^{\gamma }\ D_d \,,
\label{bgl1} \\
N_u = \tan{\beta} D_u - 
\left( \tan{\beta} +  \cot{\beta}\right)  
 V P_{j}^{\gamma }  V^\dagger\ D_u\,, \label{bgl2}
\end{eqnarray}

where $\gamma $ stands for $u$ (up) or $d$ (down) quarks, and 
$P_{j}^{\gamma }$ are the projection operators
defined \cite{Botella:2004ks} by:
\begin{equation}
P_{j}^{u}\equiv V^\dagger P_{j}^d V\ ,\qquad \left( P_{j}^d \right) _{lk}\equiv\delta _{jl}\delta _{jk}\ ,  \label{projectors1}
\end{equation}

With this notation the index $\gamma$ refers to the sector that has no
Flavour Changing Neutral Currents and $j$ refers to the row/column  of $V$ defined by the symmetry.
Note that for $\gamma$ denoting  ``up'' the index $j$
singles a row of $V$, while for $\gamma$ denoting  ``down'' the
index $j$ singles a column of $V$. A characteristic feature of
\ac{BGL} models is the fact that both matrices $N_d$,
$N_u$ involve the same projection operator.

The \acs{BGL} models are a class of models with Minimal Flavour Violation (with twelve spurions)
as a result of an abelian symmetry in this sense they are special. It was shown that
\acs{BGL} models are the only
models satisfying a set of conditions sufficient for Minimal Flavour Violation
that can be enforced by abelian symmetries\cite{Ferreira:2010ir}.

In the leptonic sector for Dirac neutrinos we have\cite{Botella:2011ne}:
 \begin{eqnarray}
N_\ell = \tan{\beta} D_\ell - \left( \tan\beta +  
\cot\beta\right)   P_{m}^{\eta }\ D_\ell \,,
\label{bgll1} \\
N_\nu = \tan{\beta} D_\nu - 
\left( \tan{\beta} +  \cot{\beta}\right)  
 U^\dagger P_{m}^{\eta }  U\ D_\nu\,, \label{bgll2}
\end{eqnarray}

where $\eta$ stands for $\ell$ (charged leptons) or $\nu$ (neutrinos), and 
$P_{m}^{\eta }$ are the projection operators
defined by:
\begin{equation}
P_{m}^{\nu}\equiv U P_{m}^\ell U^\dagger\ ,\qquad \left( P_{m}^\ell \right) _{lk}\equiv\delta _{ml}\delta _{mk}\ ,  \label{projectors1}
\end{equation}
 
In the leptonic sector, the \ac{PMNS} mixing matrix
$U\equiv  U^\dagger_{\ell L}U_{\nu L}$, has large
mixings, unlike the \ac{CKM} matrix $V$. Therefore, the Higgs mediated
Flavour Changing Neutral Currents are not strongly suppressed. However, models where the Higgs
mediated leptonic Flavour Changing Neutral Currents are present only in the neutrino sector can be
easily accommodated experimentally due to the smallness of the neutrino masses. 

We label each of the thirty six different models by the pair ($\gamma_j$, $\eta_m$): 
the generation numbers $j,m$ refer to the projectors $P_{j,m}$
involved in each sector $\gamma,\eta$. For example, the model $(\text{up}_3, \ell_2)=(t,\mu)$ 
will have no tree level neutral flavour changing couplings in the up quark
and the charged lepton sectors while the neutral flavour changing couplings
in the down quark and neutrino sectors will be controlled, respectively, by $V_{td_i}^{\phantom{\ast}}V_{td_j}^\ast$ and $U_{\mu \nu_a}^{\phantom{\ast}}U_{\mu \nu_b}^\ast$.

In BGL models the Higgs potential is constrained by the
symmetry to be of the form:
\begin{eqnarray}
V_\Phi&=&\mu_1 \Phi_1^{\dagger}\Phi_1+\mu_2\Phi_2^{\dagger}\Phi_2-\left(m_{12}\Phi_1^{\dagger}\Phi_2+\text{ h.c. }\right)+
2\lambda_3\left(\Phi^{\dagger}_1\Phi_1\right)\left(\Phi_2^{\dagger}\Phi_2\right)\nonumber\\
&+&2\lambda_4\left(\Phi_1^{\dagger}\Phi_2\right)\left(\Phi_2^{\dagger}\Phi_1\right)+
\lambda_1\left(\Phi_1^{\dagger}\Phi_1\right)^2+
\lambda_2\left(\Phi_2^{\dagger}\Phi_2\right)^2,
\end{eqnarray}
the term in $m_{12}$ is a soft symmetry breaking term. Its
introduction prevents the appearence of an would-be Goldstone boson
due to an accidental continuous global symmetry of the  potential,
which arises when the \ac{BGL} symmetry is exact. Namely, in the limit
$m_{12} \rightarrow 0$  the pseudo scalar neutral field $I$ remains
massless. 

It was shown in the literature that a potential with a sofly broken $U(1)$
symmetry does not violate \acs{CP}, neither explicitly nor spontaneously\cite{Branco:2011iw}.
Hence all the parameters of the potential can be made real.
In the absence of \acs{CP} violation the
scalar field $I$ does not mix with the fields $R$ and $H^0$, therefore
$I$ is already a physical Higgs and the mixing of $R$ and $H^0$ is
parametrized by a single  angle.  There are  two important rotations 
that define the two parameters, $\tan \beta$ and $\alpha$, widely used in the literature:
\begin{eqnarray}
 \left( \begin{array}{c} H^0 \\ R  \end{array} \right) 
 =  \frac{1}{v}\left( \begin{array}{rr}  
v_1  & v_2 \\
- v_2 & v_1
\end{array} \right)
\left( \begin{array}{c} \rho_1\\ \rho_2  \end{array} \right)
 =  \left( \begin{array}{cc}  
\cos \beta  & \sin \beta \\
- \sin \beta & \cos \beta 
\end{array} \right)
\left( \begin{array}{c} \rho_1\\ \rho_2  \end{array} \right)
\label{beta}
\end{eqnarray}
where $\rho_{1,2}$ are the $CP$ even states of the $\Phi_{1,2}$ doublets where the \acs{BGL}
symmetry is imposed.
This rotation ensures that the field $H^0$ has flavour conserving couplings to the quarks 
with strength equal to the standard model Higgs  couplings. The other rotation is:  
\begin{eqnarray}
 \left( \begin{array}{c} H \\ h  \end{array} \right) 
 =  \left( \begin{array}{cc}  
\cos \alpha  & \sin \alpha \\
- \sin \alpha & \cos \alpha
\end{array} \right)
\left( \begin{array}{c} \rho_1\\ \rho_2  \end{array} \right)
 \end{eqnarray}
relating  $\rho_1$ and $\rho_2$ to two of the neutral physical Higgs  fields.
The seven  independent real parameters of the Higgs potential $V_\Phi$ will
fix the seven observable quantities, comprising the masses of the
three neutral Higgs, the mass of the charged Higgs,  the combination
$v \equiv \sqrt{v_1^2 + v_2^2} $, $\tan \beta \equiv v_2/v_1$, and $\alpha$.

\section{A contribution for a systematic search for FCNCs}

While some authors develop different software tools to help in the search for new 
physics\cite{flavourkit,database,*MadAnalysis,*superiso,*higgsbounds,*higgssignals,*2HDMC,*checkmate,*DELPHES,*fastlim,*SModelS,*ScannerS,*Lilith,*vevacious}, others prefer at this stage to develop a more transparent approach based on 
transparent formulae to monitor the future improvements on
experimental data and lattice calculations\cite{burastalk}.
In some of the mentioned tools there is duplicated work assumed by the authors
which in some way agree that a more transparent approach is preferable. Since the study of the physics involved 
is our priority, the transparency gained often pays the duplicated work.
That was the main reason why for the study of the BGL models 
described in the next chapter we developed our own software(see the software documentation\cite{bglanalysis}).
Another reason was that when we started the study in 2011, the maturity of the 
available software was not the same it is today.

The results for the type II 2HDM presented in Figure~\ref{fig:corr} are consistent with the ones found in the literature\cite{flavourtype2} 
and so we believe our results are qualitatively correct and roughly quantitatively correct. That is,
the errors which certainly exist should not modify our conclusions which are that 
the CKM's hierarchy makes some BGL models competitive against NFC models(e.g. MSSM) 
when facing flavour physics data. There is no a priori reason to expect a worse performance 
then NFC against LHC data, this must be checked.
We will also study the correlations among the observables, to find interesting patterns.

However, we can and should be more ambitious---but with some care,
see the cartoon~\ref{fig:std}.
Suppose that you want to buy one of two apparently similar cars;
the owner of the first car shows you a complete manual about the car,
but he doesn't let you drive it before you buy it.
The owner of second car does let you drive it, but he doesn't show you
the manual of the car. Assuming that both the manual and the test drive
seemed ok for each car, which car will you pick?

When we are comparing experimental data with theory predictions--- say searching for FCNCs ---,
transparency is crucial. There are too many variables, we have to know what we are doing.
A numerical computation like $y=\sin(x)$ is transparent to us,
despite we do not know how the computer's math library calculates the function
$\sin$ at a generic point---in fact, the implementation is system dependent---
or how to calculate it without a computer.
This is because if we want we can play with the $\sin$ easily---
the input and output are easy to understand as they are related with the well known 
trigonometric function and the program runs fast---so e.g. we make a plot
and check how it looks.
So, transparency of a program for us is all about input and output easy to
understand and the program to be easy to test.

Following FlavorKit\cite{flavourkit}, to test a model against flavour data we need:

\begin{itemize}\itemsep1pt \parskip0pt \parsep0pt
	\item expressions for the masses and couplings of the fields as a function of the parameters of the model
	\item renormalization group equations for the running of the masses and couplings of the fields
        \item expressions for the Wilson coefficients corresponding to the operators of Effective Field Theory
              as a function of the masses and couplings of the fields
        \item expressions for the observables as a function of the Wilson coefficients
        \item simulation covering the model parameter space, comparing the predictions
              and measurements for the observables
\end{itemize}

In FlavorKit these different task are implemented by different modules (see the diagram~\ref{fig:std}).
\begin{figure}[h!]
   \center
   \includegraphics[width=.45\textwidth]{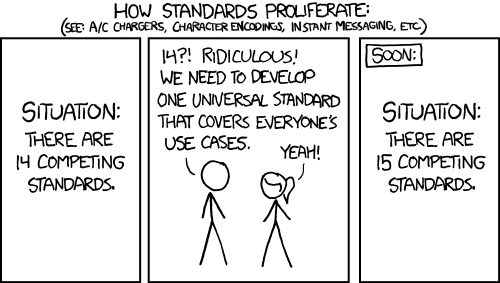}
   \includegraphics[width=.45\textwidth]{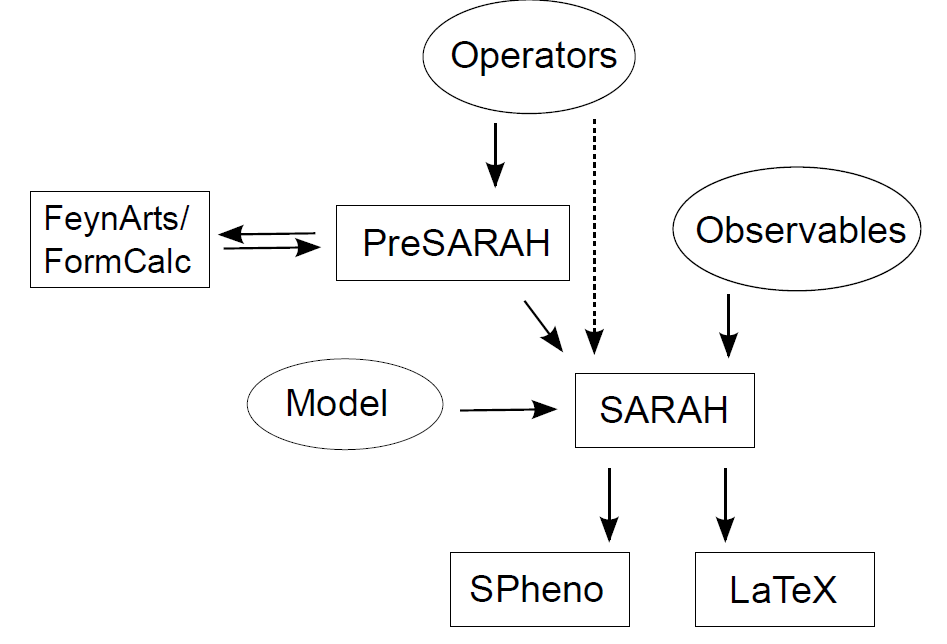}
\caption{Cartoon\cite{standards}. Diagram with Flavourkit components
\cite{flavourkit}.
}
\label{fig:std}
   \end{figure}

All the above mentioned tools are an excellent
starting point and FlavorKit goes in the good direction of a modular solution.
The contribution we do in this section is to try to design a path towards
increased transparency of the software tools. The main obstacle to transparency in some of
these programs is that the input---a model---and the output---expressions or plots which
are function of many parameters---are not easy to understand and often the running time can
be large. Are we proposing to build a whole program which tests models against data whose input/output are not
models/plots? No, of course that if we change the input/output, what the program does cannot not be the same.

Most users of these tools are capable of building their own program to test their models
and often they have to because the available tools despite helpful are not enough for all their needs
(even when the tools are extensible). So the main goal of these tools is to \emph{assist} the users in the task of
building their own programs to test their models. 
Since these tools constitute a collection of resources used to build programs, these tools are 
in fact a library. When seen as a library, these tools are a poorly designed library, which is expected
since they were not designed as a library.

Then a function whose input is a set of Effective Field Theory operators and 
its output is an expression for an observable as a function of the corresponding Wilson coefficients, 
can be made transparent. That is its input/output can be easy to handle
if the library includes the capability to manipulate symbolic expressions.
We can decompose some of the mentioned tools into many such transparent functions.
The user can then include these functions in his program, if guided by a good manual hopefully
with a lot of physical content which will give him the physical insight of what is going on.

As an example, Figure~\ref{fig:caller} illustrates that tools such as the \ac{BGL} analysis tools\cite{bglanalysis}
can be decomposed in several functions with input/output which is easy to handle.
\begin{figure}[h!]
   \center
   \includegraphics[width=\textwidth]{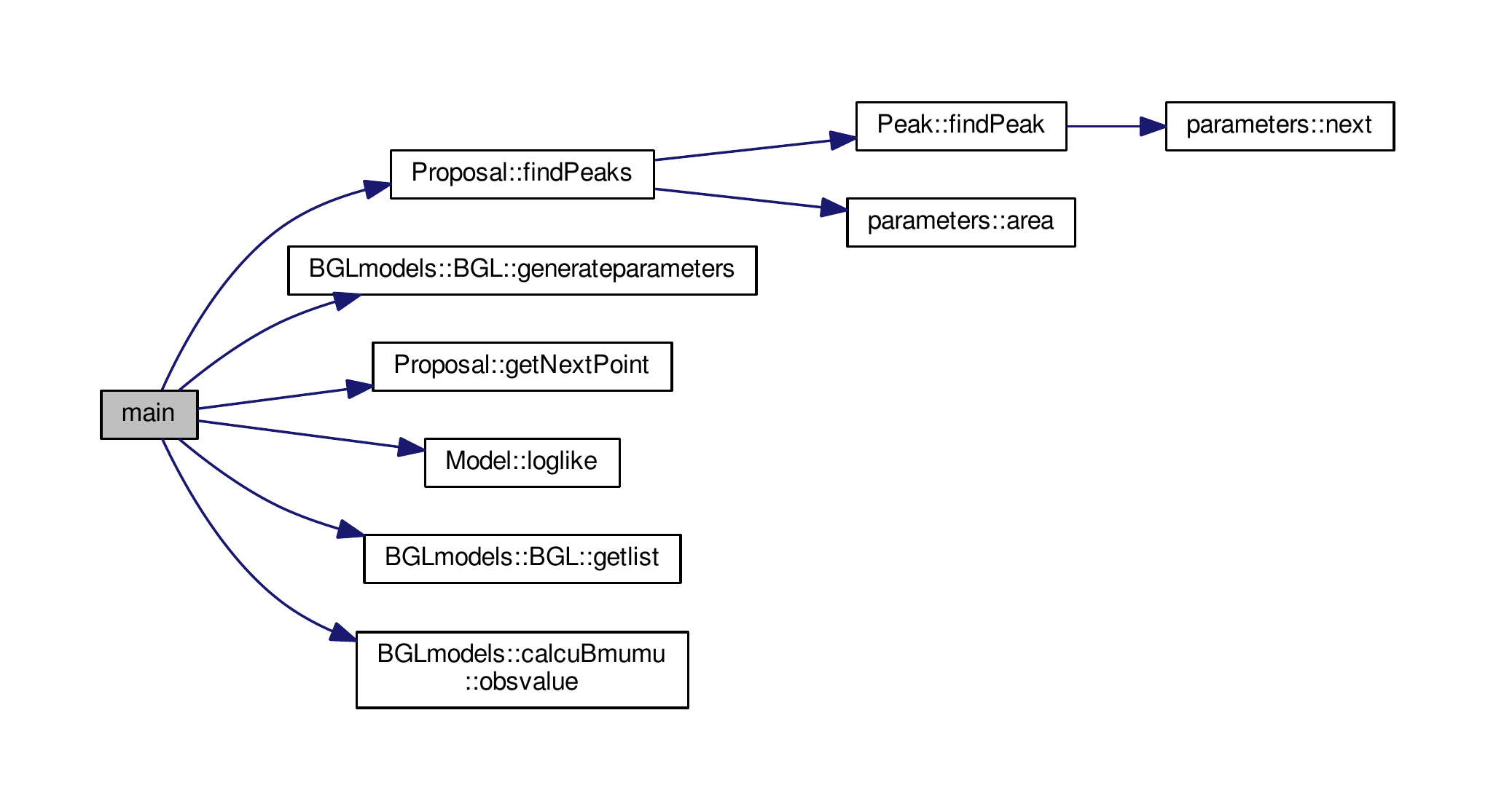}
\caption{Caller graph for the \emph{main} function of the \ac{BGL} analysis tools\cite{bglanalysis}, generated by Doxygen.}
\label{fig:caller}
      \end{figure}

The capability to manipulate symbolic expressions must be part of the library.
While there are many programs with symbolic capabilities, there are not so many libraries
that can be useful in this context but still there are---Ginac\cite{ginac} and Giac\cite{giac} (C++ algebra systems)
are the best options I am aware of. 
Moreover, there is a lot of room for improvement of GiNaC, by combining it with Giac and also
using LLVM\cite{llvm,*llvm2} the code generation of GiNaC and some symbolic capabilities
can be much improved. LLVM stands for Low Level Virtual Machine and it is basically a C++ library
providing a representation of a symbolic language well suited for computer-like operations such as logic
and numerical computations. This symbolic language can be extended with the GiNaC and Giac symbolic
capabilities resulting in an efficient and flexible library.
As an example, we show below the function which in the \ac{BGL} analysis\cite{bglanalysis} adds a new symbolic
expression to the list of constraints, with machine code generation for efficiency:

\begin{scriptsize}
\begin{lstlisting}[language=C++]
void add(const char * s, ex pred, observable * ob, bool sb=0){

	ex p=pred.subs(replacements).real_part();
	p=collect_common_factors(expand(p.evalf()));
	FUNCP_CUBA fp;
	lst l(tanb,McH,MR,MI);	
	
	for(uint i=0;i<3;i++){
		l.append(Mu[i]);
		l.append(Md[i]);
		}
	if(sb) push_back(prediction(ob,p));
	else {
		compile_ex(lst(p), l, fp);
	  	push_back(prediction(ob,fp));
		}
	}
\end{lstlisting}
\end{scriptsize}

The function above takes a symbolic expression for the prediction of an observable as an input (ex pred), optionally compiles it to machine code, and adds it to the list of observables to be calculated in the simulation.

We also need a library containing many known formulas for decays important for FCNC
as a function of the Wilson coefficients, these libraries are appearing with FlavourKit an example;
libraries making global fits can be found e.g. in ROOSTATS\cite{root,*roostats}, 
We need also libraries describing models, formulas relating the model parameters and Wilson coefficients,
library containing the experimental data distributions, these libraries are also appearing\cite{database}.

Due to the high investment of CERN in the C++ language, with the recent development of a LLVM-based
efficient interpreter Cling for the ROOT framework\cite{root}, also the availability of C++ symbolic libraries, 
it makes sense that the language connecting these libraries can be C++, that is, all the libraries should
have a C++ interface.
Note that a Mathematica package can be used as a C++ library (if Mathematica is installed).

What I am proposing is not an utopia, as the ROOT and ROOSTATS framework examples show it is possible that a collection
of libraries to be useful for many different experimental physicists, why not also for theoretical physicists?
What is clear is that it helps that an important institution backs the project, but with CERN investing in 
the LHC's open data project\cite{opening} and projecting more experiments\cite{futureexperiments} certainly there will be interest in a framework which can help
to give physical meaning to all this data.

In the following we present the classes of the \ac{BGL} analysis tools with brief descriptions, illustrating what kind of
libraries may be developed to aid in these kind of analysis. The complete source code presented as a Doxygen manual is online\cite{bglanalysis}. 

\begin{description}\itemsep1pt \parskip0pt \parsep0pt
\item[BGL] Implementation of the BGL model
\item[Boson] Gauge/Higgs boson properties
\item[Fermion] Fermion properties
\item[Meson] Meson properties
\item[calcu] Base class to do the calculus of a constraint to the model
\item[calcuba] Class to do the calculus of a constraint based on a GiNaC compiled expression
\item[calcuBmumu] Calculus of the constraints coming from the B- > mu mu decay
\item[calcubtosgamma2] Calculus of the constraints coming from the b- > s gamma decay
\item[calcuex] Class to do the calculus of a constraint based on a GiNaC symbolic expression
\item[calcuOblique] Calculus of the constraints coming from the oblique parameters
\item[discreteparameter] A parameter which will be fitted in the simulation
\item[freeparameter] A parameter which will be fitted in the simulation
\item[gauss2obs] Same as gaussobs but with a different initializer, such that the uncertainty sigma is absolute
\item[gaussobs] An experimental measure of a parameter which is a mean value and a standard deviation
\item[limitedobs] An experimental measure which is an upper limit on a parameter with a given Confidence Level
\item[Matrixx] Class to represent the mixing matrices VCKM and VPMNS
\item[measure] A class containing the value and uncertainty of an experimental measure
\item[Mixes] Definition of the couplings for the different BGL models
\item[Model] Abstract class for a model
\item[multivector < T, N >] A vector of vectors of vectors of... (N times) of class T objects
\item[multivector < T, 1 >] Specialization template class of multivector < T,N > for N=1
\item[observable] A base class representing an experimental measure
\item[parameters] Vector of parameters
\item[Peak] A class containing the parameters of a maximum of the likelihood function
\item[prediction] Theoretical expression for an experimental measure
\item[Proposal] A class containing the parameters of a proposal for the next step in the Markov Chain
\item[widthcalc] This class calculates decay widths of one lepton to 3 leptons
\end{description}

Note that the purpose of this section is to convince other people interested in doing phenomenological studies in Flavour Physics that a collaborative approach based on a library using a C++ symbolic algebra system as a glue for the different modules is possible and needed, with part of the work needed for some modules already done in the library developed for the BGL analysis. If the reader is simply looking for a ready-to-use package for phenomenological studies then there are many alternatives available\cite{flavourkit,database,*MadAnalysis,*superiso,*higgsbounds,*higgssignals,*2HDMC,*checkmate,*DELPHES,*fastlim,*SModelS,*ScannerS,*Lilith,*vevacious}
\clearpage{}
\cleartooddpage
\clearpage{}\chapter{Physical constraints on the BGL models}

\begin{epigraphs}
\qitem{This gives a unique character to the work of Branco, Grimus, and Lavoura.
They have developed the only
possible implementation of a relation between FCNSI\emph{[Flavour Changing Neutral Scalar Interactions]} 
and the CKM matrix which uses abelian symmetries and is
consistent with the sufficient conditions above.\emph{[...]}
In light of our analysis, that a BGL case was found by inspection in the 
THDM\emph{[Two Higgs Doublet Model]} is truly remarkable.}
{---\emph{P. M. Ferreira \& Joao P. Silva (2010)\cite{abelian}}}

\qitem{The so-called BGL models, proposed in \emph{\cite{Branco:1996bq}}, is a class
of two-Higgs doublet models where the strength of FCNCs in the up- or down-type sector
is unambiguously related to the off-diagonal elements of the CKM matrix. While all the
six BGL models are interesting, only one of them is compatible with the MFV principle.
 \emph{[...]} only the BGL model where  $d_i\to d_j$ FCNC transitions are proportional to 
$V_{3i}^* V_{3j}$ is an explicit example of MFV.\emph{[...]}

More precisely, this
framework coincides with the MFV construction in the limit $m^2_{c,u}/m^2_t \to 0$, which is
an excellent approximation.}{---\emph{A. Buras, M. Carlucci, S. Gori \& G. Isidori (2010)\cite{nfcmfv}}}

\qitem{
Two noteworthy features which
distinguish the t-type BGL model from others are: (i) the $\mu\nu$ final state dominates over $\tau \nu$ for $\tan\beta > 5$,
which is a distinctive characteristic of t-type BGL model unlike any of the Type I, II, X or Y models (due to family
nonuniversal BGL Yukawa couplings); (ii) for $\tan\beta> 10$, the branching ratio into $cs$ significantly dominates
over other channels including $tb$, again a unique feature of t-type BGL.\emph{[...]}

In other types of 2HDM, the bb and $\tau\tau$ final states dominate over cc and
$\mu\mu$ channels, respectively. Here,the hierarchy is reversed,which transpires from the expressions of $N_d$ and
$N_u$\emph{[...]}

The feature that makes our scenario unique is
the possibility of their \emph{[the Higgs bosons]} relative lightness as well as unconventional decay signatures.}
{---\emph{G. Bhattacharyya, D. Das \& A. Kundu (2014)\cite{Bhattacharyya:2014nja}}}
\end{epigraphs}

In this chapter, we analyse the experimental constraints on BGL type models defined in the previous chapter and discuss 
some of their phenomenological implications. 
In the next section, we explain the profile likelihood method
used in our analysis, the input data and settle the notation. In the second
section, we analyse the constraints on BGL models, derived from
experiment. Finally, in section \ref{SEC:Results} we present our results and discuss them.
 
\clearpage  
\section{Analysis details\label{AP:Analysis}\label{AP:Input}}
In this work we only consider explicitly scenarios with Dirac type neutrinos,
where no Majorana mass terms are added to the Lagrangian. However, our 
analysis of the experimental implications does not depend on the
nature of the neutrinos, i.e., Majorana or Dirac. Therefore, our
conclusions can be extended to the case of neutrinos being Majorana
fermions provided that deviations from unitarity of the $3 \times 3$
low energy leptonic mixing matrix are negligible, as it is the case in
most seesaw models. 

In our
analysis we use the current limits on Higgs masses, identifying one of 
the Higgs with the one that was discovered by ATLAS and CMS. We make
the approximation of no mixing between $R$ and $H^0$ identifying $H^0$ with 
the recently discovered Higgs and $R$ and $I$ with the additional physical 
neutral Higgs fields. This limit corresponds to $\beta - \alpha = \pi/2 $ and
with this notation $H^0$ coincides with $h$, which is the usual choice in the
literature. This approximation is justified by the fact that
the observed Higgs boson seems to behave as a standard-like Higgs
particle. The quantity $v$ is of course already fixed by experiment.
Electroweak precision tests and, in particular the $T$ and $S$
parameters, lead to constraints relating  the masses of the new Higgs
fields among themselves.  Therefore the bounds on $T$ and $S$, together 
with  direct mass limits, significantly restrict the masses of the new Higgs particles, 
once  the mass of $H^\pm$ is fixed. In our analysis we study BGL type models by combining 
the six possible implementations of the quark sector with the six
implementations of the leptonic sector. It is illustrative to plot our
results in terms of $m_{H^{\pm}}$ versus $\tan \beta$, since, as explained above
in the context of our approximation of no mixing between $R$ and
$H^0$, there is not much freedom left. Therefore with these two
parameters we may approximately scan the whole region of parameter
space. In our analysis, the presentation of our results will reflect that fact despite we scan over all R, I, $H^+$ masses.
We impose present constraints from several
relevant flavour observables, as specified in the next section.

In tables \ref{TAB:AP:MixingMatrices}, \ref{TAB:AP:TreeCharged}, \ref{TAB:AP:TreeNeutral}, \ref{TAB:AP:Loop} and \ref{TAB:AP:Misc} we collect relevant input used in the analysis,
the notation is explained in next section.

In figures \ref{fig:onea}, \ref{FIG:Results01} and \ref{FIG:Results02} we have presented 68\%, 95\% and 99\% CL allowed regions in parameter space. To wit, we represent regions where the specific BGL model is able to fit the imposed experimental information at least as well as the corresponding goodness levels. Some comments are in order. This procedure corresponds to the profile likelihood method \cite{asymptotic}. In brief, for a model with parameters $\vec p$, we compute the predictions for the considered set of observables $\vec O_{\mathrm{Th}}(\vec p)$. Then, using the experimental information $\vec O_{\mathrm{Exp}}$ available for those observables, we build a likelihood function $\mathcal L(\vec O_{\mathrm{Exp}}|\vec O_{\mathrm{Th}}(\vec p))$ which gives the probability of obtaining the experimental results $\vec O_{\mathrm{\mathrm{Exp}}}$ assuming that the model is correct. The likelihood function $\mathcal L(\vec O_{\mathrm{Exp}}|\vec O_{\mathrm{Th}}(\vec p))$ encodes all the information on how the model is able to reproduce the observed data all over parameter space. Nevertheless, the knowledge of $\mathcal L(\vec O_{\mathrm{Exp}}|\vec O_{\mathrm{Th}}(\vec p))$ in a multidimensional parameter space can be hardly represented and one is led to the problem of reducing that information to one or two-dimensional subspaces. In the profile likelihood method, for each point in the chosen subspace, the highest likelihood over the complementary, marginalized space, is retained. Let us clarify that likelihood -- or chi-squared $\chi^2\equiv -2\log \mathcal L$ -- profiles and derived regions such as the ones we represent, are thus insensitive to the size of the space over which one marginalizes; this would not be the case in a Bayesian analysis, where an integration over the marginalized space is involved. The profile likelihood method seems adequate to our purpose, which is none other than exploring where in parameter space are the different BGL models able to satisfy experimental constraints, without weighting in eventual fine tunings of the models or parameter space volumes. For the numerical computations the libraries GiNaC \cite{ginac} and ROOT \cite{root} were used.

There are two types of experimental results: the measures and the
upper limits. The contribution of the measures to $\chi^2$ is
$(\frac{r-p}{\sigma})^2$, where $r$ and $\sigma$ are the mean value
and uncertainty of the measure and $p$ is the prediction of the
model. $\sigma$ also includes the part of the
uncertainty of the prediction which is assumed to be uncorrelated with
the other predictions such as the truncation errors. The contribution of the
upper limits to $\chi^2$ is $(\frac{p-\sqrt{\pi}\rho}{2\rho})^2$, where $\rho$ is
such that the correspondent Gaussian cumulative distribution function at 
the upper limit equals the Confidence Level.

Through the generation of a large
enough set of pseudo-experiments we could construct numerically the
statistical distribution of $\chi^2$ if needed.
According to Wilks theorem, $\chi^2$
asymptotically follows a chi squared distribution, with the degrees of
freedom equal to the number of observables. Since our
purpose is only to take the qualitative conclusion on whether the BGL
models can describe better the data than the Standard Model,
considering that $\chi^2$ follows a chi squared distribution is a good
enough approximation.
We are considering as free parameters: the $tan(\beta)$, and the masses of
the Higgs bosons $H^+$, $R$, $I$.

Since the $cos(\theta_W)$ measurement is done with the muon decay width, we are not using 
it's value, we are using the direct measurement of the W mass to calculate it. 
Then we compare the predicted muon decay width with the measured one, like we do with the other decay widths.

We are considering that the corrections
introduced by the BGL model when compared with the SM, are small
enough so that the values of the mixing matrices (\ac{CKM} and \ac{PMNS}) in
some region of parameters of the BGL models are not significantly different from the ones of the mixing
matrices obtained when using the Standard Model. This assumption is a
posteriori justified by the results we obtain, since we can describe
the experimental data for a large region of the parameter space of the BGL models using such hypothesis,
and no significant deviations in that region are found with respect with the Standard Model.
This approach would not be
necessarily valid if we obtained allowed regions in BGL models
making significantly different predictions than the Standard Model, for instance
if we could explain the anomalous B decays measured in BABAR experiment with some BGL model.
\begin{table}
\begin{center}
\begin{tabular}{|c|c||c|c|}
\hline    $\lambda$ & $0.22535(65)$ & $A$ & $0.811(22)$\\ 
\hline    $\bar{\rho}$ & $0.131(26)$ & $\bar{\eta}$ & $0.345(14)$\\ 
\hline\hline    $\sin^2 \theta_{12}$& $0.320(16)$ & $\sin^2 \theta_{23}$& $0.613(22)$\\ 
\hline    $\sin^2 \theta_{13}$& $0.0246(29)$\\ 
\cline{1-2}
\end{tabular}
\caption{Input for the CKM and PMNS mixing matrices \cite{pdg}.\label{TAB:AP:MixingMatrices}}
\end{center}
\end{table}

\begin{table}
\begin{center}
\begin{tabular}{|c|c||c|c|}
\hline
    $\abs{g_\mu/g_e}^2$ & $1.0018(14)$ & $|g_{RR,\tau \mu}^{S}|$ & $<0.72$\\ 
    $|g_{RR,\tau e}^{S}|$ & $<0.70$ & $|g_{RR,\mu e}^{S}|$ & $<0.035$\\
\hline   
    $\Br(B^+ \to e^+ \nu)$ & $<9.8\cdot 10^{-7}$ & $\Br(D^+_s \to e^+ \nu)$ & $<1.2\cdot 10^{-4}$ \\
    $\Br(B^+ \to \mu^+ \nu)$ & $< 1.0\cdot 10^{-6}$ &  $\Br(D^+_s \to \mu^+ \nu)$ & $5.90(33)\cdot 10^{-3}$\\
    $\Br(B^+ \to \tau^+ \nu)$ & $1.15(23)\cdot 10^{-4}$ &  $\Br(D^+_s \to \tau^+ \nu)$ & $5.43 (31)\cdot 10^{-2}$\\
\hline
    $\Br(D^+ \to e^+ \nu)$ & $<8.8\cdot 10^{-6}$ \\
    $\Br(D^+ \to \mu^+ \nu)$ &  $3.82(33)\cdot 10^{-4}$ \\
    $\Br(D^+ \to \tau^+ \nu)$ & $<1.2\cdot 10^{-3}$\\ 
\hline
    $\frac{\Gamma(\pi^+\to e^+\nu)}{\Gamma(\pi^+\to \mu^+\nu)}$ & $1.230(4)\cdot 10^{-4}$ & $\frac{\Gamma(\tau^-\to \pi^-\nu)}{\Gamma(\pi^+\to \mu^+\nu)}$ & $9703(54)$\\
    $\frac{\Gamma(K^+\to e^+\nu)}{\Gamma(K^+\to \mu^+\nu)}$ & $2.488(12)\cdot 10^{-5}$ & $\frac{\Gamma(\tau^-\to K^-\nu)}{\Gamma(K^+\to \mu^+\nu)}$ &
    $469(7)$\\
\hline
$\frac{\Gamma(B\to D\tau\nu)_{\mathrm{NP}}}{\Gamma(B\to D\tau\nu)_{\mathrm{SM}}}$ & & $\log C$ ($K\to\pi\ell\nu$) & $0.194(11)$\\
\cline{3-4}
$\frac{\Gamma(B\to D^\ast\tau\nu)_{\mathrm{NP}}}{\Gamma(B\to D^\ast\tau\nu)_{\mathrm{SM}}}$ & \\
\cline{1-2}
\end{tabular}
\caption{Constraints on processes mediated at tree level by $H^\pm$ -- section \ref{SEC:ExpConst-sSEC:TreeCharged} --, bounds are given at 90\% CL, except the first set of four which is at 90\%CL.
\label{TAB:AP:TreeCharged}}
\end{center}
\end{table}

\begin{table}
\begin{center}
\begin{tabular}{|c|c||c|c|}
\hline
    $\Br(\tau^- \to e^-e^-e^+)$ & $<2.7\cdot 10^{-8}$ & $\Br(\tau^- \to \mu^-\mu^-\mu^+)$ & $<2.1\cdot 10^{-8}$\\ 
    $\Br(\tau^- \to e^-e^-\mu^+)$ & $<1.5\cdot 10^{-8}$ & $\Br(\tau^- \to e^-\mu^-e^+)$ & $<1.8\cdot 10^{-8}$\\ 
    $\Br(\tau^- \to \mu^-\mu^-e^+)$ & $<1.7\cdot 10^{-8}$ & $\Br(\tau^- \to \mu^-e^-\mu^+)$ & $<2.7\cdot 10^{-8}$\\ 
\cline{3-4}
    $\Br(\mu^- \to e^-e^-e^+)$ & $<1\cdot 10^{-12}$\\ 
\hline
    $2|M_{12}^K|$ & $<3.5 \cdot 10^{-15}$ \GeV & $2|M_{12}^D|$ & $<9.47\cdot 10^{-15}$ \GeV \\
\cline{3-4}
    $|\epsilon_K|_{NP}\Delta m_K$ & $<7.8 \cdot 10^{-18}$ \GeV \\
\hline
    $\re(\Delta_d)$ & $0.823(143)$ & $\re(\Delta_s)$ & $0.965(133)$ \\
    $\im(\Delta_d)$ & $-0.199(62)$ & $\im(\Delta_s)$ & $0.00(10)$\\
\hline
    $\Br(K_L \to \mu^\pm e^\mp)$ & $<4.7\cdot 10^{-12}$ & $\Br(\pi^0 \to \mu^\pm e^\mp)$ & $<3.6\cdot 10^{-10}$\\
\cline{3-4} 
    $\Br(K_L \to e^-e^+)$ & $<9\cdot 10^{-12}$ \\ 
    $\Br(K_L \to \mu^- \mu^+)$ & $<6.84\cdot 10^{-9}$\\ 
\hline
    $\Br(D^0 \to e^- e^+)$ & $< 7.9\cdot 10^{-8}$ & $\Br(B^0 \to e^+ e^-)$ & $<8.3\cdot 10^{-8}$ \\
    $\Br(D^0 \to \mu^\pm e^\mp)$ & $<2.6\cdot 10^{-7}$ & $\Br(B^0 \to \tau^\pm e^\mp)$ & $<2.8\cdot 10^{-5}$ \\
    $\Br(D^0 \to \mu^- \mu^+)$ & $<1.4\cdot 10^{-7}$ & $\Br(B^0 \to \mu^- \mu^+)$ & $3.6(1.6)\cdot 10^{-10}$ \\
\cline{1-2}
    $\Br(B_s^0 \to e^+ e^-)$ & $<2.8\cdot 10^{-7}$ & $\Br(B^0 \to \tau^\pm \mu^\mp)$ & $<2.2\cdot 10^{-5}$ \\
    $\Br(B_s^0 \to \mu^\pm e^\mp)$ & $<2\cdot 10^{-7}$ & $\Br(B^0 \to \tau^+ \tau^-)$ & $<4.1\cdot 10^{-3}$ \\
\cline{3-4}
    $\Br(B_s^0 \to \mu^- \mu^+)$ & $2.9(0.7)\cdot 10^{-9}$\\
\cline{1-2}
\end{tabular}
\caption{Constraints on processes mediated at tree level by $R$, $I$ -- section \ref{SEC:ExpConst-sSEC:TreeNeutral} --, bounds are given at 90\% CL.\label{TAB:AP:TreeNeutral}}
\end{center}
\end{table}

\begin{table}
\begin{center}
\begin{tabular}{|c|c||c|c|}
\hline    
$\Br(\mu \to e\gamma)$ & $<2.4\cdot 10^{-12}$ & $\Br(B \to X_s \gamma)_{\mathrm{SM}}^{\mathrm{NNLO}}$ & $3.15(23)\cdot 10^{-4}$\\ 
    $\Br(\tau \to e\gamma)$& $<3.3\cdot 10^{-8}$ & $\Br(B \to X_s \gamma)$ & $3.55(35)\cdot 10^{-4}$\\
\cline{3-4}
    $\Br(\tau \to \mu\gamma)$ & $<4.4\cdot 10^{-8}$ \\ 
\hline
    $\Delta T$ & $0.02(11)$ & $F_{Zb\bar b}$ & $<0.0024$ \GeV$^{-1}$ \\
\cline{3-4} 
    $\Delta S$ & $0.00(12)$ \\
\cline{1-2}
\end{tabular}
\caption{Constraints on processes mediated by $R$, $I$, $H^\pm$ at loop level -- section \ref{SEC:ExpConst-sSEC:Loop} --, bounds are given at 90\% CL.\label{TAB:AP:Loop}}
\end{center}
\end{table}

\begin{table}
\begin{center}
\begin{tabular}{|c|c||c|c||c|c|}
\hline $f_{\pi}$ & $0.132(2)$ \GeV & $f_{K}$ & $0.159(2)$ GeV & $f_{D}$ & $0.208(3)$ \GeV\\ 
\hline $f_{D_s}$ & $0.248(3)$ \GeV & $f_{B}$ & $0.189(4)$ GeV & $f_{B_s}$ & $0.225(4)$ \GeV\\ 
\hline $\delta_{\pi^+}$ & $-0.036419(78)$ & $\delta_{K^+}$ & $-0.03580(39)$ & $\delta_{\tau\pi}$ & $0.0016(14)$\\
\hline $\delta_{\tau K}$ & $0.0090(22)$ & $\Delta_{\chi PT}$ & $-3.5(8)\cdot 10^{-3}$ & $f_+^{K\pi}$ & $0.965(10)$\\
\hline
\end{tabular}
\caption{Additional theoretical input -- lattice, radiative corrections -- \cite{lattice,kaon,taucorrections,Pich:2010,Massesrunning}.\label{TAB:AP:Misc}}
\end{center}
\end{table}

\section{Confronting experimental results\label{SEC:ExpConst}}

\subsection{Generalities\label{SEC:ExpConst-sSEC:Gen}}
In the class of 2HDM considered in this chapter, the Yukawa interactions of the new scalars may produce new contributions, at tree and at loop level, that modify the SM predictions for many processes for which experimental information is available. As is customary, this will allow us to study the viability and interest of the different cases within this class of models. In terms of the New Physics (NP) and the SM leading contributions, one can organize the processes to be considered as follows.
\begin{itemize}
\item Processes with tree level NP contributions mediated by $H^\pm$ and SM tree level contributions $W^\pm$-mediated, as, for example, universality in lepton decays, leptonic and semileptonic decays of mesons like $\pi\to e \nu$, $B\to\tau\nu$ and $B\to D\tau\nu$, or $\tau$ decays of type $\tau\to M\nu$.
\item Processes with tree level NP contributions mediated by the neutral scalars $R$, $I$, and
\begin{itemize}
\item loop level SM contributions as in, for example, $K_L\to\mu^+\mu^-$, $B_s\to\mu^+\mu^-$, and $B^0\rightleftarrows\bar B^0$ oscillations,
\item highly suppressed (because of the smallness of the neutrino masses) loop level SM contributions as in, for example, $\tau^-\to\mu^-\mu^-\mu^+$ or $\mu^-\to e^-e^-e^+$.
\end{itemize}
\item Processes with loop level NP contributions and
\begin{itemize}
\item loop level SM contributions as in, for example, $B\to X_s\gamma$,
\item highly suppressed (here too because of the smallness of the neutrino masses) loop level SM contributions as in, for example, $\tau\to\mu\gamma$ or $\mu\to e\gamma$.
\end{itemize}
\end{itemize}

Besides those observables, electroweak precision information -- $Z\to b\bar b$ and the oblique parameters $S$, $T$ -- are also relevant; they involve loop level contributions from the new scalars.

Table \ref{TAB:summary} summarizes this classification of the potentially relevant observables. Notice however that the table signals the possible new contributions but for each specific model type, some of them will be absent. More detailed descriptions of each type of constraint are addressed in the following subsections. Since we focus in the flavour sector, we exclude from the analysis of the experimental implications of the BGL models processes that probe additional couplings related to the scalar potential, such as $H^0\to\gamma\gamma$, central in the Higgs discovery at the LHC, and refer the interested reader to \cite{Bhattacharyya:2013rya}.

\newcommand{\nchck}{{\checkmark}}
\newcommand{\grchck}{{\footnotesize\color{gray}\nchck}}
\begin{table}[h] 
\begin{center}
\begin{tabular}{c|cc|cc|cc|}
\cline{2-7} & \multicolumn{4}{|c|}{BGL - 2HDM} & \multicolumn{2}{||c|}{SM}\\ 
\cline{2-7} & \multicolumn{2}{|c|}{Charged $H^\pm$} & \multicolumn{2}{|c|}{Neutral $R$, $I$} & \multicolumn{1}{||c|}{\multirow{2}{*}{Tree}} & \multicolumn{1}{|c|}{\multirow{2}{*}{Loop}}\\
\cline{2-5} & \multicolumn{1}{|c|}{Tree} & \multicolumn{1}{|c|}{Loop} & \multicolumn{1}{|c|}{Tree} & \multicolumn{1}{|c|}{Loop} & \multicolumn{1}{||c|}{} & \multicolumn{1}{|c|}{} \\
\hline\multicolumn{1}{|c|}{$M\to\ell\bar\nu,M^\prime\ell\bar\nu$} & \multicolumn{1}{|c|}{\nchck} & \grchck & \multicolumn{1}{|c|}{} & \multicolumn{1}{|c|}{\grchck} & \multicolumn{1}{||c|}{\nchck} & \multicolumn{1}{|c|}{\grchck}\\
\hline\multicolumn{1}{|c|}{Universality} & \multicolumn{1}{|c|}{\nchck} & \grchck & \multicolumn{1}{|c|}{} & \multicolumn{1}{|c|}{\grchck} & \multicolumn{1}{||c|}{\nchck} & \multicolumn{1}{|c|}{\grchck}\\
\hline\hline\multicolumn{1}{|c|}{$M^0\to\ell_1^+\ell_2^-$} & \multicolumn{1}{|c|}{} & \grchck & \multicolumn{1}{|c|}{\nchck} & \multicolumn{1}{|c|}{\grchck} & \multicolumn{1}{||c|}{} & \multicolumn{1}{|c|}{\nchck}\\
\hline\multicolumn{1}{|c|}{$M^0\rightleftarrows \bar M^0$} & \multicolumn{1}{|c|}{} & \grchck & \multicolumn{1}{|c|}{\nchck} & \multicolumn{1}{|c|}{\grchck} & \multicolumn{1}{||c|}{} & \multicolumn{1}{|c|}{\nchck}\\
\hline\multicolumn{1}{|c|}{$\ell_1^-\to\ell_2^-\ell_3^+\ell_4^-$} & \multicolumn{1}{|c|}{} & \grchck & \multicolumn{1}{|c|}{\nchck} & \multicolumn{1}{|c|}{\grchck} & \multicolumn{1}{||c|}{} & \multicolumn{1}{|c|}{\grchck}\\
\hline\hline\multicolumn{1}{|c|}{$B\to X_{s}\gamma$} & \multicolumn{1}{|c|}{} & \nchck & \multicolumn{1}{|c|}{} & \multicolumn{1}{|c|}{\nchck} & \multicolumn{1}{||c|}{} & \multicolumn{1}{|c|}{\nchck}\\
\hline\multicolumn{1}{|c|}{$\ell_j\to \ell_i\gamma$}  & \multicolumn{1}{|c|}{} & \nchck  & \multicolumn{1}{|c|}{} & \multicolumn{1}{|c|}{\nchck} & \multicolumn{1}{||c|}{} & \multicolumn{1}{|c|}{\grchck}\\
\hline\hline\multicolumn{1}{|c|}{EW Precision} & \multicolumn{1}{|c|}{} & \nchck & \multicolumn{1}{|c|}{} & \multicolumn{1}{|c|}{\nchck} & \multicolumn{1}{||c|}{} & \multicolumn{1}{|c|}{\nchck}\\
\hline
\end{tabular}
\end{center}
\caption{Summary table of the different types of relevant observables; leading contributions are tagged $\nchck$ while subleading or negligible ones are tagged $\grchck$.\label{TAB:summary}}
\end{table}

The set of observables that we consider is sufficient to obtain significant constraints for the masses of the new scalars and $\tan\beta$. Notice that, since the new contributions will be typically controlled by these masses, $\tan\beta$ and the mixing matrices, with no additional parameters, we need fewer observables than would be necessary in the analysis of a more general 2HDM such as the one presented in \cite{Crivellin:2013wna}. 

Apart from the previous flavour related  observables, direct searches at colliders may be relevant. For instance, a charged Higgs decaying to $\tau^+\nu$ or $c\bar{s}$ with a mass lighter than $80$ \GeV ~was excluded\footnote{For all BGL models, in the parameter space not excluded by the previous observables, the branching ratio for the decays $H^\pm \to \tau^+\nu$ or $H^\pm \to c\bar{s}$ is $>96\%$ and thus the bound applies.}, in the context of 2HDM, at LEP \cite{LEP}. However, we do not include direct searches at colliders since the kind of analysis required goes beyond the scope of this work. As a side benefit, we are then able to explore which BGL models may be probed at colliders, in particular at the LHC, and check if flavour constraints allow light charged Higgs masses.

In the next subsections we describe in detail the different types of observables introduced above.

\subsection{Processes mediated by charged scalars at tree level\label{SEC:ExpConst-sSEC:TreeCharged}}
Since transitions mediated within the SM by a $W$ boson may receive new $H^\pm$ mediated contributions, one has to pay attention to:
\begin{itemize}
\item universality tests in pure leptonic decays $\ell_1\to\ell_2\nu\bar\nu$,\item leptonic decays of pseudoscalar mesons $M\to \ell\nu$,
\item semileptonic decays of pseudoscalar mesons $M\to M^\prime\ell\nu$,
\item $\tau$ decays of the form $\tau\to M\nu$.
\end{itemize}

\subsubsection{Universality}
 Pure leptonic decays $\ell_1\to\ell_2\nu\bar\nu$ are described by the following effective Lagrangian
\begin{multline}
{\mathcal L}_{\rm eff}=-\frac{4 G_F}{\sqrt{2}}\times \\ \sum_{\ell_\alpha,\ell_\beta=e,\mu,\tau}\sum_{i,j=1}^3 U^\ast_{\ell_\alpha\nu_i}U_{\ell_\beta \nu_j}
\left\{\left[\bar \nu_i\gamma^\mu \gamma_L \ell_\alpha\right] 
\left[\bar \ell_\beta\gamma_\mu\gamma_L\nu_j\right]+  g^{\nu_i\ell_\alpha \nu_j\ell_\beta}\left[\bar \nu_i\gamma_R \ell_\alpha\right] \left[\bar \ell_\beta\gamma_L\nu_j\right]\right\}.\label{eq:L:univ:lept}
\end{multline}
The second operator in \eq{eq:L:univ:lept} is the new contribution mediated by $H^\pm$. The coefficient $g^{\nu_i\ell_\alpha \nu_j\ell_\beta}$ depends on the specific BGL model:
\begin{equation}
g^{\nu_i\ell_\alpha \nu_j\ell_\beta}=-\frac{m_{\ell_\alpha} m_{\ell_\beta}}{m_{H^+}^2}C^{\nu_i\ell_\alpha}C^{ \nu_j\ell_\beta}\,,
\end{equation}
where, $C^{\nu_i\ell_\alpha}=-1/\tan\beta$ for models of types $\nu_i$ and $\ell_\alpha$ and $C^{\nu_i\ell_\alpha}=\tan\beta$ otherwise -- this concerns the lepton label of the model, the quark one is irrelevant here.
Following the notation in \cite{Pich:1995,Pich:2010}, we then have
\begin{align}
&\abs{g_{RR,\ell_\alpha \ell_\beta}^{S}}^2\equiv\sum_{i,j=1}^3|U_{\ell_\alpha \nu_i}|^2|U_{\ell_\beta \nu_j}|^2(g^{\nu_i\ell_\alpha \nu_j\ell_\beta})^2\,,\\
&\abs{g_{LL,\ell_\alpha \ell_\beta}^{V}}^2\equiv 1\,,\\
&\left(g_{RR,\ell_\alpha \ell_\beta}^{S}\right) \left(g_{LL,\ell_\alpha \ell_\beta}^{V}\right)^*\equiv\sum_{i,j=1}^3|U_{\ell_\alpha \nu_i}|^2|U_{\ell_\beta \nu_j}|^2 g^{\nu_i\ell_\alpha \nu_j\ell_\beta}\,.
\end{align}
We consider for example universality in $\tau$ decays,
\begin{equation}
\abs{\frac{g_\mu}{g_e}}^2\equiv \frac{\text{Br}\left(\tau\to \mu\nu\bar\nu\right)}{\text{Br}\left(\tau \to e\nu\bar\nu\right)}
\frac{f\big(\frac{m^2_e}{m^2_\tau}\big)}{f\big(\frac{m^2_\mu}{m^2_\tau}\big)}\,,
\end{equation}
where
\begin{equation}
\frac{\text{Br}(\tau\to \mu\nu\bar\nu)}{\text{Br}(\tau \to e\nu\bar\nu)}=\frac{\left(\big|{g_{LL,\tau \mu}^{V}}\big|^2+\frac{1}{4}\big|{g_{RR,\tau \mu}^{S}}\big|^2\right)
f\big(\frac{m^2_\mu}{m^2_\tau}\big)+
2{\rm Re}\left(g_{RR,\tau \mu}^{S} \left(g_{LL,\tau \mu}^{V}\right)^\ast\right)
\frac{m^2_\mu}{m^2_\tau} g\big(\frac{m^2_\mu}{m^2_\tau}\big)}{
\left(\big|{g_{LL,\tau e}^{V}}\big|^2+\frac{1}{4}\big|{g_{RR,\tau e }^{S}}\big|^2\right)
f\big(\frac{m^2_e}{m^2_\tau}\big)+2{\rm Re}\left(g_{RR,\tau e }^{S}
\left(g_{LL,\tau e }^{V}\right)^\ast\right)\frac{m^2_e}{m^2_\tau} g\big(\frac{m^2_e}{m^2_\tau}\big)}\,,\label{eq:univ:BRs}
\end{equation}
with $f(x)$ and $g(x)$ phase space functions\footnote{$f(x)=1-8x+8x^3-x^4-12x^2 \log(x)$ and $g(x)=1+9x-9x^2-x^3+6x(1+x)\log(x)$.}. One loop radiative corrections for the individual branching ratios cancel out in the ratio \eq{eq:univ:BRs}. The experimental limits on $\abs{g_{RR,\ell_\alpha\ell_\beta}^{S}}$ are collected in Section~\ref{AP:Input}.

\subsubsection{Semileptonic processes}
 Semileptonic processes may also receive tree level contributions from virtual $H^\pm$; the relevant effective Lagrangian for these processes is:
\begin{multline}
{\mathcal L}_{\rm eff} = -\frac{4 G_F}{\sqrt{2}}\ \sum_{u_i=u,c,t}\ \sum_{d_j=d,s,b}\ \sum_{\ell_a=e,\mu,\tau}\,\sum_{\nu_b=\nu_1,\nu_2,\nu_3}\ V_{u_id_j}\ U_{\ell_a \nu_b}\\
\left\{\left[\bar u_i\gamma^\mu \gamma_L d_j\right]\left[\bar \ell_a\gamma_\mu\gamma_L\nu_{b}\right] + \left[\bar u_i \left( g_L^{u_id_j \nu_b\ell_a}\,\gamma_L + g_R^{u_id_j \nu_b\ell_a}\, 
\gamma_R\right) d_j\right] \left[\bar \ell_a\gamma_L\nu_b\right]\right\}+ \mbox{h.c.}\,,\label{eq:L:semilept}
\end{multline}
where
\begin{equation}
g_L^{u_id_j \nu_b\ell_a}=\frac{ m_{u_i} m_{\ell_a}}{m_{H^+}^2}C^{u_id_j}C^{\ell_a \nu_b}\,,\qquad g_R^{u_id_j \nu_b\ell_a}=-\frac{ m_{d_j} m_{\ell_a}}{m_{H^+}^2}C^{u_id_j}C^{\nu_b\ell_a},
\end{equation}
and, $C^{u_id_j}=-1/\tan\beta$ for models of types $u_i$ and $d_j$, $C^{u_id_j}=\tan\beta$ otherwise, while $C^{ \nu_b\ell_a}=-1/\tan\beta$ for models of types $\ell_a$ and $\nu_b$, $C^{\nu_b\ell_a }=\tan\beta$ otherwise.

\begin{figure}[h]
\begin{center}
\begin{subfigure}{0.32\textwidth}
\includegraphics[width=\textwidth]{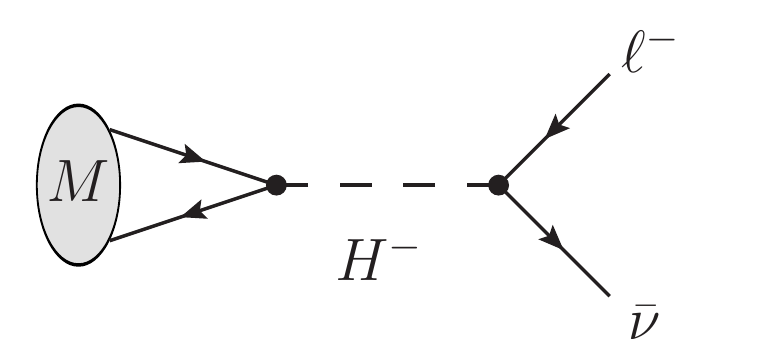}
\caption{$M\to\ell\nu$}
\end{subfigure}
\begin{subfigure}{0.32\textwidth}
\includegraphics[width=\textwidth]{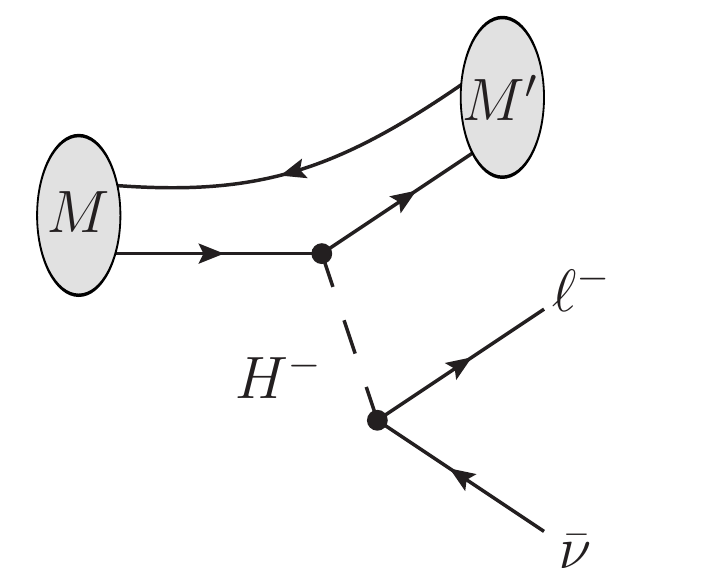}
\caption{$M\to M^\prime\ell\nu$}
\end{subfigure}
\begin{subfigure}{0.32\textwidth}
\includegraphics[width=\textwidth]{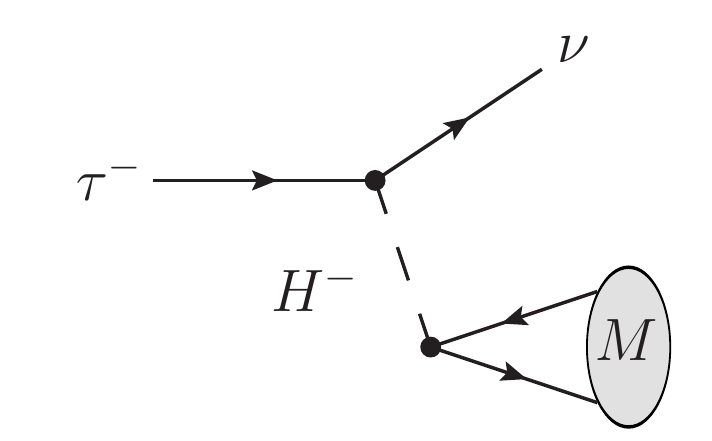}
\caption{$\tau\to M\nu$}
\end{subfigure}

\caption{Tree level $H^\pm$ mediated NP contributions to semileptonic process.\label{fig:Semileptonic}}
\end{center}
\end{figure}

The rate of the leptonic decay $M\to\ell\bar\nu$ of a pseudoscalar meson $M$, with quark content $\bar u_i d_j$, obtained from the effective Lagrangian in Eq.~\eq{eq:L:semilept}, is given by\footnote{Including electromagnetic radiative corrections 
\cite{kaons}, $\Gamma(M\to \ell\bar\nu)=(1+\delta_{\rm em})\,\Gamma_0(M\to \ell\bar\nu)$.}
\begin{equation}
\Gamma_0(M\to \ell\bar\nu)\, =\, G_F^2m_\ell^2 f_M^2 |V_{u_id_j}|^2 \,\frac{m_{M}}{8\pi} \left( 1- \frac{m_\ell^2}{m_{M}^2}\right)^2\; 
\sum_{n=1,2,3}|U_{\ell \nu_n}|^2|1-\Delta^{ \nu_n\ell}_{u_id_j}|^2\,. \label{eq:semileptonic:rate0:Mln}
\end{equation}
The scalar mediated new contribution is given by,
\begin{equation}
\Delta^{ \nu_n\ell}_{u_id_j}=C^{u_id_j}C^{\nu_n\ell}\frac{m_M^2}{m_{H^\pm}^2}\,.
\end{equation}
Since the process is helicity suppressed and receives NP contributions proportional to $m_M^2/m_{H^{\pm}}^2$, interesting channels are expected to involve heavy mesons and the $\tau$ lepton, as for example in $B^+\to\tau^+\nu$, $D_s^+\to\tau^+\nu$.
Taking into account the different possible values of $C^{u_id_j}$ and $C^{\nu_n\ell}$, we must have
\[
C^{u_id_j}C^{\nu_n\ell}\in\left\{-1,\ \tan^2\beta,\ \frac{1}{\tan^2\beta}\right\}\,.
\]
Therefore, for $m_{H^+}^2\gg m^2_M$, if $\Delta_{u_id_j}^{\nu_n\ell}$ is negative, then the NP contribution is negligible; otherwise, if the NP contribution is enhanced by $(\tan\beta)^{\pm 2}$, it will typically interfere destructively with the SM contribution. An increase with respect to SM predictions, which would be interesting for example to account for some $B^+\to\tau^+\nu$ measurements, would require a NP contribution more than twice larger than the SM one, leading to tensions in other observables. The different channels considered in the analysis are collected in Section~\ref{AP:Input} and radiative corrections are included according to \cite{kaons}.

 In the case of $\tau$ decays of type $\tau\to M\nu$, the analogue of Eq.~(\ref{eq:semileptonic:rate0:Mln}) is\footnote{Radiative corrections to $\Gamma_0(\tau\to M\nu)$ are included in the analysis \cite{taucorrections}.} 
\begin{equation}
\Gamma_0(\tau\to M\nu)\, =\, G_F^2m_\tau^3 f_M^2 |V_{u_id_j}|^2 \,\frac{3}{16\pi} \left( 1- \frac{m_{M}^2}{m_\tau^2}\right)^2\; 
\sum_{n=1,2,3}|U_{\tau \nu_n}|^2|1-\Delta^{\nu_n\tau}_{u_id_j}|^2\,. \label{eq:semileptonic:rate0:tMn}
\end{equation}
The analysis uses experimental $\tau\to\pi\nu$ and $\tau\to K\nu$ results -- see table \ref{TAB:AP:TreeCharged}.

 While $M\to\ell\bar\nu$ transitions are helicity suppressed two body decays, this is not the case anymore for $M\to M^\prime \ell\bar\nu$ decays. The corresponding decay amplitude is described by two form factors, $F_+(q^2)$ and $F_0(q^2)$ -- with $q$ the momentum transfer to the $\ell\bar\nu$ pair --, associated to the P wave and the S wave components of the amplitude $\langle 0 | \bar u_i\gamma^\mu d_j| M \bar M^\prime \rangle$. The $H^\pm$ mediated amplitude can only contribute to the S wave component. Considering for example a specific case like $B\to D\tau\nu$, where the quark level weak transition is $b\to c\tau\nu$, we have
\begin{equation}
	\frac{F_0^{(\mathrm{BGL})}(q^2,n)}{F_0^{(\mathrm{SM})}(q^2)}= 1-C^{cb}C^{\nu_n\tau}\frac{q^2}{m_{H^+}^2}\,,
\end{equation}
giving then
\begin{multline}
	\frac{\Gamma_{(\mathrm{BGL})}(B\to D \tau\nu)}{\Gamma_{(\mathrm{SM})}(B\to D \tau\nu)}=1+\\
\sum_{n=1}^3|U_{\tau\nu_n}|^2 \left(-{C_1} C^{cb}C^{\nu_n\tau}\frac{m_\tau(m_b-m_c)}{m_{H^+}^{2}}
+C_2(C^{cb}C^{\nu_n\tau})^2\frac{ m_\tau^{2}(m_b-m_c)^{2}}{ m_{H^+}^{4}}\right)\,,
\end{multline}
with coefficients $C_1\sim 1.5$ and $C_2\sim 1.0$. For $B\to D^\ast \tau\nu$, we have instead 
\begin{multline}
	\frac{\Gamma_{(\mathrm{BGL})}(B\to D^\ast \tau\nu)}{\Gamma_{(\mathrm{SM})}(B\to D^\ast \tau\nu)}=1+\\
\sum_{n=1}^3|U_{\tau\nu_n}|^2 \left(-{C_1} C^{cb}C^{\nu_n\tau}\frac{m_\tau(m_b+m_c)}{M_{H^+}^{2}}
+C_2(C^{cb}C^{\nu_n\tau})^2\frac{ m_\tau^{2}(m_b+m_c)^{2}}{ M_{H^+}^{4}}\right)\,,
\end{multline}
and $C_1\sim 0.12$ and $C_2\sim 0.05$. 
Notice that, even though BGL models still remain compatible with the present data
for the decays $B \to \tau \nu$, $B\to D \tau \nu$ and $B\to D^\ast \tau\nu$,
if the experimental anomalies observed in these processes, pointing towards
physics beyond the SM, are confirmed no two such anomalies could be simultaneously
accommodated in the BGL framework.

For $K\to\pi\ell\nu$ decays, rather than resorting to the rate or the branching fraction to constrain the NP contributions, the Callan-Treiman relation is used to relate the scalar form factor at the kinematic point $q^2_{\rm CT}=m_K^2-m_\pi^2$ to the decay constants of $K$ and $\pi$:
\begin{equation}
\frac{F_0^{(BGL)}(q^2_{\rm CT})}{F_+(0)}=\frac{f_K}{f_\pi}\frac{1}{F_+(0)}+\Delta_{\chi {\rm PT}}\equiv C\,.\label{eq:CT:KPln}
\end{equation}
$\Delta_{\chi {\rm PT}}$ is a Chiral Perturbation Theory correction. The right-hand side of Eq.~(\ref{eq:CT:KPln}), $C$, is extracted from experiment,  thus leading
to a constraint on $F_0^{(\mathrm{BGL})}(q^2_{\rm CT})$.

\subsection{Processes mediated by neutral scalars at tree level\label{SEC:ExpConst-sSEC:TreeNeutral}}
While the $H^\pm$ mediated NP contributions of the previous section compete with tree level SM amplitudes -- including suppressed ones, as in $M\to\ell\nu$ decays --, the neutral scalars $R$ and $I$ produce tree level contributions that compete with loop level SM contributions. We consider three different types of processes.
\begin{itemize}
\item Lepton flavour violating decays $\ell_1^-\to\ell_2^-\ell_3^+\ell_4^-$: in this case the SM loop contribution, proportional to neutrino masses is completely negligible and thus NP provides the only relevant one.
\item Mixings of neutral mesons, $M^0\rightleftarrows \bar M^0$, where $M^0$ could be a down-type meson $K^0$, $B^0_d$ or $B^0_s$ or the up-type meson $D^0$. The distinction among down and up-type mesons is relevant since depending on the BGL model the tree level NP contributions will appear in one or the other sector.
\item Rare decays $M^0\to\ell_1^+\ell_2^-$ (including lepton flavour violating modes $\ell_1\neq\ell_2$): again depending on the BGL model and $M^0$ being one of the previous down or up-type pseudoscalar mesons, the tree level NP contributions will be present or not. 
\end{itemize}
 
\subsubsection{Lepton flavour violating decays }

Lepton flavour violating decays of the form $\ell_1^-\to\ell_2^-\ell_3^+\ell_4^-$,
such as $\mu^-\to e^-e^+e^-$, $\tau^-\to e^-\mu^+\mu^-$ or $\tau^-\to \mu^- e^+ \mu^-$
are completely negligible in the SM, since the corresponding penguin and/or box amplitudes
are proportional to neutrino masses. In BGL models of type $(X,\nu_j)$, tree level
NP contributions mediate these decays. For muons, there is only one possible decay
of this type, while for taus there are two interesting cases: either $\ell_3^+$
belongs to the same family as one of the negatively charged leptons or not.
In the latter case the two vertices in the diagrams of figure \ref{fig:Mtoll} are flavour changing
and the SM contributes dominantly via a box diagram. Otherwise, the dominant BGL
contribution only requires one flavour changing vertex and SM penguin diagrams are
possible. In this case a connection can be established with the lepton flavour
violating processes of the type $\ell_j \to \ell_i \gamma$ considered in section \ref{SEC:ExpConst-sSEC:Loop}.

\begin{figure}[h]
\begin{center}
\includegraphics[width=0.5\textwidth]{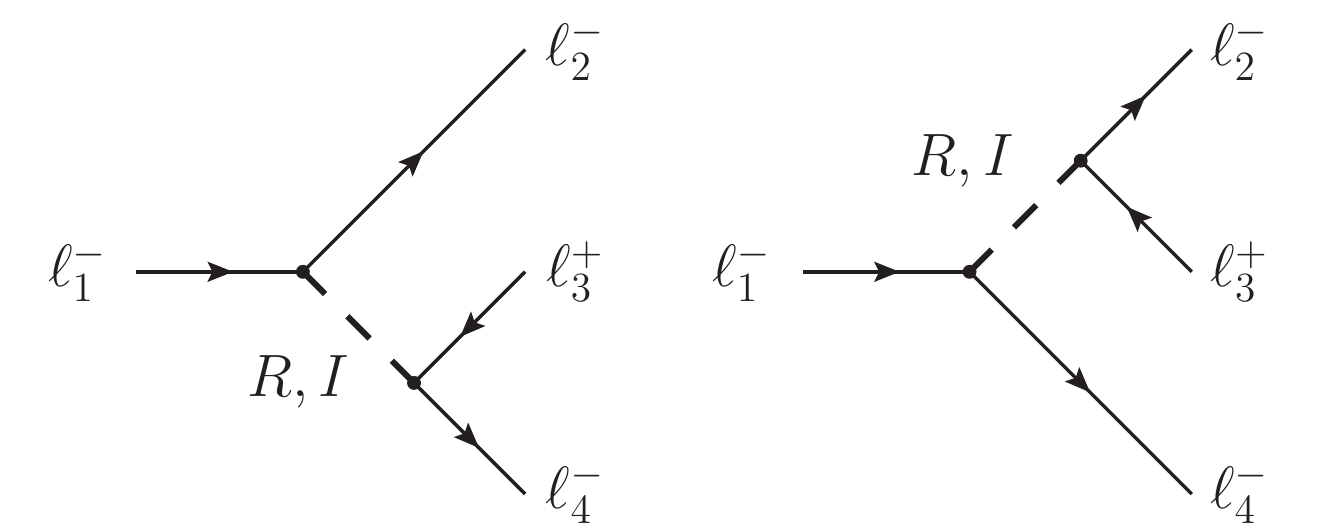}
\caption{Tree level $R,I$ mediated NP contributions to $\ell_1^-\to \ell_2^-\ell_3^+\ell_4^-$.\label{fig:Mtoll}}
\end{center}
\end{figure}
The corresponding effective Lagrangian is
\begin{multline}
\mathcal L_{\rm eff} = 
-\frac{2 G_F}{\sqrt{2}}\sum_{\chi_1,\chi_2=L,R}
\Big\{
g_{\chi_1\chi_2}^{12,34} \left[\bar\ell_2\gamma_{\chi_1} \ell_1\right]\left[\bar\ell_4\gamma_{\chi_2} \ell_3\right]+
g_{\chi_1\chi_2}^{14,32} \left[\bar\ell_4\gamma_{\chi_1} \ell_1\right]\left[\bar\ell_2\gamma_{\chi_2} \ell_3\right]
\Big\}\,,
\end{multline}
with
\begin{align*}
&g_{LL}^{ij,kl}=\frac{(N_\ell^\dagger)_{\ell_j \ell_i}(N_\ell^\dagger)_{\ell_l \ell_k}}{m_R^2}-\frac{(N_\ell^\dagger)_{\ell_j \ell_i}(N_\ell^\dagger)_{\ell_l \ell_k}}{m_I^2}\,,\ &
g_{RL}^{ij,kl}=\frac{(N_\ell)_{\ell_j \ell_i}(N_\ell^\dagger)_{\ell_l \ell_k}}{m_R^2}+\frac{(N_\ell)_{\ell_j \ell_i}(N_\ell^\dagger)_{\ell_l \ell_k}}{m_I^2}\,,\\
&g_{LR}^{ij,kl}=\frac{(N_\ell^\dagger)_{\ell_j \ell_i}(N_\ell)_{\ell_l \ell_k}}{m_R^2}+\frac{(N_\ell^\dagger)_{\ell_j \ell_i}(N_\ell)_{\ell_l \ell_k}}{m_I^2}\,,\ &
g_{RR}^{ij,kl}=\frac{(N_\ell)_{\ell_j \ell_i}(N_\ell)_{\ell_l \ell_k}}{m_R^2}-\frac{(N_\ell)_{\ell_j \ell_i}(N_\ell)_{\ell_l \ell_k}}{m_I^2}\,,
\end{align*}
and $N_\ell$ is the analogue, in the lepton sector, of $N_d$, i.e. the analogue of Eq.~(\ref{bgl1}) in the basis where $M_\ell$ is diagonal. Neglecting all masses except $m_{\ell_1}$, the width of the process is derived to be\footnote{The factor $(1+\delta_{\ell_2\ell_4})^{-1}$ takes into account the case of two identical particles in the final state.}
\begin{multline}
\Gamma(\ell_1^-\to\ell_2^-\ell_3^+\ell_4^-)=\frac{1}{1+\delta_{\ell_2\ell_4}}\frac{G_F^2 m_{\ell_1}^5}{3\cdot 2^{10}\pi^3}\times\\
\left\{
\abs{g_{LL}^{12,34}}^2+\abs{g_{LL}^{14,32}}^2+\abs{g_{RR}^{12,34}}^2+\abs{g_{RR}^{14,32}}^2+\abs{g_{LR}^{12,34}}^2+\abs{g_{LR}^{14,32}}^2\right.\\
\left.+\abs{g_{RL}^{12,34}}^2+\abs{g_{RL}^{14,32}}^2-\text{Re}\left[g_{LL}^{12,34}{g_{LL}^{14,32}}^\ast+g_{RR}^{12,34}{g_{RR}^{14,32}}^\ast\right]
\right\}\,.
\end{multline}
Experimental bounds on the corresponding branching ratios are collected in Section~\ref{AP:Input}.

\subsubsection{Neutral Meson mixings}
The NP short distance tree level contribution to the meson-antimeson transition amplitude\footnote{$M$ is the hermitian part of the effective hamiltonian describing the evolution of the two-level, meson-antimeson, system; $M_{12}$ is the dispersive transition amplitude.} $M_{12}^{NP}$ is \cite{Lavoura}
\begin{multline}
M_{12}^{NP}=\\
\sum_{H=R,I}
\frac{f_M^2m_M}{96v^2m^2_H}\left(\left(1+\left(\frac{m_M}{m_{q_1}+m_{q_2}}\right)^2\right)C_1(H)-\left(1+11\left(\frac{m_M}{m_{q_1}+m_{q_2}}\right)^2\right)C_2(H)\right)
\end{multline}
where $C_1(R)=(N_{q_2q_1}^*+N_{q_1q_2})^2$, $C_2(R)=(N_{q_2q_1}^*-N_{q_1q_2})^2$,
$C_1(I)=-(N_{q_2q_1}^*-N_{q_1q_2})^2$ and $C_2(I)=-(N_{q_2q_1}^*+N_{q_1q_2})^2$.
$q_1$ and $q_2$ refer to the valence quarks of the corresponding meson and $N$ is $N_u$ or $N_d$ for up-type or down-type quarks (and thus mesons).
\begin{figure}[h]
\begin{center}
\includegraphics[width=0.325\textwidth]{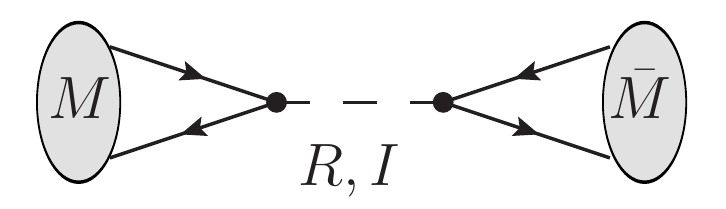}
\caption{Tree level $R,I$ mediated NP contributions to $M\to \bar M$.\label{fig:ljtoli_g}}
\end{center}
\end{figure}
For both $B^0_d$--$\bar B^0_d$ and $B^0_s$--$\bar B^0_s$ systems, the mass differences $\Delta M_{B_d}$ and $\Delta M_{B_s}$ are, to a very good approximation (namely $M_{12}^{B_q}\gg \Gamma_{12}^{B_q}$ with $\Gamma_{12}^{B_q}$ the absorptive transition amplitude), 
\[
\Delta M_{B_d}=2\abs{M_{12}^{B_d}}\,,\qquad \Delta M_{B_s}=2\abs{M_{12}^{B_s}}\,.
\]
In addition, time dependent CP violating asymmetries in $B^0_d\to J/\Psi K_S$ and $B^0_s\to J/\Psi \Phi$ decays constrain the phase of $M_{12}^{B_d}$ and $M_{12}^{B_s}$, respectively. We incorporate neutral B meson mixing constraints through the quantities
\[
\Delta_d=\frac{M_{12}^{B_d}}{[M_{12}^{B_d}]_{\rm SM}}\,,\qquad \Delta_s=\frac{M_{12}^{B_s}}{[M_{12}^{B_s}]_{\rm SM}}\,,
\]
according to \cite{CKMfitter}.

In $K^0$--$\bar K^0$, both $M_{12}^K$ and $\Gamma_{12}^K$ are relevant for the mass difference and thus we require that the NP contribution to $M_{12}^K$ 
does not exceed the experimental value of $\Delta M_K$.  In addition we take into
account the CP violating observable $\epsilon_K$,
\[
|\epsilon_K| = \frac{\text{Im}(M_{12}^K)}{\sqrt{2}\Delta M_K}\,,
\]
where the new contribution cannot exceed the experimental value.

For $D^0$--$\bar D^0$ long distance effects also prevent a direct connection between $M_{12}^D$ and $\Delta M_D$; as in $K^0$--$\bar K^0$, we then require  that the short distance NP contribution to $M_{12}^D$ does not give, alone, too large a contribution to
 $\Delta M_D$. Since this is the only existing up-type neutral meson system, the constraints on flavour changing neutral couplings arising from neutral meson mixings are tighter for neutral couplings to down quarks than they are for up quarks. The values used in the analysis are collected in Section~\ref{AP:Input}.

\subsubsection{Rare decays $M^0\to\ell_1^+\ell_2^-$}
Let us now consider mesons $M^0$ with valence quark composition $\bar q_2 q_1$.
In BGL models, the tree level induced NP terms in the effective Lagrangian relevant for the rare decays $M^0\to\ell_1^+\ell_2^-$ are:
\begin{equation}
\mathcal L_{\rm eff}^{NP} = 
-\frac{2 G_F}{\sqrt{2}}\sum_{\chi_1,\chi_2=L,R}\ c_{\chi_1\chi_2}^{12,12} \big[\bar q_2\gamma_{\chi_1} q_1\big]\left[\bar\ell_2\gamma_{\chi_2} \ell_1\right]
\end{equation}
with
\begin{align*}
&c_{LL}^{ij,kl}=\frac{(N_q^\dagger)_{q_j q_i}(N_\ell^\dagger)_{\ell_l \ell_k}}{m_R^2}-\frac{(N_q^\dagger)_{q_j q_i}(N_\ell^\dagger)_{\ell_l \ell_k}}{m_I^2}\,,\ &
c_{RL}^{ij,kl}=\frac{(N_q)_{q_j q_i}(N_\ell^\dagger)_{\ell_l \ell_k}}{m_R^2}+\frac{(N_q)_{q_j q_i}(N_\ell^\dagger)_{\ell_l \ell_k}}{m_I^2}\,,\\
&c_{LR}^{ij,kl}=\frac{(N_q^\dagger)_{q_j q_i}(N_\ell)_{\ell_l \ell_k}}{m_R^2}+\frac{(N_q^\dagger)_{q_j q_i}(N_\ell)_{\ell_l \ell_k}}{m_I^2}\,,\ &
c_{RR}^{ij,kl}=\frac{(N_q)_{q_j q_i}(N_\ell)_{\ell_l \ell_k}}{m_R^2}-\frac{(N_q)_{q_j q_i}(N_\ell)_{\ell_l \ell_k}}{m_I^2}\,.
\end{align*}
Notice that for the lepton flavour violating modes $M^0\to\ell_1^+\ell_2^-$ with $\ell_1\neq\ell_2$, the SM contribution to the effective Lagrangian is absent, this is no longer
true in $\ell_2=\ell_1$ case.
\begin{figure}[h]
\begin{center}
\includegraphics[width=0.35\textwidth]{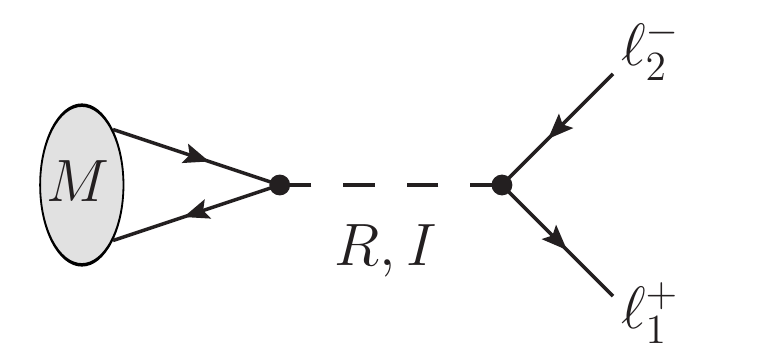}
\caption{Tree level $R,I$ mediated NP contributions to $M\to \ell_1^+\ell_2^-$.\label{fig:Mtoll2}}
\end{center}
\end{figure}

In the notation of appendix 7 of reference \cite{Crivellin:2013wna}, the Wilson coefficients read
\begin{align*}
&C_S^{q_2q_1} = -\frac{\sqrt{2}\pi^2}{G_FM_W^2}\left(c_{LR}^{12,12}+c_{LL}^{12,12}\right) \,, 
&C_P^{q_2q_1} = -\frac{\sqrt{2}\pi^2}{G_FM_W^2}\left(c_{LR}^{12,12}-c_{LL}^{12,12}\right)\,,\\
&C_S^{\prime\,q_2q_1} = -\frac{\sqrt{2}\pi^2}{G_FM_W^2}\left(c_{RR}^{12,12}+c_{RL}^{12,12}\right)\,,
&C_P^{\prime\,q_2q_1} = -\frac{\sqrt{2}\pi^2}{G_FM_W^2}\left(c_{RR}^{12,12}-c_{RL}^{12,12}\right)\,.
\end{align*}

The different modes and measurements used in the analysis are collected in Section~\ref{AP:Input}. It should be noted that while the previous type of short distance 
contributions dominate the rate for $B_s$ and $B_d$ decays, the situation is more involved in other cases. For example, for $K_L\to\mu^+\mu^-$ decays, the rate is dominated by the intermediate $\gamma\gamma$ state \cite{kaons} and NP is constrained through the bounds on the short distance SM+NP contributions.

\subsection{Loop level processes\label{SEC:ExpConst-sSEC:Loop}}
In the previous subsections we have listed observables useful to constrain the flavour changing couplings of the BGL models; their common characteristic is the possibility of having
NP contributions at tree level. In this subsection we address  two important rare decays where NP only contributes at loop level: $\ell_j\to\ell_i\gamma$ and $B\to X_s\gamma$.

\subsubsection{$\ell_j \to \ell_i \gamma$}
Lepton flavour violating (LFV) processes like $\mu\to e\gamma$ or $\tau\to \mu\gamma$ are in general a source of severe constraints for models with FCNC, like the BGL models we are considering in this work. The reason, anticipated for $\ell_1\to \ell_2\bar\ell_3\ell_4$ decays, is that these processes are negligible in the SM (their amplitudes are proportional to $m_{\nu_k}^2/m_W^2\ll 1$), while in the BGL case we expect loop contributions from neutral Higgs  flavour changing couplings proportional to
$m_{\ell_k}^2/m_{R,I}^2$. Moreover, and contrary to other 2HDM, the charged Higgs can also be relevant here, as the non-unitarity of the matrices controlling the couplings $H^{-}\bar{\ell}_j\nu_k$ and $H^{+}\bar{\nu}_k\ell_i$ leads to contributions 
proportional to $ m_{\ell_j}m_{\ell_i}/m_{H^{\pm}}^2$ (which would otherwise cancel out when summing over all generations of neutrinos running in the loop).
\begin{figure}[h]
\begin{center}
\begin{subfigure}{0.35\textwidth}
\includegraphics[width=\textwidth]{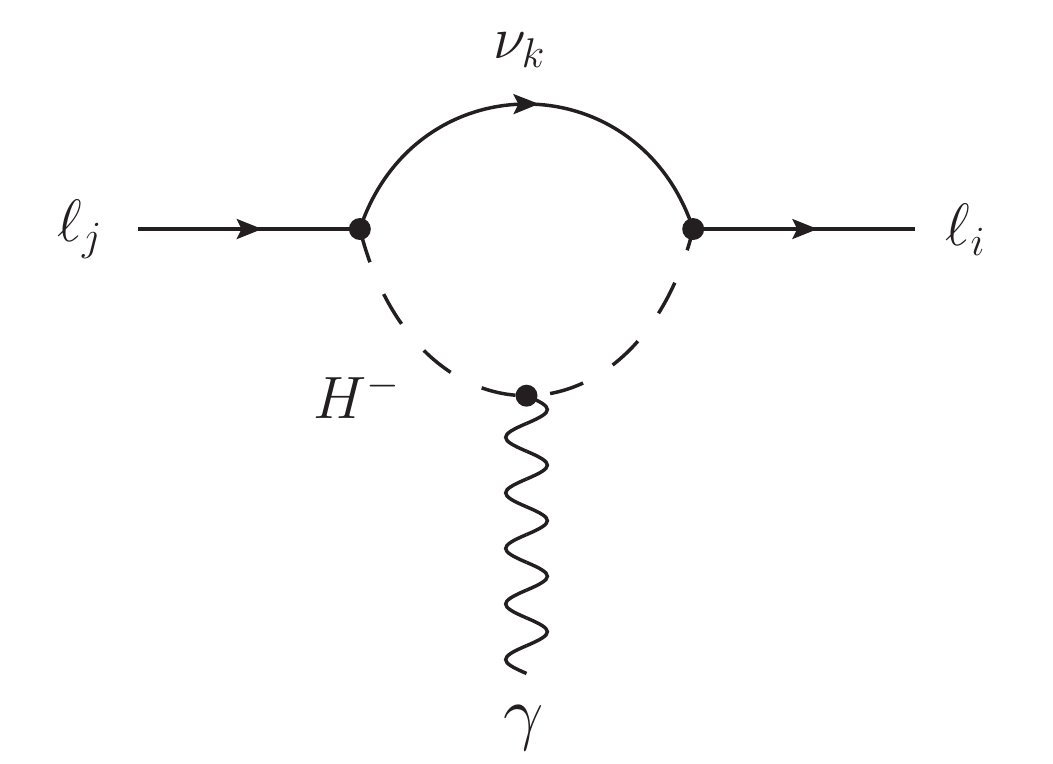}
\caption{$H^\pm$mediated.\label{fig:a}}
\end{subfigure}
\qquad
\begin{subfigure}{0.35\textwidth}
\includegraphics[width=\textwidth]{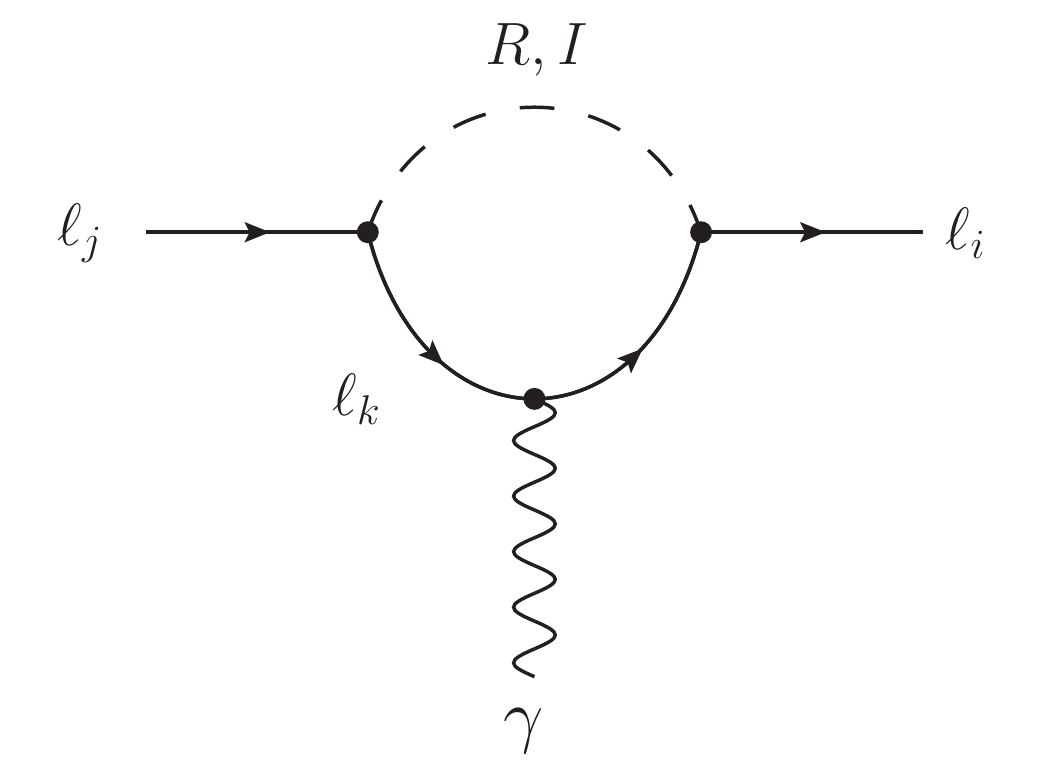}
\caption{$R,I$ mediated.\label{fig:b}}
\end{subfigure}
\caption{NP contributions to $\ell_j\to\ell_i\gamma$.\label{fig:ljtoli_g}}
\end{center}
\end{figure}
For on-shell photon and external fermions, the $\ell_j\to \ell_i\gamma$ amplitude is completely described by a dipole transition, see e.g. \cite{delAguila:2008zu},
\begin{equation}
i\mathcal{M}=ie\left[\mathcal{A}_R \gamma_R+ \mathcal{A}_L\gamma_L\right]\sigma^{\mu\nu}q_{\nu}\epsilon_{\mu},
\end{equation}
with $q^{\mu}$ the incoming photon momentum. The corresponding decay width is
\begin{equation}
\Gamma(\ell_j \to \ell_i\gamma)=\frac{\alpha m_{\ell_j}^5 G_F^2}{128\pi^4}\left[\left|\mathcal{A}_L\right|^2+\left|\mathcal{A}_R\right|^2\right]\,.
\end{equation}
Up to terms of $\mathcal{O}(m_{\ell_i}/m_{\ell_j})$ -- note that
$N_\ell^{ik}$ is proportional to $m_{\ell_k}$ --, the coefficients $\mathcal{A}_R$ and $\mathcal{A}_L$ are derived to be:
\begin{multline}
\mathcal{A}_R=
\sum_{k}\left\{\frac{1}{12 m_{R}^2}N_\ell^{ik}N_\ell^{jk\ast}
-\frac{1}{2m_{R}^2}N_\ell^{ik}N_\ell^{kj}\frac{m_{\ell_k}}{m_{\ell_j}}
\left[\frac{3}{2}+\ln\left(\frac{m_{\ell_k}^2}{m_{R}^2}\right)\right]\right.\\
+\left.\frac{1}{12 m_I^2}N_\ell^{ik}N_\ell^{jk\ast}
+\frac{1}{2m_{I}^2}N_\ell^{ik}N_\ell^{k j}\frac{m_{\ell_k}}{m_{\ell_j}}\left[\frac{3}{2}+\ln\left(\frac{m_{\ell_k}^2}{m_{I}^2}\right)\right]\right\},
\end{multline}
\begin{multline}
\mathcal{A}_L=
\sum_{k}\left\{-\frac{1}{12m_{H^{\pm}}^2}(N_\ell^{\dagger}U)^{i k}(N_{\ell}^{\dagger}U)^{j k\ast}
+\frac{1}{12m_{R}^2}N_{\ell}^{ki\ast}N_{\ell}^{kj}
+\frac{1}{12m_{I}^2}N_{\ell}^{k i\ast}N_{\ell}^{k j}\right\},
\end{multline}
where we have neglected contributions proportional to the neutrino masses $m_{\nu_k}\approx 0$ as well as subleading terms in $m_{\ell_k}^2/m_{R,I}^2$.  

In some cases, two-loop contributions for $\ell_{j}\to \ell_{i}\gamma$ can dominate over the one-loop ones \cite{Bjorken:1977vt,Chang:1993kw}. This is related to the fact that, due to the required chirality flip, we need three mass insertions at one loop level. However, there are two-loop contributions with only one chirality flip in the $\ell_j-\ell_i$ fermion line. Therefore, in some cases they can compensate the extra loop factor
by avoiding two small Yukawa couplings. We roughly estimate the two-loop contribution as 
\begin{eqnarray}
\Gamma(\ell_j\to \ell_i\gamma)_{\mathrm{2-loop}}\approx \frac{\alpha m_{\ell_j}^5 G_F^2}{128 \pi^4}\left(\frac{\alpha}{\pi}\right)^2[\left|\mathcal{C}\right|^2+\left|\mathcal{D}\right|^2],
\end{eqnarray}
where
\begin{equation}
\mathcal{C}=\frac{2}{m_{R}^2}
(N_{\ell})_{ij}(N_{u})_{tt}\frac{m_t}{m_{\ell_j}}\ln^2\left(\frac{m_t^2}{m_{R}^2}\right)\quad
\text{and}\quad
\mathcal{D}=\frac{2}{m_{I}^2}
(N_{\ell})_{ij}(N_{u})_{tt}\frac{m_t}{m_{\ell_j}}\ln^2\left(\frac{m_t^2}{m_{I}^2}\right).
\end{equation}

\subsubsection{$\bar B\to X_s\gamma$}
 The other important rare decay, now in the quark sector, is $\bar B\to X_s\gamma$, induced by the quark level transition $b\to s\gamma$.
Similarly to the LFV processes $\ell_j\to\ell_i\gamma$ considered before, NP contributions due to the exchange of both neutral and charged Higgs are present. Although the contributions coming from the latter case are naively expected to be dominant, due to the relative enhancement coming from  the top mass insertion -- i.e. proportional to
$ m_{t}^2/m_{H^{\pm}}^2$ versus $ m_{b}^2/m_{R,I}^2$ --, we cannot neglect diagrams with FCNC because this effect can be compensated by $\tan\beta$ enhancements. The effective Hamiltonian describing this transition is\cite{Misiak:2006ab, *Blanke:2011ry, *Blanke:2012tv,Buras:2011zb}:
\begin{equation}
\mathcal{H}_{\rm eff}(b\to s\gamma)=-\frac{4 G_F}{\sqrt{2}} V_{tb}V_{ts}^{\ast}\left[C_7(\mu_b)\mathcal{O}_7+C_7^{\prime}(\mu_b)\mathcal{O}_7^{\prime}+C_8(\mu_b)\mathcal{O}_8+C_8^{\prime}(\mu_b)\mathcal{O}_8^{\prime}\right],
\end{equation}
with new effective operators $\mathcal{O}_{7}^{\prime}$ and $\mathcal{O}_8^{\prime}$, which are absent in the SM besides terms $\mathcal{O}(m_s/m_b)$. $C_{7,8}(\mu)$ and $C_{7,8}^{\prime}(\mu)$ are the Wilson coefficients of the dipole operators
\begin{eqnarray}
\mathcal{O}_7&=&\frac{e}{16\pi^2}m_b \bar{s}_{L,\alpha}\sigma^{\mu\nu}b_{R,\alpha}F_{\mu\nu},\qquad	\mathcal{O}_8=\frac{g_s}{16\pi^2}m_b \bar{s}_{L,\alpha}\left(\frac{\lambda^a}{2}\right)_{\alpha\beta}\sigma^{\mu\nu}b_{R,\beta}G_{\mu\nu}^a,\\
\mathcal{O}_7^{\prime}&=&\frac{e}{16\pi^2}m_b \bar{s}_{R,\alpha}\sigma^{\mu\nu}b_{L,\alpha}F_{\mu\nu},\qquad  \mathcal{O}_8^{\prime}=\frac{g_s}{16\pi^2}m_b \bar{s}_{R,\alpha}\left(\frac{\lambda^a}{2}\right)_{\alpha\beta}\sigma^{\mu\nu}b_{L,\beta}G_{\mu\nu}^a,\qquad~
\end{eqnarray}
evaluated at the scale $\mu_b=\mathcal{O}(m_b)$, with $F_{\mu\nu}$ and $G_{\mu\nu}^a$ 
denoting the electromagnetic and gluon field strength tensors, and $\lambda^{a}$, $a=1,
\ldots,8$, standing for the Gell-Mann matrices. 
\begin{figure}[h]
\begin{center}
\begin{subfigure}{0.6\textwidth}
\includegraphics[width=0.5\textwidth]{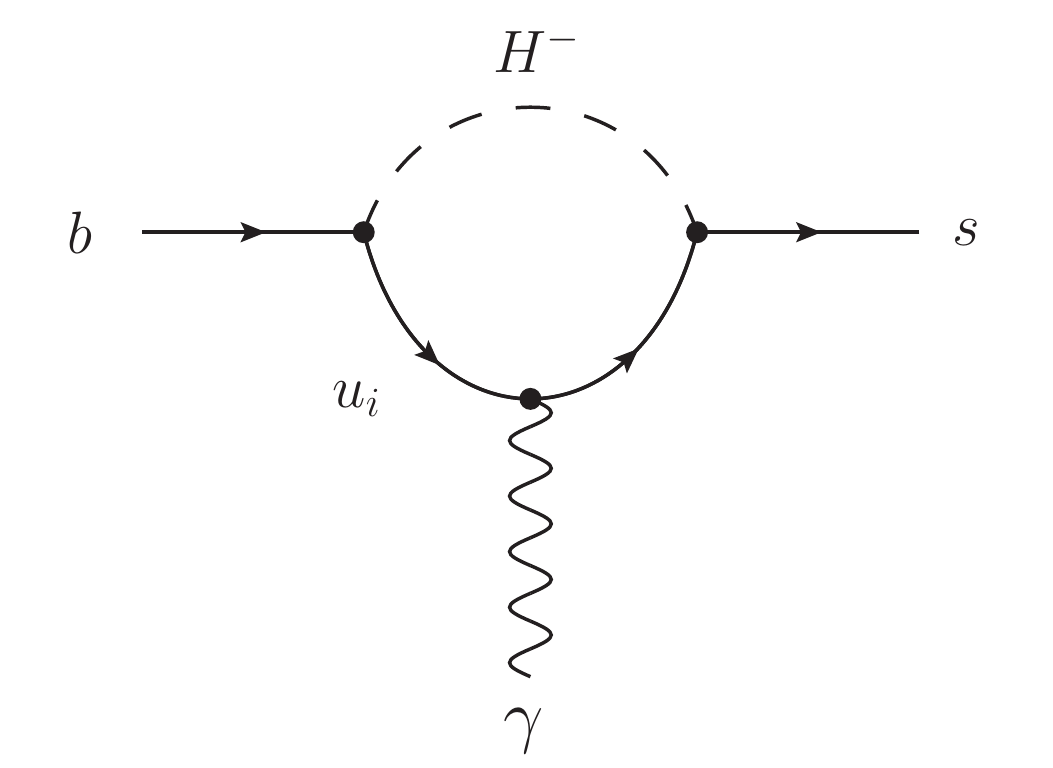}
\quad\includegraphics[width=0.5\textwidth]{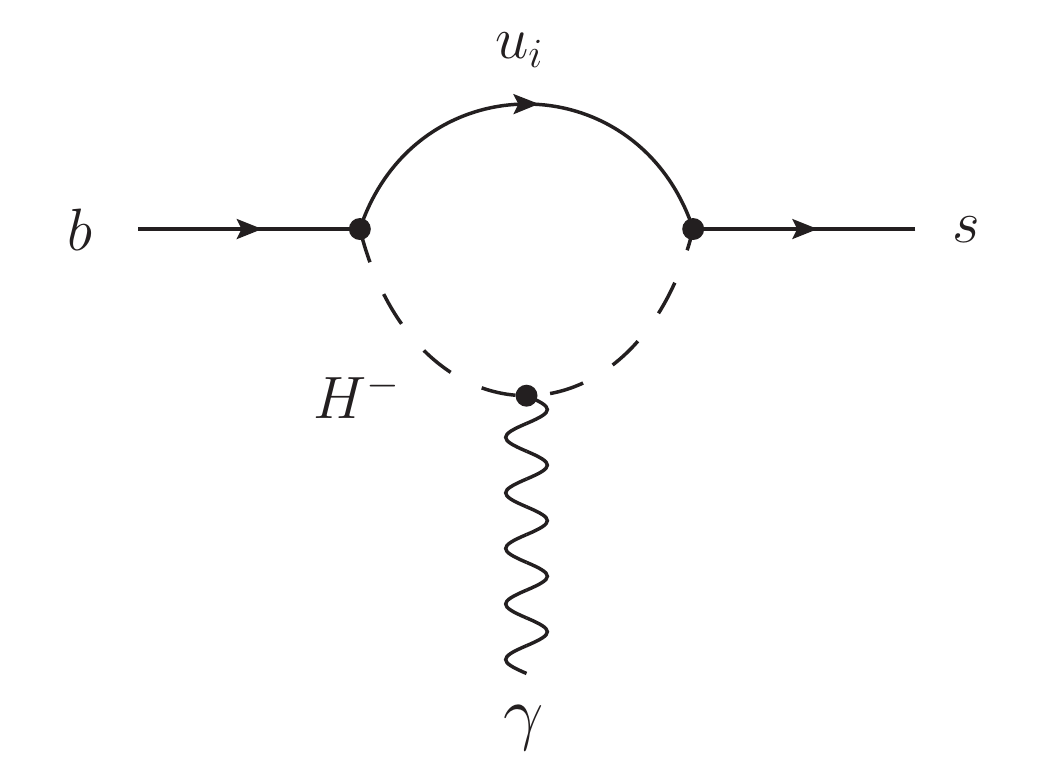}
\caption{$H^\pm$ mediated.\label{fig:1}}
\end{subfigure}
\quad
\begin{subfigure}{0.3\textwidth}
\includegraphics[width=\textwidth]{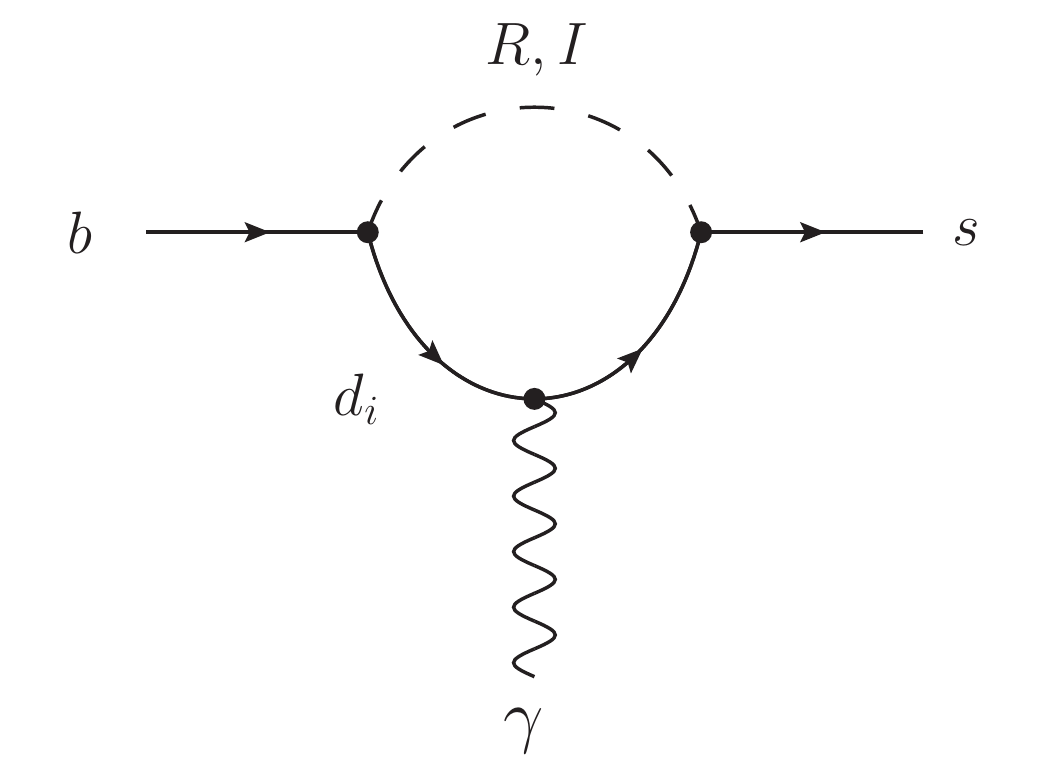}
\caption{$R,I$ mediated.\label{fig:2}}
\end{subfigure}
\caption{NP contributions to $b\to s\gamma$.\label{fig:btosg}}
\end{center}
\end{figure}

We then constrain the BGL contribution to $b\to s\gamma$ using the master formula \cite{Misiak:2006ab, *Blanke:2011ry, *Blanke:2012tv,Buras:2011zb}
\begin{equation}
\mathrm{Br}\left(\bar{B}\to X_s\gamma\right)=\mathrm{Br}_{\rm SM}+0.00247\left[|\Delta C_7(\mu_b)|^2+|\Delta C_7^{\prime}(\mu_b)|^2-0.706\mathrm{Re} \left(\Delta C_7(\mu_b)\right)\right],\label{eq:bsgamma}
\end{equation}
where $\mathrm{Br}_{\rm SM}=\mathrm{Br}(\bar{B}\to X_s\gamma)_{\rm SM}=(3.15\pm 0.23)\times 10^{-4}$ is the SM prediction at NNLO \cite{Gambino:2001ew, *Misiak:2006zs,Misiak:2006ab} and we have split the SM and the NP contributions to the relevant Wilson coefficients
\begin{equation}
C_{7}^{(\prime)}(\mu)=C_{7,\mathrm{SM}}^{(\prime)}(\mu)+\Delta C_7^{(\prime)}(\mu), \qquad 	C_{8}^{(\prime)}(\mu)=C_{8,\mathrm{SM}}^{(\prime)}(\mu)+\Delta C_8^{(\prime)}(\mu).
\end{equation}
The value obtained from equation (\ref{eq:bsgamma}) has to be compared with the experimental measurement \cite{Asner:2010qj}
\begin{equation}
\mathrm{Br}(B\to X_s \gamma)_{\rm exp}= \left(3.55 \pm 0.27\right) \times 10^{-4}.
\end{equation}
The Wilson coefficients  $\Delta C_{7,8}^{}(\mu)$ and $\Delta C_{7,8}^{\prime}(\mu)$ are computed at the high energy scale $\tilde{\mu}=\mathcal{O}(m_{H^{\pm}})\sim\mathcal{O}(m_{R,I})$ at one-loop in perturbation theory, and then run down to $\mu_b$ using RGE \cite{Grinstein:1990tj,*Buras:1993xp}:
\begin{equation}
\Delta C_7^{(\prime)}(\mu_b)\approx \eta^{\frac{16}{23}}\Delta C_7^{(\prime)}(\tilde{\mu})+\frac{8}{3}\left(\eta^{\frac{14}{13}}-\eta^{\frac{16}{23}}\right)\Delta C_8^{(\prime)}(\tilde{\mu}),
\end{equation}
where $\eta=\alpha_s(\tilde{\mu})/\alpha_s(\mu_b)$. FCNC might also affect the running of these Wilson coefficients through new operators which are not present in the SM, similarly to what happens in the case of flavour changing neutral gauge bosons \cite{Buras:2011zb}. However, the impact of this effect is expected to be subleading, and its study is well beyond the scope of this chapter. 

The relevant Wilson coefficients are derived to be:
\begin{multline}
\Delta C_7(\tilde{\mu})=\frac{1}{2}\frac{1}{V_{ts}^\ast V_{tb}}
\sum_k\left\{\frac{1}{m_{H^{\pm}}^2}(V^{\dagger}N_{u})_{sk}
\left((V^{\dagger}N_{u})^{\ast}_{bk}A^{(2)}_H(x_{H^{\pm}}^k)
+(N^{\dagger}_d V^{\dagger})^{\ast}_{bk}\frac{m_{u_k}}{m_{b}}A_H^{(3)}(x_{H^{\pm}}^k)\right)\right.\\
-(N_d)_{sk}(N_d)^{\ast}_{bk}\left(\frac{Q_{d}}{m_{R}^2}A_H^{(0)}(y_R^{k})+\frac{Q_{d}}{m_{I}^2}A^{(0)}_H(y_I^{k})\right)\\
-\left.(N_d)_{sk}(N_d)_{kb}\frac{m_{d_k}}{m_{b}}\left(\frac{Q_{d}}{m_{R}^2}A^{(1)}_H(y_R^k)-\frac{Q_{d}}{m_{I}^2}A^{(1)}_H(y_I^k)\right)\right\},
\end{multline}
\begin{multline}
\Delta C_7^{\prime}(\tilde{\mu})=\frac{1}{2}\frac{1}{V_{ts}^\ast V_{tb}}
\sum_k\left\{\frac{1}{m_{H^{\pm}}^2}(N_d^{\dagger}V^{\dagger})_{sk}(N_d^{\dagger}V^{\dagger})^{\ast}_{bk}A_{H}^{(2)}(x_{H^{\pm}}^k)\right.\\
-(N_d)^{\ast}_{ks}(N_d)_{kb}\left.\left(\frac{Q_d}{m_{R}^2}A_H^{(0)}(y_R^{k})+\frac{Q_d}{m_{I}^2}A^{(0)}_H(y_I^{k})\right)\right\},
\end{multline}
\begin{multline}
\Delta C_8(\tilde{\mu})=\frac{1}{2}\frac{1}{V_{ts}^\ast V_{tb}}
\sum_k\left\{2\frac{1}{m_{H^{\pm}}^2}(V^{\dagger}N_{u})^{sk}\left(-(V^{\dagger}N_{u})^{\ast}_{bk}A^{(0)}_H(x_{H^{\pm}}^k)+(N^{\dagger}_d V^{\dagger})^{\ast}_{bk}\frac{m_{u_k}}{m_{b}}A^{(1)}_H(x_{H^{\pm}}^k)\right)\right.\\
-(N_d)_{sk}(N_d)^{\ast}_{bk}\left(\frac{1}{m_{R}^2}A_H^{(0)}(y_R^{k})+\frac{1}{m_{I}^2}A^{(0)}_H(y_I^{k})\right)\\
-\left.(N_d)_{sk}(N_d)_{kb}\frac{m_{d_k}}{m_{b}}\left(\frac{1}{m_{R}^2}A_H^{(1)}(y_R^{k})-\frac{1}{m_{I}^2}A^{(1)}_H(y_I^{k})\right)\right\},
\end{multline}
\begin{multline}
\Delta C_8^{\prime}(\tilde{\mu})=\frac{1}{2}\frac{1}{V_{ts}^\ast V_{tb}}
\sum_k\left\{-2\frac{1}{m_{H^{\pm}}^2}(N_d^{\dagger}V^{\dagger})_{sk}(N_d^{\dagger}V^{\dagger})^{\ast}_{bk}A_{H}^{(0)}(x_{H^{\pm}}^k)\right.\\
-(N_d)^{\ast}_{ks} (N_d)_{kb}\left.\left(\frac{1}{m_R^2}A_{H}^{(0)}(y_R^k)+\frac{1}{m_I^2}A_{H}^{(0)}(y_I^k)\right)\right\},
\end{multline}
where $Q_d=-1/3$ and $x_{H^{\pm}}^k=m_{u_k}^2/m_{H^{\pm}}^2$, $y_{R,I}^k=m_{d_k}^2/m_{R,I}^2$. The loop functions $A_H^{(i)}$ are:

\begin{align*}
&A_{H}^{(0)}(x)=\frac{2+3x-6x^2+x^3+6x\ln x}{24(1-x)^4},\quad  A_H^{(2)}(x)=\frac{-7+5x+8x^2}{36(1-x)^3}+\frac{x(-2+3x)\ln x}{6(1-x)^4},\\
&A_{H}^{(1)}(x)=\frac{-3+4x-x^2-2\ln x}{4(1-x)^3},\quad A_H^{(3)}(x)=\frac{-3+8x-5 x^2-(4-6 x)\ln x}{6 (1-x)^3}.
\end{align*}

\subsubsection{Electric dipole moments and anomalous magnetic moments}
 NP induced one loop contributions to the electric dipole moments (EDM) of leptons and quarks are absent in BGL models. For example, in the leptonic case, the contribution to the flavor conserving $\ell_i\to\ell_i\gamma$ dipole transition is real while the EDM is proportional to the imaginary part.
For the anomalous magnetic moments, we checked that the NP induced one loop contributions appearing in BGL models are too small to have significant impact on the results -- once other constraints are used --, in agreement with \cite{Crivellin:2013wna}, so we are not considering their constraints to save processing time during the simulation despite the fact that the results do not change noticeably whether we include them or not.
Addressing two loop contributions to electric and magnetic dipole moments is beyond the scope of this chapter -- see for example \cite{Jung:2013hka}.
\subsubsection{Precision Electroweak Data}
 The previous subsections have covered representative flavour related low energy processes that are able to constrain the masses of the new scalar together with $\tan\beta$.
Electroweak precision data also play an important  role. The observables included in the analysis for that purpose are the $Z\bar bb$ effective vertex and the oblique parameters $S$, $T$ and $U$.

For the $Z\bar bb$ vertex probed at LEP, BGL models introduce new contributions mediated by the charged and by the neutral scalars. The effects mediated by $H^\pm$  are typically the most relevant ones, see e.g. \cite{HernandezSanchez:2012eg}.  In our case, similarly to what happens in $b\to s\gamma$, neutral contributions can also be relevant but, as a first estimate, we just consider the charged ones \cite{Pich:2010}
\begin{equation}
F_{Zb\bar b}=\frac{|C^{tb}|-0.72}{m_H^\pm}< 0.0024\text{ \GeV}^{-1}\,,
\end{equation}
where once again $C^{tb}=-1/\tan\beta$ for BGL models of quark types $t$ and $b$, and $C^{tb}=\tan\beta$ otherwise.

For the oblique parameters, as discussed in \cite{O'Neil:2009nr}, the contributions to $S$ and $U$ in 2HDM tend to be small. This is not the case for the $T$ parameter which receives corrections that can be sizable. In BGL models, the NP contribution $\Delta T$ to $T=T_{{\rm SM}}+\Delta T$  \cite{Grimus:2007if, *Grimus:2008nb} is
\begin{equation}
\Delta T=\frac{1}{16 \pi m_W^2s_W^2}\left\{F(m_{H^{\pm}}^2,m_{R}^2)-F(m_I^2,m_{R}^2)+F(m_{H^{\pm}}^2,m_{I}^2)\right\}
\label{drho2}
\end{equation}
with
\[
F(x,y)=\frac{x+y}{2}-\frac{xy}{x-y}\ln\frac{x}{y}\,,
\]
so that $F(x,x)= 0$, while for $\Delta S$
\begin{equation}
\Delta S=\frac{1}{24\pi}\left\{\left(2s_W^2-1\right)^2G(m_{H^{\pm}}^2,m_{H^{\pm}}^2,m_Z^2)+G(m_{R}^2,m_{I}^2,m_Z^2)+\ln\left[\frac{m_{R}^2m_{I}^2}{m_{H^{\pm}}^4}\right]\right\},\qquad 
\end{equation}
where
\begin{multline}
G(x,y,z)=-\frac{16}{3}+5\frac{x+y}{z}-2\frac{(x-y)^2}{z^2}+\frac{r}{z^3}f(t,r)\\
    +\frac{3}{z}\left[\frac{x^2+y^2}{x-y}-\frac{x^2-y^2}{z}+\frac{(x-y)^3}{3z^2}\right]\ln\frac{x}{y},\end{multline}
with $r=z^2-2z(x+y)+(x-y)^2$, $t=x+y-z$ and
\begin{equation}
f(t,r)=
\begin{cases}
\sqrt{r}\ln\left|\frac{t-\sqrt{r}}{t+\sqrt{r}}\right|&r>0,\\
2 \sqrt{-r}\arctan\frac{\sqrt{-r}}{t}&r<0.
\end{cases}
\end{equation}

\clearpage
\section{Results\label{SEC:Results}}
In the previous section we have presented a large set of relevant observables that can constrain the different BGL models, excluding regions of the parameter space $\{\tan\beta, m_{H^\pm},$ $m_R, m_I\}$ where the NP contributions are not compatible with the available experimental information. Following the methodology described in Section~\ref{AP:Analysis}, we apply those constraints to each one of the 36 BGL models: the main aim of this general study is to understand where could the masses of the new scalars lie and how does this depend on $\tan\beta$. However, before addressing the main results for the complete set of BGL models, an important aspect has to be settled: since we have three different scalars, we should in principle obtain allowed regions in the $\{\tan\beta, m_{H^\pm}, m_R, m_I\}$ parameter space, and then project them to the different subspaces for each BGL model, e.g. $m_{H^\pm}$ vs. $\tan\beta$, $m_{R}$ vs. $\tan\beta$, etc. The oblique parameters, in particular $\Delta T$, help us to simplify the picture. For degenerate $H^\pm$, $R$ and $I$, according to Eq.~(\ref{drho2}), $\Delta T=0$; in general, for almost degenerate $H^\pm$, $R$ and $I$, the oblique parameters are in agreement with experimental data\footnote{$\Delta T$ alone is not sufficient; considering only $\Delta T$, for $m_{H^\pm}=m_I$, $m_R$ would be free to vary but $\Delta S$ prevents it.  Analogously, for $m_{H^\pm}=m_R$, $\Delta T=0$ irrespective of
$m_I$.
In addition, in the experimental constraint, $\Delta T$ and $\Delta S$ are correlated.}. This is explored and illustrated in figure \ref{fig:onea} for one particular model: $m_R$ vs. $m_{H^\pm}$ and $m_R$ vs. $m_I$ allowed regions are displayed when the oblique parameters constraints are used.  Therefore, even though we treated all three scalar masses 
independently and on equal basis, we only present results in terms of $m_{H^\pm}$ for simplicity.

\begin{figure}[!htb]
\centering
\includegraphics[width=0.45\textwidth]{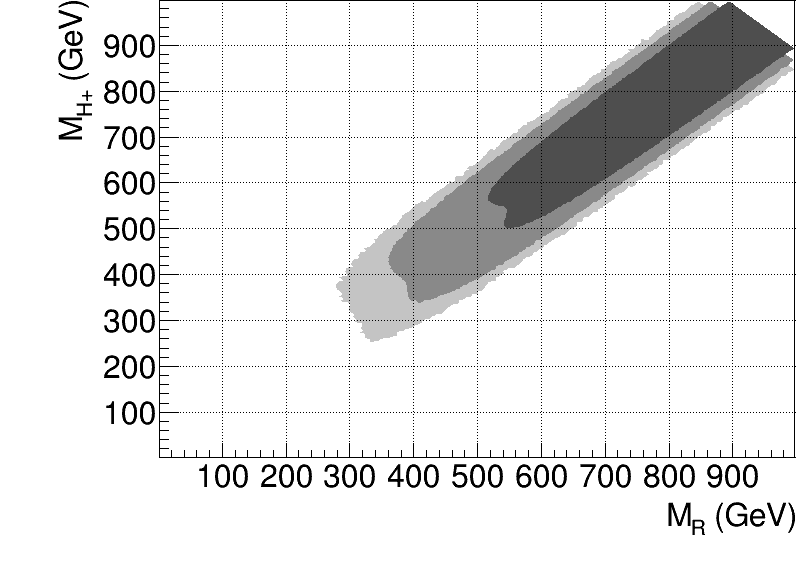}
\includegraphics[width=0.45\textwidth]{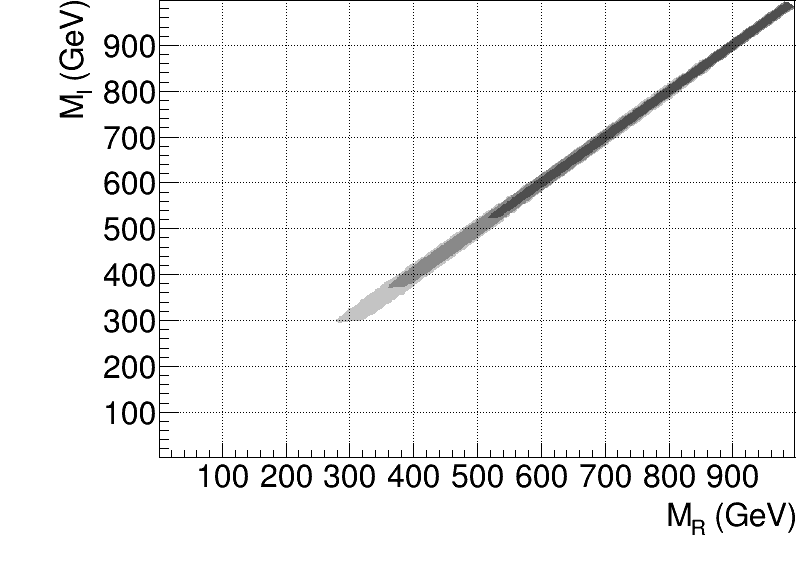}
\caption{Effect of the oblique parameters constraints in model $(t,\tau)$, all other constraints are also applied. For all other BGL models the width of the strips is the same.}
\label{fig:onea}
\end{figure}
In figures \ref{FIG:Results01} and \ref{FIG:Results02} we present the allowed regions -- corresponding to 68\%, 95\% and 99\% confidence levels (CL) -- in the $(m_{H^\pm},\tan\beta)$ plane for the 36 different BGL models. They deserve several comments.

\begin{itemize}
\item The experimental bounds for FCNC in the up sector are more relaxed than for the down sector, but the models with tree level FCNC in the up sector are not less constrained than the ones with tree level FCNC in the down sector, due to the $b\to s\gamma$ constraints on the charged Higgs mass. 

\item It should be emphasized that among the BGL models, the ones of types $t$ and $b$ guarantee a stronger suppression of the FCNC due to the hierarchical nature of the CKM matrix, so one would expect them to be less constrained. However, $b\to s\gamma$ frustrates this expectation. In fact, the models of type $d$ are less constrained than the $s$ and $b$ ones, while for up type models there is no clear trend.

\item For the leptonic part, since the experimental bounds on tree level FCNC in the neutrino sector are irrelevant -- due to the smallness of neutrino masses --, $e$, $\mu$ and $\tau$ models are typically less constrained than $\nu_i$ models. This can be seen in figure \ref{FIG:Results01}, whereas in figure \ref{FIG:Results02} differences are minute, signifying then that leptonic constraints are secondary once other constraints are imposed.

\item Lower bounds on the scalar masses lie in between 100 and 400 GeV for many models, which put them within range of direct searches at the LHC. Nevertheless some exceptions deserve attention: for models of types $s$ and $b$, the lightest masses are instead in the 400-500 GeV range. Notice in addition that in models of types $s$ and $b$ the allowed values of $\tan\beta$ span a wider range than in the rest of models.

\item One aspect that is interesting on its own but would require specific attention beyond the scope of the present work, is the following: in many models isolated allowed regions for light masses appear. That is, for the considered set of observables, the scalar masses and $\tan\beta$ can still be tuned to agree with experimental data within these reduced regions. Higher order contributions than the ones used in section \ref{SEC:ExpConst}, additional observables and direct searches may then be used to further constrain these parameter regions.

\item As a final comment it should be noticed that some of the $t$ type models, the ones that correspond to the MFV framework as defined 
in \cite{Buras:2000dm} or \cite{D'Ambrosio:2002ex}, can be very promising. However
this is not a unique feature of these implementations since, as can be seen
from our figures, there are several others that allow for light scalars.

There is an independent
analysis of one of the BGL scenarios discussed here, including
in addition the decay signatures of the new scalers\cite{Bhattacharyya:2014nja};
this analysis agrees with our conclusion concerning the feasibility of more than one light Higgs boson.

\end{itemize}

\clearpage
\begin{figure}[!htb]
\centering
\includegraphics[width=\textwidth]{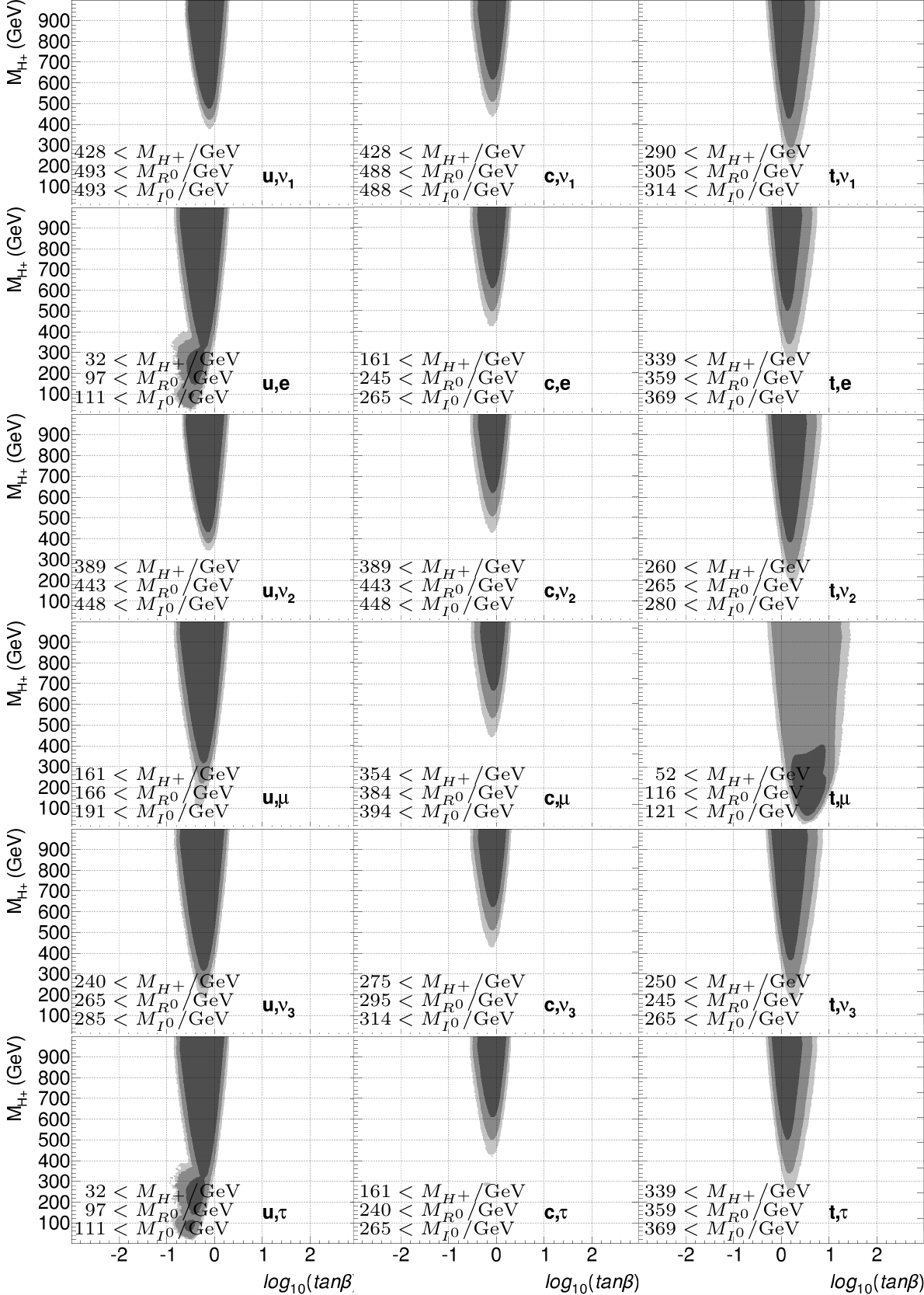}
\caption{Allowed 68\% (black), 95\% (gray) and 99\% (light gray) CL regions in $m_{H^\pm}$ vs. $\tan\beta$ for BGL models of types $(u_i,\nu_j)$ and $(u_i,\ell_j)$, i.e. for models with FCNC in the down quark sector and in the charged lepton or neutrino sector (respectively). Lower mass values corresponding to 95\% CL regions are shown in each case.\label{FIG:Results01}}
\end{figure}

\clearpage
\begin{figure}[!htb]
\centering
\includegraphics[width=\textwidth]{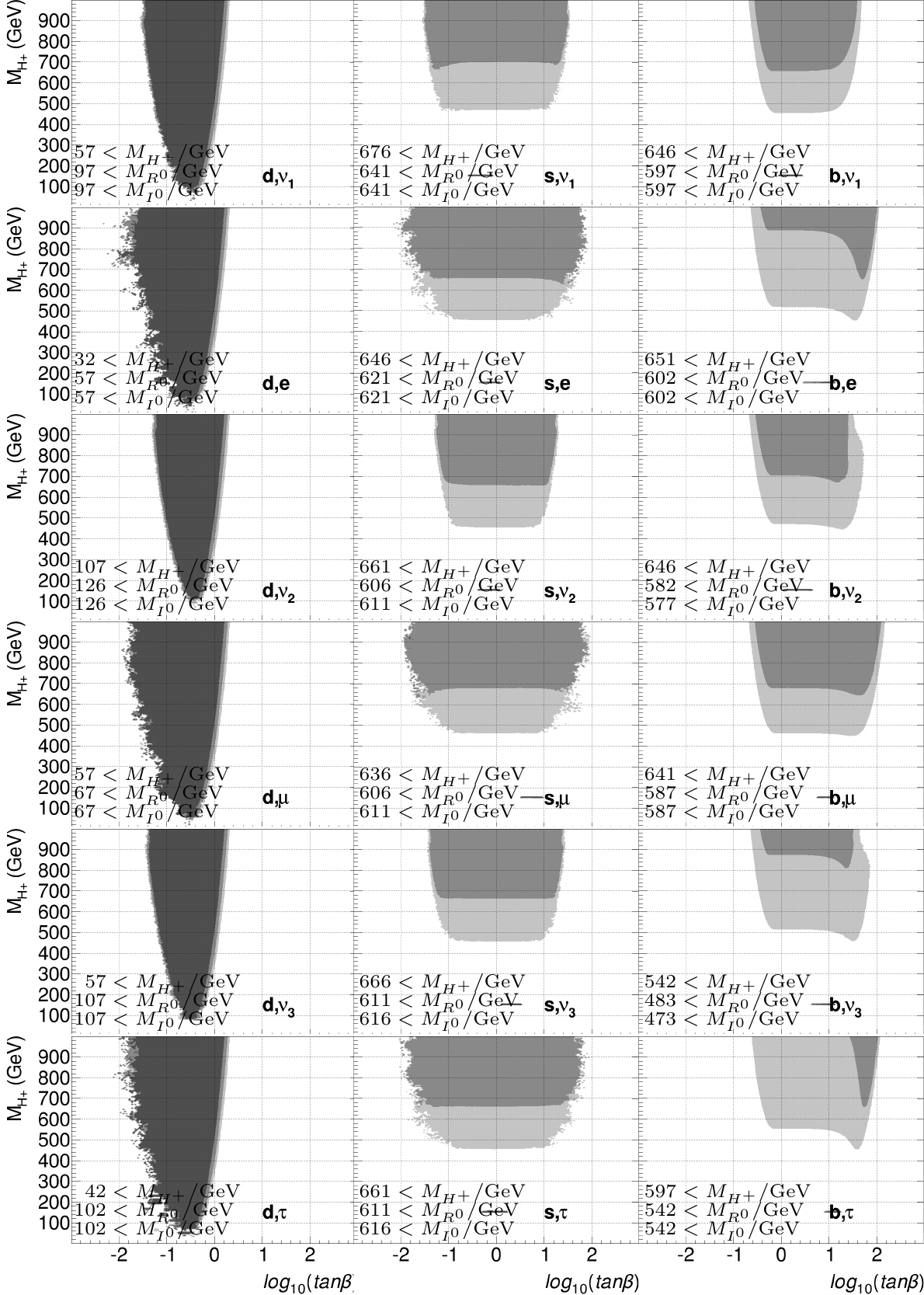}
\caption{Allowed 68\% (black), 95\% (gray) and 99\% (light gray) CL regions in $m_{H^\pm}$ vs. $\tan\beta$ for BGL models of types $(d_i,\nu_j)$ and $(u_i,\ell_j)$, i.e. for models with FCNC in the up quark sector and in the charged lepton or neutrino sector (respectively). Lower mass values corresponding to 95\% CL regions are shown in each case.\label{FIG:Results02}}
\end{figure}
\clearpage

\clearpage{}

\cleartooddpage

\clearpage{}\phantomsection
\addcontentsline{toc}{part}{II\ \ On the real representations of the Poincare
  group}
\chapter{On the real representations of the Poincare group}

\begin{epigraphs}
\qitem{A state \emph{[of a spin-0 elementary system]} which is localized
at the origin in one coordinate system, is not localized in a moving coordinate system, 
even if the origins coincide at t=0.
Hence our \emph{[position]} operators have no simple
covariant meaning under relativistic transformations.[...]

For higher but finite \emph{[spin of a massless representation]} s, beginning with s=1
(i.e. Maxwell's equations) we found that no localized states in the
above sense exist. This is an unsatisfactory, if not unexpected,
feature of our work.}
{---\textup{E.P.Wigner \& T.D.Newton (1949)\cite{newton}}}
\qitem{The concepts
of mathematics are not chosen for their conceptual simplicity---even
sequences of pairs of numbers \emph{[i.e. the real numbers]} are far from being
the simplest concepts---but for their amenability to clever
manipulations and to striking, brilliant arguments. Let us not forget
that the Hilbert space of quantum mechanics is the complex Hilbert
space, with a Hermitean scalar product. Surely
to the unpreoccupied mind, complex numbers are far from natural or simple
and they cannot be suggested by physical observations. Furthermore, the
use of complex numbers is in this case not a calculational trick of applied
mathematics but comes close to being a necessity in the formulation of
the laws of quantum mechanics.}
{---\textup{E.P.Wigner (1959)\cite{wignerquote}}}

\qitem{
Historically, confusion reigned in the relativistic
case, because situations requiring a description in
terms of many particles were squeezed into a formalism built to describe a single particle.\emph{[...]}
 
The essential result of Newton and Wigner is that
for single particles a notion of localizability and a
corresponding position observable are uniquely determined by relativistic kinematics when 
they exist at all. Whether, in fact, the position of such a particle
is observable in the sense of the quantum theory of
measurement is, of course, a much deeper problem;
that probably can only be decided within the context
of a specific consequent dynamical theory of particles.\emph{[...]}

I. For every Borel set, $S$, of three-dimensional Euclidean space, $\mathbb{R}^3$,
there is a projection operator $E(S)$ whose expectation value is the probability of finding
the system in $S$.\\
II. $E(S_1\cap S_2)=E(S_1)E(S_2)$.\\
III. $E(S_1\cup S_2)=E(S_1)+E(S_2)-E(S_1\cap S_2)$. If $S_i$, $i=1,2,...$ are disjoint sets then
$E(\cup S_i)=\sum_{i=1} E(S_i)$.\\
IV. $E(\mathbb{R}^3)=1$.\\
V'. $E(R\,S+\mathbf{a})=U(\mathbf{a},R)E(S)U(\mathbf{a},R)^{-1}$, where $R\,S+\mathbf{a}$ is the set obtained from $S$ 
by carrying out the rotation $R$ followed by the translation $\mathbf{a}$, and $U(\mathbf{a},R)$ is the unitary operator
whose application yields the wave function rotated by $R$ and translated by  $\mathbf{a}$.\\
\emph{[...]} I venture to say that any notion of localizability in three-dimensional space which does not satisfy I...V'
will represent a radical departure from present physical ideas.
}{---\textup{A.S. Whightman (1962)\cite{wightman}}}
\end{epigraphs}

\section{Introduction}
\label{intro}

\subsection{Motivation}
The Poincare
group was first introduced as the set of transformations
that leave invariant the Maxwell equations for the classical
electromagnetic field. The complex representations of the Poincare group were systematically
studied\cite{wigner,*ohnuki,*poincare,*knapp} 
and used in the definition of quantum fields\cite{feynmanrules,*symmetry}.

The formulation of quantum mechanics 
with a complex Hilbert space is equivalent to a formulation with a real Hilbert space and 
particular density matrix and observables\cite{realQM}. Moreover, for time-independent Hamiltonian, quantum mechanics 
can be defined as the eigenvalue problem $(H(\vec{x})-E)\Psi(\vec{x})=0$\cite{timeqm1}, 
in the relativistic version, the energy may be replaced by the mass squared in the equation $((\eta^{\mu\nu}\partial_\mu\partial_\nu)-m^2)\Psi(x)=0$.
Quantum Theory on real Hilbert spaces
was investigated before\cite{realqft,*realqftII,*realqftIII,*quantumstatistics,*hestenes_old},
the main conclusion was that the formulation of non-relativistic Quantum Mechanics with a 
real Hilbert space is necessarily equivalent to the complex formulation. 
We could not find in the literature a systematic study on the real
representations of the Poincare group, as it seems to be 
common assumptions:\\
1) since non-relativistic Quantum Mechanics is necessarily 
complex then the relativistic version must also be---it is hard to accept this
as relativistic causality requires the existence of anti-particles\cite{weinberg};\\
2) the energy positivity implies complex Poincare representations---it is not necessarily true as
only in a many-particles description the energy positivity is well defined, 
remember the Feynman--Stueckelberg or the Dirac sea interpretations of anti-particles\cite{diracsea,*feynmanstuckelberg};\\
3) Quantum Field Theory based on the Wightman axioms (which assume complex Poincare representations)
is the most general framework incorporating the physics principles of Quantum Mechanics and Poincare covariance---
the quantization of gauge fields does not respect Wightman axioms\cite{nonperturbativefoundations},
attempts to define non-perturbatively a Quantum Field Theory with gauge interactions involve string theory or 
space-times with dimensions lower than 4, Euclidean metric or toroidal topology, in this context studying the real 
Poincare representations cannot be considered a departure from physics principles.

The reasons motivating the study of the real representations of the Poincare group are:

1) The real representations of the Poincare group play an
important role in the classical electromagnetism and general
relativity\cite{classicalfields,etg} and in Quantum Theory---
e.g. the Higgs boson, Majorana fermion or quantum electromagnetic fields transform as real
representations under the action of the Poincare group.

2) The parity---included in the full Poincare group---and
charge-parity transformations are not symmetries of the Electroweak
interactions\cite{brancocp}. It is not clear why the charge-parity is
an apparent symmetry of the Strong interactions\cite{strongcp} or how
to explain the matter-antimatter asymmetry\cite{imbalance} through the
charge-parity violation. Since the self-conjugate finite-dimensional 
representations of the identity component of the Lorentz group are also 
representations of the parity, this work may be useful in future studies
of the parity and charge-parity violations.

3) The localization of complex irreducible
unitary representations of the Poincare
group is incompatible with causality, Poincare covariance and
energy positivity\cite{localization,causality,*hegerfeldt,thaller,*thaller2,stringfields}, while the
complex representation corresponding to the photon is not
localizable\cite{newton,vara}. 
In contrast to the classical theory, in Quantum Field Theory with gauge interactions
it impossible to define the electric charge localization of a large family of
charged states in a meaningful way\cite{nonperturbativefoundations}.
Since the free Dirac equation is 
self-conjugate in the Majorana basis, this study may be useful to the
definition of a Poincare covariant position operator as a
projection-valued measure (which Wightman considered to express the 
physical idea of localizability\cite{wightman}).

\subsection{Systems on real and complex Hilbert spaces}
The position operator in Quantum Mechanics is mathematically expressed using
a system of imprimitivity: a set of projection operators---
associated with the coordinate space---acting on a Hilbert space;
a group of symmetries acting both on the Hilbert space and on the coordinate space 
in a consistent way\cite{vara,mathQM}.

Many representations of a group---such as the finite-dimensional
representations of semisimple Lie  groups\cite{Hall} or the unitary
representations of separable locally compact
groups\cite{locallycompact}---are direct sums (or integrals) of
irreducible representations, hence the study of these representations 
reduces to the study of the irreducible representations.

If the set of normal operators commuting with an irreducible real unitary 
representation of the Poincare group is isomorphic to the quaternions or to the complex numbers,
then the most general position operator that the representation space admits is not complex linear,
but real linear. Therefore, in this case, the real irreducible representations 
generalize the complex ones and these in turn generalize the quaternionic ones.

The study of irreducible representations on complex Hilbert
spaces is in general easier than on real Hilbert spaces, because the
field of complex numbers is the algebraic closure ---where any
polynomial equation has a root--- of the field of real numbers. 
There is a well studied map, one-to-one or two-to-one and surjective
up to equivalence, from the complex to the real linear
finite-dimensional irreducible representations of a real Lie
algebra\cite{realalgebras,*realirrep}.

Section \ref{section:Systems} reviews and extends that map from the complex
to the real irreducible representations---finite-dimensional or
unitary---of a Lie group on a Hilbert space. 
Using Mackey's imprimitivity theorem, we extend the map further to systems of 
imprimitivity. This section follows closely the reference\cite{realalgebras}, with the 
addition that we will also use the Schur's lemma for unitary representations
on a complex Hilbert space\cite{schur}.

Related studies can be found in the references\cite{compactlie,spinorsrealhilbert}.

\subsection{Finite-dimensional representations of the Lorentz group}

The Poincare group, also called inhomogeneous Lorentz group, is the
semi-direct product of the translations and Lorentz Lie
groups\cite{Hall}. Whether or not the Lorentz and Poincare groups
include the parity and time reversal transformations depends on the
context and authors. To be clear, we use the prefixes full/restricted
when including/excluding parity and time reversal transformations. The
Pin(3,1)/SL(2,C) groups are double covers of the full/restricted
Lorentz group\cite{pin}. The semi-direct product of the translations with the
Pin(3,1)/SL(2,C) groups is called IPin(3,1)/ISL(2,C) Lie group --- the
letter (I) stands for inhomogeneous.

A projective representation of the Poincare group on a complex/real
Hilbert space is an homomorphism, defined up to a complex phase/sign,
from the group to the automorphisms of the Hilbert space. Since the
IPin(3,1) group is a double cover of the full Poincare group, their
projective representations are the same\cite{pin}. 
All finite-dimensional projective representations of a simply
connected group, such as SL(2,C), are usual
representations\cite{weinberg}.
Both SL(2,C) and Pin(3,1) are semi-simple Lie groups, and so all its
finite-dimensional representations are direct sums of irreducible
representations\cite{Hall}. Therefore, the study of the
finite-dimensional projective representations of the restricted
Lorentz group reduces to the study of the finite-dimensional
irreducible representations of SL(2,C).

The Dirac spinor is an element of a 4 dimensional complex vector
space, while the Majorana spinor is an element of a 4 dimensional real
vector space\cite{todorov,*irreducible,*pal,*dreiner}. The complex
finite-dimensional irreducible representations of SL(2,C) can be
written as linear combinations of tensor products of Dirac spinors.

In Section \ref{section:Finite} we will review the Pin(3,1) and
SL(2,C) semi-simple Lie groups and its relation with the Majorana,
Dirac and Pauli matrices. We will obtain all the real
finite-dimensional irreducible representations of SL(2,C) as linear
combinations of tensor products of Majorana spinors, using the map
from Section \ref{section:Systems}.
Then we will check that all these real representations are also
projective representations of the full Lorentz group, in contrast with
the complex representations which are not all projective
representations of the full Lorentz group. We could not find these results 
explicitly in the literature but they are straightforward to derive and so 
probably known by some people, the results are derived here for completeness
and explicitness.

\subsection{Unitary representations of the Poincare group}

According to Wigner's theorem, the most general transformations,
leaving invariant the modulus of the internal product of a 
Hilbert space, are: unitary or anti-unitary operators, defined up to a
complex phase, for a complex Hilbert space; unitary, defined
up to a signal, for a real Hilbert space\cite{wignertheorem,vara}. This motivates
the study of the (anti-)unitary projective representations of the full
Poincare group.

All (anti-)unitary projective representations of ISL(2,C) are, up to
isomorphisms, well defined unitary representations, because ISL(2,C)
is simply connected\cite{weinberg}. Both ISL(2,C) and IPin(3,1) are
separable locally compact groups and so all its (anti-)unitary
projective representations are direct integrals of irreducible
representations\cite{locallycompact}. Therefore, the study of the
(anti-)unitary projective representations of the restricted Poincare
group reduces to the study of the unitary irreducible representations
of ISL(2,C).

The spinor fields, space-time dependent spinors, are solutions of the
free Dirac equation\cite{Dirac}. The real/complex Bargmann-Wigner
fields\cite{BW,*allspins}, space-time dependent linear combinations of
tensor products of Majorana/Dirac spinors, are solutions of the free
Dirac equation in each tensor index. The complex unitary irreducible
projective representations of the Poincare group with discrete spin or
helicity can be written as complex Bargmann-Wigner fields.

In Section \ref{section:Unitary}, we will obtain all the real
unitary irreducible projective representations of the Poincare
group, with discrete spin or helicity, as real Bargmann-Wigner fields, using the
map from Section~\ref{section:Systems}. For each pair of complex representations (of ISL(2,C))
with positive/negative energy, there is one real representation.
We will define the Majorana-Fourier and Majorana-Hankel unitary
transforms of the real Bargmann-Wigner fields, 
relating the coordinate space with the linear and angular momenta
spaces.
We show that any localizable unitary representation of the Poincare group (ISL(2,C)), 
compatible with Poincare covariance, verifies: 
1) it is a direct sum of irreducible representations which are massive or massless with discrete helicity.
2) it respects causality;
3) if it is complex it contains necessarily both positive and negative energy subrepresentations
4) it is an irreducible representation of the Poincare group (including parity) if and only if it is: 
a)real and b)massive with spin 1/2 or massless with helicity 1/2. If a) and b) are verified the position operator 
matches the coordinates of the Dirac equation.

The free Dirac equation is diagonal in the Newton-Wigner
representation\cite{newton}, related to the Dirac representation
through a Foldy-Wouthuysen transformation\cite{revfoldy,foldy} of
Dirac spinor fields. The Majorana-Fourier transform, when applied on
Dirac spinor fields, is related with the Newton-Wigner representation
and the Foldy-Wouthuysen transformation. In the context of Clifford
Algebras, there are studies on the geometric square roots of -1 
\cite{squareroot} and on the
generalizations of the Fourier transform\cite{clifford}, with
applications to image processing\cite{image}.
It was showed before that point localized local quantum fields---operator valued distributions
satisfying the Wightman axioms--- cannot be a massless infinite spin representation
\cite{zeromass}.

The current literature related with the position operator of 
representations of the Poincare group include:
modular and string-like localization in the context of local quantum field theory 
\cite{modular0,*modular,*modular2,*modular3};
non-commutative coordinates
\cite{noncommutative,*noncommutative2};
coordinates based on equations
\cite{equations1,*equations2};
unsharp (fuzzy) localization using positive operator valued measures
\cite{povm1,*povm2,*povm3,*povm4,*povm5};
localization of the energy density
\cite{energydensity0,*energydensity1,*energydensity2,*energydensity3,*energydensity4};
two dimensional\cite{2D} or axial symmetric\cite{axial} or space-time\cite{spacetime} localization of photons; 
pseudo-hermitian representations of Quantum Mechanics\cite{pseudohermitian}. 
All the above mentioned approaches departure, in one way or another, 
from using the system of imprimitivity to implement the position operator for a unitary representation of the Poincare group.
Finally taking the Newton-Wigner position seriously has the problem that for spin one-half the position 
does not coincide with the coordinates appearing in the Dirac equation, 
with all the phenomenological consequences that it implies\cite{seriousNW,*covariance}.
The results presented in Section \ref{section:Unitary} are a motivation to not  
departure from the systems of imprimitivity to describe the position of relativistic systems.

\subsection{Energy Positivity}

While it seems that the localization of particles in either relativistic quantum mechanics\cite{modular0} or relativistic quantum field theory\cite{localizedparticles} is not possible (at least with the properties we would expect), it is not clear whether this is a limitation of the relativistic Quantum framework itself, or due to some properties of our definition of particles which are incompatible with a proper definition of localization. For instance, it is expected that by specifying the energy-momentum properties of the vacuum we then may have troubles to define the localization of all the states related with the vacuum---a consequence of the Reeh-Schlieder theorem\cite{schlieder}---, since momentum and position do not commute; that does not imply that the localization cannot be defined at all within the relativistic Quantum framework once we relax the energy-momentum properties of the vacuum.

In non-relativistic Quantum Mechanics the time is invariant under the
Galilean transformations ---excluding the time reversal transformation---and so 
the generator of translations in time is also invariant. 
Therefore, the positivity of the Energy and the localization in space of a state
can be defined simultaneously.
In relativistic Quantum Mechanics, the time is not invariant under Lorentz transformations, as a consequence 
the positivity of the Energy and the localization in space of a state cannot be defined simultaneously---
the corresponding projection operators do not commute.

In the framework of Algebraic QFT(related with the Wightman axioms), there is a definition for Quantum Field Theories with interactions in terms of formal power series in the coupling constants\cite{modular0}, which is based on the canonical quantization where the positivity of Energy is well defined
by construction and the localization problem is handled by introducing anti-particles---causality implies the
existence of anti-particles\cite{weinberg}, a related approach led Dirac to predict the positron\cite{diracsea}.
Yet, it is also possible to build a description of a many particles system where the localization in space of a state 
is well defined by construction and the Energy positivity problem can be handled with the Feynman--Stueckelberg interpretation for anti-particles, as Energy positivity and localization are complementary. Dirac himself was the first to consider an approach which do not assume the positivity of Energy by construction\cite{negativeprobs} 
and quantization in de Sitter space-time may be achieved in a related approach\cite{desitter}.

If we need to revisit known results at a deep level to non-perturbatively define a Quantum Field Theory with interactions\cite{prize}, defining the position operator with a projection-valued measure(which Wightman considered to represent the physical idea of localizability) in a many-particles system seems to be useful. In the known formulations of Algebraic QFT there are no pure states (potentially preferring an ensemble interpretation of Quantum Mechanics)\cite{modular0} and so the position operator is not defined with a projection-valued measure. 

The description of a many-particles system based on the definition of the position operator with a projection-valued measure will be discussed in the section\ref{section:Energy}.

\section{Systems on real and complex Hilbert spaces}
\label{section:Systems}

\begin{defn}[System]
A system $(M,V)$ is defined by:\\
1) the (real or complex) Hilbert space $V$;\\
2) a set $M$ of bounded endomorphisms on $V$.
\end{defn}

The representation of a symmetry is an example of a system: a representation space plus a set of operators 
representing the action of the symmetry group in the representation space\cite{representations}. 

\begin{defn}[Complexification]
Consider a system $(M,W)$ on a
real Hilbert space. The system $(M,W^c)$ is the complexification of the
system $(M,W)$, defined as 
$W^c\equiv  \mathbb{C}\otimes W$, with the multiplication by scalars
such that $a(b w)\equiv (ab)w$
for $a,b\in \mathbb{C}$ and $w\in W$.
The internal product of $W^c$ is defined---for $u_r,u_i,v_r,v_i\in W$ and $<v_r,u_r>$ the internal product of $W$---as:
\begin{align*}
<v_r+i v_i,u_r+i u_i>_c\equiv
<v_r,u_r>+<v_i,u_i>+i<v_r,u_i>-i<v_i,u_r>
\end{align*}
\end{defn}

\begin{defn}[Realification]
Consider a system $(M,V)$ on a
complex Hilbert space. 
The system$(M,V^r)$ is the realification of the
system $(M,V)$, defined as 
$V^r\equiv V$ is a real Hilbert space 
with the multiplication by scalars restricted to reals
such that $a(v)\equiv (a+i0)v$ 
for $a\in \mathbb{R}$ and $v\in V$. 
The internal product of $V^r$ is defined---for $u,v\in V$ and $<v,u>$ is the internal product of $V$---as:
\begin{align*}
<v,u>_r\equiv \frac{<v,u>+<u,v>}{2}
\end{align*}
\end{defn}

\begin{rmk}
Let $H_n$, with $n\in\{1,2\}$, be two Hilbert spaces with internal
products $<,>:H_n\times H_n\to
\mathbb{F}$,($\mathbb{F}=\mathbb{R},\mathbb{C}$).
A (anti-)linear operator $U:H_1\to H_2$ is (anti-)unitary iff:\\
1) it is surjective;\\
2) for all $x\in H_1$, $<U(x) , U(x)>=<x, x>$.
\end{rmk}

\begin{prop}
\label{prop:unitary}
Let $H_n$, with $n\in\{1,2\}$, be two complex Hilbert spaces and
$H^r_n$ its complexification.
The following two statements are equivalent:

1) The operator $U:H_1\to H_2$ is (anti-)unitary;

2) The operator $U^r:H_1^r\to H_2^r$ is (anti-)unitary, where 
$U^r(h)\equiv U(h)$, for $h\in H_1$.
\end{prop}
\begin{proof}
Since $<h,h>=<h,h>_r$ and 
$U^r(h)=U(h)$, for $h\in H_1$, we get the result.
\end{proof}

\begin{defn}[Equivalence]
Consider the systems $(M,V)$ and $(N,W)$:\\
1) A normal endomorphism of $(M,V)$ is a bounded endomorphism $S:V\to V$
commuting with $S^\dagger$ and $m$, for all $m\in M$; an anti-endomorphism in a 
complex Hilbert space is an anti-linear endomorphism;\\
2) An isometry of $(M,V)$ is a unitary operator $S:V\to V$
commuting with $m$, for all $m\in M$;\\
3) The systems $(M,V)$ and $(N,W)$ are unitary equivalent iff there is a isometry
$\alpha:V\to W$ such that $N=\{\alpha m\alpha^{\dagger}: m\in M\}$.\\
\end{defn}

We use the trivial extension of the definition of irreducibility from representations to systems.

\begin{defn}[Irreducibility]
Consider the system $(M,V)$ and let $W$ be a linear subspace of $V$:\\
1) $(M,W)$ is a (topological) subsystem of $(M,V)$
iff $W$ is closed and invariant under the system action, that is, for all $w\in W$:$(m w)\in W$, for all $m\in M$;\\
2) A system $(M,V)$ is (topologically) irreducible iff their only 
sub-systems are the non-proper $(M,V)$ or trivial  $(M,\{0\})$ sub-systems, 
where $\{0\}$ is the null space.
\end{defn}

\begin{defn}[Structures]
1) Consider a system $(M,V)$ on a
complex Hilbert space. 
A C-conjugation operator of $(M,V)$ is an 
anti-unitary involution of $V$ commuting with $m$,
for all $m\in M$;\\
2) Consider a system $(M,W)$ on a
real Hilbert space. 
A R-imaginary operator of $(M,W)$, $J$, is an isometry of
$(M,W)$ verifying $J^2=-1$. 
\end{defn}

\subsection{The map from the complex to the real systems}

\begin{defn}
Consider an irreducible system $(M,V)$ on a
complex Hilbert space:\\
1) The system is C-real iff there is a C-conjugation
operator;\\
2) The system is C-pseudoreal iff 
there is no C-conjugation operator but there is an
anti-unitary operator of $(M,V)$;\\
3) The system is C-complex iff 
there is no anti-unitary operator of $(M,V)$.
\end{defn}

\begin{defn}
\label{defn:Rsystem}
Consider the system $(M,W)$ on a
real Hilbert space and let $(M,W^c)$ be its complexification:
1) $(M,W)$ is R-real iff $(M,W^c)$ is C-real irreducible;\\
2) $(M,W)$ is R-pseudoreal iff $(M,V)$ is C-pseudoreal irreducible,
with $W^c=V\oplus \bar V$;
3) $(M,W)$ is R-complex iff $(M,V)$ is C-complex irreducible, with
$W^c=V\oplus \bar V$.
\end{defn}

\begin{prop}
Any irreducible real system is R-real or R-pseudoreal or R-complex.
\end{prop}
\begin{proof}
Consider an irreducible system $(M,W)$
on a real Hilbert space.
There is a C-conjugation operator of $(M,W^c)$, 
$\theta$, defined by $\theta(u+iv)\equiv (u-iv)$ for $u,v\in W$,
verifying $(W^c)_\theta=W$.

Let $(M,X^c)$ be a proper non-trivial subsystem of
$(M,W^c)$. Then $\theta$ is a C-conjugation operator of the 
subsystems $(M,Y^c)$ and $(M,Z^c)$, where 
$Y^c\equiv \{u+\theta v: u,v\in X^c\}$ and 
$Z^c\equiv \{u: u,\theta u\in X^c\}$. Therefore,
$Y^c=\{u+iv: u,v\in Y\}$ and $Z^c=\{u+iv: u,v\in Z\}$, 
where $Y\equiv\{\frac{1+\theta}{2}u: u\in Y^c\}$ and
$Z\equiv\{\frac{1+\theta}{2}u: u\in Z^c\}$, are invariant closed subspaces of
$W$. 
If $Y=\{0\}$ then $Z=\{0\}$ and $Y^c=X^c=\{0\}$, in contradiction
with $X^c$ being non-trivial. 
If $Z=W$ then $Y=W$ and $Z^c=X^c=W^c$, in contradiction 
with $X^c$ being proper.
Therefore $Z=\{0\}$ and $Y=W$, which implies 
$Z^c=\{0\}$ and $Y^c=W^c$.

So, $(M,W)$ is equivalent to $(M,(X^c)^r)$, 
due to the existence of the bijective linear 
map $\alpha:(X^c)^r\to W$, $\alpha(u)=u+\theta u$,
$\alpha^{-1}(u+\theta u)=u$, for $u\in(X^c)^r$.
Suppose that there is a C-conjugation operator of $(M,X^c)$,
$\theta'$. Then $(M,W_{\pm})$ is a proper non-trivial
subsystem of $(M,W)$, where
$W_\pm\equiv\{\frac{1\pm\theta'}{2}w: w\in W\}$, in contradiction with
$(M,W)$ being irreducible.
\end{proof}

\begin{prop}
\label{prop:Rirreducible}
Any real system which is R-real or R-pseudoreal or
R-complex is irreducible.
\end{prop}
\begin{proof}
Consider an irreducible system
on a complex Hilbert space $(M,V)$.
There is a R-imaginary operator $J$ of the system
$(M ,V^{r})$, defined by $J(u)\equiv i u$, for $u\in V^r$.

Let $(M,X^r)$ be a proper non-trivial subsystem of
$(M,V^{r})$. Then $J$ is an R-imaginary operator of
$(M,Y^r)$ and $(M^r,Z^r)$, 
where $Y^r\equiv \{u+J v: u, v \in X^r\}$ and 
$Z^r\equiv \{u: u,Ju\in X^r\}$.
Then $(M,Y)$ and $(M,Z)$ are
subsystems of $(M,V)$,
where the complex Hilbert spaces 
$Y\equiv Y^r$ and $Z\equiv Z^r$ have the scalar 
multiplication such that $(a+ib)(y)=ay+bJy$, 
for $a,b\in\mathbb{R}$ and $y\in Y$ or $y\in Z$.
If $Y=\{0\}$, then $Z=X^r=\{0\}$ which is in contradiction with $X^r$
being non-trivial.
If $Z=V$, then $Y=V$ and $X^r=V^r$ which is in contradiction with
$X^r$ being non-trivial.
So $Z=\{0\}$ and $Y=V$, which implies that $V=(X^r)^c$.

Then there is a C-conjugation operator of $(M,V)$, 
$\theta$, defined by $\theta(u+iv)\equiv u-iv$, for $u,v\in X^r$. 
We have $X^r=V_\theta$. 
Suppose there is a R-imaginary operator of $(M,V_\theta)$, 
$J'$. Then $(M,V_{\pm})$, where
$V_{\pm}\equiv\{\frac{1\pm iJ'}{2}v: v\in V\}$, 
are proper non-trivial subsystems of $(M,V)$, 
in contradiction with $(M,V)$ being irreducible.

Therefore, if $(M,V)$ is C-real, then $(M,V_\theta)$ is
R-real irreducible. 
If $(M,V)$ is C-pseudoreal or C-complex, then 
 $(M,V_\theta^r)$ is R-pseudoreal or R-complex, irreducible. 
\end{proof}

\subsection{Schur Systems}

\begin{defn}[Schur System]
\label{defn:schur}
A system $(M,V)$, on a complex Hilbert space $V$, 
is a Schur system if the set of normal operators of $(M,V)$ is isomorphic to $\mathbb{C}$.\\
Consider an irreducible system $(M,W)$, on a real Hilbert space $W$ and 
let $(M,W^c)$ be its complexification:
1) $(M,W)$ is Schur R-real iff $(M,W^c)$ is Schur C-real;\\
2) $(M,W)$ is Schur R-pseudoreal iff $(M,V)$ is Schur C-pseudoreal,
with $W^c=V\oplus \bar V$;\\
3) $(M,W)$ is Schur R-complex iff $(M,V)$ is Schur C-complex, with
$W^c=V\oplus \bar V$.
\end{defn}

\begin{lem}
Consider a Schur system $(M,V)$ on a complex
Hilbert space. An anti-isometry of $(M,V)$, if it exists, is unique
up to a complex phase.
\end{lem}

\begin{proof}
Let $\theta_1$,$\theta_2$ be two anti-isometries of $(M,V)$. The
product $(\theta_2\theta_1)$ is an isometry of $(M,V)$;
since $(M,V)$ is irreducible,
$(\theta_2\theta_1)=e^{i\phi}$; with $\phi\in \mathbb{R}$.

Therefore $\theta_2=\alpha \theta_1\alpha^{-1}$; where 
$\alpha\equiv e^{i\frac{\phi}{2}}$ is a complex phase.
\end{proof}

\begin{prop}
Two R-real Schur systems are
isometric iff their complexifications are isometric.
\end{prop}
\begin{proof}
Let $(M,V)$ and $(N,W)$ be C-real Schur systems,
with $\theta_M$ and $\theta_N$ the respective C-conjugation operators.
If there is an isometry $\alpha:V\to W$ such that 
$\alpha M=N\alpha$, then 
$\vartheta\equiv \alpha\theta_M\alpha^{-1}$ is an anti-isometry of
$(N,W)$. Since it is unique up to a phase, then
$\theta_N=e^{i\phi}\vartheta$. Therefore $e^{i\frac{\phi}{2}}\alpha$ is an
isometry between $(M,V_\theta)$ and $(N,W_\theta)$, where
$V_{\theta_M}\equiv\{(1+\theta_M)v: v\in V\}$.
\end{proof}

\begin{prop}
Two C-complex or C-pseudoreal Schur systems are isometric or anti-isometric 
iff their realifications are isometric.
\end{prop}
\begin{proof}
Let $(M,V)$ and $(N,W)$ be R-complex or R-pseudoreal Schur
systems, with $J_M$ and $J_N$ the respective R-imaginary
operators. If there is an isometry $\alpha:V\to W$ such that
$\alpha M=N\alpha$, then 
$K\equiv \alpha J_M\alpha^{-1}$ is a R-imaginary operator of
$(N,W)$. 
When considering $(N,W_{J_N})$ and $(M,V_{J_M})$, where 
$W_{J_N}\equiv \{(1-iJ_N) w: w\in W\}$,  we get that 
$(1-J_N K)(1-K J_N)=r$ as an operator of $W_{J_N}$, where $r$ is a
non-negative  null real scalar. If $c=0$ then $K=-J_N$ and $\alpha$
defines an anti-isometry between $(M,V_{J_M})$ and $(N,W_{J_N})$.
If $c\neq 0$ then $(1-J_N K)\alpha c^{-\frac{1}{2}}$ is an isometry
between $(M,V_{J_M})$ and $(N,W_{J_N})$.
\end{proof}

\begin{prop}
The space of normal operators of a R-real Schur system is
isomorphic to $\mathbb{R}$.
\end{prop}
\begin{proof}
Let $(M,V)$ be a C-real Schur system,
with $\theta$ the C-conjugation operator.
If there is an endomorphism $\alpha:V\to V$ such that 
$\alpha M=M\alpha$, we know that $\alpha=re^{i\varphi}$. Then the endomorphism of $V_\theta$
is a real number.
\end{proof}

\begin{prop}
The space of normal operators of a R-complex Schur system is
isomorphic to $\mathbb{C}$.
\end{prop}
\begin{proof}
Let $(M,V)$ be a R-complex Schur
system, with $J$ the R-imaginary
operator. 
If there is a normal operator $\alpha$ of $(M,V)$, then
$KK^\dagger$ is a normal operator of the C-complex Schur system
$(M,V_{J})$, where $K\equiv
(\alpha+J\alpha J)$ and $V_{J}\equiv \{(1-iJ) v: v\in V\}$.
If $KK^\dagger=r>0$, then $\frac{K}{\sqrt{r}}$ is unitary and $V_J$ is
equivalent to $\overline{V}_{J}$ which would imply that $(M,V)$ is
C-pseudoreal. Therefore $K=0$ and hence $\alpha$ is a normal operator of 
$(M,V_{J})$, so $\alpha=r e^{J\theta}$.
\end{proof}

\begin{prop}
The space of normal operators of a R-pseudoreal Schur system is
isomorphic to $\mathbb{H}$ (quaternions).
\end{prop}
\begin{proof}
Let $(M,V)$ be a R-pseudoreal Schur
system, with $J$ the R-imaginary
operator. If there is an endomorphism $\alpha$ of $(M,V)$, then
$SS^\dagger$ and $TT^\dagger$ are a self-adjoint endomorphisms of the C-complex Schur system
$(M,V_{J})$, where $S\equiv(\alpha-J\alpha J)/2$,
$T\equiv(\alpha+J\alpha J)/2$  and 
$V_{J}\equiv\{(1-iJ) v: v\in V\}$.
Let $K$ be an unitary operator of
$(M,V)$ and anti-commuting with $J$, then $K^2=e^{J\theta}$ and 
$Ke^{J\theta}=K(K^2)=(K^2)K=e^{J\theta}K$, therefore $K^2=-1$.
If $TT^\dagger=t>0$, then $\frac{T}{\sqrt{t}}$ is unitary and
anti-commutes with $J$, $TK$ is a normal endomorphism of  
$(M,V_{J})$ and therefore $T=Kc+KJd$; if $TT^\dagger=0$ then 
$c=d=0$.
If $SS^\dagger=s>0$, then $\frac{S}{\sqrt{s}}$ is unitary and
commutes with $J$, $S$ is a normal endomorphism of  
$(M,V_{J})$ and therefore $S=a+Jb$; if $SS^\dagger=0$ then 
$a=b=0$.

Therefore $\alpha=S+T=a+Jb+Kc+KJd$, 
which is isomorphic to the quaternions.
\end{proof}

\subsection{Finite-dimensional representations}
\label{section:Finite}

\begin{lem}[Schur's lemma for finite-dimensional representations\cite{schur}]
\label{lem:commuting}
Consider an irreducible finite-dimensional representation $(M_G,V)$ of
a Lie group $G$ on a complex Hilbert space $V$. If the representation
$(M_G,V)$ is irreducible then any endomorphism $S$ of $(M_G,V)$ is a
complex scalar.
\end{lem}

\begin{lem}
Consider an irreducible complex finite-dimensional representation $(M,V)$ on a
complex Hilbert space. Then there is internal product such that:
1) The system is C-real iff there is an anti-linear involution of $(M,V)$;\\
2) The system is C-pseudoreal iff 
there is not an anti-linear bounded involution of $(M,V)$, but there is an
anti-isomorphism of $(M,V)$;\\
3) The system is C-complex iff 
there is no anti-isomorphism of $(M,V)$.
\end{lem}
\begin{proof}
Let $S$ be an anti-isomorphism of an irreducible representation $(M,V)$.
Then $S^2=re^{i\varphi}$. But $S^2$ commutes with $S$ which is anti-linear, 
so $S^2=\pm r$. So, there is an internal product such that $S$ is anti-unitary.
\end{proof}

\begin{defn}
A finite-dimensional system is completely reducible iff it can
be expressed as a direct sum of irreducible systems.
\end{defn}

\begin{rmk}[Weyl theorem]
All finite-dimensional representations of a semi-simple Lie group 
(such as SL(2,C)) are completely reducible.
\end{rmk}

\subsection{Unitary representations and Systems of Imprimitivity}
\label{section:Unitary}
\begin{defn}[Normal System]
A System $(M,V)$ is normal iff $M$ is a set $M$ of normal operators on $V$
closed under Hermitian conjugation---for all $m\in M$ there is $n\in M$ such that
$n=m^\dagger$.
\end{defn}

A unitary representation or a System of Imprimitivity are examples of a normal System.

\begin{rmk}
$W^\bot$ is the orthogonal complement of the subspace $W$ of the
Hilbert space $V$ if:\\
1) $V=W \oplus W^\bot$, that is, 
all $v\in V$ can be expressed as $v=w+x$, where $w\in W$ and $x\in W^\bot$;\\
2) if $w\in W$ and $x\in W^\bot$, then $x^\dagger w=0$.
\end{rmk}

\begin{lem}
\label{lem:orthogonal}
Consider a normal system $(M,V)$. Then, for all subsystem
$(M,W)$ of $(M_G,V)$, $(M_G,W^\bot)$ is also a subsystem of
$(M,V)$, where $W^\bot$ is the orthogonal complement of the subspace $W$.
\end{lem}

\begin{proof} 
Let $(M,W)$ be a subsystem of $(M,V)$.
$W^\bot$ is the orthogonal complement of $W$.

For all $x\in W^\bot$, $w\in W$ and $m\in M$, 
$<m x, w>=<x,m^\dagger w>$. 

Since $W$ is invariant and there is $n\in M$, such
that $n=m^\dagger$, then $w'\equiv (m^\dagger w)\in W$.
 
Since $x\in W^\bot$ and $w'\in W$, then $<x, w'>=0$.

This implies that if $x\in W^\bot$), also $(m x)\in W^\bot$, for all $m\in M$.
\end{proof}

\begin{lem}
\label{lem:commuting}
Any Schur normal system on a complex Hilbert space is irreducible.
\end{lem}

\begin{proof}
Let $(M,W)$ and $(M,W^\bot)$ be sub-systems of the complex Schur system $(M,V)$,
where $W^\bot$ is the orthogonal complement of $W$.

There is a bounded endomorphism $P: V\to V$, such that, 
for $w,w'\in W$, $x,x'\in W^\bot$, $P(w+x)=w$. $P^2=P$ and $P$ is
hermitian:
\begin{align}
&<w'+x',P (w+x)>=<w',w>=
<P(w'+x'),w+x>
\end{align}
Let $w'\equiv m w\in W$ and $x'\equiv m x\in
W^\bot$:
\begin{align}
m P(w+x)&=m w=w'\\
P m(w+x)&=P(w'+x')=w'
\end{align}
Which implies that $P$ commutes with all $m\in M$, so  $P\in\{0,1\}$. 
If $P=1$, then $W=V$, if $P=0$, then $W$ is the null space.
\end{proof}

So a complex Schur normal system is irreducible, and hence, from Defns.\ref{defn:Rsystem},\ref{defn:schur} and Prop.\ref{prop:Rirreducible}, 
a real Schur normal system is also irreducible.

\begin{lem}[Schur's lemma for unitary representations\cite{schur}]
\label{lem:commuting}
Consider an irreducible unitary representation $(M,V)$ of a Lie
group $G$ on a complex Hilbert space $V$. If the representation
$(M,V)$ is irreducible then any normal operator $N$ of $(M,V)$ is a
scalar.
\end{lem}

\begin{defn}
A unitary system is completely reducible iff it can be expressed as a
direct integral of irreducible systems.
\end{defn}

\begin{rmk}
All unitary representations of a separable locally compact group 
(such as the Poincare group) are completely reducible.
\end{rmk}

\subsection{Systems of Imprimitivity}
\label{section:Imprimitivity}

\begin{defn}
Consider a measurable space $(X, M)$, where $M$ is a $\sigma$-algebra
of subsets of $X$. A projection-valued-measure, $\pi$, is a map from
$M$ to the set of self-adjoint projections on a Hilbert space $H$ such
that $\pi(X)$ is the identity operator on $H$ and the function
$<\psi,\pi(A)\psi>$, with $A\in M$ is a measure on $M$, for all
$\psi\in H$.
\end{defn}

\begin{defn}
Suppose now that $X$ is a representation of $G$. 
Then, a system of imprimitivity is a pair $(U,\pi)$, where $\pi$ is a
projection valued measure and $U$ an unitary representation of $G$ on
the Hilbert space $H$, such that $U(g)\pi(A) U^{-1}(g)=\pi(gA)$.
\end{defn}

\begin{rmk}[Imprimitivity Theorem (thrm 6.12 \cite{commutingring,*inducedreps,*mackey,*squareimprimitivity,vara})]
\label{lem:commuting} Let $G$ be a Lie group, $H$ its closed subgroup.
Let a pair $(V, E)$ be a system of imprimitivity for $G$ based on $G/H$ on a separable complex Hilbert space.
Then there exists a representation $L$ of $H$ such that $(V, E)$ is equivalent to
the canonical system of imprimitivity $(V_L,E_L)$. For any two representations $L$, $L'$ of the subgroup 
$H$ the corresponding canonical systems of imprimitivity are equivalent if and only
if $L$, $L'$ are equivalent. The sets of normal operators commuting with $(V_L, E_L )$ and $L$ are isomorphic.
\end{rmk}

\begin{lem}[Schur's lemma for systems of imprimitivity\cite{schur}]
\label{lem:commuting}
Let a pair $(V, E)$ be a system of imprimitivity for $G$ based on $G/H$ on a separable complex Hilbert space.
If $(V, E)$ is irreducible then then any normal operator $N$ commuting with $(V,E)$ is a
scalar.
\end{lem}
\begin{proof}
Consider a representation $L$ of $H$ such that $(V, E)$ is equivalent to
the canonical system of imprimitivity $(V_L,E_L)$. 
If $L$ would be reducible then there would be a non-trivial normal projection operator commuting with $L$,
but then the imprimitivity theorem implies that there would also be a non-trivial normal projection operator commuting
with $(V,E)$ which is in contradiction with the irreducibility of $(V,E)$, therefore $L$ is irreducible.
The Schur's lemma for unitary representations then implies that any normal operator commuting with $L$ is a scalar, 
the imprimitivity theorem then implies the result.
\end{proof}

So we can define a map from the real to the complex systems of imprimitivity---analogous to the 
one for unitary representations. So we extended an existing map from the complex to the real linear
finite-dimensional irreducible representations of a real Lie
algebra\cite{realalgebras,*realirrep} to the infinite-dimensional (unitary) case.

\section{Finite-dimensional representations of the Lorentz group}
\label{section:Lorentz}
We could not find the following results 
explicitly in the literature but they are straightforward to derive and so 
probably known by some people, the results are derived here for completeness and explicitness.
\subsection{Majorana, Dirac and Pauli Matrices and Spinors}
\begin{defn}
$\mathbb{F}^{m\times n}$ is the vector space of $m\times n$ matrices whose
entries are elements of the field $\mathbb{F}$.
\end{defn}

In the next remark we state the Pauli's fundamental theorem of gamma
matrices. The proof can be found in the reference\cite{diracmatrices}.
\begin{rmk}[Pauli's fundamental theorem]
\label{rem:fundamental}
Let $A^\mu$, $B^\mu$, $\mu\in\{0,1,2,3\}$, be two sets of
$4\times 4$ complex matrices verifying:
\begin{align}
A^\mu A^\nu+A^\nu A^\mu&=-2\eta^{\mu\nu}\\
B^\mu B^\nu+B^\nu B^\mu&=-2\eta^{\mu\nu}
\end{align}
Where $\eta^{\mu\nu}\equiv diag(+1,-1,-1-1)$ is the Minkowski metric.

1) There is an invertible complex matrix $S$ such that
$B^\mu=S A^\mu S^{-1}$, for all $\mu\in\{0,1,2,3\}$. 
$S$ is unique up to a non-null scalar.

2) If $A^\mu$ and $B^\mu$ are all unitary, then $S$ is unitary.
\end{rmk}

\begin{prop}
\label{prop:realsimilar}
Let $\alpha^\mu$, $\beta^\mu$, $\mu\in\{0,1,2,3\}$, be two sets of
$4\times 4$ real matrices verifying:
\begin{align}
\alpha^\mu\alpha^\nu+\alpha^\nu\alpha^\mu&=-2\eta^{\mu\nu}\\
\beta^\mu\beta^\nu+\beta^\nu\beta^\mu&=-2\eta^{\mu\nu}
\end{align}
Then there is a real matrix $S$, with $|det S|=1$, such that
$\beta^\mu=S\alpha^\mu S^{-1}$, for all  $\mu\in\{0,1,2,3\}$. $S$ is unique up to a signal. 
\end{prop}

\begin{proof}
From remark \ref{rem:fundamental}, we know that there is an
invertible matrix $T'$, unique up to a non-null scalar, such that $\beta^\mu=T'\alpha^\mu
T^{'-1}$.
Then $T\equiv T'/|det(T')|$ has $|det T|=1$ and it is unique up to a
complex phase.

Conjugating the previous equation, we get $\beta^\mu=T^*\alpha^\mu
T^{*-1}$.
Then $T^*=e^{i 2 \theta} T$ for some real number $\theta$.
Therefore $S\equiv e^{i \theta}T$ is a real matrix,
with $|det S|=1$, unique up to a signal.
\end{proof}

\begin{defn}
The Majorana matrices, $i\gamma^\mu$, $\mu\in\{0,1,2,3\}$, are $4\times
4$ complex unitary matrices verifying:
\begin{align}
(i\gamma^\mu)(i\gamma^\nu)+(i\gamma^\nu)(i\gamma^\mu)&=-2\eta^{\mu\nu}
\end{align}
The Dirac matrices are $\gamma^\mu\equiv
-i(i\gamma^\mu)$.
\end{defn}

In the Majorana bases, the Majorana matrices are $4\times 4$ real
orthogonal matrices. An example of the Majorana matrices in a
particular Majorana basis is:
\begin{align}
\begin{array}{llllll}
\label{basis}
i\gamma^1=&\left[ \begin{smallmatrix}
+1 & 0 & 0 & 0 \\
0 & -1 & 0 & 0 \\
0 & 0 & -1 & 0 \\
0 & 0 & 0 & +1 \end{smallmatrix} \right]&
i\gamma^2=&\left[ \begin{smallmatrix}
0 & 0 & +1 & 0 \\
0 & 0 & 0 & +1 \\
+1 & 0 & 0 & 0 \\
0 & +1 & 0 & 0 \end{smallmatrix} \right]&
i\gamma^3=\left[ \begin{smallmatrix}
0 & +1 & 0 & 0 \\
+1 & 0 & 0 & 0 \\
0 & 0 & 0 & -1 \\
0 & 0 & -1 & 0 \end{smallmatrix} \right]\\
\\
i\gamma^0=&\left[ \begin{smallmatrix}
0 & 0 & +1 & 0 \\
0 & 0 & 0 & +1 \\
-1 & 0 & 0 & 0 \\
0 & -1 & 0 & 0 \end{smallmatrix} \right]&
i\gamma^5=&\left[ \begin{smallmatrix}
0 & -1 & 0 & 0 \\
+1 & 0 & 0 & 0 \\
0 & 0 & 0 & +1 \\
0 & 0 & -1 & 0 \end{smallmatrix} \right]&
=-\gamma^0\gamma^1\gamma^2\gamma^3
\end{array}
\end{align}

In reference \cite{realgamma} it is proved that the set of five
anti-commuting $4\times 4$ real matrices is unique up to
isomorphisms. So, for instance, with $4\times 4$ real matrices it is not possible
to obtain the euclidean signature for the metric.

\begin{defn}
The Dirac spinor is a $4\times 1$ complex column matrix, $\mathbb{C}^{4\times 1}$.
\end{defn}

The space of Dirac spinors is a 4 dimensional complex vector space.

\begin{lem}
The charge conjugation operator $\Theta$, is an anti-linear involution
commuting with the Majorana matrices $i\gamma^\mu$. 
It is unique up to a complex phase.
\end{lem}

\begin{proof}
In the Majorana bases, the complex conjugation is a charge conjugation
operator. Let $\Theta$ and $\Theta'$ be two charge conjugation operators
operators. Then, $\Theta\Theta'$ is a complex invertible matrix commuting with
$i\gamma^\mu$, therefore, from Pauli's fundamental theorem,  
$\Theta\Theta'=c$, where $c$ is a non-null complex scalar.
Therefore $\Theta'=c^*\Theta$ and from $\Theta'\Theta'=1$, we get that
$c^* c=1$.
\end{proof}

\begin{defn}
Let $\Theta$ be a charge conjugation operator.

The set of Majorana spinors, denoted here by $Pinor$, is the set of Dirac spinors
verifying the Majorana condition (defined up to a complex phase):
\begin{align}
Pinor\equiv \{u\in \mathbb{C}^{4\times 1}: \Theta u= u\}
\end{align}
\end{defn}

The set of Majorana spinors is a 4 dimensional real vector space. 
Note that the linear combinations of
Majorana spinors with complex scalars do not verify the Majorana
condition.

There are 16 linear independent products of Majorana matrices. These
form a basis of the real vector space of endomorphisms of Majorana spinors,
$End(Pinor)$. In the Majorana bases, $End(Pinor)$ is the vector space of
$4\times 4$ real matrices.

\begin{defn}
The Pauli matrices $\sigma^k,\ k\in\{1,2,3\}$ are $2\times 2$
hermitian, unitary, anti-commuting, complex matrices.
The Pauli spinor is a  $2\times 1$ complex column matrix. The space of
Pauli spinors is denoted by $Pauli$.
\end{defn}

The space of Pauli spinors, denoted here by $Pauli$, is a 2 dimensional complex vector
space and a 4 dimensional real vector space. The realification of
the space of Pauli spinors is isomorphic to the space of Majorana
spinors.

\subsection{On the Lorentz, SL(2,C) and Pin(3,1) groups}

\begin{rmk} 
The Lorentz group, $O(1,3)\equiv\{\lambda \in \mathbb{R}^{4\times 4}: \lambda^T \eta \lambda=\eta \}$, is the set of
real matrices that leave the metric, $\eta=diag(1,-1,-1,-1)$,
invariant.

The proper orthochronous Lorentz subgroup is defined by
$SO^+(1,3)\equiv\{\lambda \in
O(1,3): det(\lambda)=1, \lambda^0_{\ 0}>0 \}$. 
It is a normal subgroup. 
The discrete Lorentz subgroup of parity and time-reversal is 
$\Delta \equiv \{1,\eta,-\eta,-1\}$.

The Lorentz group is the semi-direct product of the previous
subgroups, $O(1,3)=\Delta \ltimes SO^+(1,3)$.  
\end{rmk}

\begin{defn}
The set $Maj$ is the 4 dimensional real space of the linear
combinations of the Majorana matrices, $i\gamma^\mu$:
\begin{align}
Maj\equiv\{a_\mu i\gamma^\mu: a_\mu\in \mathbb{R},\ \mu\in\{0,1,2,3\}\}
\end{align}
\end{defn}

\begin{defn}
$Pin(3,1)$ \cite{pin} is the group of endomorphisms of Majorana
spinors that leave the space $Maj$ invariant, that is:
\begin{align}
Pin(3,1)\equiv 
\Big\{S\in End(Pinor):\ |det S|=1,\ S^{-1}(i\gamma^\mu)S\in Maj,\ \mu\in\{0,1,2,3\} \Big\}
\end{align}
\end{defn}

\begin{prop}
\label{prop:map}
The map $\Lambda:Pin(3,1)\to O(1,3)$ defined by:
\begin{align}
(\Lambda(S))^\mu_{\ \nu}i\gamma^\nu\equiv S^{-1}(i\gamma^\mu)S
\end{align}
is two-to-one and surjective. It defines a group homomorphism.
\end{prop}

\begin{proof}
1) Let $S\in Pin(3,1)$. Since the Majorana matrices are a basis of the
real vector space $Maj$, there is an unique real matrix $\Lambda(S)$ such that:
\begin{align}
(\Lambda(S))^\mu_{\ \nu}i\gamma^\nu=S^{-1}(i\gamma^\mu)S
\end{align}
Therefore, $\Lambda$ is a map with domain $Pin(3,1)$. Now we can check
that $\Lambda(S)\in O(1,3)$:
\begin{align}
&(\Lambda(S))^\mu_{\ \alpha}\eta^{\alpha\beta}(\Lambda(S))^\nu_{\
  \beta}=-\frac{1}{2}(\Lambda(S))^\mu_{\
  \alpha}\{i\gamma^\alpha,i\gamma^\beta\}(\Lambda(S))^\nu_{\
  \beta}=\\
&=-\frac{1}{2}S\{i\gamma^\mu,i\gamma^\nu\}S^{-1}=S\eta^{\mu\nu}S^{-1}=\eta^{\mu\nu}
\end{align}
We have proved that $\Lambda$ is a map from $Pin(3,1)$ to $O(1,3)$.

2) Since any $\lambda\in O(1,3)$ conserve the metric $\eta$, the matrices
$\alpha^\mu\equiv \lambda^\mu_{\ \nu} i\gamma^\nu$ verify:
\begin{align}
\{\alpha^\mu,\alpha^\nu\}=-2\lambda^\mu_{\ \alpha}\eta^{\alpha\beta}\lambda^\nu_{\ \beta}=-2\eta^{\mu\nu}
\end{align}
In a basis where the Majorana matrices are real, from Proposition
\ref{prop:realsimilar} there is a real invertible matrix $S_\lambda$,
with $|det S_\Lambda|=1$, such that $\lambda^\mu_{\ \nu} i\gamma^\nu=S^{-1}_\lambda
(i\gamma^\mu)S_\lambda$. 
The matrix $S_\Lambda$ is unique up to a sign. So, $\pm S_\lambda\in
Pin(3,1)$ and we proved that the map
 $\Lambda:Pin(3,1)\to O(1,3)$ is two-to-one and surjective.

3) The map defines a group homomorphism because:
\begin{align}
&\Lambda^\mu_{\ \nu}(S_1)\Lambda^\nu_{\
  \rho}(S_2)i\gamma^\rho=\Lambda^\mu_{\ \nu}S_2^{-1}i\gamma^\nu S_2\\
&=S_2^{-1}S_1^{-1}i\gamma^\mu S_1 S_2=\Lambda^\mu_{\ \rho}(S_1 S_2)i\gamma^\rho
\end{align}
\end{proof}

\begin{rmk}
\label{rem:SL(2,C)}
The group $SL(2,\mathbb{C})=\{e^{\theta^j
  i\sigma^j+b^j\sigma^j}: \theta^j,b^j\in
\mathbb{R},\ j\in\{1,2,3\}\}$ is simply connected. 
Its projective representations are equivalent to its ordinary representations\cite{weinberg}.

There is a two-to-one, surjective map $\Upsilon:SL(2,\mathbb{C})\to
SO^+(1,3)$, defined by:
\begin{align}
\Upsilon^{\mu}_{\ \nu}(T)\sigma^\nu\equiv T^\dagger \sigma^\mu T
\end{align}
Where $T\in SL(2,\mathbb{C})$, $\sigma^0=1$ and $\sigma^j$, $j\in\{1,2,3\}$ are the Pauli matrices.
\end{rmk}

\begin{lem}
Consider that $\{M_+,M_-,i\gamma^5 M_+,i\gamma^5 M_-\}$ and $\{P_+,P_-,iP_+,iP_-\}$ are orthonormal basis of
the 4 dimensional real vector spaces $Pinor$ and $Pauli$, respectively, verifying:
\begin{align}
\gamma^0\gamma^3 M_\pm=\pm M_\pm&,\ \sigma^3 P_\pm=\pm P_\pm
\end{align}
The isomorphism
$\Sigma:Pauli \to Pinor$ is defined by:
\begin{align}
\Sigma(P_+)=M_+,&\ \Sigma(iP_+)=i\gamma^5 M_+\\
\Sigma(P_-)=M_-,&\ \Sigma(iP_-)=i\gamma^5 M_-
\end{align}

The group $Spin^+(3,1)\equiv \{\Sigma\circ A \circ \Sigma^{-1}:A\in
SL(2,\mathbb{C})\}$ is a subgroup of $Pin(1,3)$. 
For all $S\in Spin^+(1,3)$, $\Lambda(S)=\Upsilon(\Sigma^{-1}\circ S \circ \Sigma)$.
\end{lem}

\begin{proof}
From remark \ref{rem:SL(2,C)},  $Spin^+(3,1)=\{e^{\theta^j
  i\gamma^5\gamma^0\gamma^j+b^j\gamma^0\gamma^j}: \theta^j,b^j\in
\mathbb{R},\ j\in\{1,2,3\}\}$.
Then, for all $T\in SL(2,C)$:
\begin{align}
-i\gamma^0 \Sigma \circ T^\dagger
\circ \Sigma^{-1} i\gamma^0&=\Sigma \circ T^{-1}
\circ \Sigma^{-1}
\end{align}
Now, the map $\Upsilon:SL(2,\mathbb{C})\to
SO^+(1,3)$ is given by:
\begin{align}
\Upsilon^{\mu}_{\ \nu}(T)i\gamma^\nu = (\Sigma \circ T^{-1}
\circ \Sigma^{-1}) i\gamma^\mu (\Sigma\circ T \circ \Sigma^{-1})
\end{align}
Then, all $S\in Spin^+(3,1)$ leaves the space $Maj$ invariant:
\begin{align}
S^{-1} i\gamma^\mu S=
\Upsilon^{\mu}_{\ \nu}(\Sigma^{-1}\circ S \circ \Sigma)i\gamma^\nu
\in Maj
\end{align}
Since all the products of Majorana matrices, except the identity, are
traceless, then $det(S)=1$. So,
$Spin^+(3,1)$ is a subgroup of $Pin(1,3)$ and $\Lambda(S)=\Upsilon(\Sigma^{-1}\circ S \circ \Sigma)$.
\end{proof}

\begin{defn}
The discrete Pin subgroup $\Omega\subset Pin(3,1)$ is:
\begin{align}
\Omega \equiv \{\pm 1,\pm i\gamma^0,\pm \gamma^0\gamma^5,\pm
i\gamma^5\}
\end{align}
\end{defn}

The previous lemma and the fact that $\Lambda$ is continuous, 
implies that $Spin^+(1,3)$ is a double cover of $SO^+(3,1)$.
We can check that for all $\omega\in \Omega$, $\Lambda(\pm \omega)\in
\Delta$. 
That is, the discrete Pin subgroup is the double cover of the
discrete Lorentz subgroup. Therefore, $Pin(3,1)=\Omega \ltimes Spin^+(1,3)$

Since there is a two-to-one continuous surjective group homomorphism,
$Pin(3,1)$ is a double cover of $O(1,3)$, $Spin^+(3,1)$ 
is a double cover of $SO^+(1,3)$ and $Spin^+(1,3)\cap SU(4)$ is a
double cover of $SO(3)$. We can check that $Spin^+(1,3)\cap SU(4)$ is
equivalent to $SU(2)$.

\subsection{Finite-dimensional representations of SL(2,C)}

\begin{rmk}
Since SL(2,C) is a semisimple Lie group, all its finite-dimensional 
(real or complex) representations are direct sums of irreducible
representations.
\end{rmk}

\begin{rmk} 
The finite-dimensional complex irreducible representations of SL(2,C) 
are labeled by $(m,n)$, where $2m,2n$ are natural numbers. 
Up to equivalence, the representation space $V_{(m,n)}$ is the tensor
product of the complex vector spaces $V_m^+$ and $V_n^-$, where $V_m^\pm$ is a
symmetric tensor with $2m$ Dirac spinor indexes, such that 
$\gamma^5_{\  k}v=\pm v$, where $v\in V_m^\pm$ and $\gamma^5_{\  k}$
is the Dirac matrix $\gamma^5$ acting on the $k$-th index of $v$.

The group homomorphism consists in applying the same matrix of
$Spin^+(1,3)$, correspondent to the $SL(2,C)$ group element we are
representing, to each index of $v$. 
$V_{(0,0)}$ is equivalent to $\mathbb{C}$ and the image of the group
homomorphism is the identity.

These are also projective representations of the time reversal transformation,
but, for $m\neq n$, not of the parity transformation, that is, 
under the parity transformation, $(V^+_m\otimes V^-_n)\to (V^-_m\otimes V^+_n)$ and under the time
reversal transformation $(V^+_m\otimes V^-_n)\to (V^+_m\otimes V^-_n)$.
\end{rmk}

\begin{lem}
The finite-dimensional real irreducible representations of SL(2,C) 
are labeled by $(m,n)$, where $2m,2n$ are natural numbers and $m\geq
n$.
Up to equivalence, the representation space $W_{(m,n)}$ is defined
for $m\neq n$ as:
\begin{align*}
W_{(m,n)}&\equiv \{\frac{1+(i\gamma^5)_1\otimes
   (i\gamma^5)_1}{2}w: w\in W_m\otimes W_n\}\\
W_{(m,m)}&\equiv 
\{\frac{1+(i\gamma^5)_1\otimes(i\gamma^5)_1}{2}w: w\in (W_m)^2\}
\end{align*} 
where $W_m$ is a
symmetric tensor with $m$ Majorana spinor indexes, such that 
$(i\gamma^5)_{1}(i\gamma^5)_{k}w=-w$, where $w\in W_m$; $(i\gamma^5)_{k}$
is the Majorana matrix $i\gamma^5$ acting on the $k$-th index of $w$;
$(W_m)^2$ is the space of the linear combinations of the
symmetrized tensor products $(u\otimes v+v\otimes u)$, for $u,v\in W_m$.  

The group homomorphism consists in applying the same matrix of
$Spin^+(1,3)$, correspondent to the $SL(2,C)$ group element we are
representing, to each index of the tensor. In the $(0,0)$ case,
$W_{(0,0)}$ is equivalent to $\mathbb{R}$ and the
image of the group homomorphism is the identity.

These are also projective representations of the full Lorentz group, 
that is, under the parity or time reversal transformations,
$(W_{m,n}\to W_{m,n})$.
\end{lem}

\begin{proof}
For $m\neq n$ the complex irreducible representations of SL(2,C) are
C-complex. The complexification of $W_{(m,n)}$ verifies
$W_{(m,n)}^c=(V^+_m\otimes V^-_n)\oplus (V^-_m\otimes V^+_n)$.

For $m=n$ the complex irreducible representations of SL(2,C) are
C-real. In a Majorana basis, the C-conjugation operator of
$V_{(m,m)}$, $\theta$, is defined as 
$\theta(u\otimes v)\equiv v^*\otimes u^*$, where $u\in V^+_m$ and $v\in V^-_m$. 
We can check that there is a bijection $\alpha:W_{(m,m)}\to
(V_{(m,m)})_\theta$, 
defined by $\alpha(w)\equiv \frac{1-i(i\gamma^5)_1\otimes 1}{2}w$; 
$\alpha^{-1}(v)\equiv v+v^*$, for $w\in W_{(m,m)}$, $v\in (V_{(m,m)})_\theta$.

Using the map from Section 2,
we can check that the representations 
$W_{(m,n)}$, with $m\geq n$, are the unique 
finite-dimensional real irreducible representations of SL(2,C), up to
isomorphisms.

We can check that $W_{(m,n)}^c$ is equivalent to $W_{(n,m)}^c$,
therefore, invariant under the parity or time reversal transformations.
\end{proof}

As examples of real irreducible representations of $SL(2,C)$ we have
for $(1/2,0)$ the Majorana spinor, for $(1/2,1/2)$ the linear
combinations of the matrices $\{1,\gamma^0\vec{\gamma}\}$, for $(1,0)$ the linear 
combinations of the matrices $\{i\vec{\gamma},\vec{\gamma}\gamma^5\}$. The group 
homomorphism is defined as $M(S)(u)\equiv Su$ and $M(S)(A)\equiv S A S^\dagger$, 
for $S\in Spin^+(1,3)$,
$u\in Pinor$, $A\in \{1,\vec{\gamma}\gamma^0\}$ or $A\in
\{i\vec{\gamma},\vec{\gamma}\gamma^5\}$. 

We can check that the domain
of $M$ can be extended to $Pin(1,3)$, leaving the considered vector
spaces invariant. For $m=n$, we can define the ``pseudo-representation'' 
$W_{(m,m)}'\equiv \{((i\gamma^5)_1\otimes 1) w: w\in W_{(m,m)}\}$
which is equivalent to $W_{(m,m)}$ as an $SL(2,C)$ representation, but
under parity transforms with the opposite sign.
As an example, the ``pseudo-representation'' $(1/2,1/2)$ is defined as
the linear combinations of the matrices $\{i\gamma^5,i\gamma^5\vec{\gamma}\gamma^0\}$.

\section{Unitary representations of the Poincare group}
\label{section:Poincare}

\subsection{Bargmann-Wigner fields}
\begin{defn}
\label{defn:Theta}
Consider that $\{M_+,M_-,i\gamma^0M_+,i\gamma^0M_-\}$ and $\{P_+,P_-,iP_+,iP_-\}$ are orthonormal basis 
of the 4 dimensional real vector spaces $Pinor$ and $Pauli$, respectively, verifying:
\begin{align*}
\gamma^3\gamma^5 M_\pm=\pm M_\pm&,\ \sigma^3 P_\pm=\pm P_\pm
\end{align*}
Let $H$ be a real Hilbert space. 
For all $h\in H$, the bijective linear map
$\Theta_H:Pauli\otimes_{\mathbb{R}} H\to Pinor\otimes_{\mathbb{R}}H$ is defined by:
\begin{align*}
\Theta_H(h\otimes_{\mathbb{R}} P_+)=h\otimes_{\mathbb{R}} M_+,&\ \Theta_H(h \otimes_{\mathbb{R}}
iP_+)=h\otimes_{\mathbb{R}} i\gamma^0 M_+\\
\Theta_H(h\otimes_{\mathbb{R}} P_-)=h\otimes_{\mathbb{R}} M_-,&\ \Theta_H(h\otimes_{\mathbb{R}} iP_-)=h\otimes_{\mathbb{R}}i\gamma^0 M_-
\end{align*}
\end{defn}

\begin{defn}
Let $H_n$, with $n\in\{1,2\}$, be two real Hilbert spaces 
and $U:Pauli\otimes_{\mathbb{R}} H_1\to
Pauli\otimes_{\mathbb{R}} H_2$ be an operator.
The operator $U^\Theta:Pinor\otimes_{\mathbb{R}} H_1\to
Pinor\otimes_{\mathbb{R}} H_2$ is defined as
$U^\Theta\equiv
\Theta_{H_2}\circ U\circ \Theta^{-1}_{H_1}$.
\end{defn}

The space of Majorana spinors is isomorphic to
the realification of the space of Pauli spinors.

\begin{defn}
The real Hilbert space 
$Pinor(\mathbb{X})\equiv Pinor\otimes L^2(\mathbb{X})$ 
is the space of square integrable
functions with domain $\mathbb{X}$ and image in $Pinor$.
\end{defn}

\begin{defn}
The complex Hilbert space 
$Pauli(\mathbb{X})\equiv Pauli\otimes L^2(\mathbb{X})$ 
is the space of square integrable functions with domain $\mathbb{X}$
and image in $Pauli$.
\end{defn}

\begin{defn}
The real vector space $Pinor_j$, with $2j$ a positive integer, is the
space of linear combinations of the tensor products of $2j$
Majorana spinors, symmetric on the spinor indexes. The real vector
space $Pinor_0$ is the space of linear combinations of the tensor
products of $2$ Majorana spinors, anti-symmetric on the spinor
indexes.
\end{defn}

\begin{defn}
The real Hilbert space 
$Pinor_j(\mathbb{X})\equiv Pinor_j\otimes L^2(\mathbb{X})$ is the
space of square integrable functions with domain $\mathbb{X}$ and
image in $Pinor_j$.
\end{defn}

\begin{defn}
The space of (real) Bargmann-Wigner fields $BW_{j}(\mathbb{R}^3)$ is defined as:
\begin{align*}
BW_j\equiv 
\{\Psi\in Pinor_{j}(\mathbb{R}^3):
\Big(e^{iH(\vec{x}) t}\Big)_k\Psi=\Big(e^{iH(\vec{x})t}\Big)_1 \Psi; 1\leq k\leq
2j; t\in \mathbb{R}\}
\end{align*}
\end{defn}

Note that if the equality $e^{-iH_1 t}\Psi=e^{-iH_2 t}\Psi$ holds for all differentiable
$\Psi\in H$ then for the continuous linear extension
the equality holds for all $\Psi\in H$, by the bounded linear transform theorem.

\begin{defn}
The complex Hilbert space 
$Dirac_j(\mathbb{X})\equiv Pinor_j(\mathbb{X})\otimes \mathbb{C}$
is the complexification of $Pinor_j(\mathbb{X})$.
The space of complex Bargmann-Wigner fields is the complexification of
the space of real Bargmann-Wigner fields.
\end{defn}

\begin{prop}
\label{prop:unitary}
Consider a unitary operator $U:Pinor_j(\mathbb{R}^3)\to Pinor_j(\mathbb{X})$ such that
$U \circ H^2=E^2 \circ U$, where 
\begin{align*}
iH\{\Psi\}(\vec{x})\equiv
\Big(\gamma^0\vec{\slashed \partial}+i\gamma^0m\Big)_k\Psi(\vec{x})
\end{align*}
the Majorana matrices act on some Majorana index $k$; 
 $E^2\{\Phi\}(X)\equiv E^2(X)\Phi(X)$ with  $E(X)\geq m\geq 0$ a real
number. 

Then the operator $U':Pinor(\mathbb{R}^3)\to Pinor(\mathbb{X})$ is unitary,
where $U'$ is defined by:
\begin{align*}
U'\equiv \frac{E+U H\gamma^0 U^\dagger}{\sqrt{E+m}\sqrt{2E}}
\end{align*}
\end{prop}

\begin{proof}
Note that since $E^2=U^\dagger H^2 U$, $E=\sqrt{E^2}$ commutes with $U H\gamma^0
U^\dagger$. We have that
\begin{align*}
&(U')^\dagger (U')=\frac{E+U\gamma^0 H U^\dagger}{\sqrt{E+m}\sqrt{2E}}
\frac{E+U H\gamma^0 U^\dagger}{\sqrt{E+m}\sqrt{2E}}=1
\end{align*}
We also have that $(U')(U')^\dagger=1$. Therefore, $U'$ is unitary.
\end{proof}

\subsection{Fourier-Majorana Transform}

\begin{rmk}
The Fourier Transform 
$\mathcal{F}_P: Pauli(\mathbb{R}^3)\to Pauli(\mathbb{R}^3)$ is an unitary operator defined by:
\begin{align*}
\mathcal{F}_P\{\psi\}(\vec{p})\equiv\int d^n\vec{x} \frac{e^{-i\vec{p}\cdot
  \vec{x}}}{\sqrt{(2\pi)^n}}\psi(\vec{x}),\ \psi\in Pauli(\mathbb{R}^3)
\end{align*}
Where the domain of the integral is $\mathbb{R}^3$.
\end{rmk}

\begin{rmk}
The inverse Fourier transform verifies:
\begin{align*}
-\vec{\partial}^2\
\mathcal{F}_P^{-1}\{\psi\}(\vec{x})&=
(\mathcal{F}_P^{-1}\circ R)\{\psi\}(\vec{x})\\
i\vec{\partial}_k
\ \mathcal{F}_P^{-1}\{\psi\}(\vec{x})&=
(\mathcal{F}_P^{-1}\circ R_k')\{\psi\}(\vec{x})
\end{align*}
Where $\psi\in Pauli(\mathbb{R}^3)$ and
$R,R_k':Pauli(\mathbb{R}^3)\to Pauli(\mathbb{R}^3)$, with
$k\in\{1,2,3\}$, are linear maps defined by:
\begin{align*}
R\{\psi\}(\vec{p})&\equiv (\vec{p})^2 \psi(\vec{p})\\
R_k'\{\psi\}(\vec{p})&\equiv \vec{p}_k\ \psi(\vec{p})
\end{align*}
\end{rmk}

\begin{defn}
The Fourier-Majorana transform 
$\mathcal{F}_M: Pinor_j(\mathbb{R}^3)\to Pinor_j(\mathbb{R}^3)$
is an unitary operator defined by:

\begin{align*}
\mathcal{F}_M\{\Psi\}(\vec{p})\equiv\int d^3\vec{x}
\Big(\frac{e^{-i\gamma^0\vec{p}\cdot\vec{x}}}{\sqrt{(2\pi)^3}}\Big)_1
\prod_{k=1}^{2j}\Big(\frac{E_p+H(\vec{x})\gamma^0}{\sqrt{E_p+m}\sqrt{2E_p}}\Big)_k
\Psi(\vec{x}),\ \Psi\in Pinor_j(\mathbb{R}^3)
\end{align*}
The matrices with the index $k$ apply on
the corresponding spinor index of $\Psi$.
\end{defn}

The inverse Fourier-Majorana transform verifies:
\begin{align*}
(iH(\vec{x}))_k\
\mathcal{F}_M^{-1}\{\psi\}(\vec{x})&=
(\mathcal{F}_M^{-1}\circ R)\{\psi\}(\vec{x})\\
\vec{\partial}_l\ \mathcal{F}_M^{-1}\{\psi\}(\vec{x})&=
(\mathcal{F}_M^{-1}\circ R')\{\psi\}(\vec{x})
\end{align*}
Where $\psi\in Pinor_j(\mathbb{R}^3)$ and
$R,R':Pinor_j(\mathbb{R}^3)\to Pinor_j(\mathbb{R}^3)$ are linear maps
defined by:
\begin{align*}
R\{\psi\}(\vec{p})&\equiv (i\gamma^0)_k E_p \psi(\vec{p})\\
R'\{\psi\}(\vec{p})&\equiv (i\gamma^0)_1\vec{p}_l\ \psi(\vec{p})
\end{align*}

\subsection{Hankel-Majorana Transform}

\begin{defn}
Let $\vec{x}\in \mathbb{R}^3$. The spherical coordinates
parametrization is:
\begin{align*}
\vec{x}=r(\sin(\theta)\sin(\varphi)\vec{e_1}+\sin(\theta)\sin(\varphi)\vec{e_2}+\cos(\theta)\vec{e}_3)
\end{align*}
where $\{\vec{e}_1,\vec{e}_2,\vec{e}_3\}$ is a fixed orthonormal basis of
$\mathbb{R}^3$ and $r\in [0,+\infty[$, $\theta \in [0,\pi]$, $\varphi
\in [-\pi,\pi]$.
\end{defn}

\begin{defn}
Let
\begin{align*}
\mathbb{S}^3\equiv \{(p,l,\mu):p\in \mathbb{R}_{\geq 0}; 
l,\mu \in \mathbb{Z}; l\geq 0; -l \leq \mu\leq l\}
\end{align*}
The Hilbert space $L^2(\mathbb{S}^3)$ is the real Hilbert space of real
Lebesgue square integrable functions of $\mathbb{S}^3$. The internal product is:
\begin{align*}
<f,g>=\sum_{l=0}^{+\infty}\sum_{\mu=-l}^{l-1}\int_0^{+\infty} dp f(p,l,\mu)
g(p,l,\mu),\ f,g\in L^2(\mathbb{S}^3)
\end{align*}
\end{defn}

\begin{defn}
The Spherical transform $\mathcal{H}_{P}: Pauli(\mathbb{R}^3)\to Pauli(\mathbb{S}^3)$ 
is an operator defined by:
\begin{align*}
\mathcal{H}_{P}\{\psi\}(p,l,\mu)\equiv\int r^2 dr d(\cos\theta)d\varphi
\frac{2 p}{\sqrt{2\pi}}j_l(pr)Y_{l\mu}(\theta,\varphi)\psi(r,\theta,\varphi),\ \psi\in Pauli(\mathbb{R}^3)
\end{align*}
The domain of the integral is $\mathbb{R}^3$. The spherical
Bessel function of the first kind $j_l$ \cite{bessel},
the spherical harmonics $Y_{l\mu}$\cite{harmonics} and the associated Legendre
functions of the first kind $P_{l\mu}$ are:
\begin{align*}
j_l(r)\equiv& r^l\Big(-\frac{1}{r}\frac{d}{dr}\Big)^l \frac{\sin
  r}{r}\\
Y_{l\mu}(\theta,\varphi)\equiv&\sqrt{\frac{2l+1}{4\pi}\frac{(l-m)!}{(l+m)!}}
P_{l}^\mu(\cos\theta)e^{i\mu \varphi}\\
P_{l}^\mu(\xi)\equiv&\frac{(-1)^{\mu}}{2^{l}l!}(1-\xi^{2})^{\mu/2}
\frac{\mathrm{d}^{l+\mu}}{\mathrm{d}\xi^{l+\mu}}(\xi^{2}-1)^{l}
\end{align*}
\end{defn}

\begin{rmk}
Due to the properties of spherical harmonics and Bessel functions, the
Spherical transform  is an unitary operator. The inverse Spherical
transform verifies:
\begin{align*}
-\vec{\partial}^2\
\mathcal{H}_P^{-1}\{\psi\}(\vec{x})&=
(\mathcal{H}_P^{-1}\circ R)\{\psi\}(\vec{x})\\
(-x^1i\partial_2+x^2i\partial_1)
\ \mathcal{H}_P^{-1}\{\psi\}(\vec{x})&=
(\mathcal{H}_P^{-1}\circ R')\{\psi\}(\vec{x})
\end{align*}
Where $\psi\in Pauli(\mathbb{S}^3)$ and
$R,R':Pauli(\mathbb{S}^3)\to Pauli(\mathbb{S}^3)$ 
are linear maps defined by:
\begin{align*}
R\{\psi\}(p,l,\mu)&\equiv p^2 \psi(p,l,\mu)\\
R'\{\psi\}(p,l,\mu)&\equiv \mu\ \psi(p,l,\mu)
\end{align*}
\end{rmk}

\begin{defn}
The Hilbert space $Pinor_{j,n}$, with $(j-\nu)$ an
integer and $-j\leq n \leq j$ is defined as:
\begin{align*}
Pinor_{j,n}\equiv \{\Psi\in Pinor_{j}:
\sum_{k=1}^{k=2j}(\gamma^0)_1\Big(\gamma^0\gamma^3\gamma^5\Big)_k\Psi
=2n \Psi\}
\end{align*}
Where $\Big(\gamma^3\gamma^5\Big)_k$ is the matrix 
$\gamma^3\gamma^5$ acting on the Majorana index $k$. 
\end{defn}

\begin{defn}
The Spherical transform 
$\mathcal{H}_{P}': Pinor_j(\mathbb{R}^3)\to Pinor_j(\mathbb{S}^3)$ 
is an operator defined by:
\begin{align*}
\mathcal{H}_{P}'\{\psi\}(p,l,J,\nu)\equiv
\sum_{\mu=-l}^l\sum_{n=-j}^j
<l\mu jn|J\nu>\Big(\mathcal{H}_{P}^{\Theta}\Big)_1\{\psi\}(p,l,\mu,n),
\ \psi\in Pinor_j(\mathbb{R}^3)
\end{align*}
$<l\mu jn|J\nu>$ are the Clebsh-Gordon coefficients and
$\psi(p,l,\mu,n)\in Pinor_{j,n}$ such that
$\psi(p,l,\mu)=\sum_{n=-j}^j \psi(p,l,\mu,n)$. $(j-n)$, $(J-\nu)$ and
$(J-j)$ are integers, with $-J\leq\nu\leq J$ and $|j-l|\leq J \leq
j+l$.
$\Big(\mathcal{H}_{P}^{\Theta}\Big)_1$ is the realification of the
transform $\mathcal{H}_{P}$, with the imaginary number replaced by the
matrix $i\gamma^0$ acting on the first Majorana index of $\psi$. 
\end{defn}

\begin{defn}
The Hankel-Majorana transform 
$\mathcal{H}_M: Pinor_j(\mathbb{R}^3)\to Pinor_j(\mathbb{S}^3)$
is a unitary operator defined by:
\begin{align*}
&\mathcal{H}_M\{\Psi\}(p,l,J,\nu)\equiv
\sum_{\mu=-l}^l\sum_{n=-j}^j <l\mu jn|J\nu>\int d^3\vec{x}\\
&\Big(\frac{2 p}{\sqrt{2\pi}}j_l(pr)Y_{l\mu}(\theta,\varphi)\Big)_1
\prod_{k=1}^{2j}\Big(\frac{E_p+H(\vec{x})\gamma^0}{\sqrt{E_p+m}\sqrt{2E_p}}\Big)_k
\Psi(\vec{x},n)
\end{align*}

The matrices with the index $k$ apply on
the corresponding spinor index of 
$\Psi\in Pinor_j(\mathbb{R}^3)$. 
$<l\mu jn|J\nu>$ are the Clebsh-Gordon coefficients and
$\Psi(\vec{x},n)\in Pinor_{j,n}$ such that
$\Psi(\vec{x})=\sum_{n=-j}^j \Psi(\vec{x},n)$.
\end{defn}

The inverse Hankel-Majorana transform verifies:
\begin{align*}
(iH(\vec{x}))_k\
\mathcal{H}_M^{-1}\{\psi\}(\vec{x})&=
(\mathcal{H}_M^{-1}\circ R)\{\psi\}(\vec{x})\\
(-x^1\partial_2+x^2\partial_1+\sum_{k=1}^{2j}(i\gamma^0\gamma^3\gamma^5)_k)
\ \mathcal{H}_M^{-1}\{\psi\}(\vec{x})&=
(\mathcal{H}_M^{-1}\circ R')\{\psi\}(\vec{x})
\end{align*}
Where $\psi\in Pinor_j(\mathbb{S}^3)$ and $R,R':Pinor_j(\mathbb{S}^3)\to Pinor_j(\mathbb{S}^3)$ are
linear maps defined by:
\begin{align*}
R\{\psi\}(p,l,J,\nu)&\equiv (i\gamma^0)_k E_p \psi(p,l,J,\nu)\\
R'\{\psi\}(p,l,J,\nu)&\equiv (i\gamma^0)_1\nu\ \psi(p,l,J,\nu)
\end{align*}

\subsection{Application to the momentum of Majorana spinor fields}

\begin{defn}
The Majorana-Fourier Transform $\mathcal{F}_M: Pinor(\mathbb{R}^3)\to
Pinor(\mathbb{R}^3)$ is an operator defined by:
\begin{align*}
\mathcal{F}_M\{\Psi\}(\vec{p})&\equiv \int d^3\vec{x}\ \frac{e^{-i\gamma^0\vec{p}\cdot
  \vec{x}}}{\sqrt{(2\pi)^3}}
\frac{\slashed p \gamma^0+m}{\sqrt{E_p+m}\sqrt{2E_p}}\Psi(\vec{x}),
\ \Psi\in Pinor(\mathbb{R}^3)
\end{align*}
Where the domain of the integral is $\mathbb{R}^3$, $m\geq 0$,
$E_p\equiv\sqrt{\vec{p}^2+m^2}$ and $\slashed p=
E_p\gamma^0-\vec{p}\cdot \vec{\gamma}$.
\end{defn}

\begin{prop}
The Majorana-Fourier Transform is a unitary operator.
\end{prop}

\begin{proof}
The proof is immediate using Prop.~\ref{prop:unitary}, but we will do it
in an independent more explicit way.
The Majorana-Fourier Transform can be written as:
\begin{align*}
\mathcal{F}_M\{\Psi\}(\vec{p})\equiv& \sqrt{\frac{E_p+m}{2E_p}}\Big(\int d^3\vec{x}\ \frac{e^{-i\gamma^0\vec{p}\cdot
  \vec{x}}}{\sqrt{(2\pi)^3}}\Psi(\vec{x})\Big)\\
-&\sqrt{\frac{E_p-m}{2E_p}}\frac{\vec{p}\cdot\vec{\gamma}\gamma^0}{|\vec{p}|}\Big(\int d^3\vec{x}\ \frac{e^{+i\gamma^0\vec{p}\cdot
  \vec{x}}}{\sqrt{(2\pi)^3}}\Psi(\vec{x})\Big)
\end{align*}
So, one gets:
\begin{align*}
\mathcal{F}_{M}\{\Psi\}=S \circ \mathcal{F}^\Theta_{P}\{\Psi\}
\end{align*}
Where $S:Pinor(\mathbb{R}^3)\to Pinor(\mathbb{R}^3)$ is a bijective linear map defined by:
\begin{align*}
\left[ \begin{array}{l}
S\{\Psi\}(+\vec{p})\\
S\{\Psi\}(-\vec{p})
\end{array}
\right]&\equiv
\left[ \begin{array}{cc}
\sqrt{\frac{E_p+m}{2E_p}} 
& -\sqrt{\frac{E_p-m}{2E_p}}\frac{\vec{p}\cdot\vec{\gamma}\gamma^0}{|\vec{p}|}\\
\sqrt{\frac{E_p-m}{2E_p}}\frac{\vec{p}\cdot\vec{\gamma}\gamma^0}{|\vec{p}|}
&\sqrt{\frac{E_p+m}{2E_p}}
\end{array}
\right]\ 
\left[ \begin{array}{l}
\Psi(+\vec{p})\\
\Psi(-\vec{p})
\end{array}
\right]
\end{align*}
We can check that the $2\times 2$ matrix appearing in the equation
above is orthogonal. Therefore $S$ is an unitary operator.
Since $\mathcal{F}^\Theta_{P}$ is also unitary,
$\mathcal{F}_{M}$ is unitary.
\end{proof}

\begin{prop}
The inverse Majorana-Fourier Transform verifies:
\begin{align*}
(\gamma^0\vec{\gamma}\cdot \vec{\partial}+i\gamma^0
m)\mathcal{F}_M^{-1}\{\Psi\}(\vec{x})&=(\mathcal{F}_M^{-1}\circ
R)\{\Psi\}(\vec{x})\\
\vec{\partial}_j\mathcal{F}_M^{-1}\{\Psi\}(\vec{x})&=(\mathcal{F}_M^{-1}\circ
R_j)\{\Psi\}(\vec{x})
\end{align*}
Where $\Psi\in Pinor(\mathbb{R}^3)$ and $R,R_j:Pinor(\mathbb{R}^3)\to Pinor(\mathbb{R}^3)$ are
 linear maps defined by $R\{\Psi\}(\vec{p})=i\gamma^0
 E_p\Psi(\vec{p})$ and 
$R_j\{\Psi\}(\vec{p})=i\gamma^0 \vec{p}_j\Psi(\vec{p})$ .
\end{prop}

\begin{proof}
We have $\mathcal{F}^{-1}_{M}=(\mathcal{F}^{\Theta}_{P})^{-1}\circ S^{-1}$. Then:
\begin{align*}
(\gamma^0\vec{\gamma}\cdot \vec{\partial}+i\gamma^0
m)(\mathcal{F}^\Theta_P)^{-1}\{\Psi\}(\vec{x})=((\mathcal{F}^\Theta_P)^{-1}\circ Q) \{\Psi\}(\vec{x})
\end{align*}
Where $Q:Pinor(\mathbb{R}^3)\to Pinor(\mathbb{R}^3)$ is a linear map defined by:
\begin{align*}
\left[ \begin{array}{l}
Q\{\Psi\}(+\vec{p})\\
Q\{\Psi\}(-\vec{p})
\end{array}
\right]&\equiv
\left[ \begin{array}{cc}
i\gamma^0 m 
& i\vec{p}\cdot \vec{\gamma}\\
-i\vec{p}\cdot \vec{\gamma}
& i\gamma^0 m
\end{array}
\right]\ 
\left[ \begin{array}{l}
\Psi(+\vec{p})\\
\Psi(-\vec{p})
\end{array}
\right]
\end{align*}
Now we show that $Q\circ S^{-1}=S^{-1}\circ R$:
\begin{align*}
&\left[ \begin{array}{cc}
i\gamma^0 m 
& i\vec{p}\cdot \vec{\gamma}\\
-i\vec{p}\cdot \vec{\gamma}
& i\gamma^0 m
\end{array}
\right]\ 
\left[ \begin{array}{cc}
\sqrt{\frac{E_p+m}{2E_p}} 
& \sqrt{\frac{E_p-m}{2E_p}}\frac{\vec{p}\cdot\vec{\gamma}\gamma^0}{|\vec{p}|}\\
-\sqrt{\frac{E_p-m}{2E_p}}\frac{\vec{p}\cdot\vec{\gamma}\gamma^0}{|\vec{p}|}
&\sqrt{\frac{E_p+m}{2E_p}}
\end{array}
\right]=\\
&=\left[ \begin{array}{cc}
\sqrt{\frac{E_p+m}{2E_p}} 
& \sqrt{\frac{E_p-m}{2E_p}}\frac{\vec{p}\cdot\vec{\gamma}\gamma^0}{|\vec{p}|}\\
-\sqrt{\frac{E_p-m}{2E_p}}\frac{\vec{p}\cdot\vec{\gamma}\gamma^0}{|\vec{p}|}
&\sqrt{\frac{E_p+m}{2E_p}}
\end{array}\right]\ 
 \left[ \begin{array}{cc}
 i\gamma^0 E_p & 0\\
 0 & i\gamma^0 E_p
 \end{array}
 \right]
\end{align*}

We also have that:
\begin{align*}
\vec{\partial}_j(\mathcal{F}^\Theta_P)^{-1}\{\Psi\}(\vec{x})=((\mathcal{F}^\Theta_P)^{-1}\circ R_j) \{\Psi\}(\vec{x})
\end{align*}
Where $R_j:Pinor(\mathbb{R}^3)\to Pinor(\mathbb{R}^3)$ is the linear map defined by:
\begin{align*}
\left[ \begin{array}{l}
R_j\{\Psi\}(+\vec{p})\\
R_j\{\Psi\}(-\vec{p})
\end{array}
\right]&\equiv
\left[ \begin{array}{cc}
i\gamma^0 \vec{p}_j 
& 0\\
0
& -i\gamma^0 \vec{p}_j 
\end{array}
\right]\ 
\left[ \begin{array}{l}
\Psi(+\vec{p})\\
\Psi(-\vec{p})
\end{array}
\right]
\end{align*}
It verifies $R_j\circ S^{-1}=S^{-1}\circ R_j$.
\end{proof}

\begin{defn}
The Energy Transform $\mathcal{E}: Pinor(\mathbb{R})\to
Pinor(\mathbb{R})$ is an operator defined by:
\begin{align*}
\mathcal{E}\{\Psi\}(p^0)&\equiv \int dx^0\ \frac{e^{i\gamma^0p^0x^0}}{\sqrt{2\pi}}\Psi(x^0),\ \Psi\in Pinor(\mathbb{R})
\end{align*}
Where the domain of the integral is $\mathbb{R}$, $m\geq 0$.
\end{defn}

\begin{prop}
The Energy transform is an unitary operator.
\end{prop}
\begin{proof}
The Energy transform can be written as:
\begin{align*}
\mathcal{E}\{\Psi\}(p^0)=\Theta_{L^2}\circ
\mathcal{F}_P(-p^0)\circ\Theta^{-1}_{L^2}\{\Psi\}
\end{align*}
Where $\mathcal{F}_P(-p^0)$ is a Pauli-Fourier transform over
$\mathbb{R}$ and $\Theta$ was defined in Definition \ref{defn:Theta}.
Since the Pauli-Fourier transform is unitary, so is the
Energy transform.
\end{proof}

The energy transform can be applied in the time coordinate of a
Majorana spinor field, $x^0$, after a (linear or spherical)
momentum transform on the space coordinates, $\vec{x}$, to define an
unitary energy-momentum transform:\\
- for the linear case $\mathcal{E}\circ
\mathcal{F}_M:Pinor(\mathbb{R}^4)\to Pinor(\mathbb{R}^4)$;\\
- for the spherical case $\mathcal{E}\circ
\mathcal{H}_M:Pinor(\mathbb{R}^4)\to Pinor(\mathbb{R}\times
\mathbb{S}^3)$.

\subsection{Real unitary representations of the Poincare group}

\begin{defn}
The $IPin(3,1)$ group is defined as the semi-direct product
$Pin(3,1)\ltimes \mathbb{R}^4$, with the group's product defined as
$(A,a)(B,b)=(AB,a+\Lambda(A)b)$, for $A,B\in Pin(3,1)$ and 
$a,b\in \mathbb{R}^4$ and $\Lambda(A)$ is the Lorentz transformation
corresponding to $A$.

The $ISL(2,C)$ group is isomorphic to the subgroup of $IPin(3,1)$, 
obtained when $Pin(3,1)$ is restricted to $Spin^+(1,3)$. The full/restricted
Poincare group is the representation of the $IPin(3,1)/ISL(2,C)$ group on
Lorentz vectors, defined as 
$\{(\Lambda(A),a): A\in Pin(3,1), a\in \mathbb{R}^4\}$.
\end{defn}

\begin{defn}
Given a Lorentz vector $l$, the little group $G_l$ is the subgroup
of $SL(2,C)$ such that for all $g\in G_l$, $g\slashed l=\slashed l g$. 
\end{defn}

\begin{prop}
Given a Lorentz vector $l$, consider a set of matrices $\alpha_k\in
SL(2,C)$ verifying 
$\alpha_k \slashed l=\slashed k \alpha_k$. Let $H_k \equiv \{\alpha_{\Lambda_S(k)}^{-1} S\alpha_k: S\in SL(2,C)\}$. Then $H_k=G_l$. 
\end{prop}
\begin{proof}
We can check that $H_k\subset G_l$. For any $s\in G_l$, there is
$S=\alpha_{\Lambda_S(k)} s \alpha_k^{-1}$ such that $s\in H_k$. 
\end{proof}

For $i\slashed l=i\gamma^0$, we can set 
$\alpha_p=\frac{\slashed p\gamma^0+m}{\sqrt{E_p+m}\sqrt{2m}}$ and $G_l=SU(2)$.
For $i\slashed l=(i\gamma^0+i\gamma^3)$, we can set $\alpha_p=B_vR_{p}$,
where the boost velocity is $v=\frac{E_p^2-1}{E_p^2+1}$ along $\vec{p}$ and 
$R_p=e^{-\gamma^2\gamma^1 \theta/2}e^{-\gamma^1\gamma^3 \phi/2}$ is a rotation from the $z$
axis to the axis 
$\frac{\vec{\slashed p}}{E_p}=(\sin\phi \cos\theta \gamma_1+\sin\phi \sin\theta
\gamma_2+\cos\phi \gamma_3)$; $G_l=SE(2)$
\begin{align}
SE(2)=\{(1+i\gamma^5(\gamma^1a+\gamma^2b)(\gamma^0+\gamma^3))e^{i\gamma^0\gamma^3\gamma^5\theta}:
a,b,\theta\in \mathbb{R}\}.
\end{align}

\begin{rmk}
The complex irreducible projective representations of the Poincare
group with finite mass split into positive and negative energy
representations, 
which are complex conjugate of each other. They are labeled by one number $j$, with $2j$
being a natural number.
The positive energy representation spaces $V_j$ are, up to isomorphisms, written
as a symmetric tensor product of Dirac spinor
fields defined on the 3-momentum space, verifying
$(\gamma^0)_k\Psi_j(\vec{p})=\Psi_j(\vec{p})$. The matrices with the index $k$ apply in
the corresponding spinor index of $\Psi_j$.

The representation space $V_0$ is, up to isomorphisms, written in a
Majorana basis as a complex scalar defined on the 3-momentum space.

The representation map is given by:
\begin{align*}
L_S\{\Psi\}(\vec{p})&=\sqrt{\frac{(\Lambda^{-1})^0(p)}{E_p}}\prod_{k=1}^{2j}(\alpha^{-1}_{\Lambda(p)}S\alpha_p)_k\Psi(\vec{\Lambda}^{-1}(p))\\
T_a\{\Psi\}(\vec{p})&=e^{-i p\cdot a}\Psi(\vec{p})
\end{align*}
Where $\alpha_p=\frac{\slashed p\gamma^0+m}{\sqrt{E_p+m}\sqrt{2m}}$.
\end{rmk}

\begin{prop}
The real irreducible projective representations of the Poincare
group with finite mass are labeled by one number $j$, with $2j$
being a natural number.
The representation spaces $W_j$ are, up to isomorphisms, written
as a symmetric tensor product of Majorana spinor fields defined
on the 3-momentum space, verifying
$(i\gamma^0)_k\Psi_j(\vec{p})=(i\gamma^0)_1\Psi_j(\vec{p})$. 
The matrices with the index $k$ apply in the corresponding spinor index of
$\Psi_j$.

The representation space $V_0$ is, up to isomorphisms, written in a
Majorana basis as a real scalar defined on the 3-momentum space, times
the identity matrix of a Majorana spinor space.

The representation map is given by:
\begin{align*}
L_S\{\Psi\}(\vec{p})&=\sqrt{\frac{(\Lambda^{-1})^0(p)}{E_p}}\prod_{k=1}^{2j}(\alpha^{-1}_{\Lambda(p)}S\alpha_p)_k\Psi(\vec{\Lambda}^{-1}(p))\\
T_a\{\Psi\}(\vec{p})&=e^{-i\gamma^0 p\cdot a}\Psi(\vec{p})
\end{align*}
\end{prop}

\begin{rmk}
The complex irreducible projective representations of the Poincare
group with null mass and discrete helicity split into positive and negative energy
representations, 
which are complex conjugate of each other. They are labeled by one number $j$, with $2j$
being an integer number.
The positive energy representation spaces $V_j$ are, up to isomorphisms, written
as a symmetric tensor product of Dirac spinor
fields defined on the 3-momentum space, verifying
$(\gamma^0)_k\Psi_j(\vec{p})=\Psi_j(\vec{p})$ and $(\gamma^3\gamma^5)_k\Psi_j(\vec{p})=\pm \Psi_j(\vec{p})$,
with the plus sign if $j$ is positive and the minus sign if $j$ is
negative.

The representation space $V_0$ is, up to isomorphisms, written in a
Majorana basis as a scalar defined on the 3-momentum space.

The representation map is given by:
\begin{align*}
L_S\{\Psi\}(\vec{p})&=\sqrt{\frac{(\Lambda^{-1})^0(p)}{E_p}}\prod_{k=1}^{2j}(e^{i\gamma^0\gamma^3\gamma^5\theta})_k\Psi(\vec{\Lambda}^{-1}(p))\\
T_a\{\Psi\}(\vec{p})&=e^{-i p\cdot a}\Psi(\vec{p})
\end{align*}
Where $\theta$ is the angle of the rotation of the little group $SE(2)$.
\end{rmk}

\begin{rmk}
The real irreducible projective representations of the Poincare
group with null mass and discrete helicity are labeled by one number $j$, with $2j$
being an integer number.
The positive energy representation spaces $V_j$ are, up to isomorphisms, written
as a symmetric tensor product of Majorana spinor
fields defined on the 3-momentum space, verifying
$(i\gamma^0)_k\Psi_j(\vec{p})=(i\gamma^0)_1\Psi_j(\vec{p})$ and $(\gamma^3\gamma^5)_k\Psi_j(\vec{p})=\pm \Psi_j(\vec{p})$,
with the plus sign if $j$ is positive and the minus sign if $j$ is
negative.

The representation space $V_0$ is, up to isomorphisms, written in a
Majorana basis as the realification of the complex functions defined
on the 3-momentum space, with the operator correspondent to the
imaginary unit given by the matrix $i\gamma^0$ of a Majorana spinor space.

The representation map is given by:
\begin{align*}
L_S\{\Psi\}(\vec{p})&=\sqrt{\frac{(\Lambda^{-1})^0(p)}{E_p}}\prod_{k=1}^{2j}(e^{i\gamma^0\gamma^3\gamma^5\theta})_k\Psi(\vec{\Lambda}^{-1}(p))\\
T_a\{\Psi\}(\vec{p})&=e^{-i\gamma^0 p\cdot a}\Psi(\vec{p})
\end{align*}
Where $\theta$ is the angle of the rotation of the little group $SE(2)$.
\end{rmk}

\subsection{Localization}

\begin{rmk}[Theorem 6.12 of \cite{vara}] 
There is a one-to-one correspondence
between the complex system of imprimitivity (U,P), based on $\mathbb{R}^3$,
and the representations of $SU(2)$.
The system (U,P) is equivalent to the system
induced by the representation of $SU(2)$.
\end{rmk}

\begin{defn}
A covariant system of imprimitivity is a system of imprimitivity (U,P),
where $U$ is a representation of the Poincare group and $P$ is a projection-valued
measure based on $\mathbb{R}^3$, such that for the Euclidean group  
$U(g)\pi(A) U^{-1}(g)=\pi(gA)$ and for the Lorentz group, for a state at time null at 
point  $\vec{x}=0$, $L\{\Psi\}(0)=S\Psi(0)$.
\end{defn}

\begin{defn}
A localizable real unitary representation of the Poincare group, 
compatible with Poincare covariance, consists of a system of imprimitivity on $R^3$ for which
at time null and $\vec{x}=0$, the Lorentz transformations do not act on the space coordinates.
\end{defn}

So, the localization of a state in $x=0$ is a property 
invariant under relativistic transformations.

\begin{prop}
Any localizable unitary representation of the Poincare group,
compatible with Poincare covariance, is a direct sum of irreducible representations which are massive or massless with discrete helicity.
\end{prop}

\begin{proof}
Since the system is a unitary Poincare representation, it is a direct sum of irreducible unitary Poincare representations 
and so there must be an unitary transformation $U$, such that:
\begin{align}
&\Psi(x+a)=(U e^{-J P\cdot a}U^{-1})\{\Psi\}(x)\\
&S\Psi(\Lambda(x))=(U L U^{-1})\{\Psi\}(x)
\end{align}
Where $J$ is the operator corresponding to the imaginary unit after the 
realification of the Poincare representation and $L$ is the representation of the Lorentz group, 
so $L$ commutes with $J$. $P$ is the energy-momentum operator and $S$ acts only on the index of $\Psi$.

The system of imprimitivity is a representation of $SU(2)$, hence the operator $i\gamma^0$ is well defined. 
If we make a Fourier transformation, then we get that:
\begin{align}
(U e^{J \vec{P}\cdot \vec{a}}U^{-1})\{\Psi\}(\vec{p})=e^{i\gamma^0 \vec{p}\cdot \vec{a}}\Psi(\vec{p})
\end{align}
Note that this equation is valid for all $\vec{p}$.
The system is a direct sum of irreducible unitary Poincare representations.
Then, for  $m^2<0$ only the subspace $\vec{p}^2\geq |m^2|$ is valid.
For $p=0$ only the subspace  $\vec{p}=0$ is valid. 
Since the other types of irreducible representations verify $p\neq 0$ and $m^2\geq 0$, 
the complementary subspaces $\vec{p}^2 < |m^2|$ or $\vec{p}\neq 0$ cannot be 
representation spaces and hence the representations with 
$m^2<0$ and $p=0$ cannot be subspaces of a localizable representation.

So we are left with $p\neq 0$ and $m^2\geq 0$. Now we can define a subspace for each $m^2$, such that the square of 
the generator of translations in time is given by $\vec{\partial}^2+m^2$. In each subspace there is a localizable representation.

Given a subspace with $p\neq 0$ and $m^2\geq 0$, 
$M$, we consider the subspace $N$ of the representation $M\oplus M_0$ verifying $e^{iH t}\Psi=e^{iH_0 t}\Psi$,
where $M_0$ is a spin-0 representation and $e^{iH_0 t}$ is the translation in time acting on $M_0$.
Then, $e^{iH_0(\vec{\partial}) t}U=Ue^{iH_0(J\vec{P})}$. Multiplying $U$ by 
$\alpha_p \sqrt{m/E_p}$ we can check that $J\Psi=i\gamma^0\Psi$ and so $N$ is equivalent to $M$.

Now we define the unitary transformation 
$\Lambda \{\Psi\}(p)=\sqrt{\frac{E_p}{\Lambda^0(p)}}\Psi(\Lambda^{-1}(p))$.
Then, we can check that $S\equiv L\Lambda^{-1}$ and it does not depend on $\vec{p}$.
If we redefine $U\{\Psi\}(\vec{p})= \alpha_p\sqrt{\frac{1}{\Lambda^0(p)}} U'\{ \Psi\}(\vec{p})$,
then we get that $\Lambda S \alpha_pU'\{\Psi\}(\vec{p})=\alpha_p\Lambda Q_pU'\{\Psi\}(\vec{p})$ 
and so $U'$ commutes with the Poincare representation.

If we look for subspaces where $m^2=0$ and the representation of $Q_p$ has infinite spin,
then the boost in the $z$ direction for a momenta in the $z$ direction 
multiplies the modulus of the translations of $SE(2)$ by $E_p$, 
which is in contradiction with the fact that
$S\equiv L\Lambda^{-1}$ does not depend on $\vec{p}$.

So, we are left with a direct sum of massive representations and massless with discrete helicity.
\end{proof}

\begin{prop}
For any complex localizable unitary representation of the Poincare group,
compatible with Poincare covariance, it if contains as a subspace a positive energy representation 
then it also contains the corresponding negative energy representation.
\end{prop}
\begin{proof}
The subspaces defined by the projectors involving the $i\gamma^0$s in the $Q_p$ representation are not conserved 
by the system of imprimitivity because $\gamma^0$ does not commute with the matrices $\vec{\gamma}\gamma^0$
present in the transformation from momenta to coordinate space. 
When we go back to coordinate space, the projector on the $i\gamma^0$s can be written as an equality of the time translations which is 
not part of the commuting ring of the SU(2) representation and hence it does not commute with the system of imprimitivity on $R^3$. 
\end{proof}

\begin{cor}
A localizable Poincare representation is an irreducible representation of the Poincare group (including parity) if and only if it is:
a)real and b)massive with spin 1/2 or massless with helicity 1/2.
\end{cor}
\begin{proof}
Since the subspaces defined by the projectors involving the $i\gamma^0$s in the $Q_p$ representation are not conserved 
by the system of imprimitivity, then the condition for irreducibility cannot involve such projectors, which only happens for real representations with
one spinor index. 
\end{proof}

Notice that the condition of irreducibility of the representations admits localized solutions---
the derivative of a bump function is a bump function, so we can find bump functions in the representation space---but it does not admit a position operator---the subspace of bump functions 
is not closed. Hence, we can say that a particular spin 1 state is in an arbitrarily small region of space, but the measurement of the position of an arbitrary spin 1 state might make it no longer a spin 1 state.

Going to complex systems, we can check that in the massive case, 
the condition of irreducibility does not  admit localized solutions---given a localized solution $\Psi$ in a region of space, then the result of the 
application of the projection operator to $\Psi$ is not localized in a region of space.
As for the massless representation, the condition of positive energy
does not admit localized solutions either---for the same region as above---, but the condition for a chiral irreducible representation 
does admit localized solutions. The parity operator for such a chiral irreducible representation is anti-linear.

The localizable Poincare representation is Poincare
covariant because for time $x^0=0$ at point $\vec{x}=0$, we have for the Lorentz
group $L\{\Psi\}(0)=S\Psi(0)$. The localizable Poincare representation is
compatible with causality because the propagator $\Delta(x)=0$ for $x^2<0$
(space-like $x$), where the propagator is defined for spin or helicity $1/2$ as:

\begin{align}
\Delta(x)\equiv \int \frac{d^3\vec{p}}{(2\pi)^32E_p}
\frac{\slashed p\gamma^0+m}{\sqrt{E_p+m}}e^{-i\gamma^0 p\cdot x}
\frac{\slashed p\gamma^0+m}{\sqrt{E_p+m}}
\end{align}
 And verifies:
\begin{align}
\Psi(x)=\int d^3\vec{y} \Delta(x-y)\Psi(y)
\end{align}
To show it we just need to do a Lorentz transformation such
that $x^0=0$ and then show that  $\Delta((0,\vec{x}))=0$ for
$\vec{x}\neq 0$.

\section{Energy Positivity}
\label{section:Energy}

\subsection{Density matrix and real Hilbert space}

As a consequence of Schur's lemma---related with the Frobenious theorem---, the set of normal operators commuting
with an irreducible real unitary 
representation of a Lie group is isomorphic to the reals, to the complex numbers or to the quaternions---the 
irreduciblity of a group representation on a Hilbert space is intuitively the minimization of the degrees of
freedom of the Hilbert space.
This fact turns the 
study of the Hilbert spaces over the reals, the complex or the quaternions interesting for Quantum Theory.
However, once we consider the density matrix in Quantum Mechanics, it is a simple exercise to show that the 
complex and quaternion Hilbert spaces are special cases of the real Hilbert space. 

In short, the complex Hilbert space case is achieved once we postulate that there is a 
unitary operator $J$, with $J^2=-1$, which commutes with the density matrix and all the observables.
The quaternionic Hilbert space corresponds to the case where both the unitary operators $J$ and $K$ commute with
the density matrix and all the observables, with $J^2=K^2=-1$ and $JK=-KJ$. Note that a complex 
Hilbert space is an Hilbert space over a division algebra over the real numbers, hence it has an extra layer of 
mathematical structure, which is dispensable because of the already existing density matrix in Quantum Mechanics. 

Of course, if the postulate corresponding to the complex Hilbert space is correct, 
there are practical advantages in using the complex notation. However, we should be aware that using the complex 
notation is a practical choice, not one of fundamental nature in the formalism of Quantum Mechanics. 
We cannot claim that the fact that the operator 
$J$ exists is a deductible consequence of the formalism of Quantum Mechanics with a complex Hilbert space.
It would be the same as claiming that we can derive from Newton's formalism that the
space is 3 dimensional, instead of assuming that we use 3 dimensional vectors in Newton 
mechanics because we postulate that the space has 3 dimensions.

Choosing real representations is, in practice, choosing real Majorana spinors instead of
complex scalars as the basic elements of relativistic Quantum Theory.
For instance, the state of a spin-0 elementary system is a tensor 
field of real Majorana spinors, which only in momenta space (not in coordinate space) can
be considered a complex scalar field. Note that we are assuming the position operator given by systems of imprimitivity which is suitable for unitary representations; the results are not valid for symplectic representations,
usually associated with the bosons.
 
\subsection{Many particles}
In classical mechanics, the energy of a free body of mass $m$ is $E_p=\frac{\vec{p}^2}{2m}$. 
Since it is proportional to the square of the momentum, it does not make sense to talk about a negative energy. 
However, if we consider a box in which we can insert and remove free bodies such that in both the initial and final 
states the box is empty, the insertion of a body with momentum $\vec{p}$ and negative energy 
$E_p=-\frac{\vec{p}^2}{2m}$ to the system is equivalent to the removal of a body with momentum $-\vec{p}$ positive energy 
$E_p=\frac{\vec{p}^2}{2m}$, because the equations of motion are invariant under time reversal.
But time reversal transforms the act of adding a body on the act of removing a body. 
  
So, how can we say that a body was added to the system and not that the movie of the removal of a body is playing backwards?
The solution is to identify a feature on the system that is also affected by time reversal and we use it as a reference. 
For instance, if there is one body that---we know, or we define it as if---it was added to the system, 
then the addition of that body will appear a removal if we are watching the movie backwards. 
The product of the energies of two bodies is invariant under the Galilean transformations.
Note that we can only remove a body which was previously added to the box, 
as well as only add a body which will later be removed, to keep the box empty in both the initial and final states. 

Hence, the value of any quantity  which is non-invariant under the space-time symmetries---including the sign of the Energy---
by itself does not mean much without something to compare to, such that we can compute an invariant quantity.

In non-relativistic Quantum Mechanics, 
the translations in time
are given by the operator $e^{i\frac{\vec{\partial}^2}{2m}t}$---where $t$ is time---acting on a Hilbert space of \emph{positive} energy solutions 
because there is the imaginary unit---which is invariant under Lorentz transformations and anti-commutes with the time reversal transformations---
that we use as our reference.

In relativistic Quantum Mechanics, the translations in time are given by the operator
 $e^{(\gamma^0\vec{\gamma}\cdot \vec{\partial}+i\gamma^0 m)t}$, 
which is real---in the Majorana basis---and the position operator does \emph{not} leave invariant a Hilbert space of \emph{positive} 
Energy solutions. 
In other words, if we want a coordinate space which is relativistic covariant, the imaginary unit cannot be used as our reference for 
the sign of the energy. We cannot say that by considering real Hilbert spaces we are creating a new problem about Energy positivity.
as if we insist on a covariant coordinate space, the problem about the Energy positivity does not vanish in complex Hilbert spaces.
Remember that ever since the Dirac sea (which led to the prediction of the positron) the problem about Energy positivity 
was always solved in a many particle description.

In a system of particles, 
we can compare the energy of one particle with the energy of another particle we know it is positive,
like we would do in classical mechanics. 
If our reference particle is massive and has momentum $q$, then the Poincare
invariant condition $p\cdot q>0$ will be respected by a massive or massless particle 
with momentum $p$ if and only if $p^0$ has the same sign as $q^0$. 
Instead of the momenta we can use the translations generators to define the condition for energy positivity.

\clearpage{}
\cleartooddpage
\clearpage{}
\chapter{Localization and Gauge symmetries in Quantum Field Theory}
\begin{epigraphs}
\qitem{
In our opinion, a careful analysis of the non-locality of the charged states and of
the general properties of the different quantizations is crucial for a mathematical and
non-perturbative understanding of the important physical phenomena predicted by
gauge quantum field theories.\emph{[...]}

We start by discussing the abelian case, where the Gauss law reduces to the
Maxwell equation $j_0(x) = div \mathbf{E}(x)$, and therefore, by the Gauss theorem, establishes
a tight link between the local properties of the solutions and their behavior at infinity.
In fact, at the classical level the charge of a solution of the electrodynamics equations
can be computed either by integrating the charge density, i.e., a local function of the
charge carrying fields, or by computing the flux of the electric field at space infinity.
In the quantum case, this implies that the charge carrying fields cannot be local with
respect to the (local) electric field.}{
--- \textup{F. Strocchi (2013)\cite{nonperturbativefoundations}}}

\qitem{Today, all the components of the ``standard model'' of particle
physics that so accurately describes our observations are gauge
theories.  Weyl's ``gauge principle'', that global symmetries should be
promoted to local ones, applied to the standard-model symmetry group
$SU(3)\times SU(2)\times U(1)$, is enough to yield the strong, weak and
electromagnetic interactions.

Only gravity is missing from this model. But it too shows many of the
same features. Going from special to general relativity involves
replacing the rigid symmetries of the Poincar\'e group---translations
and Lorentz transformations---by freer, spacetime dependent
symmetries. So it was natural to ask whether gravity too could not be
described as a gauge theory.  Is it possible that starting from a
theory with rigid symmetries and applying the gauge principle, we can
recover the gravitational field?  The answer turned out to be yes,
though in a subtly different way and with an intriguing twist.
Starting from special relativity and applying the gauge principle to
its Poincar\'e-group symmetries leads most directly not precisely to
Einstein's general relativity, but to a variant, originally proposed
by \'Elie Cartan, which instead of a pure Riemannian spacetime uses a
spacetime with torsion.  In general relativity, curvature is sourced
by energy and momentum.  In the Poincar\'e gauge theory, in its
basic version, there is also torsion, sourced by spin.
}{--- \textup{Tom Kibble(2012)\cite{kibble}}}

\qitem{When a single photon strikes
a photomultiplier tube and generates a
pulse of photocurrent, the photon is lost
forever. Or is it? The photocurrent may
interact with a macroscopic system of a
bulk conductor, and we may measure
the resulting voltage across the conductor. 
Sure, the photon has disappeared,
but if our detector indicates that we had
one photon, we can always create another and get the same answer again
and again, exactly like a QND\emph{[quantum non-demolition]} measurement.}{\textup{--- Christopher Monroe (2011) \cite{demolition}}}
\end{epigraphs}

Quantum Mechanics is a theoretical framework, useful to build theories of physical phenomena.
It is not by itself a theory of physical phenomena.
For instance, the Bohr radius is given by $a_0\equiv \frac{1}{m_e \alpha}$,
where $m_e$ is the electron mass and $\alpha$ is the fine structure constant, that is the electromagnetic coupling constant.
Therefore, it is possible to build a Quantum model where macroscopic Hydrogen-like atoms exist, 
we would just need to make the coupling constant sufficiently small.

Another well-known example is the Schrodinger's cat, where the assumption that there is a unitary interaction between an unstable nucleus and a 
macroscopic measuring device capable of creating a macroscopic superposition state from the nucleus superposition state, leads to the conclusion that Quantum Mechanics allows for macroscopic superposition states\cite{consistent, nobelcat}. However, our physics models do not predict such macroscopic superposition states because the concrete physical interactions---once decoherence is taken into account\cite{decoherence}--- do not allow to reproduce with the present technology the assumed unitary interaction\cite{superpositiontest}

To study gauge theories at the perturbative quantum level we need to drop basic assumptions such as positivity of the inner product computing the vacuum expectation values of the local operators. Only after all calculations, the physical observables verify our basic assumptions, but the framework itself does not and often only the particular properties of the Lagrangian will prevent inconsistent predictions---e.g. the quantum anomalies.

In other words, assertions about Quantum Mechanics or Quantum Field Theory are not necessarily assertions with physical content.

The assumption of locality at classical field theory level is crucial in the formulation of gauge quantum field theory and of the gauge principle itself.
Yet the non-localizability of free states in relativistic quantum mechanics\cite{modular0}, the non-localizability of the free states created by applying local operators to the vacuum in quantum field theory\cite{localizedparticles}
or the non-localizability of charged states in gauge quantum field theory\cite{nonperturbativefoundations},
adding to the never ending controversies surrounding quantum measurement\cite{qmnostates,*automaton}, contributed for the 
replacement of the study of the notion of position at the quantum level with mystification. 

As an example, in the most ambitious modern mathematical treatments of quantization, the free Dirac equation---which specifies the coordinate space of a spin 1/2 Poincare representation---is a postulated classical equation that is introduced in the quantum world after quantization\cite{mathematicsquantization}, while it can be derived from the requirement of a (covariant position related) projection-valued measure---an intrinsic quantum mechanical operator---as we have seen in the last chapter.

This mystification can be set apart once we consider quantum field theory as a framework, a set of mathematical and conceptual tools which we can use to define and make calculations from physics models.
In this chapter we will study some of these tools, mostly related with the notion of position and gauge symmetries.

\section{Localization in Quantum Field Theory}
\subsection{Vacuum density matrix}

The link between one-particle states and many particle states is 
not unique as there are uncountably many inequivalent representations of the 
Canonical Commutation Relations and Canonical Anti-commutation Relations. 
A complete classification of the
representations of the CCR and CAR relations is not
expected in the near future. Different
dynamics require inequivalent representations of the Canonical Relations,
this is related with renormalizability, entropy, phase transitions
\cite{mathematicsquantization,notequantization}.

Still, we can define a map between complex and real representations of C* algebras\cite{realoperatoralgebras}.

\begin{lem}[Schur's lemma for C* algebras]
\label{lem:commuting}
Consider an * representation $(M,V)$ of a C* algebra $A$ 
on a complex Hilbert space $V$. If the representation
$(M,V)$ is irreducible then any normal operator $N$ of $(M,V)$ is a
scalar.
\end{lem}

So we can define a map from the real to the complex representations of C* algebras---analogous to the 
one for unitary representations.
Of course, such map is not very interesting as the representations are necessarily R-complex and C-complex 
due to the fact that the C* algebra is complex.

Therefore, the interesting case is to study the representations of a real C*algebra \cite{realoperatoralgebras}, defined as a
real Banach algebra whose complexification is a C*algebra.
Using Prop. 5.3.7 of \cite{realoperatoralgebras}:

\begin{lem}[Schur's lemma for real C* algebras]
\label{lem:commuting}
Consider a * representation $(M,V)$ of a real C* algebra $A$ 
on a real Hilbert space $V$. 
Then $(M,V)$ is irreducible iff the commutant is isomorphic
to the reals, complex numbers or quaternions.
\end{lem}

Therefore, there is a similar 
map from the real to the complex representations of a real C* algebra.
Note that we can always embed a complex C* algebra in a real C* algebra.

Then for the real C* algebras, the \ac{GNS} theorem\cite{formalisms} is also valid, that is given a positive functional, there is always
a real representation with a distinguished cyclic state (usually associated with the vacuum in the complex case).
Then Prop. 5.3.7\cite{realoperatoralgebras}, the representations induced by a functional are irreducible iff 
the functional is a pure state.

Given a real Hilbert space $V$ with inner product $<,>$ we can always construct an associated real Clifford C* algebra \cite{spinorsrealhilbert}. Let $C(V)$ be the associated complex Clifford algebra, i.e. $C(V)$ is a unital
associative complex algebra such that there is an injective linear map $a:V\to C(V)$, verifying 
$a^2(v)=<v,v>1$ and $a^*(v)=a(v)$, $C(V)$ admits a unique involution $*$ and it is generated by the operators $a(v)$, for all $v\in V$.
The algebra has a natural norm given by $||a(v_1)...a(v_n)||\equiv \sqrt{||a(v_n)...a(v_1)a(v_1)...a(v_n)||}=||v_1||...||v_n||\cdot||1||$. The C* algebra $C[V]$ is the completion $C(V)$ with respect to its natural norm.

The subspace $R[V]$ of self-conjugate elements of $C[V]$, is a real Banach algebra whose complexification is a complex C* algebra,
hence $R[V]$ is a real C* algebra.

There is a natural functional of the Clifford C* algebra sending $1$ to $1$ and the remaining operators to $0$.
So there is a real representation with cyclic state, say $\xi$.

This cyclic state does not have the properties of a vacuum state because it is not Poincare invariant.

The vacuum energy is sometimes associated with the Casimir effect, however 
casimir forces can be calculated without reference to the vacuum and vanish as the coupling constant tends to zero\cite{casimir}.
The vacuum density matrix should be gauge invariant. These are the necessary properties of the vacuum,
because there is no way in which we can change it or interact with it---we are assuming no gravity for now.

Then to build the vacuum density matrix, we start with the projector $\xi\xi^\dagger$.
Suppose that the real Hilbert space only has two states
and the corresponding Clifford operators are $a$ and $b$ with $a^2=b^2=1$ and $ab=-ba$. Suppose that 
the $U(1)_Y$ gauge transformation is given by $e^{ab \theta}$. Then the density matrix
 $\frac{1}{2}(\xi\xi^\dagger+ab\xi\xi^\dagger ab)$ will be invariant under the $U(1)_Y$ gauge transformation. 
By an iterative process we can build in this way a vacuum density matrix which is gauge invariant.

Note that the fact that we are working with a real Clifford algebra is essential. The usual complex vacuum for the same Hilbert space would read $(1-iab)\xi$, it is this projector (in momentum space) that causes all the localization troubles (in coordinate space), namely when we act on the vacuum constructed in this way with a local operator we do not get a local state. As we showed in the last chapter such projectors are the root of the localization troubles.

Of course that we can remove such projectors and still work with complex Clifford algebras at the possible cost of irreducibility, but it will be the same as working with real Clifford algebras (self-conjugate representations are isomorphic to real representations).

But we can also introduce these projectors latter in the development of the theory in the case we need them for physical reasons.
The message  is that the most basic physical requirements for a vacuum density matrix can be fulfilled without spoiling localization.
Note that we can convert any density matrix to a pure state at the cost of irreducibility. With the density matrix we can have irreducible representations of the real Clifford C* algebra.

\subsection{Symplectic representations}

The complex representations give us two products: the real part is the inner 
product, the imaginary part is a symplectic product.
In the last chapter, to have good localization properties we dropped the imaginary part,
which will lead us sooner or later to fermions.
However, we can also have good localization properties by dropping the real part, 
which will lead us to bosons---the canonical commutation relations are conserved by
symplectic transformations.
That is, to have good localization properties we cannot keep both the real and imaginary parts,
we need to choose one of them. The easiest to study is the real part because it is an
inner product. But the theory of symplectic representations is also well developed\cite{infinitegroups,infinitesympletic,mathematicsquantization}.

In the following we will study the transformation from the momentum space to the coordinate
space which conserves a symplectic product.

The Klein-Gordon equation for a scalar field is:
\begin{align*}
(\partial^2-m^2)\Phi=0
\end{align*}

Due to the fact that the equation for the field $\Phi$ is second order, 
the first derivative in time of the field is a variable. Let the operator $\star$ define a symplectic product:
\begin{align}
\int d^3\vec{x} (f\star \varphi)(x^0_i,\vec{x})&\equiv\int d^3\vec{x}(f(\partial_0\varphi)-(\partial_0
f)\varphi)(x^0_i,\vec{x})
\end{align}

We define the 2D vector $f'(x)=\left[ \begin{smallmatrix}
f(x)\\
\partial_0f(x)\end{smallmatrix} \right]$. Then the symplectic
transform is given by:
\begin{align}
f'(x^0,\vec{p})=\int d^3\vec{x} U(\vec{p},x)f'(x)
\end{align}
Where $U$ is the 2D matrix:
\begin{align}
U(\vec{p},x)&=\left[ \begin{smallmatrix}
cos(p \cdot x)& \frac{sin(p \cdot x)}{E_p}\\
-E_p sin(p \cdot x)& cos(p \cdot x)\end{smallmatrix} \right]
\end{align}
Where $p^0=E_p$. The inverse symplectic transform is given by:
\begin{align}
f'(x)=\int \frac{d^3\vec{p}}{(2\pi)^3} U^\star(\vec{p},x)f'(x^0,\vec{p})
\end{align}
Where $U^\star$ is the matrix:
\begin{align}
U^\star(\vec{p},x)&=\left[ \begin{smallmatrix}
cos(p \cdot x)& -\frac{sin(p \cdot x)}{E_p}\\
E_p sin(p \cdot x)& cos(p \cdot x)\end{smallmatrix} \right]
\end{align}
The proof follows. We consider $y^0=x^0$:
\begin{align}
\int \frac{d^3\vec{p}}{(2\pi)^3} U^\star(\vec{p},x)U(\vec{p},y)&=\int \frac{d^3\vec{p}}{(2\pi)^3}\left[ \begin{smallmatrix}
cos(\vec{p} \cdot (\vec{x}-\vec{y}))& -\frac{sin(\vec{p} \cdot (\vec{x}-\vec{y}))}{E_p}\\
E_p sin(\vec{p} \cdot (\vec{x}-\vec{y})) & cos(\vec{p} \cdot
(\vec{x}-\vec{y}))\end{smallmatrix} \right]\\
&=\left[ \begin{smallmatrix}
\delta^3(\vec{x}-\vec{y})& 0\\
0 & \delta^3(\vec{x}-\vec{y})\end{smallmatrix} \right]
\end{align}
\begin{align}
&\int d^3\vec{x} U(\vec{p},x)U^\star(\vec{q},x)=\\
&=\int d^3\vec{x}\left[ \begin{smallmatrix}
cos(p \cdot x)cos(q \cdot x)+\frac{E_q}{E_p}sin(p \cdot x)sin(q \cdot
x)& -\frac{cos(p\cdot x)sin(q \cdot x)}{E_q}+\frac{sin(p\cdot x)cos(q \cdot x)}{E_p}\\
-E_psin(p\cdot x)cos(q \cdot x)+E_qcos(p\cdot x)sin(q \cdot x) & cos(p \cdot x)cos(q \cdot x)+\frac{E_p}{E_q}sin(p \cdot x)sin(q \cdot
x)\end{smallmatrix} \right]\\
&=\int d^3\vec{x}\left[ \begin{smallmatrix}
cos((p-q)\cdot x)\frac{E_p+E_q}{2E_p}+cos((p+q)\cdot
x)\frac{E_p-E_q}{2E_p}& sin((p-q)\cdot
x)\frac{E_p+E_q}{2E_pE_q}+sin((p+q)\cdot x)\frac{E_p-E_q}{2E_pE_q}\\
-sin((p-q)\cdot
x)\frac{E_p+E_q}{2}-sin((p+q)\cdot x)\frac{E_p-E_q}{2}
& cos((p-q)\cdot x)\frac{E_q+E_p}{2E_q}+cos((p+q)\cdot
x)\frac{E_q-E_p}{2E_q}\end{smallmatrix} \right]\\
&=\left[ \begin{smallmatrix}
(2\pi)^3\delta^3(\vec{p}-\vec{q})& 0\\
0 & (2\pi)^3\delta^3(\vec{p}-\vec{q})\end{smallmatrix} \right]
\end{align}
Where the fact that $E_p=E_{-p}$ was used.

\section{Poincare gauge theory}

It is well known that to introduce spinors in General Relativity we need to introduce tetrads,
which verify a gauge symmetry associated with the homogeneous Lorentz group.

The crucial contribution of Poincare (translations and Lorentz transformations) gauge theory 
is that the (Cartan's) tetrads  are also gauge fields, which leads to a conceptually better equipped theory
\cite{gravitypoincare,etg,gaugegravityintro}.

Within this framework, a lot can be done. 
\subsection{Unitary representations of the Poincare group in classical field theory}
One idea is to use the De Donder-Weyl 
polymomentum\cite{polymomentum,covariantphase}, combined with fields with non trivial representations of the translations\cite{translations}. Instead of the Dirac Lagrangian:

\begin{align*}
\Psi^\dagger (\gamma^0\gamma^\mu(\partial_\mu-ieA_\mu)-i\gamma^0m)\Psi
\end{align*}

We can consider instead the more involved Lagrangian:

\begin{align*}
\mathcal{L}\equiv\Phi^{\dagger\mu}(\partial_\mu-ieA_\mu-B_\mu)\Psi
\end{align*}
Where the local operator $B_\mu$ does not contain derivatives in space-time, but acts on the infinite-dimensional space of the components of $\Psi$.
We can show using the De Donder-Weyl formalism that:
\begin{align*}p^\mu&\equiv \frac{\delta\mathcal{L}}{\delta \partial_\mu\Psi}=\Phi^{\dagger\mu}\\
\mathcal{H}&\equiv p^{\mu}\partial_\mu\Psi-\mathcal{L}=p^{\mu}(ieA_\mu+B_\mu)\Psi
\end{align*}

We get the equations:
\begin{align*}
\partial_\mu\Psi&=\frac{\delta\mathcal{H}}{\delta p^\mu}=(ieA_\mu+B_\mu)\Psi\\
\partial_\mu p^\mu&=-\frac{\delta\mathcal{H}}{\delta \Psi}=-p^{\mu}(ieA_\mu+B_\mu)
\end{align*}

Note that the second equation is redundant, because if the first equation is verified then
for $p^\mu=\Psi^\dagger\gamma^0\gamma^\mu$ the second equation is also verified ($B_j$ commutes with $\gamma^0\gamma^j$, for $j=1,2,3$).

We recover the Dirac equation for $B_0=\gamma^0\gamma^jB_j-i\gamma^0m$. The generators of translations of a free
Poincare representation with spin one-half verify such equation. All this has the advantage that the equation:
\begin{align*}
(\partial_0-ieA_0-B_0)\Psi=0
\end{align*}
does not depend on the derivatives in space and so we can work with unitary representations of the Poincare group
in the equations of the gauge theories (of the Standard Model for instance),
instead of working with non-unitary representations of the Lorentz group.
The advantage of all this is that the unitary representations of the Poincare group 
are already used in the non-perturbative regime, before we assume the perturbative expansion
which may help in the non-perturbative definition of a gauge quantum field theory.

\subsection{Exploring the spin connection of the Majorana spinor}

In a Majorana basis, the Dirac equation for a free spin one-half
particle is a 4x4 real matrix differential equation.
When including the effects of the electromagnetic interaction, the
Dirac equation is a complex equation due to the presence of 
an imaginary connection in the covariant derivative, related with the
phase of the spinor.

In this subsection we study the solutions of the Dirac equation
with the null and Coulomb potentials and notice that there is a real
matrix that squares to -1, relating the imaginary and real components
of these solutions. We show that these solutions can be obtained from the
solutions of two non-linear 4x4 real matrix differential equations
with a real matrix as the connection of the covariant derivative. 

\paragraph{Real Connection}
The equations for the classical Majorana spinor fields $\psi$ and
$\chi$ and for the electromagnetic potential $A_\mu$  in
Quantum Electrodynamics, can be written as:
\begin{align}
(i\slashed \partial-m)\psi&=e i\slashed A \chi\\
(i\slashed \partial -m)\chi&=-e i\slashed A \psi\\
\partial^2 A_\mu-\partial_\mu \partial\cdot A&=
e\eta_{\mu\nu}(\psi^\dagger \gamma^0\gamma^\nu\psi+\chi^\dagger \gamma^0\gamma^\nu\chi)
\end{align}
These equations are invariant under the global Lorentz
transformations $S\in Pin(1,3)$:
\begin{align}
x &\to \Lambda(S)x\\
\psi(x)&\to S\psi(\Lambda(S)x)\\
\chi(x)&\to \gamma^0 S^{-1 \dagger}\gamma^0 \chi(\Lambda(S) x)\\
A_\mu(x)&\to e \Lambda_{\mu}^{\ \nu}(S) A_\nu(x)(\gamma^0S^{-1
  \dagger}\gamma^0 S^{-1})
\end{align}
Usually the Dirac field $\Psi\equiv \psi+i\chi$ is defined and the
equations are written as:
\begin{align}
(i\slashed \partial-\slashed A-m)\Psi&=0\\
\partial^2 A_\mu-\partial_\mu \partial\cdot A&=
e\eta_{\mu\nu} \Psi^\dagger \gamma^0\gamma^\nu\Psi
\end{align}
Now we can easily see that these equations are also invariant under
the local transformation:
\begin{align}
\Psi&\to e^{i\theta}\Psi\\
eA_\mu&\to eA_\mu-\partial_\mu \theta
\end{align}
The electromagnetic potential is then identified with an imaginary
connection, that is, the covariant derivative is written as:
\begin{align}
\partial_\mu+iA_\mu
\end{align}
Now we make the question: is there another way of obtaining the same
solutions but using a real (that is, real in a Majorana basis)
connection? If we drop the linearity requirement, then the answer is
yes. 
We need to assume that there is a real, space-time dependent, matrix
$J$ verifying $((i\gamma^0)J)^2=-1$. Note that these conditions are
invariant under the transform 
$J\to
S^\dagger J S$ for $S\in Pin(1,3)$, that is:
\begin{align}
&((i\gamma^0)J)^2\to (i\gamma^0S^\dagger J S)^2=(\pm
S^{-1}i\gamma^0 J S)^2=S^{-1}(\pm i\gamma^0 J)^2 S=-1
\end{align}
Now we have the following equations:
\begin{align}
(i\gamma^\mu (\partial_\mu -e A_\mu(x) i\gamma^0 J(x)-m)\psi(x)&=0\\
(i\gamma^\mu (\partial_\mu -e A_\mu(x) i\gamma^0 J(x))-m)i\gamma^0 J(x)\psi(x)&=0
\end{align}
The equation for $A_\mu$ can be written as:
\begin{align}
\partial^2 A^\mu-\partial^\mu\partial_\nu A^\nu&=
e\psi^\dagger\gamma^0\gamma^\mu\psi+e\psi^\dagger J^\dagger
\gamma^\mu\gamma^0 J\psi
\end{align}
We can see that for a global $S\in Pin(1,3)$ we have:
\begin{align}
x &\to \Lambda(S)x\\
\psi(x)&\to S\psi(\Lambda(S)x)\\
J(x)&\to S^{-1 \dagger} J(\Lambda x) S^{-1}\\
e A_\mu(x) i\gamma^0 J(x)&\to e \Lambda_{\mu}^{\ \nu}(S) A_\nu(x)
S i\gamma^0 J(\Lambda x)S^{-1}
\end{align}
We can write the previous two real equations as one complex equation
as:
\begin{align}
(i\gamma^\mu (\partial_\mu -e A_\mu(x) i-m)(1+\gamma^0 J(x))\psi(x)&=0
\end{align}
Now we can see that there is another transform that leaves the
equations invariant:
\begin{align}
\psi&\to e^{i\gamma^0 J\theta}\psi\\
(1+\gamma^0J)\psi&\to e^{i\theta}(1+\gamma^0J)\psi\\
eA_\mu&\to eA_\mu+\partial_\mu \theta
\end{align}
Where $\theta$ is a real function of the space-time.
Although we get a very similar equation with QED, there is a fundamental
difference: the connection is real, the equations are non-linear and
as a consequence we get, 
from the start a projector in the complex equation. In QED, this projector appears only
in the final solutions, not in the equations.

\paragraph{Free particle}
When the electromagnetic potential is null, we have:
\begin{align}
\psi_p(x)&=e^{-i\frac{\slashed
    p}{m} p\cdot x} \psi_p(0)\\
J_p(x)&=\frac{\slashed
    p \gamma^0}{m}
\end{align}
We can check that $J_p(x)$ is hermitian and that
$\psi_p(x)\to S\psi_p(\Lambda(S)x)$, $J_p(x)\to
S^{-1\dagger}J_p(\Lambda(S)x)S^{-1\dagger}$

\paragraph{Hydrogen Atom}
The Dirac equation for the Hydrogen atom is:
\begin{align}
i\gamma^0(i\slashed \partial-e\slashed A -m)\Psi=0
\end{align}
With $A_i=0$, $A_0=-\frac{e}{r}$. The term with the potential is
imaginary, therefore, the equation is complex. 

We define the matrix:
\begin{align}
\Lambda_{n l m \epsilon}&=\Big(\frac{f_{n l \epsilon}(r)}{r}+\frac{g_{n
    l \epsilon}(r)}{r}i\gamma^r\Big)\Omega_{l m}\frac{1+\epsilon \sigma^3}{2}
\end{align}
Where $\epsilon=\pm 1$. If $f$ and $g$ are such that the following equations hold:
\begin{align}
(E_{n l}+\frac{e^2}{r}-m)\frac{f_{n
    l \epsilon}(r)}{r}+(\partial_r+\frac{1-\epsilon l}{r})\frac{g_{n l
  \epsilon}(r)}{r}=0\\
(-E_{n l}-\frac{e^2}{r}-m)\frac{g_{n
    l \epsilon}(r)}{r}+(\partial_r+\frac{1+\epsilon l}{r})\frac{f_{n l
  \epsilon}(r)}{r}=0
\end{align}
We will not solve these equations here, the solution can be
seen in \cite{qm}.

Then $\Lambda$ verifies:
\begin{align}
i\gamma^0(i\vec{\slashed \partial}-m)\Lambda_{n l
  m\epsilon}\frac{1+\gamma^0}{2}=i(E_{n
l}+\frac{e^2}{r})\Lambda_{n l
  m\epsilon}\frac{1+\gamma^0}{2}
\end{align}
The solution to Dirac equation is:
\begin{align}
\Psi=\Lambda_{n l
  m\epsilon}e^{-i\gamma^0 E_{n l}x^0}\frac{1+\gamma^0}{2}\psi
\end{align}
Where $\psi$ is a fixed Majorana spinor.
We can now check that

\begin{align}
\Psi=\frac{1+\gamma^0J(x)}{2}\Lambda_{n l
  m\epsilon}e^{-i\gamma^0 E_{n l}x^0}\psi
\end{align}
Where 
\begin{align}
J(x)=\frac{\Big(f_{n l \epsilon}(r)-g_{n
    l \epsilon}(r) i\gamma^r\Big)^2}{f^2_{n l \epsilon}(r)-g^2_{n l \epsilon}(r)}
\end{align}
And we can check that $(i\gamma^0J)^2=-1$.

\section{Quantum measurement with Quantum Field Theory}

When some important contributors to the Standard Model discuss the Quantum measurement\cite{qmnostates,*automaton},
we should ask why not equip the discussion with Quantum Field Theory?

The problem with quantum measurement are non-commuting projections. However, in quantum field theory there is only one fundamental projection: the projection to the vacuum and so no problems arising from non-commutation. All other projections are built from the 
vacuum projection acting with operators of creation and destruction. These operators have a well accepted physical interpretation:
they create and destroy particles. Hence, a measurement can be interpreted as the superposition for many states of a destruction of some particle state followed by a projection to the vacuum followed by the creation of the same particle state\cite{demolition}.

If we accept the interpretation of the creation and destruction of particles, then we just need to explain the projection to the vacuum.
This can be a classical interpretation since there is just this projection hence no problems with non-commuting projections.

Another point to study is the propagators in perturbation theory as there are explicitly causal formulations\cite{inin,*inin2}
and another that may reduce drastically the number of divergences before regularization in perturbation theory\cite{krein}.
These may be useful in studies of the foundations of quantum theory.

\clearpage{}

\cleartooddpage
\clearpage{}\chapter{Conclusion}
\label{chap:Conclusion}
\epigraph{It is hard to prove general theorems, specially when they are false.}
{\textup{--- G. C. Branco, about the idea that led to BGL models, talk at Planck 2013 Bonn}}

Following the discovery of a Higgs boson consistent with the Standard Model,
there are founded claims that all experimental results in Particle Physics and Cosmology can be accounted by an effective theory based on General Relativity and the Standard Model extended with three right handed neutrinos and one inflaton field, and that this effective theory may be valid up to the Planck energy scale where a quantum theory of gravity plays a role.
The challenge we face is not so much to account for unexpected experimental results, but mostly to understand and solve the many theoretical problems of this indeed effective theory.

\subsection*{Higgs mediated Flavour Violation}

We may find solutions to the problems of the Standard Model by extending its scalar sector, 
e.g. in Grand Unified Theories or Supersymmetry. Simple extensions of the Higgs sector are also a step towards a
general understanding of the Higgs mechanism in gauge theories.
We study the two-Higgs-doublet model using Clifford matrices in a gauge invariant approach.  The conclusion is that it is possible to combine
studies based on perturbative and non-perturbative methods to study the phenomenology of extended Higgs sectors,
as we implement the correspondence between the standard gauge-dependent elementary states of the pertubative formalism
and the composite (non-abelian) gauge invariant final states of the non-perturbative formalism. 
Besides the theoretical interest, the results will be used in lattice studies of the
non-perturbative phenomenology of two-Higgs-doublet models.

In extensions of the Standard Model we are many times confronted with the problem of the suppression
of the Flavour Changing Neutral Currents, which in the Standard Model are accidentally suppressed through
the \ac{GIM} mechanism.
This motivates the question, how much does the experimental data constrain the 
Flavour Changing Neutral Currents which would signal New Physics?
Correlations between observables are important to obtain conclusive experimental results.
We discuss two approaches to this problem: renormalizable models and effective field theory.

The flavour data indicates that the  Flavour and \acs{CP} violation in Particle Physics 
follows a hierarchical pattern, well accounted  by the Standard Model's 
mixing matrices \ac{CKM} and \ac{PMNS} of the fermions and inconsistent in general with
generic extensions to the Standard Model.  
We define the Minimal Flavour Violation condition with six spurions in effective field theories,
which allows for Flavour and CP violation entirely dependent on the Standard Model's 
mixing matrices \ac{CKM} and \ac{PMNS} of the fermions,
but independent from the hierarchy of the fermion masses. Note that we can guess that the hierarchies
of fermion masses and mixings have a common origin but we do not know what is their precise relation.
We show that the Minimal Flavour Violation condition with six spurions is 
one-loop renormalization-group invariant in the two-Higgs-doublet model; we argue that the
condition must be renormalization-group invariant in general, unless there are quantum anomalies---a question left open
that needs to be addressed with more involved mathematical tools.

To do extensive phenomenological studies we need computational tools,
we describe the architecture that we find best suited for the task. The architecture is based on libraries made by 
different people, with several functions (not necessarily simple) with input/output easy to handle and test, such as formulas for
Wilson coefficients. This can be achieved by the use of libraries with an interface for the C++ language backed by another
C++ library providing the ability to manipulate symbolic expressions, such as the libraries GiNaC and Giac and 
a possible improvement using LLVM. We partly use this architecture in the implementation of a program to do an extensive phenomenological study of two-Higgs-doublet models.

We then analyse the constraints and some of the phenomenological 
implications of a class of renormalizable two-Higgs-doublet models which verify the
Minimal Flavour Violation condition with six spurions, as a result of a continuous $U(1)$ symmetry of the Lagrangian
which constrains the Yukawa couplings to have a special form.
The symmetry is softly broken in the Higgs potential and so there are no massless Goldstone bosons in the spectrum.
The models predict Higgs mediated Flavour Changing Neutral Currents at tree level, naturally suppressed by the \ac{CKM}
matrix elements, with no other flavour and \acs{CP}-violating parameters than the \ac{CKM} and \ac{PMNS} matrix elements.
The symmetry can be implemented in the quark
sector in six different ways, and the same applies to the leptonic sector, leading altogether
to thirty six different realizations of the BGL models. Due to the symmetry the models have few additional free parameters with respect to the Standard Model, but due to the hierarchies of the fermion masses and mixings the phenomenology of different models is diverse;
thus excellent to guide us in the search for New Physics in Flavour Changing Neutral Currents as they predict meaningful and diverse correlations between different observables.

We analyse a large number of processes mostly on flavour physics, 
including decays mediated by charged Higgs at tree level, 
processes involving Flavour Changing Neutral Currents at tree level, as well as loop
induced processes. We study the allowed regions in the parameter space $\tan\beta$, 
and the Higgs masses $m_{H^+}$, $m_R$,  $m_I$ and then we project, for each BGL model,
these regions into subspaces relating pairs of the above parameters.  
Our results clearly show that this class of models
allow for new physical scalars with masses as light as the standard Higgs boson, and so
reachable, for example, at the next round of experiments at the LHC.

For a long time, there was the belief that the only experimentally viable two-Higgs-doublet 
extensions of the Standard Model were those verifying the Natural Flavour Conservation condition.
The condition of Minimal Flavour Violation with six spurions provides an interesting alternative to 
both Natural Flavour Conservation.
We showed that this class of models is an example
that there are renormalizable models extending the Standard Model
with Higgs mediated Flavour Changing Neutral Currents at tree level, without introducing more
hierarchical coefficients than the ones already present in the Standard Model, for which the most
constraining experimental data on flavour physics allows all the Higgs masses to be around the Electroweak scale.
Then this proof of concept can be applied in more elaborated extensions of the 
Standard Model addressing its problems\cite{abelian,burastalk,*burascorrelations,axion} or in \acs{LHC} phenomenology\cite{hidden}.

\subsection*{On the real representations of the Poincare group}

We then towards the relation between real and complex representations in Quantum theories
(in mathematics the complex or real numbers are also called scalar fields).

The complex irreducible representations are not a generalization of the real
irreducible representations, in the same way that the complex numbers are a
generalization of the real numbers. There is a map, one-to-one or
two-to-one and surjective up to equivalence, from the complex to
the real irreducible representations of a Lie group on a Hilbert
space. 

We show that all the finite-dimensional real
representations of the identity component of the Lorentz group are also
representations of the parity, in contrast with many complex representations.

We obtained all the real
unitary irreducible projective representations of the Poincare
group, with discrete spin, as real Bargmann-Wigner fields.
For each pair of complex representations with positive/negative
energy, there is one real representation.
The Majorana-Fourier and Majorana-Hankel unitary
transforms of the Bargmann-Wigner fields relate the
coordinate space with the linear and angular momenta spaces.

The localizable (real or complex) unitary representations of the Poincare group 
(compatible with Poincare covariance and causality) are direct sums of 
irreducible representations with discrete spin and helicity, this result establishes a fundamental
difference between the representations associated to existing elementary systems and the other 
representations for which no existing elementary systems are known to be associated---
it was known that point localized local quantum fields cannot be a massless infinite spin representation
\cite{zeromass}.
Moreover, an irreducible representation of the Poincare group (including parity) is localizable if and only if it is: 
a)real and b)massive with spin 1/2 or massless with helicity 1/2. If a) and b) are verified the position operator 
matches the coordinates of the Dirac equation.

The current literature
\cite{modular0,noncommutative,equations1,povm1,energydensity0,2D,axial,spacetime,pseudohermitian}, in one way or another, 
departure from using projection operators to implement the position operator for a unitary representation of the Poincare group.
The results presented in this thesis are a motivation to not  
departure from the projection operators to describe the position of relativistic systems.

In the last chapter we addressed some questions related with Localization and gauge symmetries in Quantum Field Theory,
which we hope may soon reach the maturity of the remaining problems studied in this thesis.

\clearpage{}\cleartooddpage
\backmatter
\phantomsection
\addcontentsline{toc}{chapter}{\bibname}

\bibliographystyle{utphysMM}
\label{app:bibliography} 
\small
\singlespacing 
\bibliography{Poincare}
\normalsize
\onehalfspacing 
\end{document}